\documentclass[journal]{IEEEtran}
\usepackage{setspace}
\usepackage{cite,graphicx,amsmath,amssymb}
\usepackage{hyperref}
\usepackage{subfigure}
\usepackage{url}
\usepackage{fancyhdr}
\usepackage{balance}
\usepackage{mdwmath}
\usepackage{mdwtab}
\usepackage{balance}
\usepackage{xcolor}
\usepackage{bm}
\usepackage{amsthm}
\usepackage{threeparttable}
\usepackage{abstract}
\usepackage{algorithm}
\usepackage{algorithmic}
\usepackage{multirow}
\usepackage{flafter}
\usepackage{makecell}
\usepackage{diagbox}
\usepackage{tcolorbox}

\theoremstyle{definition}
\newtheorem{remark}{Remark}
\newtheorem{theorem}{Theorem}

\newtheorem{lemma}{Lemma}

\newtheorem{corollary}{Corollary}

\usepackage{xurl}

\providecommand{\url}[1]{#1}
\allowdisplaybreaks
\hypersetup{colorlinks=true,
linkcolor=blue,
citecolor=blue,      
urlcolor=black,
    }

\begin{document}

\title{Near-Field Communications: A Tutorial Review}

\author{Yuanwei Liu, Zhaolin Wang, Jiaqi Xu, Chongjun Ouyang,  Xidong Mu, and Robert Schober\\
\vspace{0.4cm}
\emph{(Invited Paper)
\vspace{-0.5cm}
}

\thanks{Yuanwei Liu, Zhaolin Wang, Jiaqi Xu, and Xidong Mu are with the School of Electronic Engineering and Computer Science, Queen Mary University of London, London E1 4NS, UK, (email: {yuanwei.liu, zhaolin.wang, jiaqi.xu, xidong.mu}@qmul.ac.uk).}
\thanks{Chongjun Ouyang is with the School of Information and Communication Engineering, Beijing University of Posts and Telecommunications, Beijing, 100876, China (e-mail: dragonaim@bupt.edu.cn).}
\thanks{Robert Schober is with the Institute for Digital Communications, Friedrich-Alexander-Universität Erlangen-Nürnberg, 91054 Erlangen, Germany (e-mail: robert.schober@fau.de).}

}

\twocolumn[
\begin{@twocolumnfalse}

\maketitle

{
\small
\hspace{10pt}\textbf{\emph{Abstract}---Extremely large-scale antenna arrays, tremendously high frequencies, and new types of antennas are three clear trends in multi-antenna technology for supporting the sixth-generation (6G) networks. To properly account for the new characteristics introduced by these three trends in communication system design, the near-field spherical-wave propagation model needs to be used, which differs from the classical far-field planar-wave one. As such, near-field communication (NFC) will become essential in 6G networks. In this tutorial, we cover three key aspects of NFC. 1) \textit{Channel Modelling}: We commence by reviewing near-field spherical-wave-based channel models for spatially-discrete (SPD) antennas. Then, uniform spherical wave (USW) and non-uniform spherical wave (NUSW) models are discussed. Subsequently, we introduce a general near-field channel model for SPD antennas and a Green's function-based channel model for continuous-aperture (CAP) antennas. 2) \textit{Beamfocusing and Antenna Architectures}: We highlight the properties of near-field beamfocusing and discuss NFC antenna architectures for both SPD and CAP antennas. Moreover, the basic principles of near-field beam training are introduced. 3) \textit{Performance Analysis}: Finally, we provide a comprehensive performance analysis framework for NFC. For near-field line-of-sight channels, the received signal-to-noise ratio and power-scaling law are derived. For statistical near-field multipath channels, a general analytical framework is proposed, based on which analytical expressions for the outage probability, ergodic channel capacity, and ergodic mutual information are obtained. Finally, for each aspect, topics for future research are discussed.}

\vspace{0.2cm}

\hspace{10pt}\textbf{\emph{Index Terms}---Antenna architecture, beamforcusing, channel modelling, near-field communications, performance analysis.}
}

\vspace{1cm}
\end{@twocolumnfalse}

]
\saythanks

\section{Introduction}

As the fifth-generation (5G) wireless network continues to be commercialized, research and development efforts are already underway on a global scale to investigate the possibilities for the sixth-generation (6G) wireless network. Compared to the previous generations of wireless networks, 6G is expected to be data-driven, instantaneous, ubiquitous, and intelligent, thereby facilitating new applications and services, such as extended reality (XR), holographic communication, pervasive intelligence, digital twin, and Metaverse \cite{Ericsson, Huawei, Qualcomm}.  Therefore, 6G has significantly higher performance targets compared to the past generations, such as a 100 times higher peak data rate of at least $1$ terabit per second (Tb/s), an air interface latency of $0.01\sim0.1$  milliseconds, and a 10 times larger connectivity density of up to $10^7$ devices per square kilometer \cite{saad2019vision, zhang20196g, dang2020should}. To reach these stringent targets, 6G is expected to integrate the following technical trends:

\begin{itemize}
    \item \textbf{Extremely large-scale antenna arrays}: Extremely large-scale antenna arrays (ELAAs) are essential for many candidate techniques for 6G. On the one hand, by exploiting ELAAs, supermassive multiple-input multiple-output (MIMO) \cite{zhang20196g} and cell-free massive MIMO \cite{bjornson2019massive} are capable of providing exceptionally high system capacity through their vast array gain and spatial resolution. On the other hand, reconfigurable intelligent surfaces (RISs) another revolutionary 6G technique, embody an ELAA with passive elements \cite{liu2021reconfigurable}. By manipulating the wireless propagation environment, RISs offer new opportunities for augmenting the coverage and capacity of the 6G network. 

    \item \textbf{Tremendously high frequencies}: 
    The terahertz (THz) band, spanning from $0.1$ THz to $10$ THz, holds great potential as a promising frequency band for 6G \cite{dang2020should}. Compared to the millimeter-wave (mmWave) band utilized in 5G, the THz band provides significantly more bandwidth resources on the order of tens of gigahertz (GHz) \cite{akyildiz2014terahertz}. Furthermore, due to the very small wavelengths in the THz band, an enormous number of antenna elements can be integrated into THz base stations (BSs), thus facilitating the implementation of ELAA. Due to these benefits, THz communication is expected to support very high data rates on the order of Tb/s \cite{zhang20196g}.
    
    \item \textbf{New types of antennas}: Metamaterials are powerful artificial materials that exhibit various desired electromagnetic (EM) characteristics for wireless communications \cite{cui2014coding}. In recent years, metamaterials have also been exploited in realizing (approximately) continuous transmitting and receiving apertures, and thus, facilitating holographic beamforming \cite{ shlezinger2021dynamic, deng2021reconfigurable}. Compared with conventional beamforming techniques, holographic beamforming realized by continuous-aperture (CAP) antennas has a super high spatial resolution while avoiding undesirable side lobes \cite{bjornson2019massive}.
    
\end{itemize}

The significant increase in the size of antenna arrays, extremely high frequencies, and the emerging new metamaterial-based antennas cause a paradigm shift for the EM characteristics in 6G. Generally, the EM field radiated from antennas can be divided into two regions: near-field region and far-field region \cite{kraus2002antennas, 7942128}. The EM waves in these regions exhibit different propagation properties. Specifically, the wavefront of EM waves in the far-field region can be reasonably well approximated as being planar, as shown in Fig. \ref{fig:far_vs_near}(a). Conversely, more complex wave models, such as spherical waves, see Fig. \ref{fig:far_vs_near}(b), are required to accurately depict propagation in the near-field region. The Rayleigh distance is one of the most common figures of merit to distinguish between the near-field and far-field regions and is given by $2D^2/ \lambda$ \cite{kraus2002antennas}. Here, $D$ and $\lambda$ denote the antenna aperture and the wavelength, respectively. From the first-generation (1G) to 5G wireless networks, the near-field was generally limited to a few meters or even centimeters because of the low-dimensional antenna arrays and low frequencies. Therefore, communication systems could be efficiently designed based on far-field approximation. However, given the large aperture of ELAAs and tremendously high frequencies, 6G networks exhibit a large near-field region on the order of hundreds of meters. For instance, a transmitter of size $D = 0.5$ meters operating at frequency $60$ GHz has a near-field region of $100$ meters. Furthermore, for the investigation of metamaterial-based CAP antennas, the conventional far-field plane wave assumption is also not suitable \cite{slepian1961prolate, slepian1976bandwidth}. Therefore, in 6G networks, the near-field region is not negligible, which motivates the investigation of the new near-field communication (NFC)\footnote{We note that the term “NFC” is also used to describe a technique rooted in radio-frequency identification (RFID), which enables communication between two electronic devices over a short distance \cite{coskun2013survey}, e.g., on the order of several centimetres. However, in this paper, the term “NFC” refers to wireless communications under the near-field electromagnetic effect caused by the usage of large antenna arrays and high-frequency bands.} paradigm. To shed light on the benefits of NFC, we first distinguish between near-field and far-field in terms of their EM-radiation properties.

\begin{figure}[!t]
    \centering
    \includegraphics[width=0.49\textwidth]{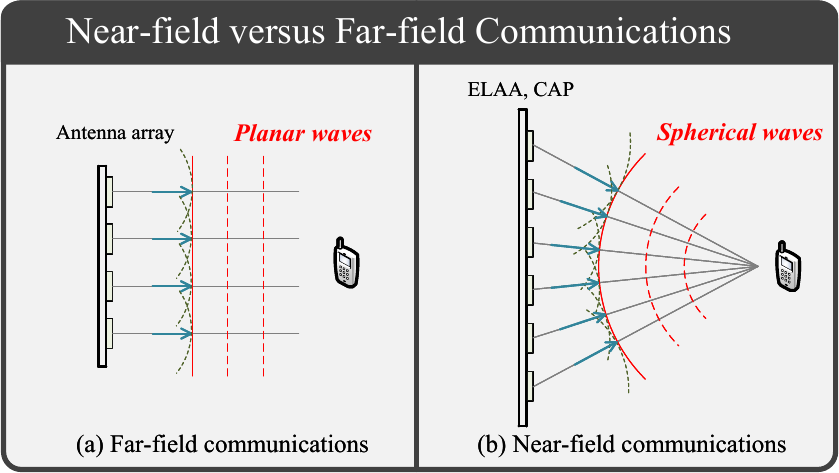}
    \caption{Near-field communications versus far-field communications.}
    \label{fig:far_vs_near}
\end{figure}

\subsection{Distinction between Near-Field and Far-Field Regions}
\subsubsection{General Field Regions}
According to EM and antenna theory, the field surrounding a transmitter can be divided into near-field and far-field regions. The near-field region can be further divided into the reactive and radiating near-field regions. The three regions are explained below \cite{yaghjian1986overview}:
\begin{itemize}
\item \textbf{{The reactive near-field region}} is limited to the space close to the antenna, where the EM fields are predominantly reactive, meaning that they store and release energy rather than propagating away from the antenna as radiating waves. Within the reactive near-field region, evanescent waves are dominant.
\item \textbf{{The radiating near-field region}} is located at a distance from the antenna that is greater than a few wavelengths. In this region, the fields have not fully developed into the planar waves that are characteristic of the far-field region. Within the radiating near-field region, the propagating waves have spherical wavefront, i.e., spherical waves are dominant.
\item \textbf{The far-field region} surrounds the radiating near-field region. In the far-field, the angular field distribution is essentially independent of the distance between the receiver and the transmitter, i.e., planar waves are dominant.
\end{itemize}
Since the reactive near-field region is typically small and the evanescent waves fall off exponentially fast with distance, in the remainder of this paper, we mainly focus on wireless communications within the radiating near-field region. For simplicity, we use the term \textit{near-field} to refer to the radiating near-field region.

\subsubsection{Boundary between Near-Field and Far-Field Regions}
The transition between the near-field and far-field regions happens gradually and there is no strict boundary between the two regions.
As a result, different works have proposed different metrics for characterising the field boundary.
For defining these metrics, two different perspectives are used, namely, the phase error perspective and the channel gain error perspective. 
\begin{itemize}
\item From the \emph{phase error} perspective, several commonly used rules of thumb, including the Rayleigh distance~\cite{kraus2002antennas}, the Fraunhofer condition~\cite{7942128}, and the extended Rayleigh distance for MIMO transceivers and RISs~\cite{danufane2020path}, have emerged. These distances mainly apply to the field boundary close to the main axis of the antenna aperture.
    \item From the \emph{channel gain error} perspective, a more accurate description of the field boundary can be given for off-axis regions. Specifically, according to the Friis formula~\cite{Lozano2018}, the channel gain falls off with the inverse of the distance squared. However, this does not hold in the near-field region. Therefore, we can define the far-field region as the region where the actual channel gain can be approximated by the Friis formula subject to a tolerable error. Exploiting this perspective, the field boundary depends not only on the aperture size and wavelength, but also on the angle of departure, angle of arrival, and shape of the transmit antenna aperture.
\end{itemize}

\begin{table}[!t]
\caption{List of Acronyms}
\label{tab:LIST OF ACRONYMS}
\centering
\resizebox{0.35\textwidth}{!}{
\begin{tabular}{!{\vrule width1pt}c!{\vrule width1pt}!{\vrule width1pt}c!{\vrule width1pt}}
\Xhline{1pt} 
AWGN & Additive Gaussian White Noise\\ \Xhline{1pt} 
BCD  & Block Coordinate Descent \\ \Xhline{1pt} 
BS   & Base Station\\ \Xhline{1pt} 
CAP  & Continuous-Aperture\\ \Xhline{1pt} 
CDF  & Cumulative Distribution Function\\ \Xhline{1pt} 
CSI  & Channel State Information\\ \Xhline{1pt} 
DFT  & Discrete Fourier Transform \\ \Xhline{1pt} 
DoF  & Degree of Freedom\\ \Xhline{1pt} 
ECC  & Ergodic Channel Capacity \\ \Xhline{1pt} 
EE   & Energy Efficiency \\ \Xhline{1pt} 
ELAA & Extremely Large-Scale Antenna Array \\ \Xhline{1pt} 
EM   & Electromagnetic\\ \Xhline{1pt} 
EMI  & Ergodic Mutual Information \\ \Xhline{1pt} 
FFC  & Far-Field Communication \\ \Xhline{1pt} 
FP   & Fractional Programming\\ \Xhline{1pt} 
GHz  & Gigahertz \\ \Xhline{1pt} 
i.i.d. & Independent and Identically Distributed \\ \Xhline{1pt} 
LoS  & Line-of-Sight \\ \Xhline{1pt} 
M-EDCF & Multi-Exponential Decay Curve Fitting\\ \Xhline{1pt} 
MI   & Mutual Information \\ \Xhline{1pt} 
MIMO & Multiple-Input Multiple-Output \\ \Xhline{1pt} 
MISO & Multiple-Input Single-Output \\ \Xhline{1pt} 
MMSE & Minimum Mean Square Error \\ \Xhline{1pt}  
mmWave & Millimeter-Wave \\ \Xhline{1pt} 
MRT  & Maximum Ratio Transmission\\ \Xhline{1pt} 
NFC  & Near-Field Communication \\ \Xhline{1pt} 
NLoS & Non-Line-of-Sight \\ \Xhline{1pt} 
NUSW & Non-Uniform Spherical Wave \\ \Xhline{1pt} 
OFDM & Orthogonal Frequency Division Multiplexing \\ \Xhline{1pt} 
OMP  & Orthogonal Matching Pursuit \\ \Xhline{1pt} 
OP   & Outage Probability \\ \Xhline{1pt} 
PDF  & Probability Density Function\\ \Xhline{1pt} 
PS   & Phase Shifter\\ \Xhline{1pt} 
QAM  & Quadrature Amplitude Modulation \\ \Xhline{1pt} 
QPSK & Quadrature Phase Shift Keying \\ \Xhline{1pt} 
RF   & Radio-Frequency\\ \Xhline{1pt} 
RIS  & Reconfigurable Intelligent Surface \\ \Xhline{1pt} 
ROC  & Rate of Convergence\\ \Xhline{1pt} 
SCA  & Successive Convex Approximation \\ \Xhline{1pt} 
SE   & Spectral Efficiency \\ \Xhline{1pt} 
SG   & Stochastic Geometry \\ \Xhline{1pt}  
SIMO & Single-Input Multiple-Output \\ \Xhline{1pt} 
SINR & Signal-to-Interference-Plus-Noise Ratio\\ \Xhline{1pt} 
SNR  & Signal-to-Noise Ratio\\ \Xhline{1pt} 
SPD  & Spatially-Discrete\\ \Xhline{1pt} 
THz  & Terahertz\\ \Xhline{1pt} 
TTD  & True Time Delayer\\ \Xhline{1pt} 
ULA  & Uniform Linear Array\\ \Xhline{1pt} 
UPA  & Uniform Planar Array\\ \Xhline{1pt} 
USW  &  Uniform Spherical Wave\\ \Xhline{1pt} 
WMMSE & Weighted Minimum Mean Square Error\\ \Xhline{1pt} 
XR   &  Extended Reality\\ \Xhline{1pt} 
1G   &  First-Generation\\ \Xhline{1pt} 
5G   &  Fifth-Generation\\ \Xhline{1pt} 
6G   &  Sixth-Generation\\ \Xhline{1pt} 
\end{tabular}}
\end{table}

\begin{table}[!t]
\caption{List of Variables}
\label{tab:LIST OF VARIABLES}
\centering
\resizebox{0.35\textwidth}{!}{
\begin{tabular}{!{\vrule width1pt}c!{\vrule width1pt}!{\vrule width1pt}c!{\vrule width1pt}}
\Xhline{1pt} 
$\mathbf{a}_{R/T}$ & Array response vector (receiver/transmitter)\\ \Xhline{1pt} 
$\mathbf{h}$ & Near/far-field MISO channel\\ \Xhline{1pt} 
$\mathbf{H}$ & Near/far-field MIMO channel\\ \Xhline{1pt} 
$\beta$ & Channel gain\\ \Xhline{1pt} 
$\mathbf{s}$& Cartesian coordinates of the transmitter\\ \Xhline{1pt} 
$\mathbf{r}$& Cartesian coordinates of the receiver\\ \Xhline{1pt} 
$\lambda$ & Wavelength of the carrier signal\\ \Xhline{1pt} 
$f_c$ & Frequency of the carrier signal\\ \Xhline{1pt} 
$k$ & Wave number of the carrier signal\\ \Xhline{1pt} 
$f_m$ & Frequency of subcarriers\\ \Xhline{1pt} 
$d$ & Antenna spacing \\ \Xhline{1pt} 
$\theta$ & Azimuth angle \\ \Xhline{1pt} 
$\phi$ & Elevation angle \\ \Xhline{1pt} 
$r$ & Propagation distance \\ \Xhline{1pt} 
$G_1(\mathbf{s},\mathbf{r})$ & Effective aperture loss\\ \Xhline{1pt} 
$G_2(\mathbf{s},\mathbf{r})$ & Polarization loss\\ \Xhline{1pt} 
$\mathbf{J}(\mathbf{s})$ & Electric (source) current\\ \Xhline{1pt} 
$\mathbf{G}(\mathbf{s},\mathbf{r})$ & Green's function \\ \Xhline{1pt} 
$K(\mathbf{s}_1,\mathbf{s}_2)$ & Kernel of the Green's function\\ \Xhline{1pt} 
$\mathbf{F}$ & Analog/digital beamformer\\ \Xhline{1pt} 
$P_{\text{max}}$ & Maximum transmit power\\ \Xhline{1pt} 
$\mathbf{T}$ & TTD analog beamformer\\ \Xhline{1pt} 
$\mathbf{\Psi}$ & Metasurface analog beamformer \\ \Xhline{1pt} 
$\mathbf{Q}$ & Waveguide propagation matrix of metasurface \\ \Xhline{1pt} 
$\gamma$ & Received SNR\\ \Xhline{1pt} 
$F_{\lVert{\mathbf{h}}\rVert^2}(\cdot)$ & CDF of MISO channel gain\\ \Xhline{1pt} 
$f_{\lVert{\mathbf{h}}\rVert^2}(\cdot)$ & PDF of MISO channel gain\\ \Xhline{1pt} 
$\mathcal{P}$ & Outage probability\\ \Xhline{1pt} 
$\bar{\mathcal{C}}$ & Ergodic channel capacity\\ \Xhline{1pt} 
$\mathcal{X}$ & Finite constellation alphabet\\ \Xhline{1pt} 
$\bar{\mathcal{I}}_{\mathcal{X}}$ & Ergodic mutual information\\ \Xhline{1pt} 
$I_{\mathcal{X}}(\cdot)$ & Mutual information of Gaussian channels\\ \Xhline{1pt} 
$\rm{MMSE}_{\mathcal{X}}(\cdot)$ & MMSE of Gaussian channels\\ \Xhline{1pt} 
${\mathcal{M}}(\cdot;\cdot)$ & Mellin transform\\ \Xhline{1pt} 
\end{tabular}}
\end{table}

\subsection{Related Overview Articles}
As discussed before, because of the quite small near-field region due to the use of small-scale antenna arrays and low operating frequencies, NFC has not been relevant for 1G-5G wireless networks, and hence, the related literature is very sparse. So far, only a few magazine papers \cite{nepa2017near, cui2023near, 10068140, han2023cross} that provide an introduction to NFC have been published. The authors of \cite{nepa2017near} presented the basic working principle and applications of near-field-focused (NFF) microwave antennas for short-range wireless systems. They introduced various metrics for NFF performance evaluation, including the 3 dB focal spot, focal shift, focusing gain, and side lobe level. The authors of \cite{cui2023near} provided an overview of NFC, covering aspects such as field boundaries, challenges, potential applications, and future research directions. They offered a high-level introduction to NFC. Furthermore, the authors of \cite{10068140} studied the difference between far-field beamsteering and near-field beamfocusing. They emphasized the significant power gain and novel application opportunities that arise from near-field beamfocusing. Finally, the authors of \cite{han2023cross} addressed the cross-far-field and near-field issues in THz communications. They discussed the relevant channel model, channel estimation techniques, and hybrid beamforming approaches for cross-field communications. While \cite{nepa2017near, cui2023near, 10068140, han2023cross} review the general concepts of NFC, fundamental aspects of NFC, including basic channel models, antenna structures, and analytical foundations are not covered. Moreover, a comprehensive tutorial on NFC that specifically caters to the needs of graduate students and researchers seeking to gain a fundamental understanding of NFC is not available in the literature.

\subsection{Motivation and Contributions}
NFC will play a significant role in 6G and fundamental knowledge gaps have to be closed to fully exploit the new opportunities and to address the new challenges arising for NFC. 
However, a comprehensive tutorial review on NFC is missing in the literature.
This is the motivation for this paper and its main contributions can be summarized as follows:
\begin{itemize}
    \item We start by reviewing the basic near-field channel models for both SPD and CAP antennas. \emph{For SPD antennas}, near-field spherical-wave-based channel models are introduced for both multiple-input single-output (MISO) and MIMO systems, where the specific characteristics of near-field channels compared with far-field channels are highlighted. Furthermore, we discuss uniform spherical wave (USW) and non-uniform spherical wave (NUSW) near-field channel models. 
    \emph{For CAP antennas}, we introduce a Green’s function-based near-field channel model.
    \item We study the properties of near-field beamfocusing and antenna architectures for NFC. We commence with the MISO case, for which we discuss hybrid beamforming architectures based on phase shifts and true time delayers for narrowband and wideband systems, respectively. Then, we propose to exploit practical metasurface-based antennas to approximate CAP antennas. As a further advance, we consider the MIMO case, where the dynamic degrees of freedom (DoFs) of near-field channels are addressed. Finally, near-field beam training is discussed, which can help to significantly decrease the complexity of channel estimation and analog beamforming design.
    \item We provide a comprehensive performance analysis framework for NFC for both deterministic line-of-sight (LoS) and statistical multipath channels. \emph{For near-field LoS channels}, we derive new expressions for the received signal-to-noise ratio (SNR) and power scaling laws for both SPD and CAP antennas. \emph{For near-field statistical channels}, we propose a general theoretical framework for analyzing the outage probability (OP), ergodic channel capacity (ECC), and ergodic mutual information (EMI). Important insights for NFC are unveiled, including the diversity order, array gain, high-SNR slope, and high-SNR power offset. 
\end{itemize}

\begin{figure}[!t]
    \centering
    \includegraphics[width=0.48\textwidth]{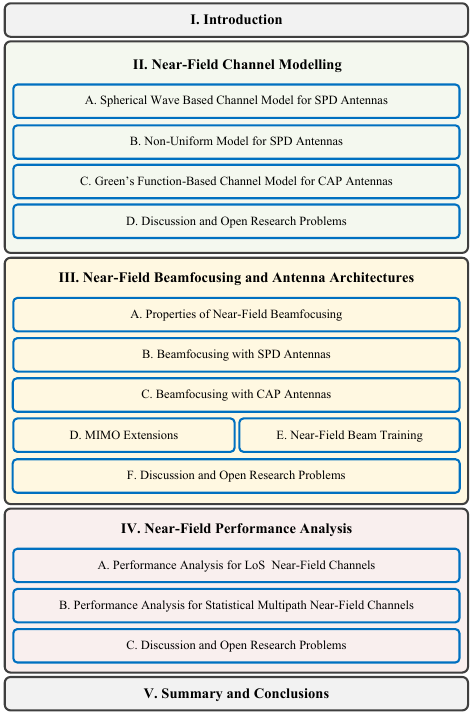}
    \caption{Condensed overview of this tutorial and outlook on NFC.}
    \label{fig:organisation}
\end{figure}

\subsection{Organization}
The remainder of this paper is structured as follows. 
Section II presents the fundamental near-field channel models for both SPD and CAP antennas. In Section III, the basic principles of near-field beamfocusing and beam training are introduced and the related antenna architectures for MISO and MIMO systems are provided. The performance of NFC in respectively LoS channels and statistical channels is analysed in Section IV. Finally, Section V concludes this paper. Lists with the most important abbreviations and variables are provided in Tables \ref{tab:LIST OF ACRONYMS} and \ref{tab:LIST OF VARIABLES}, respectively. Fig. \ref{fig:organisation} illustrates the organization of this tutorial.

\subsection{Notations}
Throughout this paper, for any matrix $\mathbf A$, $[\mathbf A]_{m,n}$, ${\mathbf{A}}^{\mathsf T}$, ${\mathbf{A}}^{*}$, ${\mathbf{A}}^{\mathsf H}$, and $\lVert{\mathbf{A}}\rVert_{\rm{F}}$ denote the $(m,n)$-th entry, transpose, conjugate, conjugate transpose, and Frobenius norm of $\mathbf A$, respectively. The matrix inequality ${\mathbf A}\succeq{\mathbf 0}$ indicates positive semi-definiteness of $\mathbf{A}$. $[\mathbf a]_{i}$ denotes the $i$-th entry of vector $\mathbf a$, and $\mathrm{diag}\{\mathbf a\}$ returns a diagonal matrix whose main diagonal elements are the entries of $\mathbf a$. Also, $\mathrm{blkdiag}\{{\cdot}\}$ represents a block diagonal matrix, $\mathbf{I}$ is the identity matrix, $\mathbf{0}$ is the zero matrix, $\lVert\cdot\rVert$ denotes the Euclidean norm of a vector, $\lvert\cdot\rvert$ denotes the norm of a scalar, $\mathbb{C}$ stands for the complex plane, $\mathbb{R}$ stands for the real plane, and ${\mathbb{E}}\{\cdot\}$ represents mathematical expectation.
The big-$\mathcal{O}$ notation, $f(x)={\mathcal{O}}\left(g(x)\right)$, means that $\limsup_{x\rightarrow\infty}\frac{\left|f(x)\right|}{g(x)}<\infty$. ${\mathcal{CN}}\left({\bm\mu},{\bm\Sigma}\right)$ is used to denote the complex Gaussian distribution with mean ${\bm\mu}$ and covariance matrix ${\bm\Sigma}$. Finally, $\otimes$ and $\odot$ denote the Kronecker product and Hadamard product, respectively.

\section{Near-Field Channel Modelling}\label{NFC_Tutorial_Channel_Model}

\begin{figure*}[!t]
    \centering
    \includegraphics[width=0.8\textwidth]{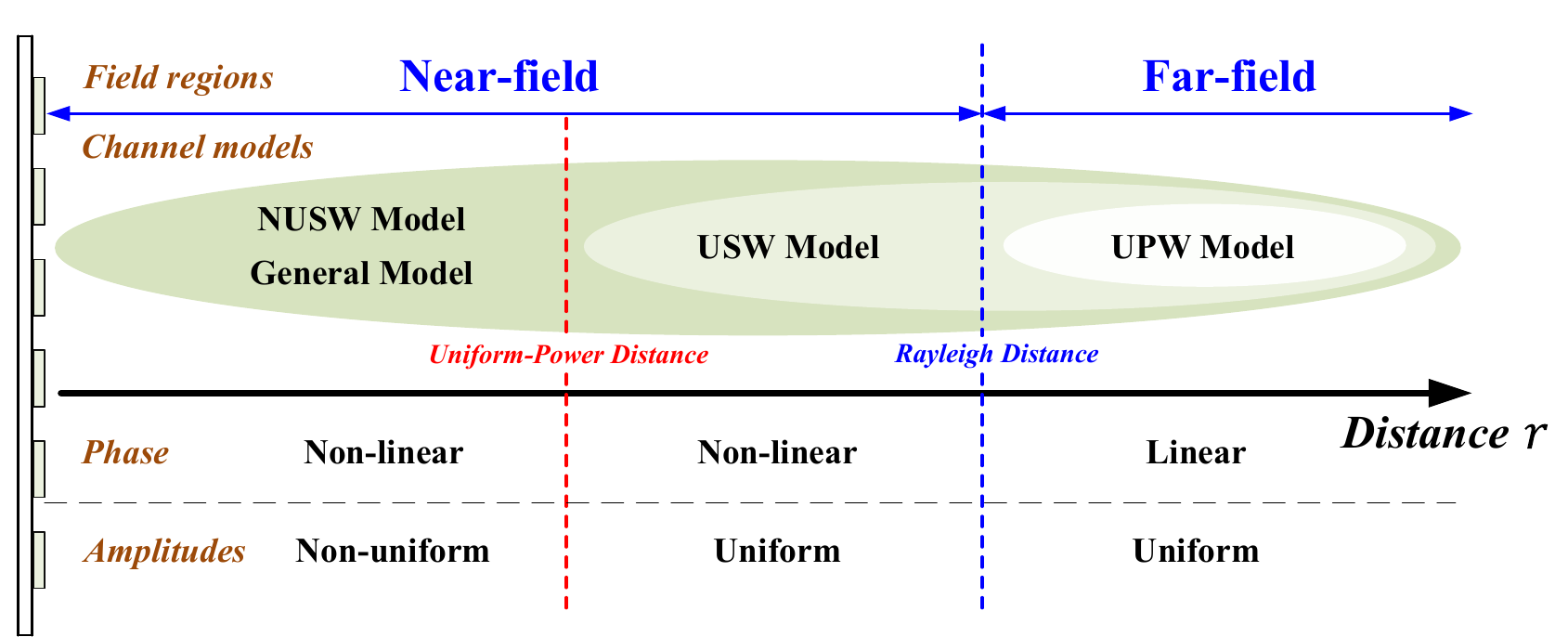}
    \caption{Field regions with respect to an antenna array.}
    \label{regions}
\end{figure*}

In this section, we introduce the fundamental models for near-field channels. As illustrated in Fig. \ref{regions}, the space surrounding an antenna array can be divided into three regions, with two distances, namely \emph{Rayleigh distance} \cite{sherman1962properties, kraus2002antennas} and \emph{uniform power distance} \cite{sherman1962properties, 9617121}. As previously discussed, the Rayleigh distance can be used to separate the near-field and far-field regions, and it is mostly relevant for characterizing the behaviour of the phase of an emitted signal. Generally speaking, if the propagation distance of the signal is larger than the Rayleigh distance, the far-field planar-wave-based channel model can be employed, leading to a \emph{linear phase} of the signals. Conversely, if the signal's propagation distance is less than the Rayleigh distance, the near-field spherical-wave-based channel model has to be utilized, resulting in a \emph{non-linear phase} of the signal. Moreover, within the near-field region, the uniform power distance can be used to differentiate between regions where the signal amplitude is \emph{uniform} and \emph{non-uniform}, respectively. 
Based on the Rayleigh and uniform power distances, the channel models suitable for characterizing signal propagation can be classified loosely into three categories: 1) uniform planar wave (UPW) model, 2) uniform spherical wave (USW) model, and 3) non-uniform spherical wave (NUSW) model, as shown in Fig. \ref{regions}.
As we will explain in the following sections, due to the different assumptions made, these near-field channel models have different levels of accuracy.
Regarding scatterers and small-scale fading, near-field channel models can be categorized into deterministic and statistical models, with deterministic models utilizing ray tracing, geometric optics, or electromagnetic wave propagation theories for precise channel gain determination, primarily for line-of-sight (LoS) or near-field channels with a limited number of paths. Statistical models, on the other hand, capture the average behavior, fading effects, and time-varying characteristics of the channel, making them suitable for characterizing rich-scattering environments.
Concerning transceiver types, near-field channel models can be classified into models for spatially-discrete (SPD) antennas and CAP antennas.
Given the intricate taxonomy of channel models, a comprehensive and systematic overview of the various models is needed.
In the following, we introduce the USW and NUSW channel models for NFC systems equipped with SPD antennas. Furthermore, we introduce a general model to accurately capture the impact of the free-space path loss, effective aperture, and polarization mismatch. Moreover, we present a near-field channel model for CAP antennas. Finally, important open research problems in near-field modelling are discussed.

\subsection{Spherical Wave Based Channel Model for SPD Antennas}
Fig. \ref{fig:far_vs_near} highlights the primary distinction between near-field and far-field channels for SPD antennas. Specifically, far-field channels are characterized by planar waves, whereas near-field channels are characterized by spherical waves. Consequently, for far-field channels, the angles of the links between each antenna element and the receiver are approximated to be identical. In this case, the propagation distance of each link linearly increases with the antenna index, resulting in a linear phase for far-field channels. However, for near-field channels, each link has a different angle, leading to a non-linear phase. In the following, we review the near-field spherical-wave-based channel models for MISO and MIMO systems, and highlight the primary differences compared to the far-field planar-wave-based channel model.

\begin{figure}[!t]
    \centering
    \includegraphics[width=0.4\textwidth]{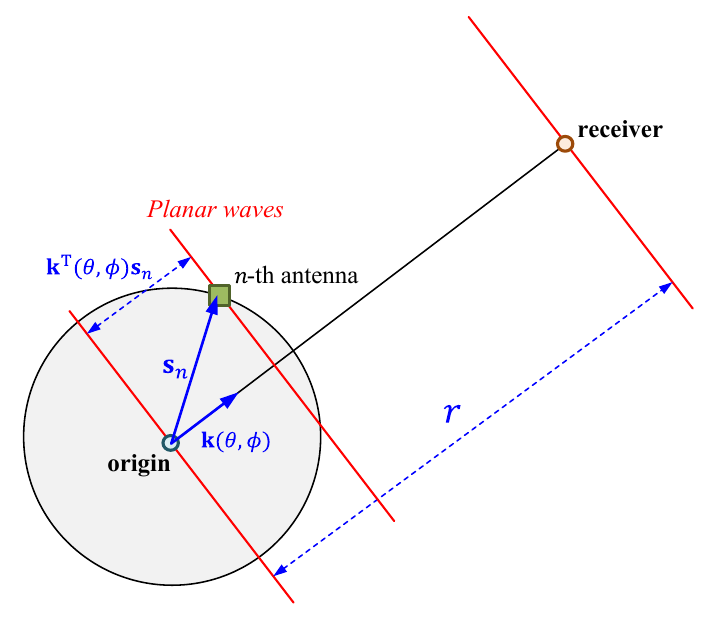}
    \caption{Far-field channel model.}
    \label{fig:far_field_model}
\end{figure}

\subsubsection{MISO Channel Model}
Let us consider a MISO system that comprises an $N$-antenna transmitter, where $N = 2\tilde{N} + 1$, and a single-antenna receiver. The antenna index of the transmit antenna array is given by $ n \in \{ -\tilde{N},\dots,\tilde{N}\}$. Let $\mathbf{r} = [r_x, r_y, r_z]^{\mathsf{T}}$ and $\mathbf{s}_n = [s_x^n, s_y^n, s_z^n]^{\mathsf{T}}$ denote Cartesian coordinates of the receive antenna and the $n$-th element of the transmit antenna array, respectively. Here, we set $\mathbf{s}_0 = [0,0,0]^T$ as the origin of the coordinate system. Accordingly, $r = \|\mathbf{r} - \mathbf{s}_0\|$ denotes the distance between the receiver and the central element of the transmit antenna array.

As shown in Fig. \ref{fig:far_field_model}, for the planar-wave-based far-field channel, the links between $\mathbf{s}_n$ and $\mathbf{r}$ are assumed to have identical angles. Let $\theta$ and $\phi$ denote the azimuth and elevation angles of the receiver with respect to the $x \textendash z$ plane, respectively. Then, the propagation direction vector of the signals from the transmitter to the receiver is given by \cite{manikas2004differential}
\begin{equation}
    \mathbf{k}(\theta, \phi) = [\cos \theta \sin \phi, \sin \theta \sin \phi, \cos \phi]^{\mathsf{T}}.
\end{equation} 
Then, according to the planar-wave assumption of far-field channels, the propagation distance for the link between the $n$-th transmit antenna and the receiver can be calculated as $r_n = r - \mathbf{k}^{\mathsf{T}}(\theta, \phi) \mathbf{s}_n$, resulting in the following channel coefficient \cite{tse2005fundamentals}:
\begin{equation}
    h_{\mathrm{far}}^n(\theta, \phi, r) = \beta_n e^{-j \frac{2\pi}{\lambda} r_n} = \beta_n e^{-j \frac{2\pi}{\lambda} r} e^{j \frac{2\pi}{\lambda} \mathbf{k}^{\mathsf{T}}(\theta, \phi) \mathbf{s}_n},
\end{equation}  
where $\beta_n$ denotes the channel gain (amplitude) for the $n$-th link and $\lambda$ denotes the wavelength of the carrier signal. In the far-field region, the propagation distance $r$ is typically beyond the \emph{uniform-power distance}\footnote{The formal definition of the uniform-power distance is given in Section~\ref{sub:UPD}}, where the channel gain disparity of each link is negligible \cite{sherman1962properties, 9617121}. In this case, we have $\beta_{-\tilde{N}} \approx \beta_{-\tilde{N}+1} \approx \dots \approx \beta_{\tilde{N}} = \beta$. 
Therefore, the far-field LoS MISO channel can be modelled in the following simplified form:
\begin{equation} \label{channel:upw_MISO}
    \mathbf{h}_{\mathrm{far}}^{\mathrm{LoS}} = \beta e^{-j\frac{2\pi}{\lambda} r} \big[ e^{j \frac{2\pi}{\lambda} \mathbf{k}^{\mathsf{T}}(\theta, \phi) \mathbf{s}_{-\tilde{N}}}, \dots, e^{j \frac{2\pi}{\lambda} \mathbf{k}^{\mathsf{T}}(\theta, \phi) \mathbf{s}_{\tilde{N}}} \big]^{\mathsf{T}}.
\end{equation}
The above far-field channel model is referred to as the \emph{UPW} model \cite{ertel1998overview, manikas2004differential}. According to \eqref{channel:upw_MISO}, the array response vector for far-field channels is given by 
\begin{center}
    \begin{tcolorbox}[title = Far-Field Array Response Vector]
    {\setlength\abovedisplayskip{2pt}
    \setlength\belowdisplayskip{2pt}
    \begin{equation} \label{far_field_vector}
        \mathbf{a}_{\mathrm{far}}(\theta, \phi) = \big[ e^{j \frac{2\pi}{\lambda} \mathbf{k}^{\mathsf{T}}(\theta, \phi) \mathbf{s}_{-\tilde{N}}}, \dots, e^{j \frac{2\pi}{\lambda} \mathbf{k}^{\mathsf{T}}(\theta, \phi) \mathbf{s}_{\tilde{N}}} \big]^{\mathsf{T}}.
    \end{equation}
    }\end{tcolorbox}
\end{center}
It can be observed that for the elements of the far-field array response vector, the phases are \emph{linear functions} of the positions, $\mathbf{s}_n$.

On the other hand, for the near-field spherical-wave-based channel, the propagation distances of the links between the transmit and receive antennas cannot be calculated assuming identical azimuth and elevation angles, since different links have different angles. Therefore, the propagation distance of the link between the $n$-th transmit antenna element and the receiver needs to be calculated as $r_n = \| \mathbf{r} - \mathbf{s}_n \|$, resulting in the following channel coefficient \cite{tse2005fundamentals}:
\begin{equation}
    h_{\mathrm{near}}^n(\mathbf{s}_n, \mathbf{r}) = \beta_n e^{-j \frac{2\pi}{\lambda} \|\mathbf{r} - \mathbf{s}_n\|},
\end{equation}
Then, by assuming that propagation distance $r$ is larger than the \emph{uniform-power distance}, we have $\beta_{-\tilde{N}} \approx \beta_{-\tilde{N}+1} \approx \dots \approx \beta_{\tilde{N}} = \beta$. In this case, the near-field LoS MISO channel can be modelled as follows:
\begin{equation}
    \mathbf{h}_{\mathrm{near}}^{\mathrm{LoS}} = \beta e^{-j\frac{2\pi}{\lambda} r} \big[ e^{-j \frac{2\pi}{\lambda} (\|\mathbf{r} - \mathbf{s}_{-\tilde{N}}\|- r)},\dots, e^{-j \frac{2\pi}{\lambda} (\|\mathbf{r} - \mathbf{s}_{\tilde{N}}\|-r)} \big]^{\mathsf{T}}.
\end{equation}  
The above near-field channel model is referred to as the \textbf{\emph{USW}} model \cite{starer1994passive, bjornson2021primer}. The corresponding array response vector is given by 
\begin{center}
    \begin{tcolorbox}[title = Near-Field Array Response Vector]
    {\setlength\abovedisplayskip{2pt}
    \setlength\belowdisplayskip{2pt}
    \begin{equation} \label{near_field_vector}
        \mathbf{a}(\mathbf{r}) = \big[ e^{-j \frac{2\pi}{\lambda} (\|\mathbf{r} - \mathbf{s}_{-\tilde{N}}\|- r)},\dots, e^{-j \frac{2\pi}{\lambda} (\|\mathbf{r} - \mathbf{s}_{\tilde{N}}\|-r)} \big]^{\mathsf{T}}.
    \end{equation}
    }\end{tcolorbox}
\end{center}
In contrast to the far-field array response vector, the phase of the $n$-th entry of the above near-field array response vector is a \emph{non-linear function} of $\mathbf{s}_n$. 

\begin{figure}[t!]
    \centering
    \includegraphics[width=0.3\textwidth]{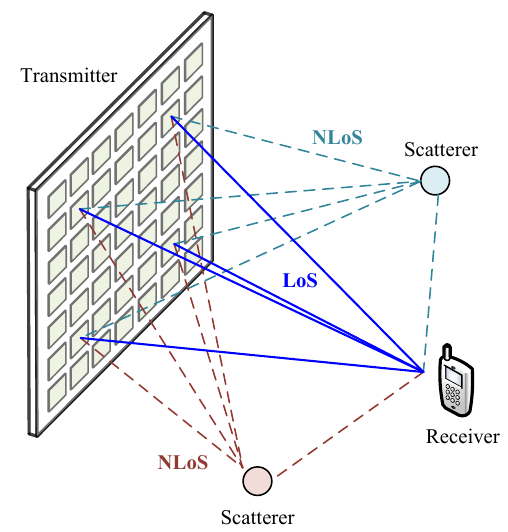}
    \caption{Near-field multipath MISO channel model.}
    \label{fig:near_MISO_multipath}
\end{figure}

\begin{remark}
    \textbf{(Rayleigh Distance and Fresnel Distance)} It is worth noting that the far-field UPW channel model is an approximation of the near-field USW channel model. More specifically, the coordinate of the receiver is given by $\mathbf{r} = [r\cos \theta \sin \phi, r\sin \theta \sin \phi, r\cos \phi]^{\mathsf{T}} = r \mathbf{k}(\theta, \phi)$. Therefore, the near-field propagation distance for the $n$-th antenna array can be calculated as 
    \begin{align} \label{taylor}
        r_n =& \|\mathbf{r}-\mathbf{s}_n \| = \|r \mathbf{k}(\theta, \phi) - \mathbf{s}_n \| \nonumber \\
        = &\sqrt{ r^2  - 2 r \mathbf{k}^{\mathsf{T}}(\theta, \phi) \mathbf{s}_n + \|\mathbf{s}_n\|^2 } \nonumber \\
        = &r - \mathbf{k}^{\mathsf{T}}(\theta,\phi)\mathbf{s}_n + \frac{\|\mathbf{s}_n\|^2-(\mathbf{k}^{\mathsf{T}}(\theta, \phi)\mathbf{s}_n)^2}{2r} \nonumber\\
        &+ \frac{\mathbf{k}^{\mathsf{T}}(\theta, \phi)\mathbf{s}_n \|\mathbf{s}_n\|^2}{2r^2} - \frac{\|\mathbf{s}_n\|^4}{8r^3} + \cdots,
    \end{align}
    where the last step is obtained by the Taylor expansion $\sqrt{1+x} = 1 + \frac{1}{2}x - \frac{1}{8}x^2+\cdots$ for $x = (-2 r \mathbf{k}^{\mathsf{T}}(\theta, \phi)\mathbf{s}_n + \|\mathbf{s}_n\|^2)/r^2$. For the \emph{far-field approximation}, only the first two terms in \eqref{taylor} are considered, which leads to
    \begin{align}
        r_n \approx r - \mathbf{k}^{\mathsf{T}}(\theta,\phi)\mathbf{s}_n.
    \end{align}
    The \emph{Rayleigh distance} is defined as the distance required such that the phase error of the channel caused by the far-field approximation does not exceed $\pi/8$ \cite{kraus2002antennas}. When $\theta = \phi = \frac{\pi}{2}$, it is given by $r_{\mathrm{R}} = \frac{2D^2}{\lambda}$. Moreover, if the first three terms in \eqref{taylor} are considered, the following \emph{Fresnel approximation} can be obtained \cite{7942128}
    \begin{align} \label{fresnel}
        r_n &\approx r - \mathbf{k}^{\mathsf{T}}(\theta,\phi)\mathbf{s}_n + \frac{\|\mathbf{s}_n\|^2-(\mathbf{k}^{\mathsf{T}}(\theta, \phi)\mathbf{s}_n)^2}{2r}.
    \end{align}
    The \emph{Fresnel distance} is defined as the distance required such that the phase error of the channel caused by the Fresnel approximation does not exceed $\pi/8$ \cite{7942128}. When $\theta = \phi = \frac{\pi}{2}$, it is given by $r_{\mathrm{F}} = 0.5\sqrt{\frac{D^3}{\lambda}}$.
\end{remark}

As shown in Fig. \ref{fig:near_MISO_multipath}, scatterers in the environment can cause multipath propagation in near-field channels, where the receiver also receives signals reflected by scatterers via non-line-of-sight (NLoS) paths. The randomness of these multipath NLoS components results in the stochastic behaviour of channels and consequently requires a statistical channel model.
Specifically, the channel between the transmitter and the scatterers can be regarded as a MISO channel. Let $L$ denote the total number of scatterers, $\tilde{\mathbf{r}}_{\ell}$ denote the coordinate of the $\ell$-th scatterer, and $\tilde{\beta}_{\ell}$ denote the channel gain, which also includes the impact of the random reflection coefficient of the $\ell$-th scatterer. Then, the near-field multipath channel can be modelled as follows:
\begin{center}
    \begin{tcolorbox}[title = Near-Field Multipath MISO Channel (LoS + NLoS)]
    {\setlength\abovedisplayskip{2pt}
    \setlength\belowdisplayskip{2pt}
    \begin{equation} \label{H_NFC_MISO}
        \mathbf{h}_{\mathrm{near}} = \underbrace{\beta \mathbf{a}(\mathbf{r})}_{\mathrm{LoS}} + \sum_{\ell=1}^L \underbrace{\tilde{\beta}_{\ell} \mathbf{a}(\tilde{\mathbf{r}}_{\ell})}_{\mathrm{NLoS}}. 
    \end{equation}
    }\end{tcolorbox}
\end{center}
In \eqref{H_NFC_MISO}, the random phase of $\tilde{\beta}_{\ell}$ 
is assumed to be independent and identically distributed (i.i.d.) and uniformly distributed in $(-\pi,\pi]$. The performance of NFC systems in near-field multipath channels will be discussed in detail in Section~\ref{Sec:per}.

The above analysis reveals that near-field MISO channels are characterized by the array response vector in \eqref{near_field_vector}. However, it is non-trivial to capture the characteristics of near-field channels based on the general expression \eqref{near_field_vector}, which entails mathematical difficulties for performance analysis and system design. To obtain more insights, we discuss the near-field array response vectors of two popular antenna array geometries, namely uniform linear array (ULA) and uniform planar array (UPA).

\begin{figure}[!t]
    \centering
    \includegraphics[width=0.4\textwidth]{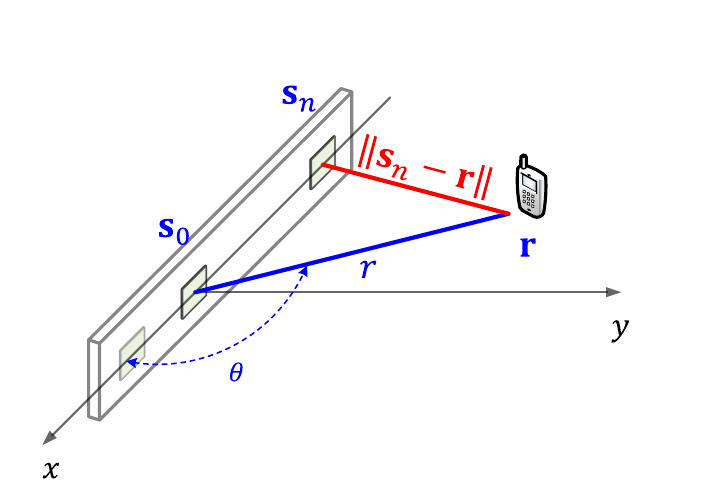}
    \caption{System layout of near-field MISO system with ULA.}
    \label{fig:ULA}
\end{figure}

$\bullet$ \textbf{\emph{Uniform Linear Array:}}
A ULA is a one-dimensional antenna array arranged linearly with equal antenna spacing. To derive the simplified array response vector, we consider a MISO system with $N$ antennas, where $N = 2\tilde{N} + 1$, at the transmitter. The spacing between adjacent antenna elements is denoted by $d$. For the ULA, we can always create a coordinate system such that all transmit antenna elements and the receiver are located in the $x \textendash y$ plane as shown in Fig. \ref{fig:ULA}. Therefore, we can ignore the $z$ axis and set $\phi = 90^\circ$. Then, by putting the origin of the coordinate system into the center of the ULA, the coordinates of the receiver and the $n$-th element of the ULA are given by $\mathbf{r} = [r \cos \theta, r \sin \theta]^{\mathsf{T}}$ and $\mathbf{s}_n = [nd, 0]^{\mathsf{T}}, \forall n \in \{ -\tilde{N},\dots,\tilde{N} \}, $ respectively. 
In this case, the propagation distance $\|\mathbf{r} - \mathbf{s}_n\|$ can be approximated as follows:
\begin{align} \label{taylor2}
    \lVert\mathbf{r} - \mathbf{s}_n\rVert &= \sqrt{r^2 + n^2d^2 - 2rnd\cos\theta}\nonumber\\
    &\approx r- nd \cos\theta + \frac{n^2 d^2 \sin^2\theta }{2r}. 
\end{align}
Here, we exploit the Fresnel approximation \eqref{fresnel} in the last step.
Then, we obtain the following simplified near-field array response vector for ULAs by substituting \eqref{taylor2} into \eqref{near_field_vector}:
\begin{center}
\begin{tcolorbox}[title = Near-Field Array Response Vector for ULAs]
{
\begin{equation}\label{eq_a_ula}
[\mathbf{a}_{\mathrm{ULA}}(\theta,r)]_n = e^{-j\frac{2\pi}{\lambda}(-nd\cos\theta + \frac{n^2d^2 \sin^2 \theta}{2r})}.
\end{equation}
}\end{tcolorbox}
\end{center}

\begin{figure}[!t]
    \centering
    \includegraphics[width=0.4\textwidth]{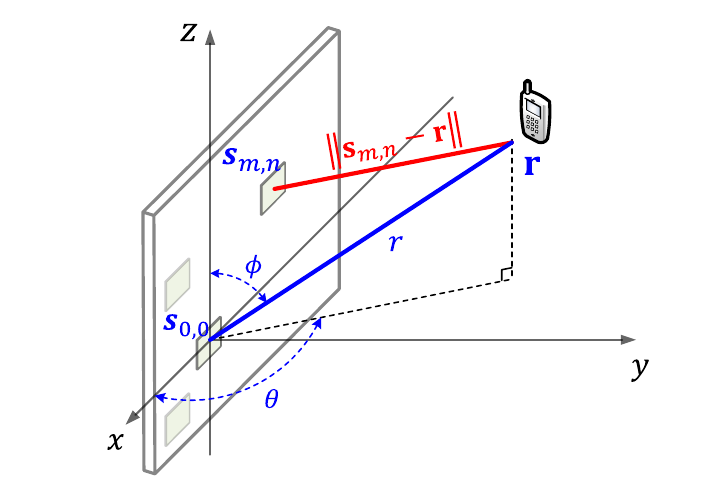}
    \caption{System layout of near-field MISO system with UPA.}
    \label{fig:UPA}
\end{figure}

$\bullet$ \textbf{\emph{Uniform Planar Array:}}
A UPA is a two-dimensional array of antennas uniformly arranged in a rectangular grid. We consider a MISO system with a UPA deployed in the $xz$-plane as illustrated in Fig. \ref{fig:UPA}. Assuming the UPA is located in the $x \textendash z$ plane and is composed of $N = N_x \times N_z$ antenna elements, where $N_x = 2\tilde{N}_x + 1$ and $N_z = 2\tilde{N}_z + 1$. The antenna spacings along the two directions are denoted by $d_x$ and $d_z$, respectively. 
Then, the Cartesian coordinates of the receiver and the $(m,n)$-th element of the transmit antenna array are given by $\mathbf{r} = (r \cos\theta \sin\phi, r \sin\theta \sin\phi, r \cos \phi)$ and $\mathbf{s}_{m,n} = (nd_x, 0, md_z), \forall n \in \{-\tilde{N}_x,\dots,\tilde{N}_x\}, m \in \{-\tilde{N}_z,\dots,\tilde{N}_z\},$ respectively. 
Assuming $d_x/r \ll 1$ and $d_z/r \ll 1$,
the distance $\|\mathbf{r} - \mathbf{s}_{m,n}\|$ can be approximated as follows:
\begin{align}
    &\lVert\mathbf{r} - \mathbf{s}_{m,n}\rVert \nonumber \\ 
    & \hspace{-0.1cm} = \sqrt{r^2 + n^2d_x^2 + m^2d_z^2  - 2r nd_x \cos\theta \sin \phi - 2 r md_z \cos \phi} \nonumber \\
    & \hspace{-0.1cm} \approx r - nd_x \cos \theta \sin \phi + \frac{n^2 d_x^2 (1 - \cos^2 \theta \sin^2 \phi) }{2r} \nonumber \\
    & \hspace{0.6cm} - md_z \cos \phi + \frac{m^2 d_z^2 \sin^2 \phi}{2r},
\end{align}
where the last step is obtained by exploiting Fresnel approximation \eqref{fresnel} and omitting the bilinear term. This approximation is sufficiently accurate for the USW model. After removing the constant phase $e^{-j \frac{2\pi}{\lambda} r}$, the phase of the array response vector can be divided into two components, namely $- nd_x \cos \theta \sin \phi + \frac{n^2 d_x^2(1 - \cos^2 \theta \sin^2 \phi)}{2r}$ and $- md_z \cos \phi + \frac{m^2 d_z^2 \sin^2 \phi}{2r}$, which only depend on $m$ and $n$, respectively. Then, we obtain the following result:
\begin{center}
\begin{tcolorbox}[title = Near-Field Array Response Vector for UPAs]
{
\begin{subequations} \label{eq_a_upa}
\begin{align}
&\mathbf{a}_{\mathrm{UPA}}(\theta, \phi, r) = \mathbf{a}_x(\theta, \phi, r) \otimes \mathbf{a}_z(\phi, r), \\
&[\mathbf{a}_x(\theta, \phi, r)]_n = e^{ -j \frac{2\pi}{\lambda} (- nd_x \cos \theta \sin \phi + \frac{n^2 d_x^2(1 - \cos^2 \theta \sin^2 \phi)}{2r})}, \\
&[\mathbf{a}_z(\phi, r)]_m = e^{-j \frac{2\pi}{\lambda} (- md_z \cos \phi + \frac{m^2 d_z^2 \sin^2 \phi}{2r})}.
\end{align}
\end{subequations}
}\end{tcolorbox}
\end{center}

Furthermore, according to \eqref{far_field_vector}, the well-known far-field array response vectors of ULAs and UPAs can be calculated, which are respectively given by:
\begin{align}
    \label{far_ULA}
    &\mathbf{a}_{\mathrm{ULA}}^{\mathrm{far}}(\theta) =  \big[ e^{-j\frac{2\pi}{\lambda} \tilde{N} d \cos\theta},\dots,e^{j\frac{2\pi}{\lambda} \tilde{N} d \cos\theta}   \big]^\mathsf{T}, \\
    \label{far_UPA}
    &\mathbf{a}_{\mathrm{UPA}}^{\mathrm{far}}(\theta, \phi) = \big[ e^{-j\frac{2\pi}{\lambda} \tilde{N}_x d_x \cos\theta \sin \phi },\dots,e^{j\frac{2\pi}{\lambda} \tilde{N}_x d_x \cos\theta \sin \phi}   \big]^\mathsf{T} \nonumber \\
    & \hspace{1.8cm} \otimes \big[ e^{-j\frac{2\pi}{\lambda} \tilde{N}_z d_z \cos\phi},\dots,e^{j\frac{2\pi}{\lambda} \tilde{N}_z d_z \cos\phi}   \big]^\mathsf{T}.
\end{align}
By comparing the far-field array response vectors in \eqref{far_ULA}, \eqref{far_UPA} with the near-field array response vectors in \eqref{eq_a_ula} and \eqref{eq_a_upa}, we obtain the following insights:
\begin{itemize}
    \item \textbf{Near-field channels dependent on both angle and distance.} This is the major difference between near-field and far-field channels. Compared with far-field channels, the additional distance dependence provides additional DoFs for NFC system design.
    
    \item \textbf{Far-field channels are special cases of near-field channels.} It can be observed that the far-field array response vectors in \eqref{far_ULA} and \eqref{far_UPA} can be obtained by omitting the phase terms that include $\frac{d^2}{r}$ in \eqref{eq_a_ula} and \eqref{eq_a_upa}. This implies that when $r$ is sufficiently large such that these terms become negligible, the near-field channel reduces to a far-field channel.
\end{itemize}

\subsubsection{MIMO Channel Model}
We continue to discuss MIMO systems. Let us consider a MIMO system consisting of an $N_T$-antenna transmitter and an $N_R$-antenna receiver, where $N_R = 2\tilde{N}_R+1$ and $N_T = 2\tilde{N}_T + 1$. The antenna indices of the transmit and receive antenna arrays are given by $n \in \{ -\tilde{N}_T,\dots,\tilde{N}_T\}$ and $m \in \{ -\tilde{N}_R,\dots,\tilde{N}_R\}$. Let $\mathbf{r}_m = [r_x^m, r_y^m, r_z^m]^{\mathsf{T}}$ and $\mathbf{s}_n = [s_x^n, s_y^n, s_z^n]^{\mathsf{T}}$ denote the Cartesian coordinates of the $m$-th element of the receive antenna array and the $n$-th element of the transmit antenna array, respectively, and let $r = \|\mathbf{r}_0 - \mathbf{s}_0\|$ denote the distance between the central elements of the receive and the transmit antenna array. Furthermore, $\mathbf{s}_0 = [0,0,0]^T$ is again the origin of the coordinate system.   

For the planar-wave-based far-field channel, let $\theta$ and $\phi$ denote the azimuth and elevation angles of the receiver with respect to the $x \textendash z$ plane, respectively. Then, for the far-field link between the $m$-th receive antenna and the $n$-th transmit antenna, the propagation distance is given by $r_{m,n} = r + \mathbf{k}^{\mathsf{T}}(\theta, \phi) \mathbf{s}_n + \mathbf{k}^{\mathsf{T}}(\theta, \phi) (\mathbf{r}_m - \mathbf{r}_0)$. By defining $\tilde{\mathbf{r}}_m = \mathbf{r}_m - \mathbf{r}_0$, the resulting channel coefficient is given by    
\begin{equation} \label{far_MIMO_link}
    h_{\mathrm{far}}^{m,n}(\theta, \phi, r) = \beta_{m,n} e^{-j \frac{2\pi}{\lambda} r} e^{j \frac{2\pi}{\lambda} \mathbf{k}^{\mathsf{T}}(\theta, \phi) \mathbf{s}_n } e^{j \frac{2\pi}{\lambda} \mathbf{k}^{\mathsf{T}}(\theta, \phi) \tilde{\mathbf{r}}_n},
\end{equation}  
Similar to the MISO case, the channel gains $\beta_{m,n}$ are assumed to have identical values of $\beta$ due to the large propagation distance in the far-field region. According to \eqref{far_MIMO_link}, the far-field LoS MIMO channel can be divided into the transmit-side component $e^{-j \frac{2\pi}{\lambda} \mathbf{k}^{\mathsf{T}}(\theta, \phi) \mathbf{s}_n }$ and the receive-side component $e^{-j \frac{2\pi}{\lambda} \mathbf{k}^{\mathsf{T}}(\theta, \phi) \tilde{\mathbf{r}}_n}$. Therefore, it can be modelled as the multiplication of transmit and receive array response vectors as follows:
\begin{center}
    \begin{tcolorbox}[title = Far-Field LoS MIMO Channel]
    {\setlength\abovedisplayskip{2pt}
    \setlength\belowdisplayskip{2pt}
    \begin{equation} \label{far_MIMO_LoS}
        \mathbf{H}_{\mathrm{far}}^{\mathrm{LoS}} = \tilde{\beta} \mathbf{a}_{\mathrm{far}}^R(\theta, \phi)  (\mathbf{a}_{\mathrm{far}}^T(\theta, \phi))^{\mathsf{T}}.
    \end{equation} 
    }\end{tcolorbox}
\end{center} 
Here, $\tilde{\beta} = \beta e^{-j \frac{2\pi}{\lambda} r}$ denotes the complex channel gain, and $\mathbf{a}_{\mathrm{far}}^T(\theta, \phi)$ and $\mathbf{a}_{\mathrm{far}}^R(\theta, \phi)$ denote the transmit and receive array response vectors, respectively, which are given as follows: 
\begin{align}
    &\mathbf{a}_{\mathrm{far}}^T(\theta, \phi) = [e^{j \frac{2\pi}{\lambda} \mathbf{k}^{\mathsf{T}}(\theta, \phi) \mathbf{s}_{-\tilde{N}_T}}, \dots, e^{j \frac{2\pi}{\lambda} \mathbf{k}^{\mathsf{T}}(\theta, \phi) \mathbf{s}_{\tilde{N}_T}} ]^{\mathsf{T}}, \nonumber \\ 
    &\mathbf{a}_{\mathrm{far}}^R(\theta, \phi) = [e^{j \frac{2\pi}{\lambda} \mathbf{k}^{\mathsf{T}}(\theta, \phi) \tilde{\mathbf{r}}_{-\tilde{N}_R}}, \dots, e^{j \frac{2\pi}{\lambda} \mathbf{k}^{\mathsf{T}}(\theta, \phi) \tilde{\mathbf{r}}_{\tilde{N}_R}} ]^{\mathsf{T}}.
\end{align}  
It is worth noting that the far-field LoS MIMO channel in \eqref{far_MIMO_LoS} always has rank one, leading to \emph{low DoFs}, i.e.,
\begin{equation}
    \mathrm{DoF}^{\mathrm{LoS}}_{\mathrm{far}} = \mathrm{rank}\{\mathbf{H}_{\mathrm{far}}^{\mathrm{LoS}}\} = 1.
\end{equation}

\begin{figure}[!t]
    \centering
    \includegraphics[width=0.35\textwidth]{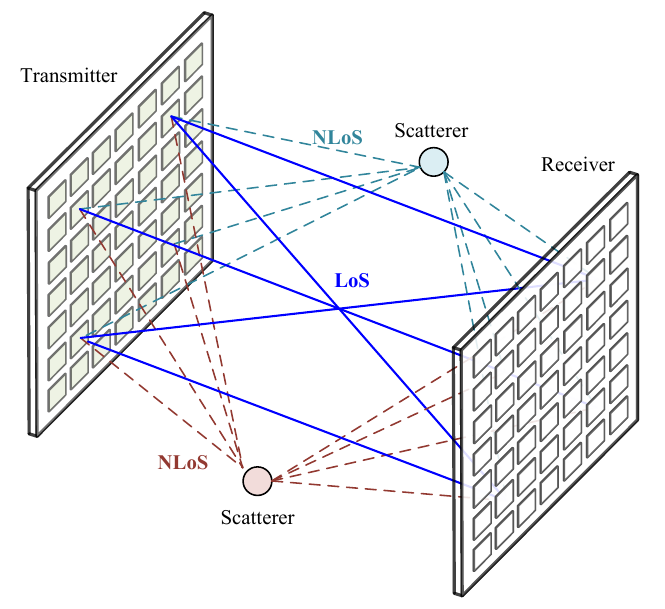}
    \caption{Near-field multipath MIMO channel model.}
    \label{fig:near_MIMO_multipath}
\end{figure}

For the near-field spherical-wave-based channel, similar to the MISO case, the propagation distance of each link in the LoS channel has to be calculated more accurately as $r_{m,n} = \|\mathbf{r}_m - \mathbf{s}_n\|$, resulting in the following channel coefficient:  
\begin{equation}
    h_{\mathrm{near}}^{m,n}(\mathbf{s}_n, \mathbf{r}_m) = \beta_{m,n} e^{-j \frac{2\pi}{\lambda} \|\mathbf{r}_m - \mathbf{s}_n\|}.
\end{equation}
If distance $r$ is larger than the \emph{uniform-power distance}, the channel gains $\beta_{m,n}$ are approximated to have the same value $\beta$. In contrast to the far-field LoS MIMO channel, the near-field LoS MIMO channel cannot be decomposed into transmit-side and receive-side components. Therefore, the near-field LoS MIMO channel needs to be modelled as 
\begin{center}
    \begin{tcolorbox}[title = Near-Field LoS MIMO Channel]
    {\setlength\abovedisplayskip{2pt}
    \setlength\belowdisplayskip{2pt}
    \begin{equation} \label{near_MIMO_LoS}
        [\mathbf{H}_{\mathrm{near}}^{\mathrm{LoS}}]_{m,n} = \beta e^{-j \frac{2\pi}{\lambda} \|\mathbf{r}_m - \mathbf{s}_n\|}.
    \end{equation} 
    }\end{tcolorbox}
\end{center} 
In particular, the near-field LoS MIMO channel matrix typically has high rank due to the non-linear phase \cite{miller2000communicating}, resulting in \emph{high DoFs}, i.e., 
\begin{equation}
    \mathrm{DoF}^{\mathrm{LoS}}_{\mathrm{near}} = \mathrm{rank}\{\mathbf{H}_{\mathrm{near}}^{\mathrm{LoS}}\} \ge 1.
\end{equation} 

Furthermore, multipath propagation will also occur in near-field MIMO channels if scatterers are present in the environment. Near-field multipath propagation is illustrated in Fig. \ref{fig:near_MIMO_multipath}. As can be observed, the NLoS MIMO channel can be regarded as the combination of two MISO channels with respect to the transmitter and receiver, respectively. Therefore, it can be written as the multiplication of the near-field array response vectors at the transmitter and receiver. Therefore, the near-field multipath channel can be modelled as follows:
\begin{center}
    \begin{tcolorbox}[title = Near-Field Multipath MIMO Channel (LoS + NLoS)]
    {\setlength\abovedisplayskip{2pt}
    \setlength\belowdisplayskip{2pt}
    \begin{equation}\label{NLoS_MIMO}
        \mathbf{H}_{\mathrm{near}} = \mathbf{H}_{\mathrm{near}}^{\mathrm{LoS}} + \sum_{\ell=1}^L \underbrace{\tilde{\beta}_{\ell} \mathbf{a}_R(\tilde{\mathbf{r}}_{\ell}) \mathbf{a}_T^{\mathsf{T}}(\tilde{\mathbf{r}}_{\ell})}_{\mathrm{NLoS}}. 
    \end{equation}
    }\end{tcolorbox}
\end{center}
Here, the near-field transmit array response vectors $\mathbf{a}_T(\tilde{\mathbf{r}}_{\ell})$ and receive array response vector $\mathbf{a}_R(\tilde{\mathbf{r}}_{\ell})$ are defined as follows:
\begin{align}
    &\mathbf{a}_T(\tilde{\mathbf{r}}_\ell) = \big[ e^{-j \frac{2\pi}{\lambda} \|\tilde{\mathbf{r}}_\ell - \mathbf{s}_{-\tilde{N}_T}\|},\dots, e^{-j \frac{2\pi}{\lambda} \|\tilde{\mathbf{r}}_\ell - \mathbf{s}_{\tilde{N}_T}\|} \big]^{\mathsf{T}}, \\
    &\mathbf{a}_R(\tilde{\mathbf{r}}_\ell) = \big[ e^{-j \frac{2\pi}{\lambda} \|\tilde{\mathbf{r}}_\ell - \mathbf{r}_{-\tilde{N}_R}\|},\dots, e^{-j \frac{2\pi}{\lambda} \|\tilde{\mathbf{r}}_\ell - \mathbf{r}_{\tilde{N}_R}\|} \big]^{\mathsf{T}}.
\end{align}
In a rich scattering environment, i.e., $L \gg 1$, the MIMO channel in \eqref{NLoS_MIMO} can achieve full rank due to the random phase-shifts imposed by scatterers, i.e., for the case of SDP antennas, the DoFs are given by:
\begin{equation}\label{Dof_s}
    \mathrm{DoF}^{\mathrm{NLoS}}_{\mathrm{near}} = \mathrm{rank}\{\mathbf{H}_{\mathrm{near}}^{\mathrm{NLoS}}\} = \min \{N_T, N_R\}.
\end{equation}

It can be observed that near-field NLoS MIMO channels have a similar structure as far-field MIMO channels, which can be written as the multiplication of transmit and receive array response vectors. However, near-field LoS MIMO channels have a significantly different structure. In the following, to obtain more insights into near-field LoS MIMO channels, we discuss two specific antenna array geometries, namely parallel ULAs and parallel UPAs.

\begin{figure}[!t]
    \centering
    \includegraphics[width=0.4\textwidth]{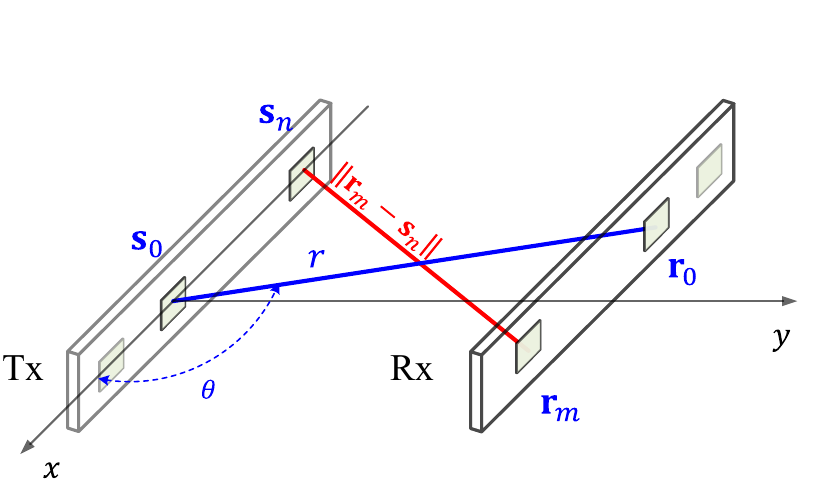}
    \caption{System layout of near-field MIMO system with two parallel ULAs.}
    \label{fig:ULA_MIMO}
\end{figure}

$\bullet$ \textbf{\emph{Parallel ULAs:}}
We consider the MIMO system shown in Fig. \ref{fig:ULA_MIMO}, with an $N_T$-antenna ULA at the transmitter and an $N_R$-antenna ULA at the receiver, where $N_T = 2\tilde{N}_T + 1$ and $N_R = 2 \tilde{N}_R + 1$. In particular, the two ULAs are parallel to each other. The spacing between adjacent transmit and receive antennas is denoted by $d_T$ and $d_R$, respectively. The angle and distance of the center of the receive ULA with respect to the center of the transmit ULA are denoted by $\theta$ and $r$, respectively. According to the system layout in Fig. \ref{fig:ULA_MIMO}, the Cartesian coordinates of the $n$-th elements at the transmitter and the $m$-th elements at the receiver are $\mathbf{s}_n = (n d_T, 0), \forall 
n \in \{-\tilde{N}_T \dots \tilde{N}_T\},$ and $\mathbf{r}_m = (r \cos \theta - m d_R, r \sin \theta), \forall 
m \in \{-\tilde{N}_R \dots \tilde{N}_R\}$, respectively. Therefore, the distance $\| \mathbf{r}_m - \mathbf{s}_n \|$ can be approximated as follows:
\begin{align}\label{eq:three_com}
    &\|\mathbf{r}_m - \mathbf{s}_n \| \nonumber\\
    &= \sqrt{r^2 + (nd_T + md_R)^2 - 2 r (nd_T + md_R) \cos \theta}   \nonumber \\
    &\overset{(a)}{\approx} r - (nd_T + md_R) \cos \theta + \frac{(nd_T + md_R)^2 \sin^2 \theta}{2 r} \nonumber \\
    &= r - nd_T \cos \theta + \frac{n^2 d_T^2 \sin^2 \theta}{2r} - md_R \cos \theta + \frac{m^2 d_R^2 \sin^2 \theta}{2r} \nonumber \\
    &\hspace{0.6cm}+ \frac{n d_T m d_R \sin^2 \theta}{r},
\end{align}
where the Fresnel approximation \eqref{fresnel} is exploited in step $(a)$. It can be observed that \eqref{eq:three_com} involves three components, namely $- nd_T \cos \theta + \frac{n^2 d_T^2 \sin^2 \theta}{2r}$, $- md_R \cos \theta + \frac{m^2 d_R^2 \sin^2 \theta}{2r}$, and $\frac{n d_T m d_R \sin^2 \theta}{r}$, where the first two components depend only on $n$ and $m$, respectively, while the last one involves both $n$ and $m$. Therefore, the near-field LoS MIMO channel matrix for parallel ULAs can be expressed via the ULA array response vectors in \eqref{eq_a_ula} and an additional coupled component as follows:
\begin{center}
\begin{tcolorbox}[title = Near-Field LoS MIMO Channel for Parallel ULAs]
{
\begin{subequations} \label{18a}
\begin{align}
&\mathbf{H}_{\mathrm{ULA}}^{\mathrm{LoS}} = \tilde{\beta} \mathbf{a}_{\mathrm{ULA}}^R(\theta, r) (\mathbf{a}_{\mathrm{ULA}}^T(\theta, r))^{\mathsf{T}} \odot \mathbf{H}_c, \\
&[\mathbf{H}_c]_{m,n} = e^{-j \frac{2\pi}{\lambda r}nd_T md_R \sin^2 \theta},
\end{align}
\end{subequations}
}\end{tcolorbox}
\end{center}
where $\tilde{\beta} = \beta e^{-j \frac{2\pi}{\lambda} r}$ denotes the complex channel gain.

\begin{remark}
    As can be observed in \eqref{18a}, the near-field LoS MIMO channel matrix between parallel ULAs includes an additional coupled component, i.e., $\mathbf{H}_c$, which cannot be decomposed into the multiplication of transmitter-side and receiver-side array response vectors. Due to the presence of this coupled component, near-field LoS MIMO channels exhibit higher DoFs than far-field LoS MIMO channels.
\end{remark}

For two parallel ULAs, the DoFs of the near-field LoS MIMO channel can be calculated through diffraction theory or eigenfunction analysis, which leads to \cite{miller2000communicating}:
\begin{equation}\label{dof_ula}
\mathrm{DoF}_{\mathrm{near}}^{\mathrm{ULA}} = \min\Big\{\frac{(N_T-1)d_T(N_R-1)d_R}{\lambda r}, N_T, N_R \Big\}.
\end{equation}
As can be observed, if the numbers of transmit and receive antennas are large enough, the DoFs between parallel ULAs are given by $\frac{(N_T-1)d_T(N_R-1)d_R}{\lambda r}$.

\begin{table*}[!h]
    \caption{Comparison between near-field and far-field channel models for SPD antennas.}
    \label{table:far_vs_near}
    \small
    \centering
    \begin{tabular}{!{\vrule width1pt}c!{\vrule width1pt}c!{\vrule width1pt}c!{\vrule width1pt}c!{\vrule width1pt}c!{\vrule width1pt}l!{\vrule width1pt}}
    \Xhline{1pt} 
    \textbf{System}         & \textbf{Channel Model}                          & \textbf{Category}             & \textbf{Main factors}               & \textbf{DoFs}    &  \textbf{Characterisitic}                      \\ \Xhline{1pt} 
    \multirow{3}{*}{MISO}   & Far-field                                       & LoS/NLoS                      & angle                               & $1$              &  \makecell[l]{Linear phase}                         \\ \cline{2-6}
                            & Near-field                                      & LoS/NLoS                      & angle, distance                     & $1$              &  \makecell[l]{Non-linear phase\\Involve the far-field model as a special case} \\ \Xhline{1pt} 
    \multirow{4}{*}{MIMO}   & \multirow{2}{*}{\makecell[c]{Far-field}}         & LoS                          & \multirow{2}{*}{angle}              & $1$                &\multirow{2}{*}{\makecell[l]{Linear phase\\Low DoFs of LoS channels}} \\ \cline{3} \cline{5}
                            &                                                  & NLoS                         &                                     & $\ge 1$          &              \\ \cline{2-6}
                            & \multirow{2}{*}{\makecell[c]{Near-field}}        & LoS                          & \multirow{2}{*}{angle, distance}    & $\ge 1$          & \multirow{2}{*}{\makecell[l]{Non-Linear phase\\High DoFs of LoS channels}}  \\ \cline{3} \cline{5}
                            &                                                   & NLoS                        &                                     & $\ge 1$          &                       \\ \Xhline{1pt} 
    \end{tabular}
\end{table*}

$\bullet$ \textbf{\emph{Parallel UPAs:}}
The near-field LoS MIMO channel matrix for two parallel UPAs can be calculated in a similar manner as that for two parallel ULAs. Assume that the transmit UPA is deployed in the $xz$-plane and is composed of $N_T = N^x_T \times N^z_T$ antenna elements with spacing $d^x_T$ and $d^z_T$ along the $x$- and $z$- directions, and the receive UPA is parallel to the transmit UPA and is composed of $N_R = N^x_R \times N^z_R$ antenna elements with spacing $d^x_R$ and $d^z_R$ along the two directions. More particularly, $N^x_{i} = 2\tilde{N}^x_{i} + 1$ and $N^z_{i} = 2\tilde{N}^z_{i} + 1, \forall i \in \{T,R\}$. The antenna element indices of transmitter and receiver are denoted as $(m,n), \forall m \in \{-\tilde{N}^x_T,\dots \tilde{N}^x_T\}, n \in \{-\tilde{N}^z_T,\dots \tilde{N}^z_T\},$ and $(p,q), \forall p  \in \{-\tilde{N}^x_R,\dots \tilde{N}^x_R\}, p  \in \{-\tilde{N}^z_R,\dots \tilde{N}^z_R\}$, respectively. Then, we have the following results:
\begin{center}
\begin{tcolorbox}[title = Near-Field LoS MIMO Channel for Parallel UPAs]
{
\begin{subequations} \label{19a}
\begin{align}
&\mathbf{H}_{\mathrm{UPA}}^{\mathrm{LoS}} = \tilde{\beta} \mathbf{a}_{\mathrm{UPA}}^R(\theta, \phi, r) (\mathbf{a}_{\mathrm{UPA}}^T(\theta, \phi, r))^{\mathsf{T}} \odot (\mathbf{H}_c^x \otimes \mathbf{H}_c^z), \\
&[\mathbf{H}_c^x]_{q,n} = e^{-j \frac{ 2 \pi }{\lambda r} n d^x_T q d^x_R (1 - \cos^2 \theta \sin^2 \phi)}, \\
&[\mathbf{H}_c^z]_{p,m} = e^{-j \frac{ 2 \pi}{\lambda r} m d^z_T p d^z_R \sin^2 \phi}.
\end{align}
\end{subequations}
}\end{tcolorbox}
\end{center}
Similar to the case of parallel ULAs, the above near-field LoS MIMO channel matrix for parallel UPAs also involves a coupled component, i.e., $\mathbf{H}_c^x \otimes \mathbf{H}_c^z$, thus resulting in high DoFs. 
For two parallel UPAs, the DoFs are given by \cite{miller2000communicating}:
\begin{equation}\label{dof_upa}
    \mathrm{DoF}_{\mathrm{near}}^{\mathrm{UPA}} = \min\Big\{\frac{2A_TA_R}{(\lambda r)^2}, N_T, N_R \Big\},
\end{equation}
where $A_T = (N_T^x-1)d_T^x(N_T^z-1)d_T^z$ and $A_R = (N_R^x-1)d_T^x(N_R^z-1)d_T^z$ are the apertures of the transmitting and receiving UPAs, respectively. Similar to the case of the parallel ULAs, if $N_T$ and $N_R$ are sufficiently large, the DoFs are given by $\frac{2A_TA_R}{(\lambda r)^2}$. 

The main differences between near-field and far-field channels are summarized in Table \ref{table:far_vs_near}.

\subsection{Non-Uniform Channel Model for SPD Antennas} \label{sec_non_uniform}
In the previous subsection, we reviewed the near-field channel model based on spherical waves and highlighted its major differences compared with the far-field channel model. Recall that the near-field channel coefficient between a transmit antenna $\mathbf{s}_n$ and a receive antenna $\mathbf{r}_m$ is given by 
\begin{equation}
    h(\mathbf{s}_n, \mathbf{r}_m) = \beta_{m,n} e^{-j \frac{2\pi}{\lambda} \|\mathbf{r}_m - \mathbf{s}_n\|},
\end{equation}  
In the previous subsection, we assumed that the propagation distance $r = \|\mathbf{r}_0 - \mathbf{s}_0\|$ with respect to the central elements of the antenna arrays is larger than the \emph{uniform-power distance}, resulting in negligible variations of the channel gains of different links. However, when $r$ is smaller than then \emph{uniform-power distance} or the antenna aperture is extremely large, the channel gain variations can no longer be neglected. In this case, more accurate channel gain models of near-field channels are required. In the following, we discuss different channel gain models valid under different assumptions.

\subsubsection{USW Model \cite{starer1994passive, bjornson2021primer}} We first briefly review the USW model defined in the previous section. In this model, the propagation distance $r$ is larger than the \emph{uniform-power distance}, resulting in \emph{uniform} channel gains, i.e., $\beta_{m,n} \approx \beta_{0,0}, \forall m,n$. The channel gain $\beta_{0,0}$ is mainly determined by the free-space path loss, which is given by $\beta_{0,0} = \frac{1}{\sqrt{4 \pi r^2}}$. Therefore, the USW model for near-field channels is given by 
\begin{center}
    \begin{tcolorbox}[title = USW Model of Near-Field Channels]
    {\setlength\abovedisplayskip{2pt}
    \setlength\belowdisplayskip{2pt}
    \begin{equation}\label{USW_Model_Channel_Coefficient}
        h_{\mathrm{U}}(\mathbf{s}_n, \mathbf{r}_m) = \frac{1}{\sqrt{4 \pi r^2}} e^{-j \frac{2\pi}{\lambda} \|\mathbf{r}_m - \mathbf{s}_n\|}.
    \end{equation}  
    }\end{tcolorbox}
\end{center}

\subsubsection{NUSW Model \cite{zhou2015spherical, friedlander2019localization}} In this model, the propagation distance $r$ is smaller than the \emph{uniform-power distance}, where the channel gain variations are not negligible. As a result, the channel gains of different links are \emph{non-uniform} and have to be calculated separately. More specifically, the channel gain can also be calculated according to the free-space path loss, which leads to $\beta_{m,n} = \frac{1}{\sqrt{4\pi \| \mathbf{r}_m - \mathbf{s}_n\|^2}}$. The NUSW model for near-field channels is given as follows:
\begin{center}
    \begin{tcolorbox}[title = NUSW Model of Near-Field Channels]
    {\setlength\abovedisplayskip{2pt}
    \setlength\belowdisplayskip{2pt}
    \begin{equation}\label{NUSW_Model_Channel_Coefficient}
        h_{\mathrm{N}}(\mathbf{s}_n, \mathbf{r}_m) = \frac{1}{\sqrt{4 \pi \| \mathbf{r}_m - \mathbf{s}_n\|^2}} e^{-j \frac{2\pi}{\lambda} \|\mathbf{r}_m - \mathbf{s}_n\|}.
    \end{equation}  
    }\end{tcolorbox}
\end{center}

\subsubsection{A General Model} Although the NUSW model is more accurate than the USW model within the \emph{uniform-power distance}, it still fails to capture the loss in channel gain caused by the effective antenna aperture and polarization mismatch, especially when the antenna arrays are of considerable size. The effective antenna aperture characterizes how much power is captured from an incident wave\footnote{Effective aperture or effective area characterizes the received power of an antenna \cite{roederer2014ieee,Balanis2012,Balanis2015,bevelacqua2019antenna}. Assume that the incident wave has the same polarization as the receive antenna and is travelling towards the antenna in the antenna's direction of maximum radiation (the direction from which the most power would be received). Then, the effective aperture describes how much power is captured from a given incident wave. Let $p_0$ be the power density of the incident wave (in ${\text{Watt}}/{\text{m}}^2$). Then, the antenna's received power (in Watts) is given by $p_0A_e$. An antenna's effective area or aperture is defined for reception. However, due to reciprocity, an antenna's directivities for reception and transmission are identical, so the power transmitted by an antenna in different directions is also proportional to the effective area \cite{roederer2014ieee,Balanis2012,Balanis2015,bevelacqua2019antenna}.}, and the polarization mismatch means the angular difference in polarization between the incident wave and receiving antenna \cite{Balanis2012,Balanis2015}. To this end, we introduce a general model for near-field channels, which has three components for the channel gain, namely \emph{free-space path loss}, \emph{effective aperture loss} $G_1$, and \emph{polarization loss} $G_2$. In the following, we explain how to calculate $G_1$ and $G_2$. Moreover, we assume that the transmit antenna has a unit effective area and the receive antenna is a hypothetical isotropic antenna element whose effective area is given by $\frac{\lambda^2}{4\pi}$ \cite{roederer2014ieee,Balanis2012,Balanis2015}.
\begin{itemize}
\item \textbf{Effective Aperture Loss}: As the signals sent by different array elements are observed by the receiver from different angles, the resulting effective antenna area varies over the array. The effective antenna area equals the product of the maximal value of the effective area and the projection of the array normal to the signal direction. Let ${\hat{\mathbf{u}}}_{\mathbf{s}}$ denote the normalized normal vector of the transmitting array at point $\mathbf{s}_n$. For example, when the transmitting array is placed in the $x \textendash z$ plane, we have ${\hat{\mathbf{u}}}_{\mathbf{s}}=[0,1,0]^{\mathsf{T}}$. Then, the power gain due to the effective antenna aperture is given as follows~\cite{Lu2022communicating}:
\begin{align}\label{p_proj_general}
    G_1(\mathbf{s}_n, \mathbf{r}_m) =\frac{(\mathbf{r}_m-\mathbf{s}_n)^{\mathsf{T}}{\hat{\mathbf{u}}}_{\mathbf{s}}}{\| \mathbf{r}_m-\mathbf{s}_n \|}.
\end{align}

\item \textbf{Polarization Loss}: The polarization loss is also caused by the fact that the receiver sees the signals sent by different array elements from different angles. The power loss due to polarization is defined as the squared norm of the inner product between the receiving-mode polarization vector at the receive antenna and the transmitting-mode polarization vector of the transmit antenna.

\begin{lemma} \label{lemma:polar}
    The polarization gain factor for transmit antenna $\mathbf{s}$ and receive antenna $\mathbf{r}$ can be expressed as follows:
    \begin{align}\label{eq_proj}
        G_2(\mathbf{s}, \mathbf{r}) &= \frac{\left\lvert{\bm\rho}_{\mathrm w}^{\mathsf T}({\mathbf{r}})\mathbf{e}( \mathbf{s}, \mathbf{r} ) \right\rvert^2}{\left\lVert\mathbf{e}( \mathbf{s}, \mathbf{r} )\right\rVert^2}, \\
        \mathbf{e}( \mathbf{s}, \mathbf{r} ) &= \left({\mathbf{I}}-\frac{({\mathbf r}-{\mathbf s})({\mathbf r}-{\mathbf s})^{\mathsf T}}{\lVert{\mathbf r}-{\mathbf s}\rVert^2}\right)\hat{\mathbf J}({\mathbf{s}}),
    \end{align}
    where ${\bm\rho}_{\mathrm w}^{\mathsf T}({\mathbf{r}})$ denotes the normalized receiving-mode polarization vector of the receive antenna and $\hat{\mathbf J}(\mathbf{s})$ denotes the normalized electric current vector at the transmit antenna.
\end{lemma}

\begin{proof}
    Please refer to Appendix \ref{proof_polarization}.
\end{proof}
\begin{remark}
It is worth noting that the influence of polarization loss was also considered in \cite{9139337,Bjornson2020power}. Yet, the derived results apply when the receiving-mode polarization vector of the receive antenna and the normalized electric current vector is along the $y$-axis. By contrast, \eqref{eq_proj} applies to arbitrary ${\bm\rho}_{\mathrm w}({\mathbf{r}})$ and $\hat{\mathbf J}(\mathbf{s})$, which yields more generality.
\end{remark}
\end{itemize}

\begin{figure}[!t]
    \centering
    \includegraphics[width=0.4\textwidth]{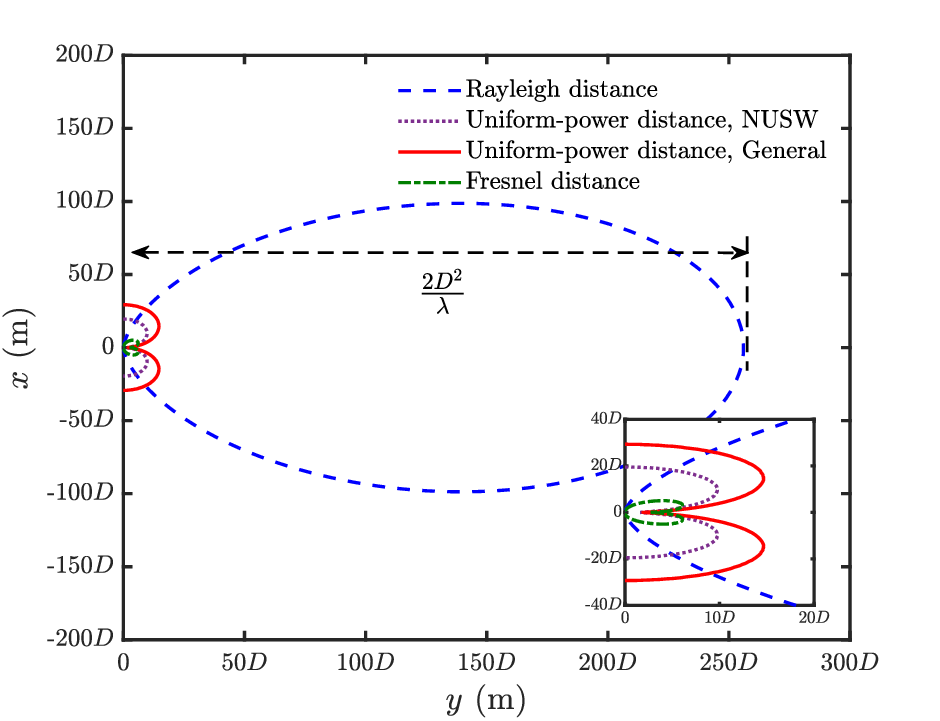}
    \caption{Illustration of Rayleigh distance, uniform-power distances, and Fresnel distance, where the BS is equipped with a ULA with $N=257$ antennas and operates at a frequency of $28$ GHz. The antenna spacing is set to $d = \frac{\lambda}{2} = 0.54$ cm. The aperture of the ULA is $D = (N-1)d = 1.37$ m}
    \label{fig:distance}
\end{figure}

Taking into account the \emph{free-space path loss}, \emph{effective aperture loss} $G_1$, and \emph{polarization loss} $G_2$, the general model for the near-field channel coefficients is given by 

\begin{center}
\begin{tcolorbox}[title = A General Model of Near-Field Channels]
{
\begin{equation}\label{proposed_h}
h_{\mathrm{G}}(\mathbf{s}_n, \mathbf{r}_m)= \sqrt{ \frac{ G_1(\mathbf{s}_n, \mathbf{r}_m) G_2(\mathbf{s}_n, \mathbf{r}_m) }{4 \pi \| \mathbf{r}_m - \mathbf{s}_n\|^2} } e^{-j\frac{2\pi}{\lambda}\lVert{\mathbf{r}_m - \mathbf{s}_n}\rVert}.
\end{equation}
}\end{tcolorbox}
\end{center}
Note that all three loss components are functions of the position of the transmit antenna, $\mathbf{s}_n$, and the position of the receive antenna, $\mathbf{r}_m$. In fact, the above-mentioned USW and NUSW channel models are special cases of the general model given in \eqref{proposed_h}, which is further explained in detail in Section~\ref{Sec:per}.

\subsubsection{Uniform-Power Distance}\label{sub:UPD}
According to the previous discussion, the \textbf{\emph{uniform-power distance}} is an important figure of merit distinguishing the region where the USW model is sufficiently accurate compared with the NUSW model and general model. In particular, the uniform-power distance can be defined based on the ratio of the weakest and strongest channel gains of the NUSW model or the general model. Let us take the general model as an example, where the channel gains of the links are given by
\begin{equation}
    \beta_{m,n}^{\mathrm{G}} = \sqrt{\frac{ G_1(\mathbf{s}_n, \mathbf{r}_m) G_2(\mathbf{s}_n, \mathbf{r}_m) }{4 \pi \| \mathbf{r}_m - \mathbf{s}_n\|^2} }.
\end{equation}
Then, the uniform-power distance $r_{\mathrm{UPD}}$ can be defined as follows \cite{sherman1962properties, 9617121}:
\begin{subequations}
    \begin{align}
        r_{\mathrm{UPD}} = &\arg\min_{r} \quad r \\
        &\mathrm{s.t.} \quad \frac{\min_{m,n} \beta_{m,n}^{\mathrm{G}} }{\max_{m,n} \beta_{m,n}^{\mathrm{G}}} \ge \Gamma,
    \end{align}
\end{subequations}
where $\Gamma$ is the minimum threshold for the ratio. Generally, the value of $\Gamma$ is selected to be slightly smaller than $1$. In this case, when $r \ge r_{\mathrm{UPD}}$, all channel gains $\beta_{m,n},\forall m,n,$ have comparable values, indicating that the USW model is sufficiently accurate.

Additionally, the uniform-power distance $r_{\mathrm{UPD}}$ can also be approximately calculated based on the NUSW model, for which the channel gain is given as follows:
\begin{equation}
    \beta_{m,n}^{\mathrm{N}} = \frac{ 1 }{\sqrt{4 \pi \| \mathbf{r}_m - \mathbf{s}_n\|^2}}.
\end{equation} 
In Fig. \ref{fig:distance}, we illustrate the Rayleigh distance, uniform-power distance, and Fresnel distance with respect to a BS equipped with a ULA. In particular, the uniform-power distance is calculated by setting $\Gamma=0.95$. It can be observed that the Fresnel distance is much smaller than the Rayleigh distance, which indicates that the \emph{Fresnel approximation} is sufficiently accurate in most of the near-field region. Additionally, it is important to note that the uniform-power distance is not always smaller than the Rayleigh distance since they are defined based on different criteria: channel gain variations and phase errors, respectively. By further comparing the conventional NUSW model and the general model, the uniform-power distance obtained by the general model, which is more accurate, is slightly larger. Finally, it can be concluded that the USW model, which is based on Fresnel approximation, can adequately address most scenarios occurring in communication networks, as both the uniform-power distance and Fresnel distance are in close proximity to the BS.

\subsection{Green's Function-Based Channel Model for CAP Antennas}

Near-field channel modelling for CAP antennas is much more challenging than that for SPD antennas.
In this subsection, we consider the scenario where both transmitter and receiver are equipped with CAP antennas, which is an analogy to the MIMO scenario for SPD antennas.
In contrast to the case of SPD antennas, CAP antennas support a continuous distribution of source currents, denoted by $\mathbf{J}(\mathbf{s})$, where $\mathbf{s}$ is the source point within the transmitting volume $V_T$.
The electric radiation field $\mathbf{E}(\mathbf{r})$ can be formulated as follows \cite{9433568}:
\begin{equation}\label{erg}
    \mathbf{E}(\mathbf{r}) = \int_{V_T} \mathbf{G}(\mathbf{s},\mathbf{r}) \mathbf{J}(\mathbf{s}) \mathrm{d}\mathbf{s},
\end{equation}
where $\mathbf{G}(\mathbf{s},\mathbf{r})$ is the tensor Green's function and $\mathbf{r}$ denotes the field point (location of the receiver).
For simplicity, we consider the case where the wireless signal is vertically polarized.
As a result, the equivalent electric currents induced within $V_T$ are in $y$-direction, i.e., $\mathbf{J}(\mathbf{s})=\hat{\mathbf{u}}_yJ_y(\mathbf{s})$ and the electric field they generate at the receiver is $\mathbf{E}(\mathbf{r}) = \hat{\mathbf{y}}E_y(\mathbf{r})$. According to \eqref{erg}, this received field is given as follows~\cite{xu2022modeling}:
\begin{equation}\label{ey}
    E_y(\mathbf{r}) = \int_{V_T} G_{yy}(\mathbf{s},\mathbf{r}) J_y(\mathbf{s}) \mathrm{d}\mathbf{s},
\end{equation}
where $G_{yy}$ is the $(y,y)$-th element of the Green's tensor. For free-space transmission, we have:
\begin{equation}\label{gyy}
    G_{yy}(\mathbf{s},\mathbf{r}) = -\left(j\omega\mu_0 + \frac{k_0^2}{j\omega\epsilon_0}\right)\frac{e^{-j \frac{2\pi}{\lambda}\lVert\mathbf{r}-\mathbf{s}\rVert}}{4\pi\lVert\mathbf{r}-\mathbf{s}\rVert},
\end{equation}
where $\mu_0$ and $\epsilon_0$ are the free space permeability and permittivity, respectively, $\omega$ is the angular frequency of the signal, and $k_0$ is the wave number.
It is worth noting that the Green's function in \eqref{gyy} can be further simplified using different approximations. In fact, the approximations used to arrive at the UPW, USW, and NUSW models for SPD antennas can also be applied to the Green's function for CAP antennas. We will further elaborate on this in Section~\ref{Section_SNR Analysis for CAP-Antennas}.
To obtain the channel gain between CAP antennas, we evaluate the overall received signal power over the receiving volume ($V_R$), i.e.,
\begin{equation}\label{A}
    |h_{\text{CAP}}|^2 = \int_{V_R} E_y^*(\mathbf{r})E_y(\mathbf{r})\mathrm{d}\mathbf{r}.
\end{equation}
Then, by substituting \eqref{ey} into \eqref{A} and denoting the normalized current distribution within $V_T$ by $\tilde{J}_y(\mathbf{s})$,
we obtain the end-to-end channel gain between two CAP antennas as follows~\cite{miller2000communicating}:
\begin{center}
\begin{tcolorbox}[title = Green's Function-Based Near-Field Channel Gain for CAP Antennas]
{
\begin{equation}\label{A_222}
    |h_{\text{CAP}}|^2 = \int_{V_T} \tilde{J}^*_y(\mathbf{s}_1) \int_{V_T} K(\mathbf{s}_1,\mathbf{s}_2)\tilde{J}_y(\mathbf{s}_2)\mathrm{d}\mathbf{s}_1 \mathrm{d}\mathbf{s}_2,
\end{equation}
}\end{tcolorbox}
\end{center}
where $<\tilde{J}_y(\mathbf{s}),\tilde{J}_y(\mathbf{s})> = \int_{V_t} \tilde{J}_y(\mathbf{s})\tilde{J}^*_y(\mathbf{s}) \mathrm{d}\mathbf{s} = 1$ holds, $V_R$ denotes the receiving volume, and
\begin{equation}\label{K_def}
    K(\mathbf{s}_1,\mathbf{s}_2) = \int_{V_R} G^*_{yy}(\mathbf{s}_1,\mathbf{r})G_{yy}(\mathbf{s}_2,\mathbf{r})\mathrm{d}\mathbf{r}.
\end{equation}

Due to the presence of multiple orthogonal current distributions across the transmit and receive apertures, we can encode data streams into these orthogonal currents to realize $D_{\text{CAP}}$ parallel channels~\cite{9139337}:
\begin{equation}
    y_n = h_n x_n + w_n,
\end{equation}
where $h_n$ is the channel coefficient of the $n$-th parallel channel, $x_n$ is the input symbol for the $n$-th channel which is associated with current distribution $\mathbf{J}_n(\mathbf{s})$, $w_n$ is the additive noise at the receiver, and $n\in \{1,\cdots,D_{\text{CAP}}\}$ is the index of the parallel channels.
Note that the relation between $|h_n|^2$ and $\mathbf{J}_n(\mathbf{s})$ is given by \eqref{A_222}. 
In order to determine the orthogonal currents $\mathbf{J}_n(\mathbf{s})$, which is necessary to exploit the corresponding DoF, we consider the following eigenvalue problem, where kernel function $K$ is a Hermitian operator:
\begin{equation}\label{eigen}
    |h_n|^2\cdot J_n(\mathbf{s}_1) =  \int_{V_T} K(\mathbf{s}_1,\mathbf{s}_2)J_n(\mathbf{s}_2)\mathrm{d}\mathbf{s}_2.
\end{equation}
Since $K(\mathbf{s}_1,\mathbf{s}_2)$ is a Hermitian operator, the eigenvalue problem in \eqref{eigen} has eigenfunctions $\{J_1, J_2, \cdots\}$, and the corresponding eigenvalues are $\{h_1, h_2, \cdots\}$.

The DoFs of the near-field channel model for CAP antennas are analyzed as follows.
The channel gains between two CAP antennas can be calculated as the eigenvalues of \eqref{eigen}. Thus, the DoFs are equal to the number of eigenfunctions ($J_n$) corresponding to non-negligible eigenvalues ($h_n$), so the resulting channel is useful. Particularly, for wireless communication between two rectangular prism CAP antennas, the available DoFs are given as follows~\cite{xu2022modeling}:
\begin{equation}\label{dof_cap}
    \mathrm{DoF}_{\mathrm{near}}^{\mathrm{CAP}} = \frac{2V_TV_R}{(\lambda r)^2 \Delta z_T \Delta z_R},
\end{equation}
where $r$ is the communication distance (defined as the distance between the centers of the two CAP antennas) and $\Delta z_{T/R}$ is the width of the transmitting/receiving volume of the CAP antennas. As can be observed from \eqref{dof_cap}, the DoFs of CAP antennas not only increase with the aperture size but also depend on the carrier frequency (inverse of wavelength) and the communication distance.

\subsection{Discussion and Open Research Problems}
In this section, we reviewed the most important near-field channel models and introduced general channel models for SPD and CAP antennas. 
The discussed near-field models are interrelated, specifically, the simplified USW models for ULAs and UPAs are derived by adopting the Fresnel approximation for the USW model. In addition, as shown in Fig.~\ref{regions}, the UPW and USW models are special cases of the NUSW model and the general model.
Thus, the UPW, USW, NUSW, and general models have increasing levels of accuracy and complexity.
Although accurate models for free-space deterministic near-field channels have been established in this section, 
see also \cite{starer1994passive, bjornson2021primer, zhou2015spherical, friedlander2019localization, miller2000communicating, 9433568, xu2022modeling}, statistical channel models for near-field multipath fading still remain an open problem.
Further research on statistical near-field channel modelling is required to fully describe the behaviour of near-field channels. Some of the key research directions are as follows:
\begin{itemize}
\item \textbf{Accurate and compact statistical channel models for SPD antennas}: New statistical models need to be developed which capture the complex dynamics of the near-field channel, such as the impact of obstacles, reflections, and diffractions. However, it is difficult to explicitly model each multipath component of a near-field NLoS channel. Compact statistical channel models are needed to capture both the multipath effect and the near-field effects. Another challenge is to develop accurate models for the reactive near-field region, where evanescent waves are dominant.
\item \textbf{Statistical channel models for CAP antennas}: For CAP antennas, existing channel models are based on the Green's function method~\cite{9433568,xu2022modeling}. However, the Green's function method is non-trivial to use in scattering environments. This is because the explicit modelling of the signal sources is difficult for the multipath components caused by scatterers. Developing statistical channel models for CAP antennas remains an open problem. 
\item \textbf{Validation of existing models by channel measurements}: Further validation of statistical channel models for NFC is required using empirical measurements~\cite{9206044}. Specifically, channel measurements involve measuring the received signal strength, time delay, and phase shift of the signals. Overall, validating and verifying channel models using channel measurements is an iterative process that requires careful experimentation, analysis, and comparison with the discussed models.
\end{itemize}

\section{Near-Field Beamfocusing and \\Antenna Architectures}
In wireless communications, beamforming is used to enhance signal strength and quality by directing the signal towards the intended receiver using an array of antennas. This requires adjusting the phase and amplitude of each antenna element to create a constructive radiation pattern, rather than radiating the signal uniformly. Compared with FFC beamforming, NFC introduces a new beamforming paradigm, referred to as \emph{beamfocusing}. In this section, we present the properties of near-field beamfocusing and then discuss various antenna architectures for narrowband and wideband communication systems, which is followed by an introduction to near-field beam training.

\begin{figure}[!t]
    \centering
    \subfigure[$N=65$]{
        \includegraphics[width=0.245\textwidth]{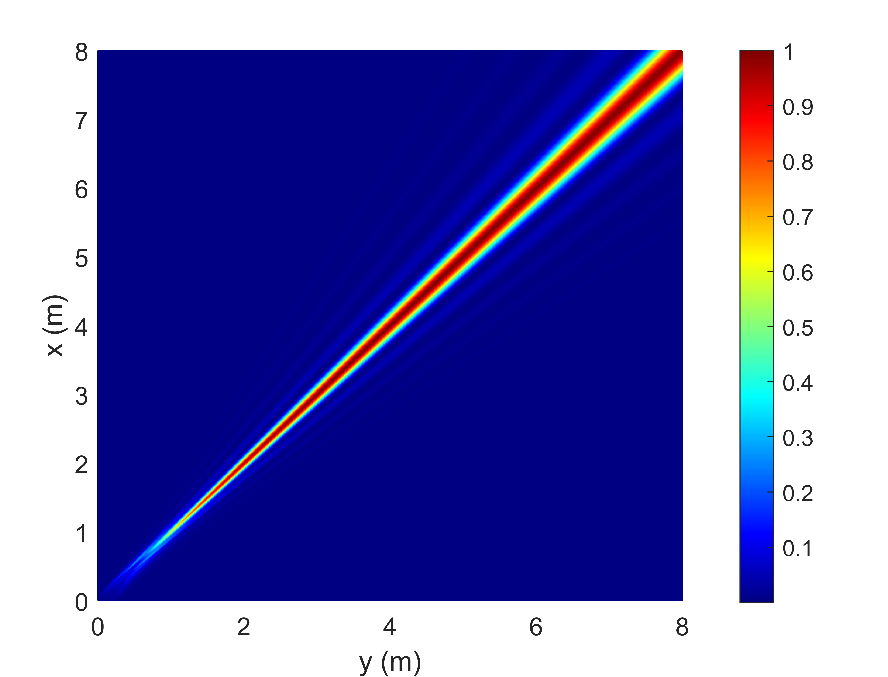}
    }%
    \subfigure[$N=129$]{
        \includegraphics[width=0.245\textwidth]{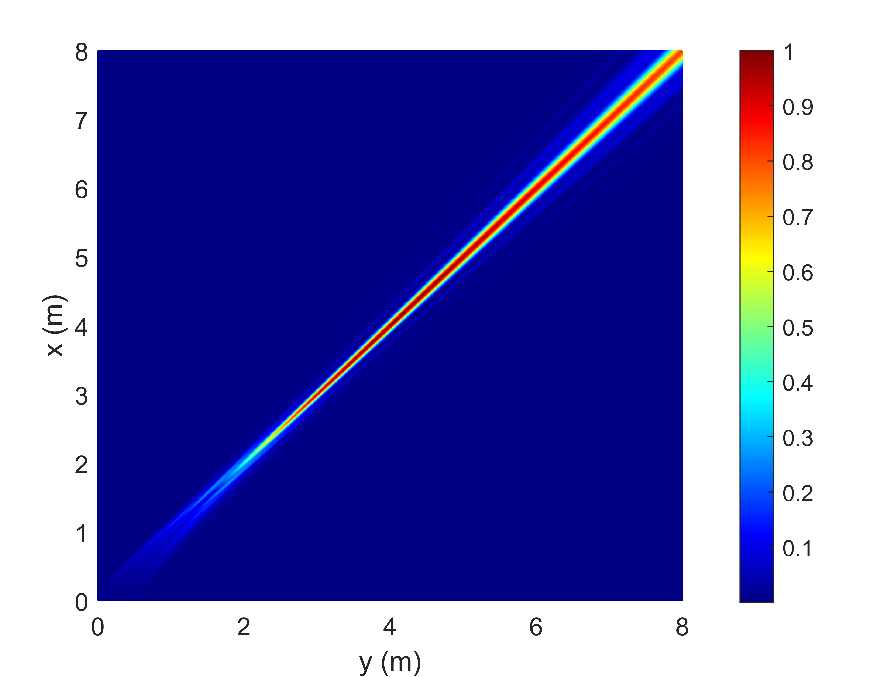}
    }
    \subfigure[$N=257$]{
        \includegraphics[width=0.245\textwidth]{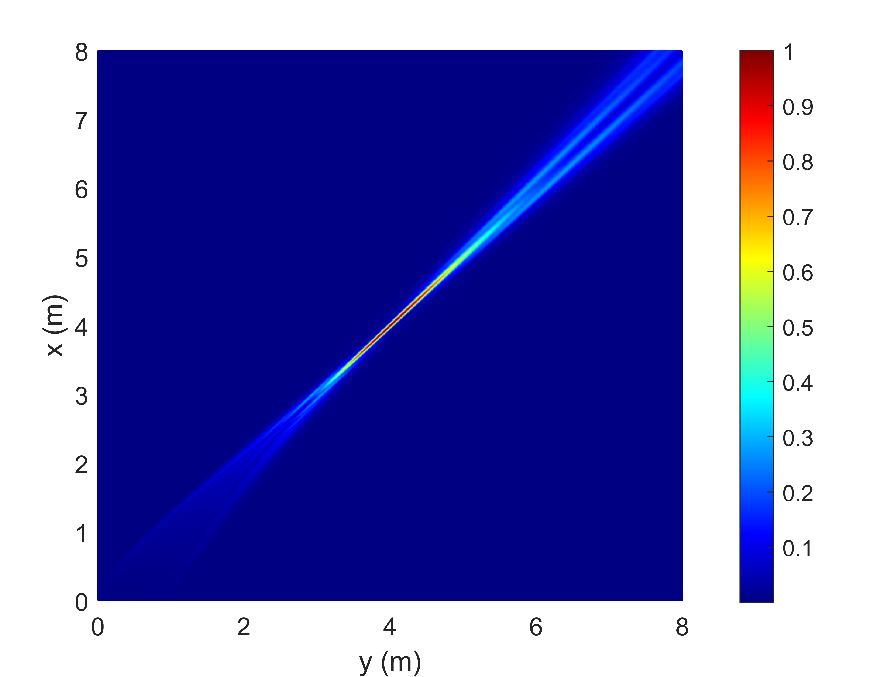}
    }%
    \subfigure[$N=513$]{
        \includegraphics[width=0.245\textwidth]{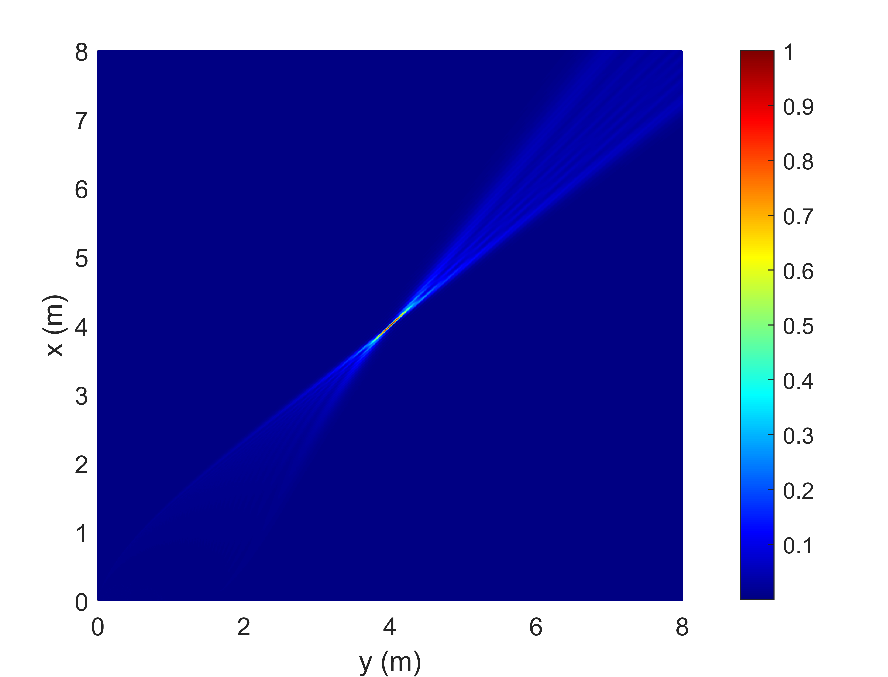}
    }
    \caption{Correlation of array response vectors for different sizes of the antenna array with half-wavelength antenna spacing in a system operating at frequency $28$ GHz.}
    \label{fig:beamfocusing}
\end{figure}

\subsection{Near-Field Beamfocusing}
The near-field array response vector under the spherical wave assumption depends on both the angle and distance between transmitter and receiver, c.f., \eqref{near_field_vector}, \eqref{eq_a_ula}, and \eqref{eq_a_upa}. 
By taking advantage of this property, NFC beamforming can be designed to act like a spotlight, allowing focusing on a specific location in the polar domain. This is known as \emph{beamfocusing}, and is different from FFC beamforming. In FFC, beamforming can only be used to steer the transmitted signal in a specific direction in the angular domain, similar to a flashlight, which is known as \emph{beamsteering}. To further elaborate, we consider a ULA with $N$ transmit antenna element, where $N = 2 \tilde{N}+1$. For illustration, we take the USW model in \eqref{eq_a_ula} as an example. In this case, the $n$-th element of the near-field array response vector, where $n \in \{-\tilde{N}, \dots, \tilde{N}\}$, can be written as 
\begin{equation} \label{eqn:beamfocusing_vector}
    [\mathbf{a}(\theta, r)]_n = e^{-j \frac{2 \pi}{\lambda} (- n d \cos \theta + \frac{n^2 d^2 \sin^2 \theta }{2r})}.
\end{equation}
By taking the above array response vector model as an example, in the following, we introduce two important properties of near-field beamfocusing, namely \emph{asymptotic orthogonality} and \emph{depth of focus}.

\subsubsection{Asymptotic Orthogonality}
The near-field array response vector demonstrates an asymptotic orthogonality \cite{wu2022multiple}, which implies that the correlation between two array response vectors tends to zero when the number of antenna elements $N$ is sufficiently large. Mathematically, this can be expressed as follows:
\begin{align} \label{eqn:orthogonality}
    \lim_{N \rightarrow +\infty} \frac{1}{N} |\mathbf{a}^\mathsf{T}(\theta_1, r_1) \mathbf{a}^*(\theta_2, r_2)| = 0, \nonumber \\ \text{for } \theta_1 \neq \theta_2 \text { or } r_1 \neq r_2.
\end{align}
In Fig. \ref{fig:beamfocusing}, we depict the correlation of array response vector $\mathbf{a}(\theta, r)$ with all other possible array response vectors in the two-dimensional space when $\theta=\frac{\pi}{4}$ and $r = 5.5$ m. As observed, as $N$ increases, vector $\mathbf{a}(\theta, r)$ gradually becomes orthogonal to all other array response vectors, particularly in the distance domain. This asymptotic orthogonality is fundamental to near-field beamfocusing. For instance, consider a two-user communication system, where two single-antenna users are located at $(\theta_1, r_1)$ and $(\theta_2, r_2)$, respectively. By assuming an LoS-only communication channel, the received signal at user $1$ can be expressed as 
\begin{equation}
    y_1 = \beta_1 \mathbf{a}^\mathsf{T}(\theta_1, r_1) \mathbf{f}_1 s_1 + \beta_1 \mathbf{a}^\mathsf{T}(\theta_1, r_1) \mathbf{f}_2 s_2 + n_1,
\end{equation}
where $\beta_1$ is the channel gain, $\mathbf{f}_1$ and $\mathbf{f}_2$ represent the beamformers for the two users, $s_1$ and $s_2$ are the desired signals of the two users, and $n_1$ denotes the complex Gaussian noise at user $1$ with a power of $\sigma_1^2$. Then, the signal-to-interference-plus-noise ratio (SINR) for decoding $s_1$ at user $1$ is given by 
\begin{equation}
    \gamma_1 = \frac{ |\beta_1 \mathbf{a}^\mathsf{T}(\theta_1, r_1) \mathbf{f}_1|^2}{ |\beta_1 \mathbf{a}^\mathsf{T}(\theta_1, r_1) \mathbf{f}_2|^2 + \sigma_1^2 }.
\end{equation}
According to the asymptotic orthogonality in \eqref{eqn:orthogonality}, the beamformers can be designed as $\mathbf{f}_1 = \sqrt{\frac{P_1}{N}}\mathbf{a}^*(\theta_1, r_1)$ and $\mathbf{f}_2 =  \sqrt{\frac{P_2}{N}} \mathbf{a}^*(\theta_2, r_2)$, where $P_1$ and $P_2$ denotes transmit powers allocated to user $1$ and user $2$, respectively. Consequently, for a sufficiently large $N$, the SINR can be calculated as 
\begin{align} \label{eqn:SINR}
    \gamma_1  = & \frac{ \frac{|\beta_1|^2 P_1}{N} |\mathbf{a}^\mathsf{T}(\theta_1, r_1) \mathbf{a}^*(\theta_1, r_1)|^2}{ \frac{|\beta_1|^2 P_2}{N} |\mathbf{a}^\mathsf{T}(\theta_1, r_1) \mathbf{a}^*(\theta_2, r_2)|^2 + \sigma_1^2 } \nonumber \\
    = & \frac{|\beta_1|^2 P_1 }{ |\beta_1|^2 P_2 \frac{1}{N^2} |\mathbf{a}^\mathsf{T}(\theta_1, r_1) \mathbf{a}^*(\theta_2, r_2)|^2 + \frac{1}{N}\sigma_1^2 } \approx \frac{|\beta_1|^2 P_1 N }{\sigma_1^2 },
\end{align}
where the second equality stems from the fact that $|\mathbf{a}^\mathsf{T}(\theta_1, r_1) \mathbf{a}^*(\theta_1, r_1)| = N$ and the approximation in the last step is due to the asymptotic orthogonality in \eqref{eqn:orthogonality}. Drawing from the aforementioned analysis, we can deduce that in NFC, it is possible to focus the desired signal of a specific user precisely at the intended location without introducing interference to other users situated elsewhere. This implies that \emph{beamfocusing} can be achieved. Compared to far-field beamsteering, for near-field beamfocusing, the distance can also contribute to the asymptotic orthogonality of array response vectors. Thus, inter-user interference can also be effectively mitigated even if the users are located in the same direction.

\begin{figure}[!t]
    \centering
    \includegraphics[width=0.4\textwidth]{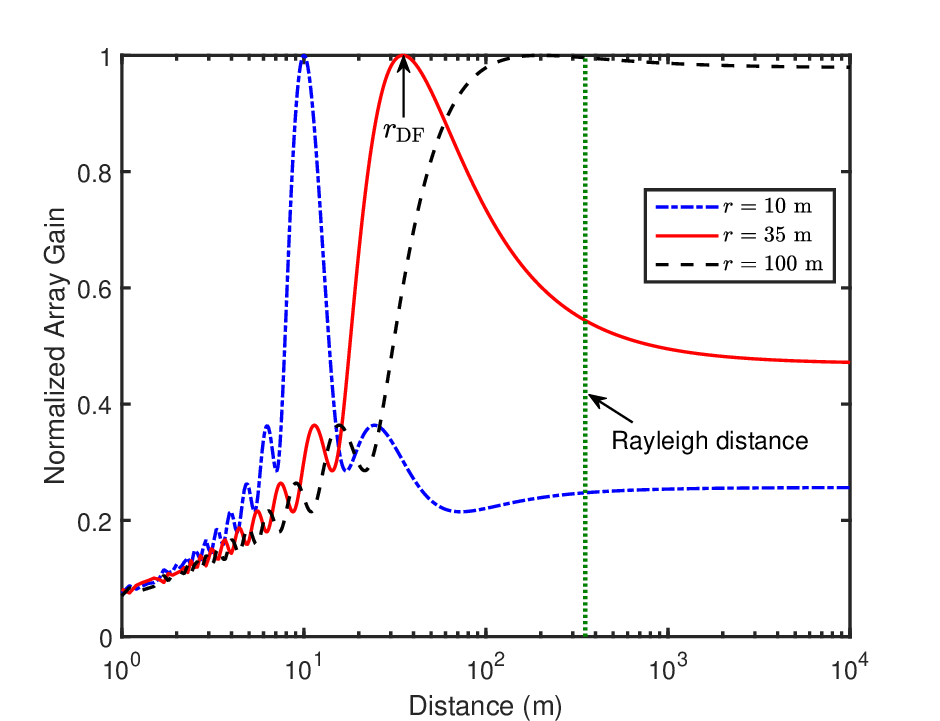}
    \caption{Depth of focus of a beamformer focused at difference distance $r$ in direction $\theta = \frac{\pi}{2}$. Here, we assume the BS is equipped with an ULA with $N = 257$ antennas and operates at a frequency of $28$ GHz. The antenna spacing is set to half-wavelength. Therefore, we have $r_{\mathrm{R}} \approx 350$ m and $r_{\mathrm{DF}} \approx 35$ m.}     \label{fig:depth}
\end{figure}

\subsubsection{Depth of Focus}
In the previous section, we discussed the asymptotic orthogonality when the number of antennas tends to infinity. However, in practice, the number of antennas is limited, which implies that the orthogonality between two different near-field array response vectors cannot be strictly achieved. Depth of focus is an important metric for evaluating the attainability of the orthogonality of near-field array response vectors in the distance domain \cite{nepa2017near, bjornson2021primer}. This characteristic sets near-field beamfocusing apart from far-field beamsteering.

Let us take the beamformer $\mathbf{f} = \sqrt{\frac{P}{N}}\mathbf{a}^*(\theta, r)$ as an example, where $P$ denotes the transmit power. We aim to find an interval $[r_{\min}, r_{\max}]$ of maximum length, such that for $r_0 \in [r_{\min}, r_{\max}]$, the following condition holds:
\begin{equation}
    \frac{1}{N} | \mathbf{a}^\mathsf{T}(\theta, r_0) \mathbf{a}^*(\theta, r) | \ge \Gamma_{\mathrm{DF}},
\end{equation}
where $\Gamma_{\mathrm{DF}}$ is a desired threshold. Then, the depth of focus is defined as the length of this interval, i.e., $\mathrm{DF} = r_{\max} - r_{\min}$. From the SINR perspective, when a user is located in the same direction but out of the depth of focus, the interference generated by beamformer $\mathbf{f}$ is relatively small. Consequently, a smaller depth of focus indicates better beamfocusing performance. Typically, the depth of focus is calculated based on a 3 dB criterion, i.e., $\Gamma_{\mathrm{DF}} = \frac{1}{2}$. The 3 dB depth of focus of beamformer $\mathbf{f}$ in different directions is given in the following lemma.

\begin{lemma} \label{lemma_focus}
    \textbf{(Depth of focus)} The 3 dB depth of focus of beamformer $\mathbf{f} = \sqrt{\frac{P}{N}}\mathbf{a}^*(\theta, r)$ in direction $\theta$ is given by
    \begin{equation}
        \mathrm{DF}_{3 \mathrm{dB}} = \begin{cases}
            \frac{2r^2 r_{\mathrm{DF}}}{r_{\mathrm{DF}}^2 - r^2 }, & r < r_{\mathrm{DF}}, \\
            \infty, & r \ge r_{\mathrm{DF}},
        \end{cases}
    \end{equation}
    where $r_{\mathrm{DF}} \approx \frac{N^2 d^2 \sin^2 \theta}{2 \lambda \eta^2_{3 \mathrm{dB}}}$ and $\eta_{3 \mathrm{dB}} = 1.6$. 
\end{lemma}

\begin{proof}
    Please refer to Appendix \ref{proof_focus}.
\end{proof}

From Lemma \ref{lemma_focus}, we observe that the depth of focus tends to infinity if the focus distance $r$ is larger than the threshold $r_{\mathrm{DF}}$, as shown in Fig. \ref{fig:depth}. This implies that beamfocusing degenerates to beamsteering, since the orthogonality in the distance domain is almost lost. Therefore, the region within distance $r_{\mathrm{DF}}$ is referred to as the \emph{focusing region}, where beamfocusing is achievable. 
To obtain more insights, we compare $r_{\mathrm{DF}}$ with the Rayleigh distance. Recall that for $\theta = \frac{\pi}{2}$, the Rayleigh distance is given by
\begin{equation}
    r_{\mathrm{R}} = \frac{2D^2}{\lambda},
\end{equation}
where $D = (N-1)d$ denotes the aperture of the ULA. When $\theta = \frac{\pi}{2}$, $r_{\mathrm{DF}}$ is given by 
\begin{equation}
    r_{\mathrm{DF}} = \frac{N^2 d^2}{2 \lambda \eta^2_{3 \mathrm{dB}}} \approx \frac{D^2}{2 \lambda \eta^2_{3 \mathrm{dB}}} \approx \frac{1}{10} r_{\mathrm{R}}.
\end{equation}
The above result suggests that beamfocusing is not universally achievable in the near-field region. Instead, it can only be achieved within a limited fraction, specifically within one-tenth, of the near-field region. For example, consider a BS with a Rayleigh distance of $350$ m and a focusing region confined to just $35$ m according to the 3 dB depth of focus.
Therefore, ELAAs are crucial for near-field beamfocusing as they can realize both a large focusing region and a small depth of focus. When the number of antennas tends to infinity, it can be easily shown that $\mathrm{DF}_{3\mathrm{dB}} \rightarrow 0$ and $r_{\mathrm{DF}} \rightarrow \infty$, which implies optimal near-field beamfocusing can be achieved in the full-space.
In the following, we discuss beamfocusing for two different ELAA architectures, namely SPD and CAP antennas.

\subsection{Beamfocusing with SPD Antennas}
Conventionally, the signal processing in multi-antenna systems is carried out in the baseband, i.e., fully-digital signal processing. However, this is not possible for near-field beamfocusing in practice since ELAAs and extremely high carrier frequencies are needed, where the number of power-hungry radio-frequency (RF) chains has to be kept at a minimum \cite{heath2016overview}. On the other hand, purely analog processing causes a loss in performance compared to digital processing. As a remedy, hybrid beamforming has been proposed as a practical solution to address this issue. In hybrid beamforming, a limited number of RF chains are utilized by a low-dimensional digital beamformer, followed by a high-dimensional analog beamformer \cite{el2014spatially, yu2016alternating, shi2018spectral, sohrabi2016hybrid}. Generally, the analog components are power-friendly and easy to implement. In the following, we discuss different hybrid beamforming architectures for narrowband and wideband systems to facilitate near-field beamfocusing.

\begin{figure}[!t]
    \centering
    \subfigure[Fully-connected.]{
        \includegraphics[height=0.25\textwidth]{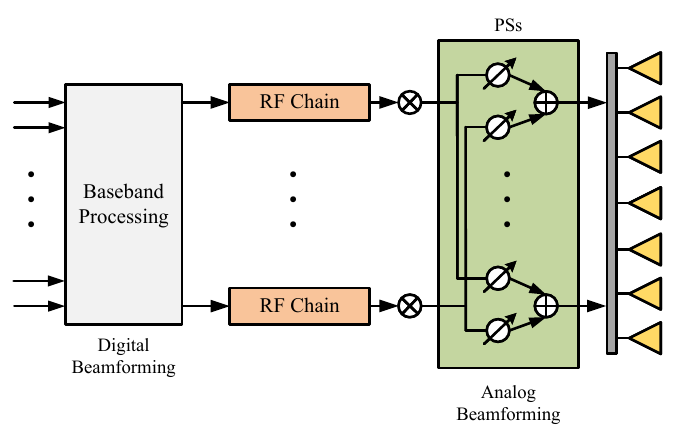}
        \label{fig:hybrid_beamforming_full}
    }
    \subfigure[Sub-connected.]{
        \includegraphics[height=0.25\textwidth]{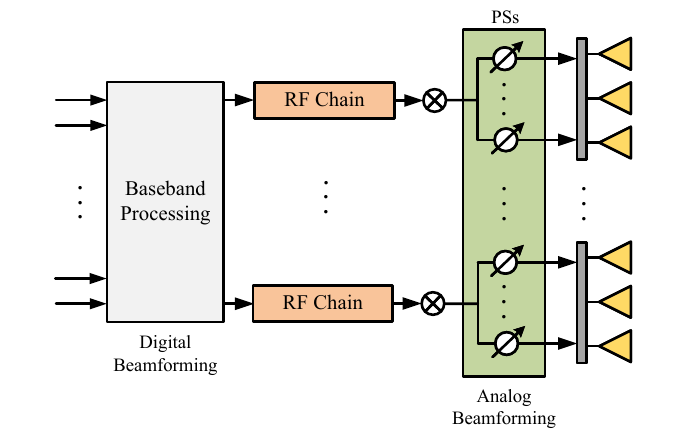}
        \label{fig:hybrid_beamforming_sub}
    }
    \caption{Architectures for PS-based hybrid beamforming.}
    \label{fig:hybrid_beamforming}
\end{figure}

\subsubsection{Narrowband Systems}
Let us consider a BS equipped with an $N$-antenna ULA, where $N = 2\tilde{N} + 1$, serving $K$ single-antenna communication users. There are only $N_{\mathrm{RF}}$ RF chains equipped at the BS, where $N_{\mathrm{RF}} \ll N$. In the context of narrowband flat fading, the discrete-time baseband model for the received signal at user $k$ is given by 
\begin{equation}
    y_k = \mathbf{h}_k^{\mathsf{T}} \mathbf{F}_{\mathrm{RF}} \mathbf{F}_{\mathrm{BB}} \mathbf{x} + n_k,
\end{equation}
where $\mathbf{x} \in \mathbb{C}^{K \times 1}$ contains the information symbols for the $K$ users and $n_k \sim \mathcal{CN}(0, \sigma_k^2)$ is additive Gaussian white noise (AWGN). $\mathbf{F}_{\mathrm{RF}} \in \mathbb{C}^{N \times N_{\mathrm{RF}}}$ and $\mathbf{F}_{\mathrm{BB}} \in \mathbb{C}^{N_{\mathrm{RF}} \times K}$ denote the analog beamformer and the baseband digital beamformer, respectively. $\mathbf{h}_k \in \mathbb{C}^{N \times 1}$ denotes the multipath NFC channel comprising one LoS path and $L_k$ resolvable NLoS paths for user $k$. According to \eqref{H_NFC_MISO} and \eqref{eqn:beamfocusing_vector}, adopting the USW model and ULAs, the multipath channel $\mathbf{h}_k$ is given as follows:
\begin{equation}
    \mathbf{h}_k = \beta_k \mathbf{a}(\theta_k, r_k) +  \sum_{\ell=1}^{L_k} \tilde{\beta}_{k,\ell} \mathbf{a}( \tilde{\theta}_{k,\ell}, \tilde{r}_{k,\ell}),
\end{equation} 
where $\beta_k$ and $\tilde{\beta}_{k,\ell}$ denote the complex channel gain of the LoS link and the $\ell$-th NLoS link, respectively, $(\theta_k, r_k)$ is the location of user $k$, and $(\tilde{\theta}_{k,\ell}, \tilde{r}_{k,\ell})$ is the location of the $\ell$-th scatterer for user $k$. In NFC, the analog beamformer realizes an array gain by generating beams focusing on specific locations, such as the locations of the users and scatterers, while the digital beamformer is designed to realize a multiplexing gain. The analog beamformer $\mathbf{F}_{\mathrm{RF}}$ is subject to specific constraints imposed by the hardware architecture of the analog beamforming network. In the following, we first briefly review the conventional phase-shifter (PS) based hybrid beamforming architecture.

In PS-based hybrid beamforming architectures, the RF chains can be either connected to all antennas (referred to as \emph{fully-connected} architectures) or a subset of antennas (referred to as \emph{sub-connected} architectures) via PSs, as shown in Fig. \ref{fig:hybrid_beamforming_full} and Fig. \ref{fig:hybrid_beamforming_sub}, respectively. The respective analog beamformers can be expressed as 
\begin{align}
    &\mathbf{F}_{\mathrm{RF}}^{(\mathrm{full})} = \big[ \mathbf{f}_{\mathrm{RF},1}^{(\mathrm{full})}, \dots, \mathbf{f}_{\mathrm{RF},N_\mathrm{RF}}^{(\mathrm{full})} \big], \\
    &\mathbf{F}_{\mathrm{RF}}^{(\mathrm{sub})} = \mathrm{blkdiag}\left( 
        \big[\mathbf{f}_{\mathrm{RF},1}^{(\mathrm{sub})},\dots,\mathbf{f}_{\mathrm{RF},N_\mathrm{RF}}^{(\mathrm{sub})}  \big]
     \right),
\end{align}
where $\mathbf{f}_{\mathrm{RF},n}^{(\mathrm{full})} \in \mathbb{C}^{N \times 1}$ and $\mathbf{f}_{\mathrm{RF},n}^{(\mathrm{sub})} \in \mathbb{C}^{\frac{N}{N_{\mathrm{RF}}} \times 1}$ denotes the analog beamformers for the $n$-th RF chain in the fully-connected and sub-connected architectures, respectively. The required number of PSs for these two architectures is $N_{\mathrm{RF}} N$ and $N$, respectively.
Recalling the fact that PSs can only adjust the phase of a signal, the unit-modulus constraint has to be satisfied, which implies
\begin{equation}
    \big| [\mathbf{f}_{\mathrm{RF},n}^{(\mathrm{full/sub})}]_m \big| = 1, \forall m,n.
\end{equation}
For the PS-based hybrid beamforming architecture, the achievable rate for each user $k$ is given by
\begin{align}
    R_k = \log_2 \left( 1 + \frac{ | \mathbf{h}_k^{\mathsf{T}} \mathbf{F}_{\mathrm{RF}} \mathbf{f}_{\mathrm{BB},k}|^2 }{ \sum_{i \neq k} | \mathbf{h}_k^{\mathsf{T}} \mathbf{F}_{\mathrm{RF}} \mathbf{f}_{\mathrm{BB},i}|^2 + \sigma_k^2} \right),
\end{align} 
where $\mathbf{f}_{\mathrm{BB},k}$ denotes the $k$-th column of $\mathbf{F}_{\mathrm{BB}}$. Then, the analog and digital beamformers can be designed to achieve different objectives with respect to $R_k$, such as maximizing spectral efficiency or energy efficiency. However, due to the coupling between $\mathbf{F}_{\mathrm{RF}}$ and $\mathbf{F}_{\mathrm{BB}}$ and the constant-modulus constraint on the elements of $\mathbf{F}_{\mathrm{RF}}$, it is not easy to directly obtain the optimal analog and digital beamformers. In the following, two alternative approaches are introduced to solve hybrid beamforming optimization problems, namely \emph{fully-digital approximation} \cite{heath2016overview, el2014spatially, yu2016alternating, shi2018spectral} and \emph{heuristic two-stage optimization} \cite{alkhateeb2015limited, zhu2016adaptive, zhang2020energy, wei2018multi}.
\begin{itemize}
    \item \textbf{Fully-digital Approximation}:
    This approach aims to minimize the distance between the hybrid beamformer and the unconstrained fully-digital beamformer $\mathbf{F}^{\mathrm{opt}}$, which can be effectively obtained via existing methods such as successive convex approximation (SCA) \cite{sun2016majorization}, weighted minimum mean square error (WMMSE) \cite{christensen2008weighted}, and fractional programming (FP) \cite{shen2018fractional}. Then, we only need to solve the following optimization problem, where we take the fully-connected hybrid beamforming architecture as an example with $P_{\max}$ denoting the maximum transmit power:
    \begin{center}
        \begin{tcolorbox}[title = {$\mathcal{P}_{\mathrm{HB}}$: Hybrid Beamforming Optimization}]
        \vspace{-5mm}
        \begin{align} \label{problem:hybrid beamforming}
        \underset{\begin{subarray}{c}
            \mathbf{F}_{\mathrm{RF}}, \mathbf{F}_{\mathrm{BB}}
            \end{subarray}}{\mathrm{min}} \quad &\big\| \mathbf{F}^{\mathrm{opt}} - \mathbf{F}_{\mathrm{RF}} \mathbf{F}_{\mathrm{BB}} \big\|_F^2 
        \\
        \mathrm{s. t.} \quad &\big|[\mathbf{F}_{\mathrm{RF}}]_{m,n} \big| = 1, \forall m,n, \nonumber \\
        &  \big\| \mathbf{F}_{\mathrm{RF}} \mathbf{F}_{\mathrm{BB}} \big\|_F^2 \le P_{\max}. \nonumber 
        \end{align}
        \vspace{-4mm}\par\noindent
    \end{tcolorbox}
    \end{center}
    The existing literature suggests that fully-digital approximation can achieve near-optimal performance \cite{heath2016overview}. In particular, several methods have been proposed to solve problem \eqref{problem:hybrid beamforming} through orthogonal matching pursuit (OMP) \cite{el2014spatially}, manifold optimization \cite{yu2016alternating}, and block coordinate descent (BCD) \cite{shi2018spectral}. Furthermore, it has been proved in \cite{sohrabi2016hybrid} that the Frobenius norm in problem \eqref{problem:hybrid beamforming} can be made exactly zero when the number of RF chains is not less than twice the number of information streams, i.e., $N_\mathrm{RF} \ge 2K$. However, it is important to note that the complexity of this approach can be exceedingly high, especially when the number of antennas is extremely large, e.g., $N = 512$. On the one hand, this approach requires the acquisition of the optimal $\mathbf{F}^{\mathrm{opt}}$, which has the potentially large dimension $N \times K$. On the other hand, the design of $\mathbf{F}_{\mathrm{RF}}$ by solving problem \eqref{problem:hybrid beamforming} also involves a large number of optimization variables. The resulting potentially high computational complexity can present a practical challenge in real-world applications.

    \item \textbf{Heuristic Two-stage Optimization}: This approach involves a two-step process for designing beamformers. Firstly, the analog beamformer $\mathbf{F}_{\mathrm{RF}}$ is implemented using a heuristic design to generate beams focused at specific locations. Following this, the digital beamformer $\mathbf{F}_{\mathrm{BB}}$ is optimized with reduced dimension $N_{\mathrm{RF}} \times K$. One popular solution for the analog beamformer is to design each column $\mathbf{F}_{\mathrm{RF}}$ such that the corresponding beam is focused at the location of one of the users through the LoS path. Hence, for $N_{\mathrm{RF}} = K$, the following closed-form solution can be obtained for the full-connected and sub-connected architectures, respectively:
    \begin{align}
        \mathbf{f}_{\mathrm{RF}, k}^{(\mathrm{full})} = &\mathbf{a}^*(\theta_k, r_k), \forall k, \\
        \mathbf{f}_{\mathrm{RF},k}^{(\mathrm{sub})} = &\left[ \mathbf{a}^*(\theta_k, r_k) \right]_{\frac{(k-1)N}{N_{\mathrm{RF}}} + 1 : \frac{k N}{N_{\mathrm{RF}}}}, \forall k.
    \end{align} 
    For the resulting $\mathbf{F}_{\mathrm{RF}}$, the equivalent channel $\mathbf{g}_k \in \mathbb{C}^{N_{\mathrm{RF}} \times 1}$ for baseband processing can be obtained as follows:
    \begin{equation}
        \mathbf{g}_k = \mathbf{F}_{\mathrm{RF}}^{\mathsf{T}} \mathbf{h}_k.
    \end{equation} 
    Then, the achievable rate of user $k$ is given as follows:
    \begin{equation}
        R_k(\mathbf{F}_{\mathrm{BB}}) = \log_2 \left( 1 + \frac{ | \mathbf{g}_k^{\mathsf{T}} \mathbf{f}_{\mathrm{BB},k}|^2 }{ \sum_{i \neq k} | \mathbf{g}_k^{\mathsf{T}} \mathbf{f}_{\mathrm{BB},i}|^2 + \sigma_k^2} \right).
    \end{equation} 
    As a result, the optimization of $\mathbf{F}_{\mathrm{BB}}$ can be regarded as a reduced-dimension fully-digital beamformer design problem. It can be solved with existing methods \cite{sun2016majorization, christensen2008weighted, shen2018fractional}. Compared to the fully-digital approximation approach, the heuristic two-stage approach exhibits much lower computational complexity due to the closed-form design of the analog beamformers and the low-dimensional optimization of the digital beamformers.
\end{itemize} 

\subsubsection{Wideband Systems}

\begin{figure}[!t]
    \centering
    \includegraphics[width=0.48\textwidth]{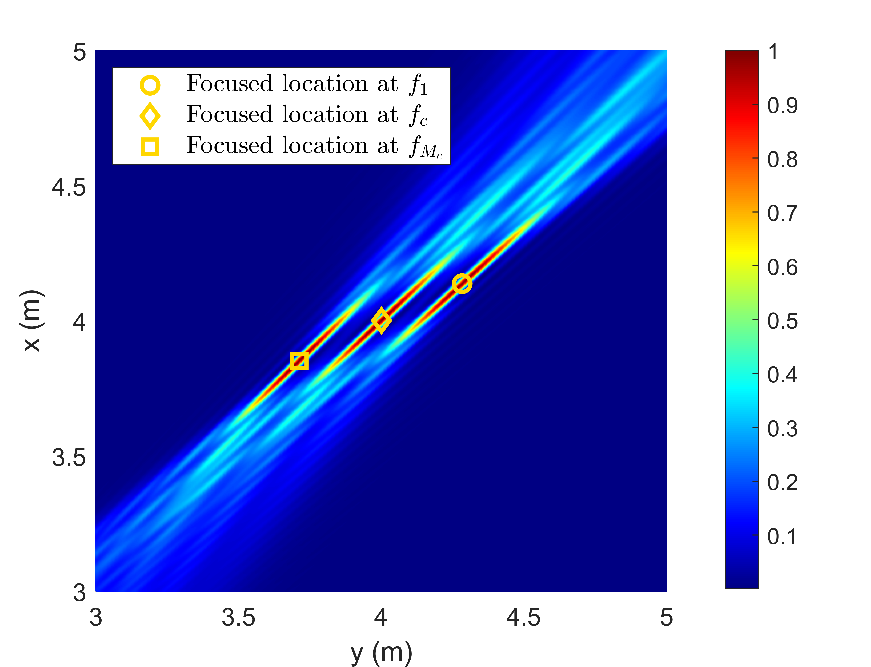}
    \caption{Near-field beam split in wideband OFDM system, where $N = 513$, $f_c = 28$ GHz, $W = 1$ GHz, and $M_c = 128$. The beam is generated such that it is focused on location $(\theta=\frac{\pi}{4}, r=5.5 \text{m})$ for the central frequency.}
    \label{fig:beamsplit}
\end{figure}

In wideband communication systems, orthogonal frequency division multiplexing (OFDM) is usually adopted to effectively exploit the large bandwidth resources and overcome frequency-selective fading. 
Let $W$, $f_c$, and $M_c$ denote the system bandwidth, central frequency, and the number of subcarriers of the OFDM system, respectively. The frequency for subcarrier $m$ is $f_m = f_c + \frac{W(2m-1-M_c)}{2M_c}$.
Adopting the USW model for ULAs, the frequency-domain channel for user $k$ at subcarrier $m$ after discrete Fourier transform (DFT) is given by \cite{tse2005fundamentals, cui2022near}
\begin{equation} \label{wideband_channel}
    \tilde{\mathbf{h}}_{m,k} = \beta_{m,k} \mathbf{a}(f_m, \theta_k, r_k) + \sum_{\ell=1}^{L_k} \tilde{\beta}_{m,k,\ell} \mathbf{a}(f_m, \tilde{\theta}_{k,\ell}, \tilde{r}_{k,\ell}),
\end{equation}
where $\beta_{m,k}$ and $\tilde{\beta}_{m,k,\ell}$ represent the complex channel gain for subcarrier $m$ of the LoS and NLoS paths, respectively.
We note that the in \eqref{wideband_channel}, near-field array response vector $\mathbf{a}(f_m, \theta, r)$ is rewritten as a function of the carrier frequency due to its frequency-dependence. More specifically, for the array response vector in \eqref{eqn:beamfocusing_vector}, the phase is inversely proportional to the signal wavelength $\lambda$, and thus it is proportional to the signal frequency $f$. In practice, the antenna spacing is usually fixed to half a wavelength of the central frequency $d = \frac{\lambda_c}{2} = \frac{c}{2 f_c}$, where $c$ denote the speed of light. Thus, the array response vector for subcarrier $m$ can be expressed as follows: 
\begin{align} \label{eqn:67}
    \mathbf{a}(f_m, \theta, r) = \left[e^{-j \pi \frac{f_m}{f_c} \delta_1(\theta, r) },\dots, e^{-j \pi \frac{f_m}{f_c} \delta_N(\theta, r)}\right]^{\mathsf{T}},
\end{align}
where $\delta_n(\theta, r) = -n \cos \theta + \frac{n^2d \sin^2 \theta }{2r}$. From \eqref{eqn:67}, it can be observed that the wideband near-field array response vector is \emph{frequency-dependent}, thus leading to \emph{frequency-dependent} communication channels for different OFDM subcarriers. Then, if the conventional PS-based hybrid beamforming architecture is used, the received signal at subcarrier $m$ of user $k$ is given by 
\begin{equation}
    \tilde{y}_{m,k} = \tilde{\mathbf{h}}_{m,k}^{\mathsf{T}} \mathbf{F}_{\mathrm{RF}} \mathbf{F}_{\mathrm{BB}}^m \tilde{\mathbf{x}}_{m,k} + \tilde{n}_{m,k},
\end{equation}  
where $\tilde{\mathbf{x}}_{m,k} \in \mathbb{C}^{K \times 1}$ and $\tilde{n}_{m,k} \sim \mathcal{CN}(0, \sigma_{m,k}^2)$ denote the vector containing the information symbols for the $K$ users and the AWGN at subcarrier $m$, respectively. $\mathbf{F}_{\mathrm{RF}} \in \mathbb{C}^{N \times N_{\mathrm{RF}}}$ and $\mathbf{F}_{\mathrm{BB}}^m \in \mathbb{C}^{N_{\mathrm{RF}} \times K}$ denote the analog beamformer and the low-dimensional digital beamformer for subcarrier $m$, respectively. It can be observed that the analog beamformer $\mathbf{F}_{\mathrm{RF}}$ realized by PSs is \emph{frequency-independent}, which causes a mismatch with respect to the \emph{frequency-dependent} wideband communication channel. This mismatch leads to the so-called \emph{near-field beam split} effect \cite{cui2023near, liu2023near}. In other words, a PS-based analog beamformer focusing on a specific location for one subcarrier will focus on different locations for other subcarriers. To elaborate on this phenomenon, we assume that one column $\mathbf{f}_{\mathrm{RF}}$ of $\mathbf{F}_{\mathrm{RF}}$ is designed such that the beam is focused on location $(\theta_c, r_c)$ at the central frequency $f_c$, i.e., 
\begin{equation}
    \mathbf{f}_{\mathrm{RF}} = \mathbf{a}^*(f_c, \theta_c, r_c).
\end{equation}  
Then, the impact of the near-field beam split effect is unveiled in the following lemma.
\begin{lemma} \label{lemma:beam spit}
    \textbf{(Near-field beam split \cite{cui2022near})} In wideband NFC, the beam generated by PS-based beamformer $\mathbf{f}_{\mathrm{RF}} = \mathbf{a}^*(f_c, \theta_c, r_c)$ is focused on location $(\theta_m, r_m)$ for subcarrier $m$, where
    \begin{equation}
        \theta_m = \arccos \left( \frac{f_c}{f_m} \cos \theta_c \right), \quad r_m = \frac{f_m \sin^2\theta_m}{f_c \sin^2\theta_c} r_c.
    \end{equation}
\end{lemma}

\begin{proof}
    Please refer to Appendix \ref{proof_split}.
\end{proof}

Lemma \ref{lemma:beam spit} implies that the beam generated by $\mathbf{f}_{\mathrm{RF}} = \mathbf{a}^*(f_c, \theta_c, r_c)$ will focus on location $(\theta_m, r_m)$ for subcarrier $m$, rather than the desired location $(\theta_c, r_c)$. The resulting \emph{beam split} effect is illustrated in Fig. \ref{fig:beamsplit}. If subcarrier frequency $f_m$ deviates significantly from the central frequency $f_c$, a substantial loss of array gain occurs due to the beam split effect. 

\begin{figure}[!t]
    \centering
    \subfigure[Fully-connected.]{
        \includegraphics[height=0.25\textwidth]{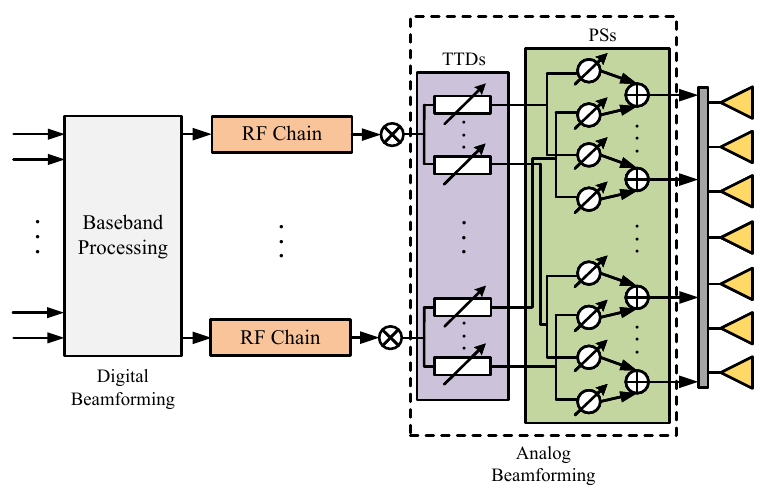}
        \label{fig:TTD_Hybrid_full}
    }
    \subfigure[Sub-connected.]{
        \includegraphics[height=0.25\textwidth]{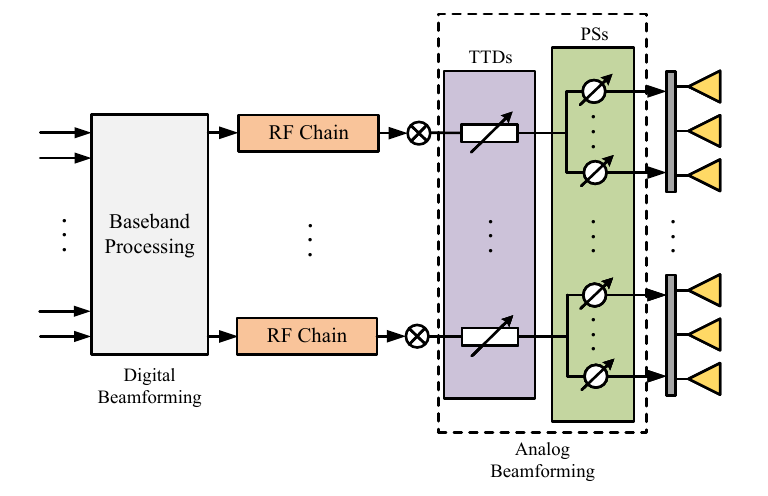}
        \label{fig:TTD_Hybrid_sub}
    }
    \caption{Architectures for TTD-based hybrid beamforming.}
    \label{fig:TTD_Hybrid}
\end{figure}

To mitigate the performance degradation induced by beam split, an efficient method is to exploit true time delayers (TTDs) \cite{carrasco2008bandwidth, rotman2014wideband, rotman2016true}, which, similar to PSs, are analog components. TTDs enable the realization of \emph{frequency-dependent} phase shifts by introducing variable time delays of the signals. Specifically, a time delay $t$ in the time domain corresponds to a \emph{frequency-dependent} phase shift of $e^{-j2\pi f_m t}$ at subcarrier $m$ in the frequency domain. Therefore, TTDs can be exploited to realize \emph{frequency-dependent} analog beamforming, which is a viable approach to mitigate the near-field beam split phenomenon. A straightforward way to implement such beamformers is to replace all PSs with TTDs in the conventional hybrid beamforming architecture. However, the cost and power consumption of TTDs is much higher than that of PSs, which makes this strategy impractical. As a remedy, several TTD-based hybrid beamforming architectures leveraging both PSs and TTDs have been proposed \cite{dovelos2021channel, gao2021wideband, dai2022delay, cui2021near}, as depicted in Fig. \ref{fig:TTD_Hybrid}\footnote{In Fig. \ref{fig:TTD_Hybrid}, we only show one possible architecture, where TTDs are arranged in a parallel manner. TTDs can also be arranged in a serial manner as shown in \cite{zhai2020thzprism}}. This architecture features an additional low-dimensional time delay network comprising a few TTDs, which is inserted between the digital and high-dimensional PS-based analog beamforming architectures. For the TTD-based hybrid beamforming architecture, the received signal at subcarrier $m$ of user $k$ is given by 
\begin{equation}
    \tilde{y}_{m,k} = \tilde{\mathbf{h}}_{m,k}^{\mathsf{T}} \mathbf{F}_{\mathrm{PS}} \mathbf{T}_m \mathbf{F}^m_{\mathrm{BB}} \tilde{\mathbf{x}}_m + \tilde{n}_{m,k},
\end{equation}  
where $\mathbf{F}_{\mathrm{PS}}$ denotes the \emph{frequency-independent} analog beamformer realized by PSs and $\mathbf{T}_m$ denotes the \emph{frequency-dependent} analog beamformer realized by TTDs.   

Similar to PS-based hybrid beamforming, TTD-based hybrid beamforming can be classified into two architectures:
\begin{itemize}
    \item \textbf{Fully-connected Architecture}: As illustrated in Fig. \ref{fig:TTD_Hybrid_full}, in the \emph{fully-connected} architecture, each RF chain is connected to all antenna elements through TTDs and PSs based on the following strategy. Each RF chain is first connected to $N_T$ TTDs, and then each TTD is connected to a subarray of $N/N_T$ antenna elements via $N/N_T$ PSs. Therefore, the number of PSs in the TTD-based hybrid beamforming architecture is the same as in the conventional PS-based hybrid beamforming architecture, i.e., $N_{\mathrm{PS}} = N_{\mathrm{RF}}N$. The number of TTDs in this architecture is given by $N_{\mathrm{TTD}} = N_{\mathrm{RF}} N_T$. The analog beamformers realized by the PSs and TTDs in the \emph{fully-connected} architecture can be expressed as, respectively, 
    \begin{align}
        &\mathbf{\mathbf{F}}_{\mathrm{PS}}^{(\mathrm{full})} = \big[\mathbf{F}_1^{(\mathrm{full})}, \dots, \mathbf{F}_{N_\mathrm{RF}}^{(\mathrm{full})}\big], \\
        &\mathbf{T}_m^{(\mathrm{full})} = \mathrm{blkdiag}\Big( [ e^{-j 2 \pi f_m \mathbf{t}_1^{(\mathrm{full})}}, \dots, e^{-j 2 \pi f_m \mathbf{t}_{N_\mathrm{RF}}^{(\mathrm{full})}} ] \Big).
    \end{align}
    Here, $\mathbf{F}_n^{(\mathrm{full})} \in \mathbb{C}^{N \times N_T}$ represents the PS-based analog beamformer connected to the $n$-th RF chain via TTDs and is given by 
    \begin{equation}
        \mathbf{F}_n^{(\mathrm{full})} = \mathrm{blkdiag}\left( \big[\mathbf{f}_{n,1}^{(\mathrm{full})}, \dots, \mathbf{f}_{n,N_\mathrm{T}}^{(\mathrm{full})} \big] \right),    
    \end{equation} 
    with $\mathbf{f}_{n,\ell}^{(\mathrm{full})} \in \mathbb{C}^{\frac{N}{N_T} \times 1}$ denoting the PS-based analog beamformer connecting the $\ell$-th TTD and the $n$-th RF chain. The constant-modulus constraint needs to be satisfied for each element of $\mathbf{f}_{n,\ell}^{(\mathrm{full})}$, which implies
    \begin{equation}
        \big|[\mathbf{f}_{n,\ell}^{(\mathrm{full})}]_i \big|=1, \forall n,\ell,i.
    \end{equation}
    Furthermore,  $\mathbf{t}_n^{(\mathrm{full})} \in \mathbb{R}^{N_T \times 1}$ denotes the time delays realized by the TTDs connected to the $n$-th RF chain. In practice, the maximum time delay that can be achieved by TTDs is limited, yielding the following constraint:
    \begin{equation}
        [\mathbf{t}_n^{(\mathrm{full})}]_\ell \in [0, t_{\max}], \forall n, \ell,
    \end{equation} 
    where $t_{\max}$ denotes the maximum delay that can be realized by the TTDs. For ideal TTDs, we have $t_{\max} = +\infty$.

    \item \textbf{Sub-connected Architecture}: As shown in Fig. \ref{fig:TTD_Hybrid_sub}, in the \emph{sub-connected architecture}, each RF chain is connected to a subarray of antenna elements via TTDs and PSs \cite{liu2023near}. The small number of antenna elements in each subarray reduces the beam split effect across the subarray. Consequently, the number of TTDs required for each RF chain can be significantly reduced to the point where only one TTD may suffice. In the following, we present the signal model for the simplest case, where each RF chain is connected to a single TTD. In particular, the analog beamformers realized by PSs and TTDs are, respectively, given by
    \begin{align}
        &\mathbf{F}_{\mathrm{PS}}^{(\mathrm{sub})} = \mathrm{blkdiag}\Big( 
        \big[\mathbf{f}_1^{(\mathrm{sub})},\dots,\mathbf{f}_{N_\mathrm{RF}}^{(\mathrm{sub})}  \big]
        \Big), \\
        &\mathbf{T}_m^{(\mathrm{sub})} = \mathrm{diag}\Big( [e^{-j2\pi t_1^{(\mathrm{sub})}}, \dots, e^{-j 2 \pi t_{\mathrm{RF}}^{(\mathrm{sub})}}]^{\mathsf{T}} \Big),
    \end{align}
    where $\mathbf{f}_n^{(\mathrm{sub})} \in \mathbb{C}^{\frac{N}{N_T} \times 1}$ and $t_n^{(\mathrm{sub})} \in [0, t_{\max}]$ denote the coefficients of the PSs and TTD connected to the $n$-th RF chain, respectively.

\end{itemize}

By exploiting TTD-based hybrid beamforming, for both fully-connected and sub-connected architectures, the achievable rate of user $k$ at subcarrier $m$ is given by 
\begin{align}
    R_{m,k} = \log_2 \left( 1 + \frac{ |\tilde{\mathbf{h}}_{m,k}^{\mathsf{T}} \mathbf{F}_{\mathrm{PS}} \mathbf{T}_m \mathbf{f}_{\mathrm{BB}}^{m,k} |^2 }{ \sum_{i \neq k} |\tilde{\mathbf{h}}_{m,k}^{\mathsf{T}} \mathbf{F}_{\mathrm{PS}} \mathbf{T}_m \mathbf{f}_{\mathrm{BB}}^{m,i} |^2 + \sigma_{m,k}^2 } \right).
\end{align} 
TTD-based hybrid beamforming can again be designed based on \emph{fully-digital approximation} and \emph{heuristic two-stage optimization}. These two approaches are detailed as follows.
\begin{itemize}
    \item \textbf{Fully-digital Approximation}:
    In this approach, the optimal unconstrained fully-digital beamformer $\mathbf{F}_m^{\mathrm{opt}}$ is designed for each subcarrier $m$. Then, the TTD-based hybrid beamformer is optimized to minimize the distance to $\mathbf{F}_m^{\mathrm{opt}}$. The resulting optimization problem is given as follows, where we take the fully-connected architecture as an example with $P_{\max}^m$ denoting the maximum transmit power available at subcarrier $m$:
    \begin{center}
        \begin{tcolorbox}[title = {$\mathcal{P}_{\mathrm{TTD}}$: TTD-based Hybrid Beamforming Optimization}]
        \vspace{-5mm}
        \begin{align} \label{problem:TTD hybrid beamforming}
        \underset{\begin{subarray}{c}
            \mathbf{F}_{\mathrm{BB}}, \mathbf{f}_{n,j}, \mathbf{t}_n
            \end{subarray}}{\mathrm{min}} \quad & \sum_{m=1}^M \big\| \mathbf{F}_m^{\mathrm{opt}} - \mathbf{F}_{\mathrm{PS}} \mathbf{T}_m \mathbf{F}_{\mathrm{BB}}^m \big\|_F^2 
        \\
        \mathrm{s. t.} \quad &\big|[\mathbf{f}_{n,\ell}]_i\big| = 1, \forall n,\ell,i, \nonumber \\
        & [\mathbf{t}_n]_\ell \in [0, t_{\max}], \forall n, \ell, \nonumber \\
        & \big\| \mathbf{F}_{\mathrm{RF}} \mathbf{T}_m \mathbf{F}_{\mathrm{BB}}^m \big\|_F^2 \le P_{\max}^m, \forall m. \nonumber 
        \end{align}
        \vspace{-4mm}\par\noindent
    \end{tcolorbox}
    \end{center}
    Compared to the conventional PS-based hybrid beamforming design problem for narrowband systems in \eqref{problem:hybrid beamforming}, problem \eqref{problem:TTD hybrid beamforming} is more challenging due to the following reasons. On the one hand, the three sets of optimization variables $\mathbf{F}_{\mathrm{BB}}$, $\mathbf{f}_{n,j}$, and $\mathbf{t}_n$ are deeply coupled. On the other hand, the exponential form of $\mathbf{T}_m$ with respect to $\mathbf{t}_n$ adds a further challenge to this problem. Furthermore, this approach may lead to high complexity because it requires the design of the large-dimensional $\mathbf{F}_m^{\mathrm{opt}}$ over a large number of subcarriers and optimizing the large-dimensional $\mathbf{F}_{\mathrm{PS}}$ and $\mathbf{T}_m$. The development of efficient algorithms for solving problem \eqref{problem:TTD hybrid beamforming} is still in its infancy.

    \item \textbf{Heuristic Two-stage Optimization}: In this approach, the complexity is significantly reduced by designing the analog beamformer heuristically in closed form. Then, the low-dimensional digital beamformer can be optimized with low complexity. In contrast to conventional hybrid beamforming, the analog beamformers realized by PSs and TTDs need to be jointly designed in TTD-based hybrid beamforming, such that the beams of all subcarriers are focused on the desired location. For wideband FFC systems, several heuristic designs for analog beamformers have been reported in recent studies \cite{gao2021wideband,dai2022delay} to address the issue of beam split effect. However, these designs are not applicable in near-field scenarios. As a further advance, the authors of \cite{cui2021near} developed a novel heuristic approach to mitigate the near-field beam split effect based on a far-field approximation for each antenna subarray. However, this approach only accounts for a single RF chain. Therefore, additional research is needed to develop general heuristic designs for TTD-based hybrid beamforming architectures in wideband NFC.
\end{itemize}

\begin{figure}[!t]
    \centering
    \subfigure[Spectral efficiency versus SNR.]{
        \includegraphics[width=0.4\textwidth]{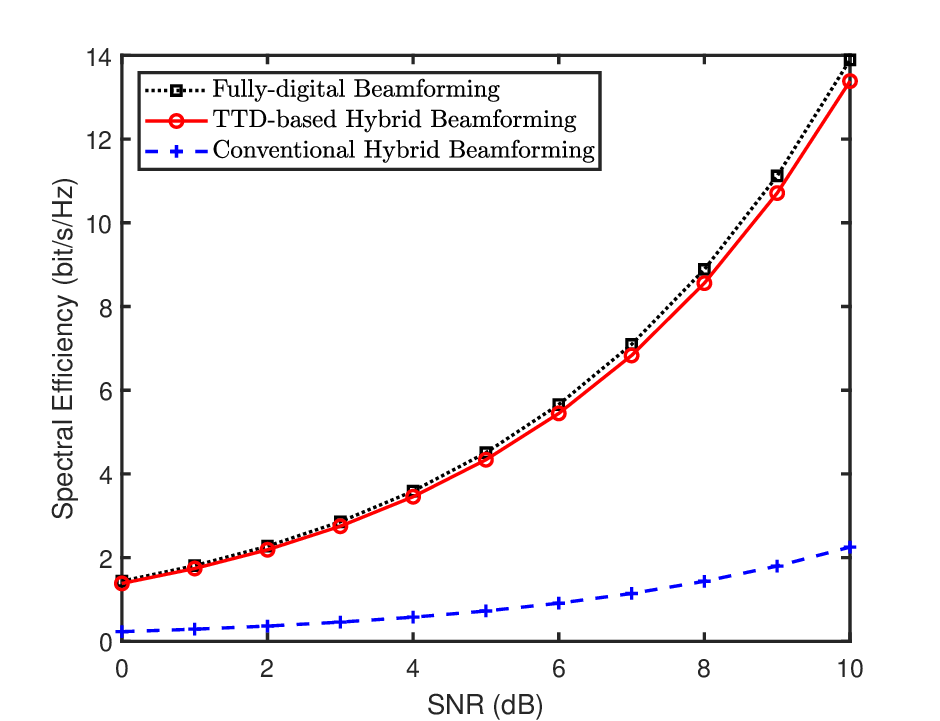}
        \label{fig:performance_TTD_1}
    }
    \subfigure[Spectral efficiency versus frequency of subcarriers.]{
        \includegraphics[width=0.4\textwidth]{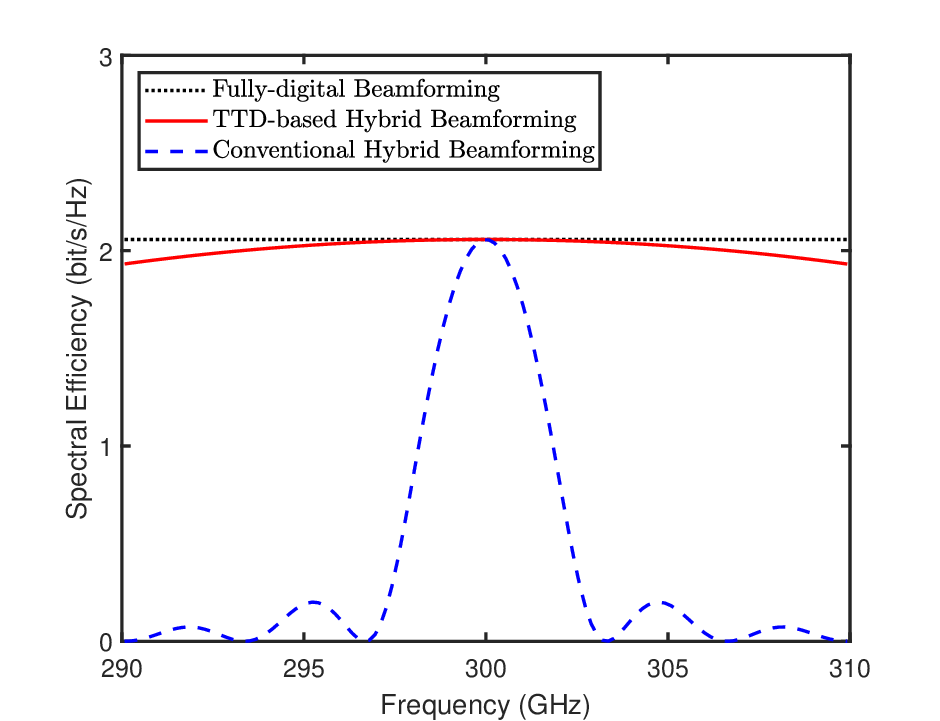}
        \label{fig:performance_TTD_2}
    }
    \caption{Performance of TTD-based hybrid beamforming architecture in near-field wideband OFDM systems, where $N = 256$, $N_{\mathrm{RF}} = 1$, $N_T = 16$, $f_c = 0.3$ THz, $W = 20$ GHz, and $M_c=128$. The location of the user is $(\theta=\frac{\pi}{4}, r=5.5 \text{m})$.}
    \label{fig:performance_TTD}
\end{figure}

In Fig. \ref{fig:performance_TTD_1}, we compare the spectral efficiency of TTD-based hybrid beamforming with fully-digital beamforming and conventional hybrid beamforming for a wideband OFDM system as a function of the SNR. In particular, the conventional hybrid beamformer is configured to generate a beam focused on the user's location at the central frequency. The TTD-based hybrid beamformer is designed based on the far-field approximation proposed in \cite{cui2021near}.
As can be observed, the fully-digital beamformer achieves the highest spectral efficiency due to its ability to generate a dedicated beam for each subcarrier. However, this is at the cost of extremely high power consumption. Thanks to its capability to realize \emph{frequency-dependent} analog beamforming, the TTD-based hybrid beamformer achieves a performance comparable to that of the fully-digital beamformer. In contrast, the conventional PS-based hybrid beamforming architecture has a low performance due to the significant beam split effect. The benefits of TTD-based hybrid beamforming are further illustrated in Fig. \ref{fig:performance_TTD_2}, where the spectral efficiencies of different subcarriers are illustrated. As can be observed, TTD-based hybrid beamforming achieves high spectral efficiency for all subcarriers. However, for conventional PS-based hybrid beamforming, only subcarriers in the proximity of the central frequency can achieve good performance. These results further underscore the importance of TTD-based hybrid beamforming architectures for wideband NFC systems.

\subsection{Beamfocusing with CAP Antennas}
In this subsection, we focus our attention on near-field beamfocusing with CAP antennas. Although CAP antennas generally consist of a continuous radiating surface, achieving independent and complete control of each radiating point on the surface proves to be an insurmountable task. Therefore, sub-wavelength sampling of the continuous surface becomes crucial \cite{smith2017analysis}. 

Metasurface antennas are a design approach to approximate CAP antennas and are implemented with metamaterials \cite{smith2017analysis, shlezinger2021dynamic, deng2021reconfigurable}. In the context of conventional antenna arrays, it is customary to employ an antenna spacing that is not less than half of the operating wavelength $\lambda$. This is mainly attributed to the practical constraints posed by the sub-wavelength size of conventional antennas and the associated mutual coupling effects caused by the closely spaced deployment. However, metamaterials with their unique electromagnetic properties can be utilized to overcome the aforementioned limitations. By exploiting metamaterials, a metasurface antenna with numerous metamaterial radiation elements can be realized allowing for reduced power consumption and ultra-close element spacing on the order of $\frac{1}{10}\lambda \sim \frac{1}{5} \lambda$ \cite{smith2017analysis}. The transmit signals are generally fed to the metasurface antenna by a waveguide \cite{shlezinger2021dynamic, deng2021reconfigurable}. The signal from each RF chain propagates in the waveguide to excite the metamaterial elements. Subsequently, the excited metamaterial elements tune the amplitude or phase of the signal before emitting it. Therefore, metasurface antennas can be regarded as a kind of hybrid beamforming architecture, which we refer to as \emph{metasurface-based hybrid beamforming architecture}.

\begin{figure}[!t]
    \centering
    \includegraphics[height=0.25\textwidth]{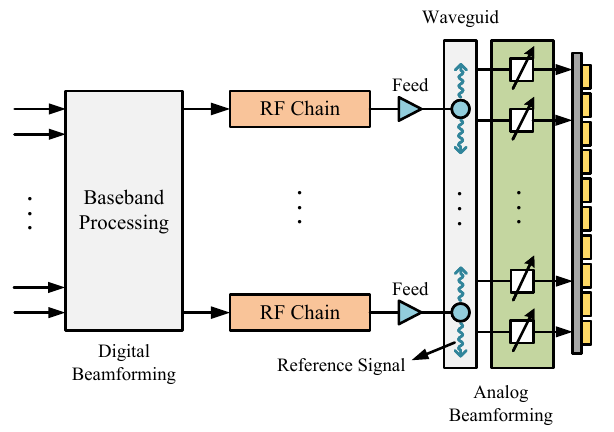}
    \caption{Architecture for metasurface-based hybrid beamforming.}
    \label{fig:metasurface}
\end{figure}

\subsubsection{Narrowband Systems}
We first focus on narrowband systems. 
As shown in Fig. \ref{fig:metasurface}, the metasurface-based hybrid beamforming architecture is generally composed of $N_{\mathrm{RF}}$ feeds connected to $N_{\mathrm{RF}}$ RF chains and $N$ reconfigurable metamaterial elements which can tune the signal. Specifically, the signal of each RF chain is fed into a waveguide as a \emph{reference signal}. Then, the metamaterial element is excited by the reference signal propagating in the waveguide and radiates the \emph{object signal}. Let $\tilde{x}_i$ denote the signal generated by the digital beamformer and emitted by the $i$-th feed. The reference signal propagated to the $n$-th metamaterial element through the waveguide is given by \cite{smith2017analysis}
\begin{equation}
    x_n^r = \sum_{i=1}^{N_{\mathrm{RF}}} \tilde{x}_i e^{-j \frac{2\pi n_r}{\lambda} d_{n,i}},
\end{equation}
where $n_r$ denotes the refractive index of the waveguide, and $d_{n,i}$ denotes the distance between the $i$-th feed and the $n$-th metamaterial element. Then, the object signal radiated by the $n$-th metamaterial element is given by 
\begin{equation}
    x_n^o = \psi_n x_n^r,
\end{equation}
where $\psi_n$ denotes the configurable weight of the $n$-th metamaterial antenna element. Therefore, the overall object signal $\mathbf{x}^o = [x_1^o, x_2^o,\dots, x_N^o]$ after the analog and digital beamformer can be modelled as follows
\begin{equation}
    \mathbf{x}^o = \mathbf{\Psi} \mathbf{Q} \mathbf{F}_{\mathrm{BB}} \mathbf{x},
\end{equation}
where 
\begin{align}
    &\mathbf{\Psi} = \mathrm{diag}\left( [ \psi_1,\dots,\psi_N ]^{\mathsf{T}} \right), \\
    &[\mathbf{Q}]_{n,i} = e^{-j \frac{2\pi n_r}{\lambda} d_{n,i}}.
\end{align}
In particular, the configurable weight $\psi_n$ represents the Lorentzian resonance response of metamaterial elements, which depending on the implementation can be controlled while meeting the following three sets of constraints \cite{smith2017analysis}:
\begin{itemize}
    \item \textbf{Continuous amplitude control}: In this case, the metamaterial element is assumed to be near resonance, which provides the modality to tune the amplitude of each element without significant phase shifts. Therefore, the configurable weight $\psi_i$ is constrained by
    \begin{equation} \label{meta:amplitude}
        \psi_n \in [0, 1].
    \end{equation}
    \item \textbf{Discrete amplitude control}: In this case, the amplitude of each metamaterial element can be adjusted based on a set of discrete values. The corresponding constraint of $\psi_i$ is given by 
    \begin{equation} \label{meta:amplitude_discrete}
        \psi_n \in \left\{0, \frac{1}{C-1}, \frac{2}{C-1}, \dots, 1 \right\},
    \end{equation}
    where $C$ denotes the number of candidate amplitudes.
    \item \textbf{Lorentzian-constrained phase-shift control}: In this case, phase shift tuning of each metamaterial element can be achieved. However, the phase shift and the amplitude of each element are coupled due to the Lorentzian resonance. Therefore, phase shifts and amplitudes need to be jointly controlled based on the following constraint:
    \begin{equation} \label{meta:phase}
        \psi_n = \frac{j + e^{j \vartheta_n}}{2}, \quad \vartheta_n \in [0, 2 \pi].
    \end{equation}
    It can be verified that the phase of $\psi_n$ is restricted to the range $[0, \pi]$ and the amplitude is coupled with the phase, i.e., $|\psi_n| = | \cos(\vartheta_n/2) |$.   
\end{itemize}
By exploiting metasurface-based hybrid beamforming in narrowband communication, the received signal at user $k$ is given by 
\begin{equation}
    y_k = \mathbf{h}_k^{\mathsf{T}} \mathbf{\Psi} \mathbf{Q} \mathbf{F}_{\mathrm{BB}} \mathbf{x} + n_k.
\end{equation} 
We note that since the CAP antenna is sampled by the discrete metamaterial elements of the metasurface, the communication channel $\mathbf{h}_k$ can be approximated by the near-field multipath MISO channel model developed for SPD antennas in \eqref{H_NFC_MISO}.
Then, the achievable rate at user $k$ is obtained as follows:
\begin{align}
    R_k = \log_2 \left( 1 + \frac{ | \mathbf{h}_k^{\mathsf{T}} \mathbf{\Psi} \mathbf{Q} \mathbf{f}_{\mathrm{BB},k} |^2 }{ \sum_{i \neq k} |\mathbf{h}_k^{\mathsf{T}} \mathbf{\Psi} \mathbf{Q} \mathbf{f}_{\mathrm{BB},i} |^2 + \sigma_k^2 } \right).
\end{align}
The analog beamformer $\mathbf{\Psi}$ and the digital beamformer $\mathbf{F}_{\mathrm{BB}}$ need to be designed properly to maximize system performance. However, unlike the hybrid beamforming for SPD antennas, an unconstrained fully-digital equivalent does not exist for metasurface-based hybrid beamforming in practice, which makes the fully-digital approximation approach for analog beamfocusing design almost infeasible. Therefore, the analog beamformer $\mathbf{\Psi}$ and the digital beamformer $\mathbf{F}_{\mathrm{BB}}$ need to be optimized directly. Let $f(\mathbf{\Psi}, \mathbf{F}_{\mathrm{BB}})$ denote some suitable performance metric, such as spectral efficiency or energy efficiency. Then, we have the following optimization problem, where $\mathcal{F}_{\psi}$ denotes the feasible set of $\psi_n$ given by \eqref{meta:amplitude}, \eqref{meta:amplitude_discrete}, or \eqref{meta:phase}:
\begin{center}
    \begin{tcolorbox}[title = {$\mathcal{P}_{\mathrm{Meta}}$: Metasurface-based Hybrid Beamforming Optimization}]
    \vspace{-5mm}
    \begin{align} \label{problem:Meta hybrid beamforming}
    \underset{\begin{subarray}{c}
        \mathbf{\Psi}, \mathbf{F}_{\mathrm{BB}}
        \end{subarray}}{\mathrm{max}} \quad & f(\mathbf{\Psi}, \mathbf{F}_{\mathrm{BB}})
    \\
    \mathrm{s. t.} \quad &\mathbf{\Psi} = \mathrm{diag}\left( [ \psi_1,\dots,\psi_N ]^T \right), \nonumber \\
    & \psi_n \in \mathcal{F}_{\psi}, \forall n, \nonumber \\
    & \big\| \mathbf{\Psi} \mathbf{Q} \mathbf{F}_{\mathrm{BB}} \big\|_F^2 \le P_{\max}. \nonumber 
    \end{align}
    \vspace{-4mm}\par\noindent
\end{tcolorbox}
\end{center}
The main challenges in solving the above problem stem from the high dimensionality of $\mathbf{\Psi}$ and the intractable constraints imposed on $\mathbf{\Psi}$, especially for the Lorentzian-constrained phase shift. How to develop efficient algorithms to solve \eqref{problem:Meta hybrid beamforming} is still an open problem, which requires additional research efforts. 

\subsubsection{Wideband Systems}

For wideband NFC systems, the challenges arising from the \emph{frequency-dependent} array response vector also exist for metasurface-based hybrid beamforming. However, the \emph{frequency-dependence} of the waveguide has the potential to mitigate the near-field beam split effect. To elaborate, in wideband systems, the received signal at subcarrier $m$ of user $k$ is given by 
\begin{equation}
    \tilde{y}_{m,k} = \tilde{\mathbf{h}}_{m,k}^\mathsf{T} \mathbf{\Psi} \mathbf{Q}_m \mathbf{F}_{\mathrm{BB}}^m \tilde{\mathbf{x}}_m + \tilde{n}_{m,k},
\end{equation} 
where $\mathbf{Q}_m$ is defined as
\begin{equation}
    [\mathbf{Q}_m]_{n,i} = e^{-j \frac{2 \pi f_m n_r}{c} d_{n,i}}.
\end{equation}    
Based on this \emph{frequency-dependent} waveguide propagation matrix, an interference-pattern-based design was proposed in \cite{di2021reconfigurable} to mitigate the far-field beam split effect by exploiting the metasurface-based hybrid beamforming architecture. How to use the frequency-dependence of the waveguide to mitigate the near-field beam split is an interesting direction for future research.

\begin{table*}[!h]
\caption{Summary of hybrid beamforming architectures for NFC.}
\label{table_arc}
\small
\centering
\begin{tabular}{!{\vrule width1pt}c!{\vrule width1pt}c!{\vrule width1pt}c!{\vrule width1pt}c!{\vrule width1pt}c!{\vrule width1pt}}
\Xhline{1pt} 
\textbf{Antenna}                                            & \textbf{Channel Category}         & \textbf{Bandwidth}        & \textbf{Analog Beamforming}               & \textbf{Digital Beamforming}                           \\ \Xhline{1pt} 
\multirow{4}{*}{\makecell[c]{SPD\\Antennas}}                & \multirow{2}{*}{MISO}             & Narrowband                & PS-based                                  & \multirow{2}{*}{ Fixed RF chains }            \\ \cline{3-4}
                                                            &                                   & Wideband                  & TTD-based                                 &                                               \\ \cline{2-5} 
                                                            & \multirow{2}{*}{MIMO}             & Narrowband                & PS-based                                  & \multirow{2}{*}{Dynamic RF chains}           \\ \cline{3-4}
                                                            &                                   & Wideband                  & TTD-based                                 &                                              \\ \Xhline{1pt} 
\multirow{2}{*}{\makecell[c]{CAP\\Antennas}}                & MISO                              & -                         &\multirow{2}{*}{Metasurface-based}         & Fixed RF chains                               \\ \cline{2-3} \cline{5}
                                                            & MIMO                              & -                         &                                           & Dynamic RF chains                             \\ \Xhline{1pt} 
\end{tabular}
\end{table*}

\subsection{MIMO Extensions}

\begin{figure}[t!]
    \centering
    \includegraphics[height=0.25\textwidth]{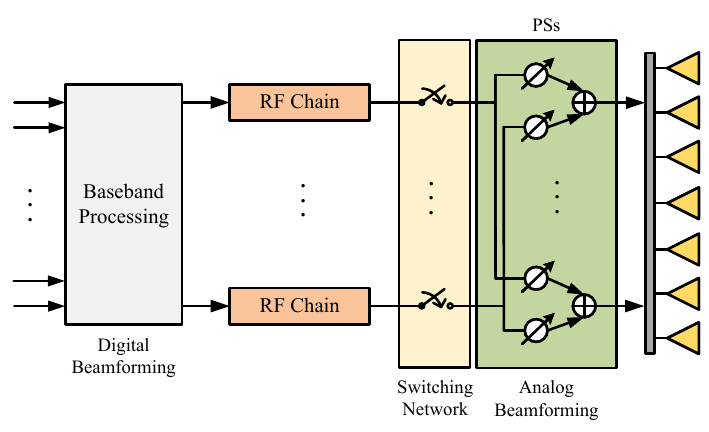}
    \caption{Architecture of hybrid beamforming with dynamic RF chains.}
    \label{fig:switch_hybrid}
\end{figure}

As discussed in Section II, near-field MIMO channels can provide more DoFs than far-field MIMO channels, especially in the LoS-dominant case. Therefore, MIMO systems can yield substantial performance gains compared to MISO systems in NFC.
In particular, in the far-field, the LoS MIMO channel matrix is always rank-one, c.f., \eqref{far_MIMO_LoS}, and thus can support only a single data stream and cannot fully exploit the benefits of MIMO. By contrast, in NFC, the high-rank LoS MIMO channel matrix is capable of supporting multiple data streams, thus providing more multiplexing gain and enhancing communication performance. Let us consider a near-field narrowband single-user MIMO system with two parallel ULAs with $N_T$ and $N_R$ antenna elements at the transmitter and receiver, respectively. Employing the hybrid beamforming architecture at the BS, the received signal at the user is given by 
\begin{equation}
    \mathbf{y} = \mathbf{H}_{\mathrm{ULA}}^{\mathrm{LoS}} \mathbf{F}_{\mathrm{RF}} \mathbf{F}_{\mathrm{BB}} \mathbf{x} + \mathbf{n},
\end{equation}
where $\mathbf{H}_{\mathrm{ULA}}^{\mathrm{LoS}} \in \mathbb{C}^{N_R \times N_T}$ denotes the near-field LoS channel matrix for parallel ULAs given in \eqref{18a} and $\mathbf{x} \in \mathbb{C}^{N_{\mathrm{s}} \times 1}$ collects the information symbols of the $N_{\mathrm{s}}$ data streams. The maximum number of data streams that can be supported is determined by the DoFs of channel matrix $\mathbf{H}_{\mathrm{ULA}}^{\mathrm{LoS}}$. However, the primary obstacle in such a system is that the rank of $\mathbf{H}_{\mathrm{ULA}}^{\mathrm{LoS}}$ is not a constant but rather is contingent upon the distance $r$ between the BS and the user. Specifically, as unveiled by \eqref{dof_ula}, when $N_T$ and $N_R$ are sufficiently large, the DoFs of $\mathbf{H}_{\mathrm{ULA}}^{\mathrm{LoS}}$ are given by $\frac{(N_T-1)d_T(N_R-1)d_R}{\lambda r}$, which decrease gradually as distance $r$ increases, reducing the number of data streams that can be supported. Additionally, as pointed out in \cite{sohrabi2016hybrid}, for the fully-connected hybrid beamforming architecture, $2N_s$ RF chains are sufficient to achieve the same performance as a fully-digital beamformer, which implies that deploying more than $2N_s$ RF chains will only result in more power consumption without enhancing communication performance. Therefore, in near-field MIMO systems, the number of RF chains needs to be dynamically changed based on distance $r$ to achieve the best possible performance \cite{wu2022distance}. This can be achieved by the hybrid beamforming architecture shown in Fig. \ref{fig:switch_hybrid}, where a switching network is introduced to control the number of active RF chains. Let $N_{\mathrm{RF}}^{\max}$ denote the maximum number of available RF chains. Then, the received signal at the user can be rewritten as follows:
\begin{equation}
    \mathbf{y} = \mathbf{H}_{\mathrm{ULA}}^{\mathrm{LoS}} \tilde{\mathbf{F}}_{\mathrm{RF}} \tilde{\mathbf{F}}_{\mathrm{S}} \tilde{\mathbf{F}}_{\mathrm{BB}} \tilde{\mathbf{x}} + \mathbf{n},
\end{equation}
where $\tilde{\mathbf{x}} \in \mathbb{C}^{N_{\mathrm{RF}}^{\max} \times 1}$, $\tilde{\mathbf{F}}_{\mathrm{RF}} \in \mathbb{C}^{N_T \times N_{\mathrm{RF}}^{\max}}$, $\tilde{\mathbf{F}}_{\mathrm{BB}} \in \mathbb{C}^{N_{\mathrm{RF}}^{\max} \times N_{\mathrm{RF}}^{\max}}$ denote the full-dimensional candidate data steams, analog beamformer, and digital beamformer, respectively. $\tilde{\mathbf{F}}_{\mathrm{S}} \in \mathbb{R}^{N_{\mathrm{RF}}^{\max} \times N_{\mathrm{RF}}^{\max}}$ represents the ON/OFF status of the switching network, and is given by 
\begin{equation}
    \tilde{\mathbf{F}}_{\mathrm{S}} = \mathrm{diag} \big( \big[ \alpha_1, \alpha_2, \dots, \alpha_{N_{\mathrm{RF}}^{\max}} \big]^{\mathsf{T}} \big).
\end{equation}
In particular, $\alpha_n = 1$ means that the $n$-th RF chain is active while $\alpha_n = 0$ indicates that the $n$-th RF chain is switched off. The achievable rate of the user is given by 
\begin{equation}
    R = \log_2 \det \Big(  \mathbf{I} + \frac{1}{\sigma^2} \mathbf{H}_{\mathrm{ULA}}^{\mathrm{LoS}} \mathbf{Q} (\mathbf{H}_{\mathrm{ULA}}^{\mathrm{LoS}})^{\mathsf{H}} \Big),
\end{equation}
where $\mathbf{Q} = \tilde{\mathbf{F}}_{\mathrm{RF}} \tilde{\mathbf{F}}_{\mathrm{S}} \tilde{\mathbf{F}}_{\mathrm{BB}} \tilde{\mathbf{F}}_{\mathrm{BB}}^{\mathsf{H}} \tilde{\mathbf{F}}_{\mathrm{S}}^{\mathsf{H}} \tilde{\mathbf{F}}_{\mathrm{RF}}^{\mathsf{H}}$ denotes the covariance matrix of the transmit signal. Since the purpose of the switching network is to avoid the waste of energy, the power consumption of the RF chains and the PSs needs to be considered. In particular, the total number of active RF chains and PSs is $\sum_{n=1}^{N_{\mathrm{RF}}^{\max}} \alpha_n$ and $\sum_{n=1}^{N_{\mathrm{RF}}^{\max}} \alpha_n N_T$, respectively. The corresponding power consumption is given by 
\begin{align}
    P = \sum_{n=1}^{N_{\mathrm{RF}}^{\max}} \alpha_n (P_{\mathrm{RF}} + N_T P_{\mathrm{PS}}),
\end{align}
where $P_{\mathrm{RF}}$ and $P_{\mathrm{PS}}$ denote the power consumed by one RF chain and one PS, respectively. Then, the optimal $\tilde{\mathbf{F}}_{\mathrm{RF}}$, $\tilde{\mathbf{F}}_{\mathrm{S}}$, and $\tilde{\mathbf{F}}_{\mathrm{BB}}$ can be obtained by solving the following multi-objective optimization problem:    
\begin{center}
    \begin{tcolorbox}[title = {$\mathcal{P}_{\mathrm{MIMO}}$: Dynamic RF Chain based Hybrid Beamforming Optimization}]
    \vspace{-5mm}
    \begin{align} \label{problem:NFC_MIMO}
    \underset{\begin{subarray}{c}
        \tilde{\mathbf{F}}_{\mathrm{RF}}, \tilde{\mathbf{F}}_{\mathrm{S}}, \tilde{\mathbf{F}}_{\mathrm{BB}}
        \end{subarray}}{\mathrm{max}} \quad & R - \mu P
    \\
    \mathrm{s. t.} \quad &\big|[\tilde{\mathbf{F}}_{\mathrm{RF}}]_{m,n}\big| = 1, \forall m,n, \nonumber \\
    & \tilde{\mathbf{F}}_{\mathrm{S}} = \mathrm{diag} \big( \big[ \alpha_1, \alpha_2, \dots, \alpha_{N_{\mathrm{RF}}^{\max}} \big]^{\mathsf{T}} \big), \nonumber \\
    & \alpha_n \in \{0, 1\}, \forall n, \nonumber \\
    & \big\| \tilde{\mathbf{F}}_{\mathrm{RF}} \tilde{\mathbf{F}}_{\mathrm{S}} \tilde{\mathbf{F}}_{\mathrm{BB}} \big\|_F^2 \le P_{\max}, \nonumber 
    \end{align}
    \vspace{-4mm}\par\noindent
\end{tcolorbox}
\end{center}
where $\mu \ge 0$ is a weight factor that can be adjusted to emphasize the importance of high and low power consumption. Furthermore, the dynamic RF chain concept can also be applied in the context of TTD-based and metasurface-based hybrid architectures. The corresponding optimization problems can be formulated in a similar manner as \eqref{problem:NFC_MIMO}. Note that problem \eqref{problem:NFC_MIMO} is a mixed-integer nonlinear programming (MINLP) problem, which is NP-hard. The branch-and-bound (BnB) algorithm \cite{tawarmalani2005polyhedral} can be used to obtain the globally optimal solution of MINLP problems, but has exponential complexity. Furthermore, some low-complexity algorithms, such alternating direction method of multipliers (ADMM) \cite{gabay1976dual} and machine-learning-based methods \cite{shen2019lorm}, can be employed to obtain a high-quality sub-optimal solution.

Based on the discussion in previous subsections, the hybrid beamforming architectures for NFC are summarized in Table \ref{table_arc}.

\begin{figure}[!t]
    \centering
    \includegraphics[width=0.48\textwidth]{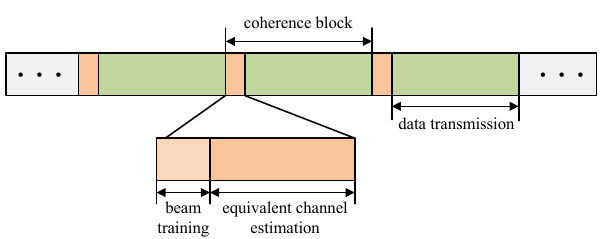}
    \caption{Communication protocol including beam training \cite{hur2013millimeter, zhang2016tracking, tan2021wideband}.}
    \label{fig:communication_block}
\end{figure}

\subsection{Near-Field Beam Training}

In the previous subsections, we have discussed several hybrid beamforming architectures for realizing near-field beamfocusing. However, to optimize the hybrid beamformer, channel state information (CSI) is required. Conventionally, CSI is obtained via channel estimation. However, in the case of NFC employing ELAAs, the complexity of conventional channel estimation techniques significantly increases. To address this challenge, beam training is proposed as a fast and efficient method for reducing the complexity of CSI acquisition and obtaining high-quality analog beamformers \cite{heath2016overview}. Rather than estimating the complete CSI of the high-dimensional near-field channel, beam training establishes a training procedure between the BS and the users to estimate the physical locations of the channel paths, where an optimal codeword yielding the largest received power at the user is selected from a pre-designed beam codebook. Then, the optimal codeword is exploited as analog beamformer for transmission. Once the analog beamformer is selected, conventional channel estimation methods can be employed to estimate the low-dimension equivalent channel comprising the original channel and the analog beamformer. The communication protocol including beam training is shown in Fig. \ref{fig:communication_block} \cite{hur2013millimeter, zhang2016tracking, tan2021wideband}.

\begin{figure}[!t]
    \centering
    \includegraphics[width=0.48\textwidth]{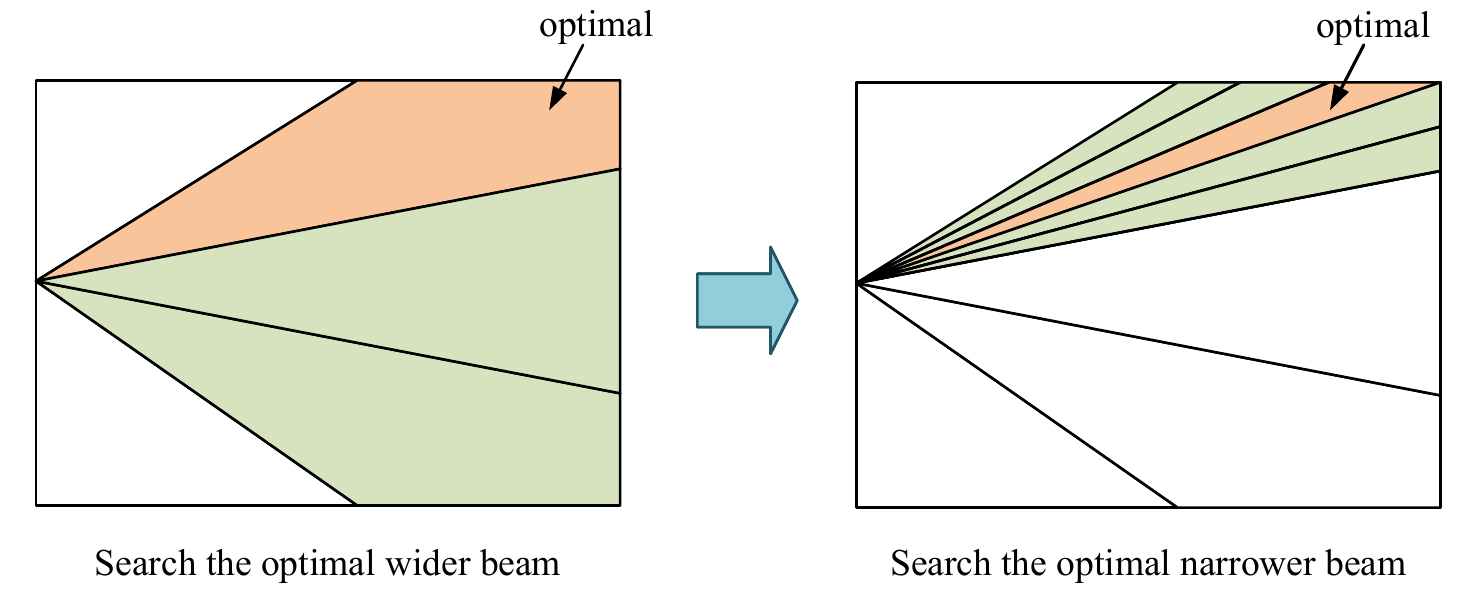}
    \caption{Far-field beam training process.}
    \label{fig:far_field_training}
\end{figure}

\begin{figure}[!t]
    \centering
    \includegraphics[width=0.48\textwidth]{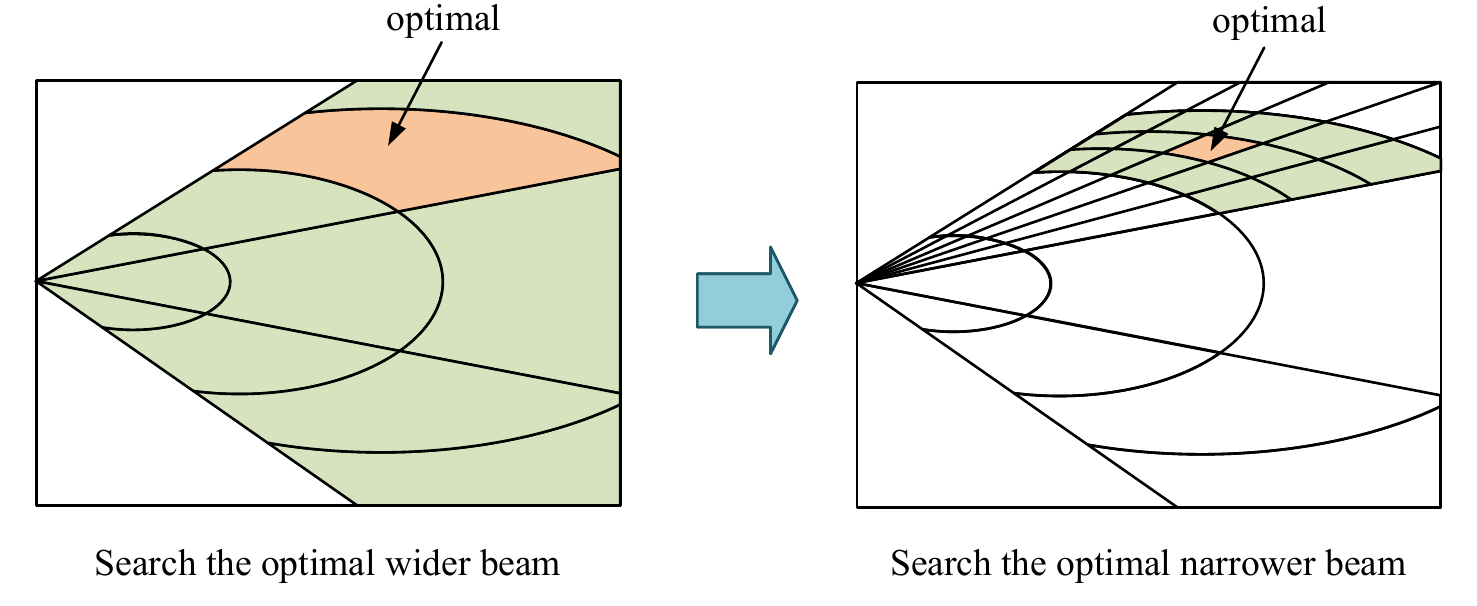}
    \caption{Near-field beam training process.}
    \label{fig:near_field_training}
\end{figure}

Compared to the far field, near-field beam training imposes new challenges. To elaborate, let us first provide a brief summary of far-field beam training based on hierarchical codebooks \cite{xiao2016hierarchical}. Generally, for conventional far-field beam training, the hierarchical codebook contains codewords that correspond to directional beams with varying beamwidths, ranging from wide to narrow. As shown in Fig. \ref{fig:far_field_training}, these codewords are then searched in a tree architecture to determine the optimal one. In this approach, wider beams are designated as “nodes” in the tree architecture, while narrower beams that are encompassed by the wider beams are regarded as “leaves” of the wider beam. The beam training starts with the BS transmitting a pilot signal to the user using the wider beams and selecting the optimal one with the highest received power at the user. The BS then transmits a pilot signal using the narrow beams that are “leaves” of the selected wider beam and repeats this process until the last level of the tree architecture is reached. The design of beam training involves two parts, namely \emph{codebook design} and \emph{training protocol design}, which have been widely investigated for FFC \cite{nitsche2014ieee, alkhateeb2014channel, xiao2016hierarchical, lin2017subarray, noh2017multi, chen2019two, 10033092, ning2023beamforming}.

However, far-field beam training methods are not directly applicable for near-field channels due to the absence of distance information in far-field codebooks. Consequently, it is necessary to revise the beam training for NFC. The near-field beam training process based on a hierarchical codebook is illustrated in Fig. \ref{fig:near_field_training}. Notably, compared to far-field beam training, near-field beam training requires a significantly larger codebook as near-field channels have to the modelled in the polar domain rather than the angular domain. This results in increased complexity of the near-field beam training process. Therefore, employing a low-complexity training protocol is critical for near-field beam training. Furthermore, the design of the codebook for near-field beam training is more intricate than that for far-field beam training due to the \emph{non-linear phase} or even \emph{non-uniform amplitude} of near-field channels. In the following, we discuss beam training for narrowband and wideband systems, respectively.
\begin{itemize}
    \item \textbf{Narrowband systems}: Some recent efforts have been devoted to the design of narrowband near-field beam training \cite{cui2022channel, zhang2022fast, liu2022deep, wu2023two}. Specifically, a new near-field codebook design was proposed in \cite{cui2022channel}, where the codewords are sampled uniformly in the angular domain and non-uniformly in the distance domain. To reduce the complexity of near-field beam training, the authors of \cite{zhang2022fast} proposed a two-phase training protocol. In this protocol, the candidate angle is first obtained via conventional far-field beam training in the first phase, which is followed by the estimation of the effective distance in the second phase. As a further advance, the authors of \cite{liu2022deep} conceived a deep learning based near-field beam training method, where a pair of neural networks were designed to determine the optimal near-field beam based on the received signals of the far-field wide beam. Finally, the authors of \cite{wu2023two} developed a two-stage hierarchical beam training method. In the first stage, a coarse direction is obtained through conventional far-field beam training. In this second stage, based on the obtained coarse direction, joint angle-and-distance estimation is carried out over a fine grid in the polar domain. 
    
    \item \textbf{Wideband systems}: For wideband beam training, the near-field beam split effect has to be considered. Unlike in the data transmission stage, where the beam split effect leads to a performance loss, it has the potential to speed up beam training as a single beam can cover several directions or locations at different OFDM subcarriers. When the TTD-based hybrid beamforming architecture is employed, flexible control of near-field beam split can be achieved by properly designing the time delays. Therefore, the size of the codebook and the complexity of the training process can be significantly reduced. To take advantage of this, a fast wideband near-field beam training method was proposed in \cite{cui2022near}. This method is based on an antenna architecture that only utilizes TTDs to realize the analog beamforming. However, how to design the wideband near-field beam training for the hybrid beamforming architecture employing both TTDs and PSs is still an open research problem.
\end{itemize}

Finally, all the aforementioned beam training methods are based on SPD antennas. For SPD antennas, the analog beamformer is generally subject to a unit-norm constraint. Therefore, the codebooks can be designed based on near-field array response vectors. On the other hand, for CAP antennas realized by metasurfaces, the complex constraints on the analog beamformers, c.f., \eqref{meta:amplitude}-\eqref{meta:phase}, makes the design of near-field beam training challenging, and more research is needed.

\subsection{Discussion and Open Research Problems}
We have introduced several fundamental antenna architectures for near-field beamfocusing and the basic principles of near-field beam training. However, to fully realize the advantages of near-field beamfocusing and beam training in practice, it is essential to overcome several key challenges. In the following, we discuss some of the primary open research problems that require attention:

\begin{itemize}
    \item \textbf{Near-field channel estimation}: Accurate channel estimation plays a key role in guaranteeing the performance of near-field beamfocusing. In conventional far-field communication, channel estimation techniques often rely on the sparsity of the channels in the angular domain. However, for near-field channels, the sparsity in the angular domain no longer holds. Therefore, it is important to unveil the sparsity of near-field channels in an appropriate domain for channel estimation. This task can be challenging, especially considering the additional distance dimension of near-field channels.
    \item \textbf{Near-field beamfocusing with finite-resolution DAC/ADC}: An alternative approach to reducing the complexity of hybrid beamforming is to adopt finite-resolution digital-to-analog converters (DACs) and analog-to-digital converters (ADCs) \cite{heath2016overview}. In this case, the power consumption of the RF chains can be reduced substantially. However, finite-resolution DAC/ADC may result in different design challenges, such as limited signal constellations. How to achieve near-field beamfocusing in this case is an open problem.
    \item \textbf{Multi-functional near-field beamfocusing}: 
    Next-generation wireless networks are envisioned to transcend the communication-only paradigm, incorporating a range of additional functionalities such as computing, sensing, secure transmission, and wireless power transfer. However, integrating such diverse functions poses a challenge as different near-field beamfocusing designs may be required to optimize performance for each functionality. Addressing this challenge is essential for unlocking the full potential of NFC in next-generation wireless networks and enabling the seamless integration of diverse functionalities.
    \item \textbf{Hybrid-field beamforming and beam training}: In practical scenarios, it is highly probable that users are situated in different field regions of the BS, resulting in a combination of near-field and far-field channels. However, the design of hybrid-field beamforming and beam training techniques accounting for the different channel characteristics of the near-field and far-field regions remains an open research challenge.
    \item \textbf{Dynamic switching between near-field and far-field communications}: Whether a user is located in the near-field or far-field is determined by the aperture of the employed antenna arrays and the frequency band. This classification is generally fixed in the current system designs, which poses a challenge for dynamic utilization of the benefits offered by near-field and far-field communications, namely high communication rate and low complexity, respectively. To overcome this challenge, it is imperative to develop new strategies that allow dynamic switching between these regions, thus enabling the exploitation of the full potential of near-field and far-field communications.
\end{itemize}

\begin{figure}[!t]
\setlength{\abovecaptionskip}{0pt}
\centering
\includegraphics[width=0.4\textwidth]{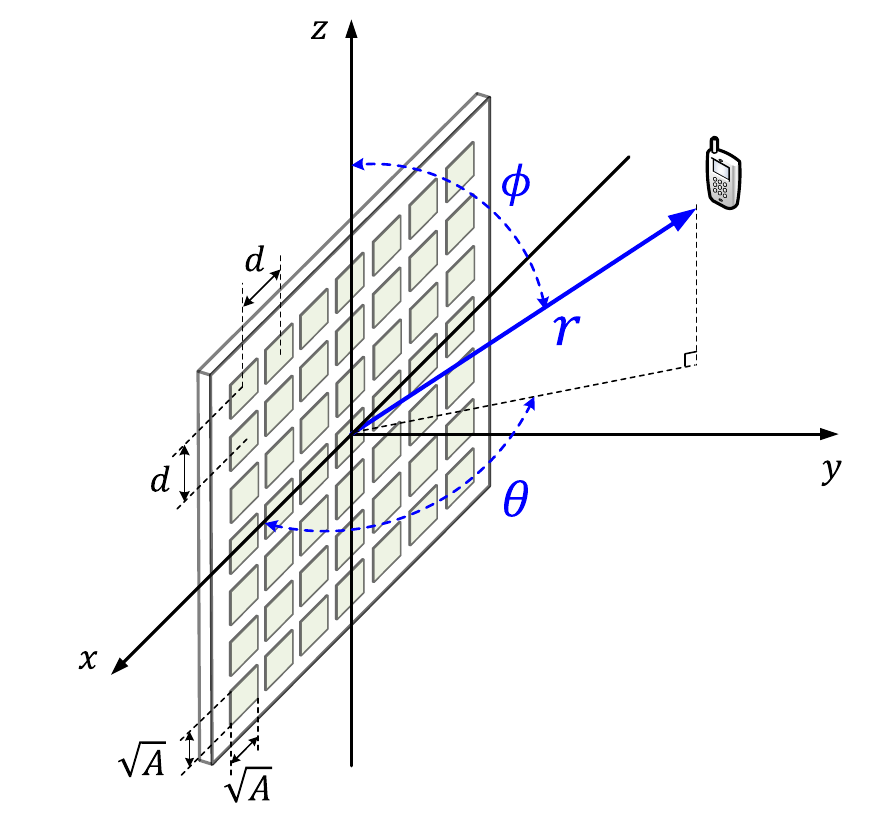}
\caption{System layout for NFC under LoS channels.}
\label{LoS_3D_Model}
\end{figure}
\section{Near-Field Performance Analysis}\label{Sec:per}
In this section, we provide a comprehensive performance analysis of NFC based on the near-field channel models discussed in Section \ref{NFC_Tutorial_Channel_Model}. We commence with the performance analysis for basic free-space LoS propagation and then shift our attention to statistical multipath channel models. For LoS channels, the received SNR and the corresponding power scaling law are analyzed for both SPD and CAP antennas. The derived results contribute to a deeper understanding of the performance limits of NFC. For multipath channels, we provide a general analytical framework for three typical performance metrics, namely the outage probability (OP), ergodic channel capacity (ECC), and ergodic mutual information (EMI).
\subsection{Performance Analysis for LoS Near-Field Channels}
\subsubsection{System Model}
As illustrated in {\figurename} {\ref{LoS_3D_Model}}, we consider a typical near-field MISO channel, where the BS is equipped with an UPA\footnote{Since UPA includes ULA as a special case, the results derived in this subsection can be straightforwardly extended to the case when the BS is equipped with an ULA.} containing $N\gg1$ antenna elements and the user is a single-antenna device. The UPA is placed on the $x$-$z$ plane and centered at the origin. Here, $N=N_{x}N_{z}$, where $N_{x}$ and $N_{z}$ denote the number of antenna elements along the $x$- and $z$-axes, respectively. For the sake of brevity, we assume that $N_{x}$ and $N_{z}$ are both odd numbers with $N_x=2\tilde{N}_x+1$ and $N_z=2\tilde{N}_z+1$. The physical dimensions of each BS antenna element along the $x$- and $z$-axes are given by $\sqrt{A}$, and the inter-element distance is $d$, where $d\geq\sqrt{A}$. It is worth noting that when $d=\sqrt{A}$, the SPD UPA turns into a CAP surface. The central location of the $(m,n)$-th BS antenna element is denoted by ${\mathbf{s}}_{m,n}=[nd,0,md]^{\mathsf{T}}$, where $n\in{\mathcal{N}}_x\triangleq\left\{-\tilde{N}_x,...,\tilde{N}_x\right\}$ and $m\in{\mathcal{N}}_z\triangleq\left\{-\tilde{N}_z,...,\tilde{N}_z\right\}$. As a result, the physical dimensions of the UPA along the $x$- and $z$-axes are $L_x\approx N_xd$ and $L_z\approx N_zd$, respectively. Let $e_a\in(0,1]$ denote the aperture efficiency of each antenna element. An antenna's aperture efficiency is defined as the ratio between the antenna's effective aperture and the antenna's physical aperture. By definition, the effective antenna aperture of each antenna element is given by $A_e=Ae_a$. The aperture efficiency is a dimensionless parameter between 0 and 1 that measures how close the antenna comes to using all the radio wave power intersecting its physical aperture \cite{Lu2022communicating}. If the aperture efficiency were 100\%, then all the wave's power falling on the antenna's physical aperture would be converted to electrical power delivered to the load attached to its output terminals. For the hypothetical isotropic antenna element, we have $A=A_e=\frac{\lambda^2}{4\pi}$ \cite{Lu2022communicating}.

As for the user, let $r$ denote its distance from the center of the antenna array, and $\theta\in[0,\pi]$ and $\phi\in[0,\pi]$ denote the azimuth and elevation angles, respectively, cf. {\figurename} {\ref{LoS_3D_Model}}. Consequently, the user location can be expressed as ${\mathbf{r}}=[r\Phi,r\Psi,r\Omega]^{\mathsf{T}}$, where $\Phi\triangleq\sin{\phi}\cos{\theta}$, $\Psi\triangleq\sin\phi\sin\theta$, and $\Omega\triangleq\cos{\phi}$. Moreover, we assume that the user is equipped with a hypothetical isotropic antenna element to receive the incoming signals and the receiving-mode polarization vector is fixed to be ${\bm\rho}\in{\mathbb C}^{3\times 1}$.

We next evaluate the performance of the considered MISO near-field channel by analyzing the received SNR at the user. Specifically, the SNR achieved by SPD antennas will be examined under the USW, NUSW, and the general near-field channel models. Subsequently, the SNR achieved by CAP antennas will be investigated.

\subsubsection{SNR Analysis for SPD Antennas}
When the BS is equipped with SPD antennas, the received signal at the user can be expressed as follows:
\begin{align}
y=\sqrt{p}{\mathbf{h}}^{\mathsf{T}}{\mathbf{w}}s+{n},
\end{align}
where $s\in{\mathbb C}$ is the transmitted data symbol with zero mean and unit variance, ${\mathbf{h}}\in{\mathbb C}^{N\times1}$ is the channel vector between the user and the BS, $p$ is the transmit power, ${\mathbf{w}}$ is the beamforming vector with $\lVert{\mathbf{w}}\rVert^2=1$, and ${{n}}\in{\mathcal{CN}}({{0}},\sigma^2)$ is AWGN. We assume that the BS has perfect knowledge of ${\mathbf{h}}$ and exploits maximum ratio transmission (MRT) beamforming to maximize the received SNR of the user, i.e., ${\mathbf{w}}$ is given by
\begin{align}
{\mathbf{w}}=\frac{{\mathbf{h}}^{*}}{\lVert{\mathbf{h}}\rVert}.
\end{align}
Therefore, the received SNR is given as follows:
\begin{align}\label{Section_Performance_Analysis_Statistical_Channel_MISO_SNR_Expression}
\gamma=\frac{p}{\sigma^2}\lVert{\mathbf{h}}\rVert^2.
\end{align}
For LoS channels, the channel gain $\lVert{\mathbf{h}}\rVert^2$ can be expressed as follows \cite{Lu2022communicating,Bjornson2020power}:
\begin{equation}\label{Section_Performance_Analysis_SNR_Most_General_Expression}
\begin{split}
\lVert{\mathbf{h}}\rVert^2=\sum\nolimits_{n\in{\mathcal{N}}_x}\sum\nolimits_{m\in{\mathcal{N}}_z}\lvert h_{m,n}^{i}({\mathbf{r}}) \rvert^2,
\end{split}
\end{equation}
where $h_{m,n}^{i}({\mathbf{r}})$ denotes the effective channel coefficient from the $(m,n)$-th BS antenna element to the user for $i\in\{{\rm U},{\rm N},{\rm G}\}$. Section \ref{NFC_Tutorial_Channel_Model} provides analytical expressions to determine $\lvert h_{m,n}^{i}({\mathbf{r}}) \rvert^2$ for several near-field LoS channel models including the USW model ($i={\rm{U}}$, Eq. \eqref{USW_Model_Channel_Coefficient}), the NUSW model ($i={\rm{N}}$, Eq. \eqref{NUSW_Model_Channel_Coefficient}), and the introduced general model ($i={\rm{G}}$, Eq. \eqref{proposed_h}). For each LoS model, we derive a closed-form expression for the received SNR, based on which the power scaling law in terms of the number of antennas, $N$, can be unveiled.

$\bullet$ \emph{\textbf{USW Channel Model:}} For the USW model ($i={\rm{U}}$), $\left\lVert h_{m,n}^{\rm{U}}({\mathbf{r}})\right\rVert^2$ can be derived from \eqref{USW_Model_Channel_Coefficient}, and the received SNR is provided in the following theorem.

\begin{theorem}\label{Section_Performance_Analysis_SNR_USW_Expression_Theorem}
The received SNR for the USW model is
\begin{equation}\label{Section_Performance_Analysis_SNR_USW_Expression}
\begin{split}
\gamma_{\rm{USW}}=\frac{p}{\sigma^2}\beta_0^2 N,
\end{split}
\end{equation}
where $\beta_0=\sqrt{A e_a \frac{1}{4\pi r^2}
G_1({\mathbf{0}},{\mathbf{r}})G_2({\mathbf{0}},{\mathbf{r}})}$.
\end{theorem}

\begin{proof}
Please refer to Appendix \ref{Section_Performance_Analysis_SNR_USW_Expression_Theorem_Proof}.
\end{proof}

\begin{remark}
Considering the terms appearing in $\beta_0$, $G_1({\mathbf{0}},{\mathbf{r}})$ and $G_2({\mathbf{0}},{\mathbf{r}})$ model the effective aperture loss and the polarization loss, respectively, whose expressions are given by \eqref{p_proj_general} and \eqref{eq_proj}, respectively. Moreover, $ \frac{A e_a}{4\pi r^2}$ models the free-space path loss. As mentioned in Section \ref{NFC_Tutorial_Channel_Model}, the near-field channel gain is mainly determined by the free-space path loss, and thus the terms $G_1(\cdot,\cdot)$ and $G_2(\cdot,\cdot)$ are not included in \eqref{USW_Model_Channel_Coefficient} for the sake of brevity. As outlined in Appendix \ref{Section_Performance_Analysis_SNR_USW_Expression_Theorem_Proof}, under the USW channel model, the powers radiated by different transmit antenna elements are affected by identical polarization mismatches, projected antenna apertures, and free-space path losses. In order to provide a general expression for the received SNR, we have kept the terms related to the effective aperture loss and the polarization loss in \textbf{Theorem} \ref{Section_Performance_Analysis_SNR_USW_Expression_Theorem}.
\end{remark}

Next, by letting $N_x, N_z\rightarrow\infty$, i.e., $N=N_xN_z\rightarrow\infty$, the asymptotic SNR for the USW channel model is given in the following corollary.

\begin{corollary}
As $N_x, N_z\rightarrow\infty$ ($N\rightarrow\infty$), the asymptotic SNR for the USW channel model satisfies
\begin{equation}\label{Section_Performance_Analysis_SNR_USW_Asymptotic}
\begin{split}
\lim_{N\rightarrow\infty}\gamma_{\rm{USW}}\simeq{\mathcal{O}}(N).
\end{split}
\end{equation}
\end{corollary}

\begin{remark}
The result in \eqref{Section_Performance_Analysis_SNR_USW_Asymptotic} suggests that for the USW model, the received SNR increases linearly with the total number of transmit antenna elements. In other words, by increasing the number of antenna elements, it is possible to increase the link gain to any desired level, which may even exceed the transmit power. This thereby breaks the law of conservation of energy. The reason for this behaviour is that when $N$ tends to infinity, the uniform amplitude assumption inherent to the USW model cannot capture the exact physical properties of near-field propagation.
\end{remark}

$\bullet$ \emph{\textbf{NUSW Channel Model:}}
For the NUSW model ($i={\rm{N}}$), $\left\lVert h_{m,n}^{\rm{N}}({\mathbf{r}})\right\rVert^2$ can be derived from \eqref{NUSW_Model_Channel_Coefficient}, and the received SNR is provided in the following theorem.

\begin{theorem}\label{Section_Performance_Analysis_SNR_NUSW_Expression_Theorem}
The received SNR for the NUSW model is given by
\begin{align}\label{Section_Performance_Analysis_SNR_NUSW_Expression}
\gamma_{\rm{NUSW}}=
\frac{p}{\sigma^2}\beta_0^2
\sum_{n\in{\mathcal{N}}_x}\sum_{m\in{\mathcal{N}}_z}\frac{1}{4\pi\lVert{\mathbf{s}}_{m,n}-{\mathbf{r}}\rVert^2},
\end{align}
where $\beta_0=\sqrt{Ae_aG_1({\mathbf{0}},{\mathbf{r}})G_2({\mathbf{0}},{\mathbf{r}})}$.
\end{theorem}

\begin{proof}
Please refer to Appendix \ref{Section_Performance_Analysis_SNR_NUSW_Expression_Theorem_Proof}.
\end{proof}

\begin{remark}
Similar to \eqref{Section_Performance_Analysis_SNR_USW_Expression}, $G_1({\mathbf{0}},{\mathbf{r}})$ and $G_2({\mathbf{0}},{\mathbf{r}})$ in \eqref{Section_Performance_Analysis_SNR_NUSW_Expression} model the effective aperture loss and the polarization loss, respectively. For the sake of brevity, these two terms are not included in \eqref{NUSW_Model_Channel_Coefficient}. As explained in Appendix \ref{Section_Performance_Analysis_SNR_NUSW_Expression_Theorem_Proof}, under the NUSW channel model, the powers radiated by different transmit antenna elements are affected by the same polarization mismatches and have the same projected antenna apertures. In order to provide a general expression for the received SNR, we have kept the terms related to the effective aperture loss and the polarization loss in Theorem \ref{Section_Performance_Analysis_SNR_NUSW_Expression_Theorem}.
\end{remark}

By letting $N_x, N_z\rightarrow\infty$, the asymptotic SNR for the NUSW model is obtained as follows.

\begin{corollary}\label{Section_Performance_Analysis_SNR_NUSW_Asymptotic_Theorem}
As $N_x, N_z\rightarrow\infty$ ($N\rightarrow\infty$), the asymptotic SNR for the NUSW model satisfies
\begin{equation}\label{Section_Performance_Analysis_SNR_NUSW_Asymptotic}
\begin{split}
\lim_{N\rightarrow\infty}\gamma_{\rm{NUSW}}\simeq\mathcal{O}(\log{N}).
\end{split}
\end{equation}
\end{corollary}

\begin{proof}
Please refer to Appendix \ref{Section_Performance_Analysis_SNR_NUSW_Asymptotic_Theorem_Proof}.
\end{proof}

\begin{remark}
The result in \eqref{Section_Performance_Analysis_SNR_NUSW_Asymptotic} suggests that by taking into account the non-uniform amplitude, the received SNR for the NUSW model no longer scales linearly with $N$, as was the case for the USW model. Instead, it scales only logarithmically with $N$. However, such a logarithmic scaling law still breaks the law of conservation of energy when $N\rightarrow\infty$. The reason for this is that when $N$ tends to infinity, the variations of the projected apertures and polarization losses across the array elements cannot be neglected, which is not considered in the NUSW model.
\end{remark}

$\bullet$ \emph{\textbf{The General Channel Model:}}
Under the general model ($i=\rm{G}$), $\left\lVert h_{m,n}^{\rm{G}}({\mathbf{r}})\right\rVert^2$ can be derived from \eqref{proposed_h}, and the received SNR is given in the following theorem.

\begin{theorem}\label{Section_Performance_Analysis_SNR_Proposed_Vert_General_Expression_Theorem}
The received SNR for the general model is given by
\begin{equation}\label{Section_Performance_Analysis_SNR_Proposed_Vert_General_Expression}
\begin{split}
\gamma_{\rm{General}}&=
\frac{p}{\sigma^2}Ae_a\sum_{n\in{\mathcal{N}}_x}\sum_{m\in{\mathcal{N}}_z}\frac{G_1({\mathbf{s}}_{m,n},{\mathbf{r}})G_2({\mathbf{s}}_{m,n},{\mathbf{r}})}{4\pi\lVert{\mathbf{s}}_{m,n}-{\mathbf{r}}\rVert^2}.
\end{split}
\end{equation}
\end{theorem}

\begin{proof}
Please refer to Appendix \ref{Section_Performance_Analysis_SNR_Proposed_Vert_General_Expression_Theorem_Proof}.
\end{proof}
Although \eqref{Section_Performance_Analysis_SNR_Proposed_Vert_General_Expression} can be utilized to calculate the SNR, deriving the power scaling law based on this expression is a challenging task. Thus, for mathematical tractability and to gain insights for system design, we consider a simplified case when the receiving-mode polarization vector at the user and the electric current induced in the UPA are both in $x$ direction, i.e., ${\bm\rho}=\hat{\mathbf J}({\mathbf{s}}_{m,n})=[1,0,0]^{\mathsf{T}}$, $\forall m,n$. In this case, the SNR can be approximated as shown in the following corollary.

\begin{corollary}\label{Section_Performance_Analysis_SNR_Proposed_Vert_General_Approximation_Expression_Y_Polar_Corollary}
When ${\bm\rho}=\hat{\mathbf J}({\mathbf{s}}_{m,n})=[1,0,0]^{\mathsf{T}}$, $\forall m,n$, the received SNR for the general model satisfies
\begin{equation}\label{Section_Performance_Analysis_SNR_Proposed_Vert_General_Approximation_Expression_Y_Polar}
\begin{split}
\gamma_{\rm{General}}&\approx\frac{pAe_{{a}}}{4\pi d^2\sigma^2}\sum_{x\in{{\mathcal{X}}_1}}\sum_{z\in{\mathcal{Z}}_1}\\
&\times\left(\frac{\Psi xz}{3(\Psi^2+x^2)\sqrt{\Psi^2+x^2+z^2}}\right.\\
&\left.+\frac{2}{3}\arctan\left(\frac{xz}{\Psi\sqrt{\Psi^2+x^2+z^2}}\right)\right),
\end{split}
\end{equation}
where ${\mathcal{X}}_1=\left\{\frac{N_x}{2}\epsilon-\Phi,\frac{N_x}{2}\epsilon+\Phi\right\}$, ${\mathcal{Z}}_1=\left\{\frac{N_z}{2}\epsilon-\Omega,\frac{N_z}{2}\epsilon+\Omega\right\}$, and $\epsilon=\frac{d}{r}$.
\end{corollary}

\begin{proof}
Please refer to Appendix \ref{Section_Performance_Analysis_SNR_Proposed_Vert_General_Approximation_Expression_Y_Polar_Corollary_Proof}.
\end{proof}
We further investigate the asymptotic SNR when $N_x, N_z\rightarrow\infty$, which is given in the following corollary.

\begin{figure}[!t]
    \setlength{\abovecaptionskip}{0pt}
    \centering
    \includegraphics[width=0.4\textwidth]{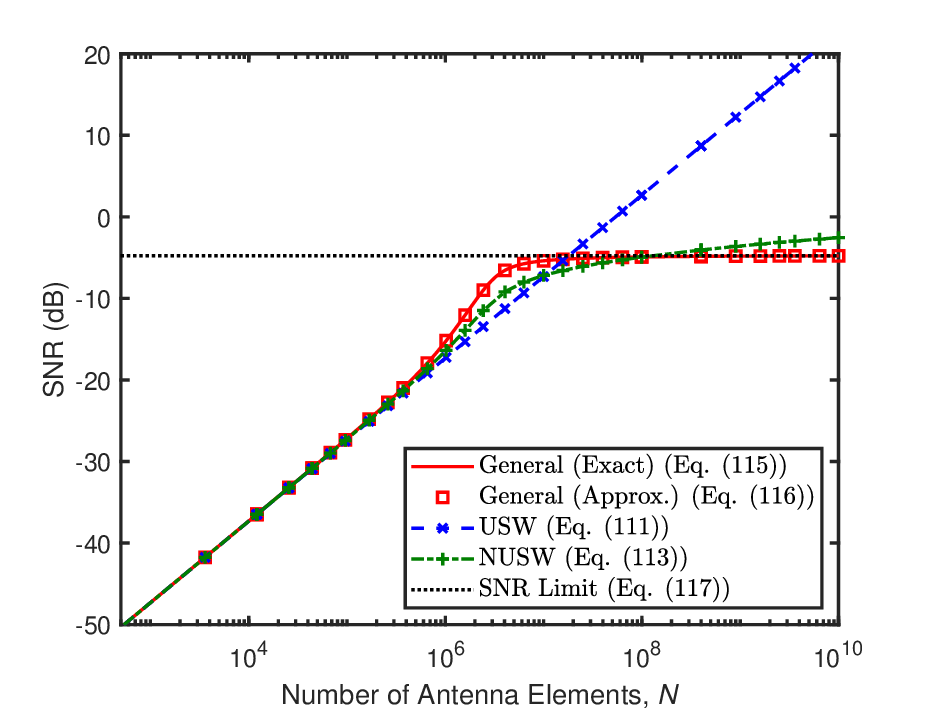}
    \caption{Comparison of SNRs for different channel models versus the number of antennas $N$ for a UPA with SPD antennas. The system is operating at a frequency of $28$ GHz. $\frac{p}{\sigma^2}=0$ dB, $N_x=N_z=\sqrt{N}$, $(\theta,\phi)=(\frac{\pi}{6},\frac{\pi}{6})$, $d=\frac{\lambda}{2}$, $r=5$ m, $A=\frac{\lambda^2}{4\pi}$, $e_a=1$, ${\bm\rho}=\hat{\mathbf J}({\mathbf{s}}_{m_x,m_z})=[1,0,0]^{\mathsf{T}}$, $\forall m_x,m_z$.}
    \label{Performance Analysis Figure: SNR_Discrete}
\end{figure}

\begin{corollary}\label{Section_Performance_Analysis_SNR_Proposed_Asymptotic_Cororllary}
Let ${\bm\rho}=\hat{\mathbf J}({\mathbf{s}}_{m,n})=[1,0,0]^{\mathsf{T}}$, $\forall m,n$. As $N_x, N_z\rightarrow\infty$ ($N\rightarrow\infty$), the asymptotic SNR for the general channel model satisfies
\begin{equation}\label{Section_Performance_Analysis_SNR_Proposed_Asymptotic}
\begin{split}
\lim_{N\rightarrow\infty}\gamma_{\rm{General}}=\frac{pAe_{{a}}}{4\pi d^2\sigma^2}\left(\frac{2}{3}\cdot\frac{\pi}{2}\cdot4\right)=\frac{pAe_{{a}}}{3 \sigma^2d^2}.
\end{split}
\end{equation}
\end{corollary}

\begin{proof}
This corollary can be directly obtained by using the following results
\begin{align}
&\lim_{N_x,N_z\rightarrow\infty}\frac{\Psi xz}{(\Psi^2+x^2)\sqrt{\Psi^2+x^2+z^2}}=0,\\
&\lim_{N_x,N_z\rightarrow\infty}\arctan\left(\frac{xz}{\Psi\sqrt{\Psi^2+x^2+z^2}}\right)=\frac{\pi}{2},
\end{align}
for $x\in{\mathcal{X}}_1$ and $z\in{\mathcal{Z}}_1$.
\end{proof}
Recalling the fact that $d\geq\sqrt{A}$ and $e_a\in(0,1]$ yields
\begin{align}\label{Section_Performance_Analysis_SNR_Proposed_Asymptotic_Limiting_Value}
\frac{pAe_{a}}{3 \sigma^2d^2}\leq\frac{pe_{a}}{3 \sigma^2}\leq\frac{p}{3 \sigma^2}<\frac{p}{\sigma^2},
\end{align}
and thus $\lim_{M\rightarrow\infty}\gamma_{\rm{General}}<\frac{p}{\sigma^2}$.

\begin{remark}
The results in \eqref{Section_Performance_Analysis_SNR_Proposed_Asymptotic} suggest that for the general channel model, the asymptotic SNR approaches a constant value $\frac{pAe_{{a}}}{3 \sigma^2d^2}$ as the UPA grows in size rather than increasing unboundedly as for the USW and NUSW channel models. Eq. \eqref{Section_Performance_Analysis_SNR_Proposed_Asymptotic_Limiting_Value} further shows that even with an infinitely large array, only at most $\frac{1}{3}$ of the total transmitted power can be received by the user, validating that $\gamma_{\rm{General}}$ satisfies the law of conservation of energy when $N\rightarrow\infty$. An intuitive explanation for why the received power is only one-third of the transmitted power, even though the UPA is infinitely large, is that each newly added BS antenna is deployed further away from the user, which reduces the effective area and increases the polarization loss \cite{Bjornson2020power}.
\end{remark}

\begin{remark}
The array occupation ratio is defined as $\zeta\triangleq\frac{A}{d^2}\in(0,1]$ and measures the fraction of the total UPA area that is occupied by array elements \cite{Lu2022communicating}. The result in \eqref{Section_Performance_Analysis_SNR_Proposed_Asymptotic} suggests that the asymptotic SNR increases linearly with $\zeta$, because more effective antenna area is available to radiate the signal power for a larger $\zeta$. As stated before, an SPD UPA becomes a CAP surface when $\zeta=1$. In this case, the asymptotic SNR is given by $\frac{pe_{a}}{3 \sigma^2}$, which will be further validated in Section \ref{Section_SNR Analysis for CAP-Antennas}.
\end{remark}

$\bullet$ \emph{\textbf{Numerical Results:}} To further verify our results, we show the SNRs obtained for the different channel models versus $N$ in {\figurename} {\ref{Performance Analysis Figure: SNR_Discrete}}. Particularly, $\gamma_{\rm{USW}}$, $\gamma_{\rm{NUSW}}$, $\gamma_{\rm{General}}$ (exact), and $\gamma_{\rm{General}}$ (approximation) are obtained with \eqref{Section_Performance_Analysis_SNR_USW_Expression}, \eqref{Section_Performance_Analysis_SNR_NUSW_Expression}, \eqref{Section_Performance_Analysis_SNR_Proposed_Vert_General_Expression}, and \eqref{Section_Performance_Analysis_SNR_Proposed_Vert_General_Approximation_Expression_Y_Polar}, respectively. The asymptotic SNR limit given in \eqref{Section_Performance_Analysis_SNR_Proposed_Asymptotic} is also included. As can be observed, for small and moderate $N$, the SNRs obtained for all models increase linearly with $N$. This is because, in this case, the user is located in the far field, where all considered models are accurate. However, for a sufficiently large $N$, the projected antenna apertures and polarization losses vary across the transmit antenna array. In this case, $\gamma_{\rm{USW}}$ and $\gamma_{\rm{NUSW}}$ violate the law of conservation of energy and approach infinity. By contrast, {\figurename} {\ref{Performance Analysis Figure: SNR_Discrete}} confirms that when $N\rightarrow\infty$, the received SNR obtained for the general model, i.e., $\gamma_{\rm{General}}$, approaches a constant, i.e., the SNR limit, and obeys the law of conservation of energy. Furthermore, {\figurename} {\ref{Performance Analysis Figure: SNR_Discrete}} shows that $N=10^6$ antennas are needed before the difference between $\gamma_{\rm{General}}$ and $\gamma_{\rm{USW}}$ (or $\gamma_{\rm{NUSW}}$) becomes noticeable. In this case, the array size is $\text{5.35}~\text{m}\times\text{5.35}~\text{m}$ which is a realistic size for future conformal arrays, e.g., deployed on facades of buildings. In conclusion, although the USW and NUSW models are applicable in some NFC application scenarios, they are not suitable for studying the asymptotic performance in the limit of large $N$.

\subsubsection{SNR Analysis for CAP Antennas}\label{Section_SNR Analysis for CAP-Antennas}
The above results obtained for SPD antennas can be directly extended to the case of CAP antennas. In the sequel, we assume that the UPA illustrated in {\figurename} {\ref{LoS_3D_Model}} is a CAP surface of size $L_x\times L_z$, which is placed on the $x$-$z$ plane and centered at the origin. The signal received by the user is given by
\begin{align}
y=\sqrt{p}h_{\rm{CAP}}({\mathbf{0}},{\mathbf{r}})s+n,
\end{align}
where $h_{\rm{CAP}}({\mathbf{0}},{\mathbf{r}})\in{\mathbb C}$ is the effective channel between the transmit CAP surface and the user. Therefore, the received SNR is given by
\begin{align}\label{Section_Performance_Analysis_SNR_Most_General_Expression_CAP}
\gamma=\frac{p}{\sigma^2}\lvert h_{\rm{CAP}}({\mathbf{0}},{\mathbf{r}}) \rvert^2,
\end{align}
where $\lvert h_{\rm{CAP}}({\mathbf{0}},{\mathbf{r}}) \rvert^2$ is the effective power gain. On the basis of \eqref{Section_Performance_Analysis_SNR_Most_General_Expression_CAP}, the received SNR for the USW, NUSW, and general channel models for CAP antennas can be derived, based on which the power scaling law in terms of the surface size, $S_{\rm{CAP}}\triangleq L_xL_z$, can be unveiled.

$\bullet$ \emph{\textbf{USW Channel Model:}}
We commence with the USW model for CAP antennas, for which the received SNR can be calculated as follows.

\begin{theorem}\label{Section_Performance_Analysis_SNR_USW_Expression_CAP_Theorem}
The received SNR for the USW model for CAP antennas is given by
\begin{equation}\label{Section_Performance_Analysis_SNR_USW_Expression_CAP}
\begin{split}
\gamma_{\rm{USW}}^{\rm{CAP}}=\frac{p}{\sigma^2} S_{\rm{CAP}}  \beta_0^2,
\end{split}
\end{equation}
where $\beta_0=\sqrt{e_a \frac{1}{4\pi r^2}
G_1({\mathbf{0}},{\mathbf{r}})G_2({\mathbf{0}},{\mathbf{r}})}$.
\end{theorem}

\begin{proof}
Please refer to Appendix \ref{Section_Performance_Analysis_SNR_UPW_Expression_CAP_Theorem_Proof}.
\end{proof}
By setting $L_x, L_z\rightarrow\infty$, i.e., $S_{\rm{CAP}}=L_xL_z\rightarrow\infty$, the asymptotic SNR for the USW model for CAP antennas is obtained and provided in the following corollary.

\begin{corollary}
As $L_x, L_z\rightarrow\infty$ ($S_{\rm{CAP}}\rightarrow\infty$), the asymptotic SNR for the USW model for CAP antennas satisfies
\begin{equation}\label{Section_Performance_Analysis_SNR_USW_Asymptotic_CAP}
\begin{split}
\lim_{S_{\rm{CAP}}\rightarrow\infty}\gamma_{\rm{USW}}^{\rm{CAP}}\simeq{\mathcal{O}}(S_{\rm{CAP}}).
\end{split}
\end{equation}
\end{corollary}

\begin{remark}
The result in \eqref{Section_Performance_Analysis_SNR_USW_Asymptotic_CAP} reveals that, for the USW model, the received SNR for CAP antennas scales linearly with the transmit surface area, which leads to the violation of the law of conservation of energy when $S_{\rm{CAP}}\rightarrow\infty$. The reason for this lies in the fact that when $S_{\rm{CAP}}$ approaches infinity, the uniform amplitude assumed in the USW model cannot capture the exact physical properties of near-field propagation.
\end{remark}

$\bullet$ \emph{\textbf{NUSW Channel Model:}} For the NUSW model for CAP antennas, we have the following result.
\begin{theorem}
The received SNR for the NUSW model for CAP antennas is given by
\begin{equation}\label{Section_Performance_Analysis_SNR_NUSW_Expression_CAP}
\begin{split}
\gamma_{\rm{NUSW}}^{\rm{CAP}}=
\int_{-\frac{L_z}{2}}^{\frac{L_z}{2}}\int_{-\frac{L_x}{2}}^{\frac{L_x}{2}}
\frac{\frac{p}{\sigma^2}\beta_0^2\frac{1}{4\pi r^2}{\rm{d}}x{\rm{d}}z}{\Psi^2+(\frac{x}{r}-\Phi)^2+(\frac{z}{r}-\Omega)^2},
\end{split}
\end{equation}
where $\beta_0=\sqrt{e_a G_1({\mathbf{0}},{\mathbf{r}})G_2({\mathbf{0}},{\mathbf{r}})}$.
\end{theorem}

\begin{proof}
The proof resembles the proofs of \textbf{Theorem} \ref{Section_Performance_Analysis_SNR_NUSW_Expression_Theorem} and \textbf{Theorem} \ref{Section_Performance_Analysis_SNR_USW_Expression_CAP_Theorem}.
\end{proof}

By setting $L_x, L_z\rightarrow\infty$, the asymptotic SNR for the NUSW model for CAP antennas is obtained and given in the following corollary.

\begin{corollary}
As $L_x, L_z\rightarrow\infty$ ($S_{\rm{CAP}}\rightarrow\infty$), for CAP antennas, the asymptotic SNR for the NUSW model satisfies
\begin{equation}\label{Section_Performance_Analysis_SNR_NUSW_Asymptotic_CAP}
\begin{split}
\lim_{S_{\rm{CAP}}\rightarrow\infty}\gamma_{\rm{NUSW}}^{\rm{CAP}}\simeq{\mathcal{O}}(\log(S_{\rm{CAP}})).
\end{split}
\end{equation}
\end{corollary}

\begin{proof}
The proof resembles the proof of \textbf{Corollary} \ref{Section_Performance_Analysis_SNR_NUSW_Asymptotic_Theorem}.
\end{proof}

\begin{remark}
The result in \eqref{Section_Performance_Analysis_SNR_NUSW_Asymptotic_CAP} reveals that by taking into account non-uniform amplitudes, the received SNR for the NUSW model increases logarithmically with $S_{\rm{CAP}}$, which differs from the linear scaling law obtained for the USW model. However, since the impact of varying projected apertures and polarization losses is not considered across the CAP, $\gamma_{\rm{NUSW}}^{\rm{CAP}}$ can exceed $\frac{p}{\sigma^2}$ when $S_{\rm{CAP}}\rightarrow\infty$, thereby violating the law of energy conservation.
\end{remark}

$\bullet$ \emph{\textbf{The General Channel Model:}} The SNR expression for the general channel model for CAP antennas is provided in the following theorem.

\begin{theorem}
For CAP antennas, the received SNR for the general channel model is given by
\begin{equation}\label{Section_Performance_Analysis_SNR_Proposed_Genarl_Expression_CAP}
\begin{split}
\gamma_{\rm{General}}^{\rm{CAP}}&=\frac{p}{\sigma^2}e_a \int_{-\frac{L_z}{2}}^{\frac{L_z}{2}}\int_{-\frac{L_x}{2}}^{\frac{L_x}{2}}\frac{1}
{4\pi\lVert{\mathbf{r}}-[x,0,z]\rVert^2}\\
&\times
G_1([x,0,z]^{\mathsf{T}},{\mathbf{r}})
G_2([x,0,z]^{\mathsf{T}},{\mathbf{r}})
{\rm{d}}x{\rm{d}}z.
\end{split}
\end{equation}
\end{theorem}

\begin{proof}
The proof resembles the proofs of \textbf{Theorem} \ref{Section_Performance_Analysis_SNR_Proposed_Vert_General_Expression_Theorem} and \textbf{Theorem} \ref{Section_Performance_Analysis_SNR_USW_Expression_CAP_Theorem}.
\end{proof}

Deriving the power scaling law based on \eqref{Section_Performance_Analysis_SNR_Proposed_Genarl_Expression_CAP} is a challenging task. Therefore, for convenience, in the following corollary, we consider the special case of ${\bm\rho}=\hat{\mathbf J}([x,0,z]^{\mathsf{T}})=[1,0,0]^{\mathsf{T}}$, $\forall x,z$, based on which some interesting conclusions can be drawn.

\begin{corollary}
When ${\bm\rho}=\hat{\mathbf J}([x,0,z]^{\mathsf{T}})=[1,0,0]^{\mathsf{T}}$, $\forall x,z$, the received SNR for the general channel model for CAP antennas satisfies
\begin{equation}\label{Section_Performance_Analysis_SNR_Proposed_Vert_General_Approximation_Expression_Y_Polar_CAP}
\begin{split}
\gamma_{\rm{General}}^{\rm{CAP}}&=\frac{pe_{{a}}}{4\pi\sigma^2}\sum_{x\in{{\mathcal{X}}}_2}\sum_{z\in{\mathcal{Z}}_2}\\
&\times\left(\frac{\Psi xz}{3(\Psi^2+x^2)\sqrt{\Psi^2+x^2+z^2}}\right.\\
&\left.+\frac{2}{3}\arctan\left(\frac{xz}{\Psi\sqrt{\Psi^2+x^2+z^2}}\right)\right),
\end{split}
\end{equation}
where ${\mathcal{X}}_2=\left\{\frac{L_x}{2r}-\Phi,\frac{L_x}{2r}+\Phi\right\}$ and ${\mathcal{Z}}_2=\left\{\frac{L_z}{2r}-\Omega,\frac{L_z}{2r}+\Omega\right\}$. As $L_x,L_z\rightarrow\infty$ ($S_{\rm{CAP}}\rightarrow\infty$), the asymptotic SNR satisfies
\begin{equation}\label{Section_Performance_Analysis_SNR_Proposed_Asymptotic_CAP}
\begin{split}
\lim_{S_{\rm{CAP}}\rightarrow\infty}\gamma_{\rm{General}}^{\rm{CAP}}=\frac{pe_{a}}{4\pi \sigma^2}\left(\frac{2}{3}\cdot\frac{\pi}{2}\cdot4\right)=\frac{p e_{{a}}}{3 \sigma^2}.
\end{split}
\end{equation}
\end{corollary}

\begin{proof}
The proof resembles the proofs of \textbf{Corollary} \ref{Section_Performance_Analysis_SNR_Proposed_Vert_General_Approximation_Expression_Y_Polar_Corollary} and \textbf{Corollary} \ref{Section_Performance_Analysis_SNR_Proposed_Asymptotic_Cororllary}.
\end{proof}
Recalling that $e_a\in(0,1]$, we obtain
\begin{align}\label{Section_Performance_Analysis_SNR_Proposed_Asymptotic_CAP_Limiting_Value}
\frac{p e_{{a}}}{3 \sigma^2}<\frac{p}{3 \sigma^2}<\frac{p}{\sigma^2},
\end{align}
and thus $\lim_{S_{\rm{CAP}}\rightarrow\infty}\gamma_{\rm{General}}^{\rm{CAP}}<\frac{p}{\sigma^2}$.

\begin{remark}
The results in \eqref{Section_Performance_Analysis_SNR_Proposed_Asymptotic_CAP} and \eqref{Section_Performance_Analysis_SNR_Proposed_Asymptotic_CAP_Limiting_Value} suggest that $\gamma_{\rm{General}}^{\rm{CAP}}$ satisfies the law of conservation of energy when $S_{\rm{CAP}}\rightarrow\infty$. Note that \eqref{Section_Performance_Analysis_SNR_Proposed_Asymptotic_CAP} can be also obtained from \eqref{Section_Performance_Analysis_SNR_Proposed_Asymptotic} by setting the inter-element distance $d$ to $d=\sqrt{A}$, i.e., $\zeta=\frac{A}{d^2}=1$. This is expected since a CAP surface is equivalent to an SPD UPA, which is fully covered by array elements.
\end{remark}

\begin{figure}[!t]
\setlength{\abovecaptionskip}{0pt}
\centering
\includegraphics[width=0.4\textwidth]{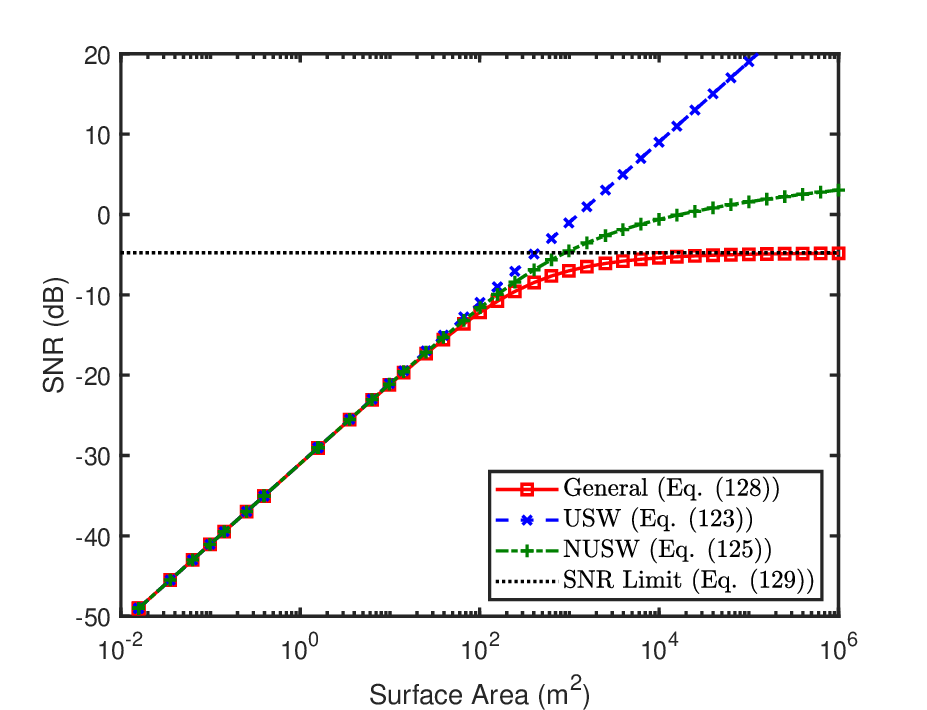}
\caption{Comparison of SNRs for different channel models versus the transmit surface area $S_{\rm{CAP}}$ for CAP antennas. $\frac{p}{\sigma^2}=0$ dB, $L_x=L_z=\sqrt{S_{\rm{CAP}}}$, $(\theta,\phi)=(\frac{\pi}{2},\frac{\pi}{2})$, $r=10$ m, $e_a=1$, ${\bm\rho}=\hat{\mathbf J}([x,0,z]^{\mathsf{T}})=[1,0,0]^{\mathsf{T}}$, $\forall x,z$.}
\label{Performance Analysis Figure: SNR_CAP}
\end{figure}

\begin{table*}[!t]
\caption{Comparison of SNRs Obtained for Different Near-Field LoS Channel Models.}
\label{tab:Section_Performance_Analysis_SNR_Table}
\centering
\begin{tabular}{!{\vrule width1pt}c!{\vrule width1pt}c!{\vrule width1pt}c!{\vrule width1pt}c!{\vrule width1pt}c!{\vrule width1pt}}
\Xhline{1pt} 
Antenna                                                                 & Channel Model & SNR Expression & Scaling Law & Conservation of Energy \\ \Xhline{1pt} 
\multirow{3}{*}{\begin{tabular}[c]{@{}c@{}}SPD\\ Antenna\end{tabular}} & USW            & \eqref{Section_Performance_Analysis_SNR_USW_Expression}              & \eqref{Section_Performance_Analysis_SNR_USW_Asymptotic}: $\mathcal{O}(N)$           & Violated when $N\rightarrow\infty$                      \\ \cline{2-5}
                                                                        & NUSW           & \eqref{Section_Performance_Analysis_SNR_NUSW_Expression}              & \eqref{Section_Performance_Analysis_SNR_NUSW_Asymptotic}: $\mathcal{O}(\log{N})$           & Violated when $N\rightarrow\infty$                      \\ \cline{2-5}
                                                                        & General       & \eqref{Section_Performance_Analysis_SNR_Proposed_Vert_General_Expression}, \eqref{Section_Performance_Analysis_SNR_Proposed_Vert_General_Approximation_Expression_Y_Polar}              & \eqref{Section_Performance_Analysis_SNR_Proposed_Asymptotic}: $\mathcal{O}(1)$           & Satisfied (Always)                      \\ \Xhline{1pt} 
\multirow{3}{*}{\begin{tabular}[c]{@{}c@{}}CAP\\ Antenna\end{tabular}} & USW            & \eqref{Section_Performance_Analysis_SNR_USW_Expression_CAP}              & \eqref{Section_Performance_Analysis_SNR_USW_Asymptotic_CAP}: $\mathcal{O}({{{S}}_{\rm{CAP}}})$           & Violated when ${{{S}}_{\rm{CAP}}}\rightarrow\infty$                       \\ \cline{2-5}
                                                                        & NUSW           & \eqref{Section_Performance_Analysis_SNR_NUSW_Expression_CAP}              & \eqref{Section_Performance_Analysis_SNR_NUSW_Asymptotic_CAP}: $\mathcal{O}(\log{{{S}}_{\rm{CAP}}})$           & Violated when ${{{S}}_{\rm{CAP}}}\rightarrow\infty$                      \\ \cline{2-5}
                                                                        & General       & \eqref{Section_Performance_Analysis_SNR_Proposed_Genarl_Expression_CAP}, \eqref{Section_Performance_Analysis_SNR_Proposed_Vert_General_Approximation_Expression_Y_Polar_CAP}              & \eqref{Section_Performance_Analysis_SNR_Proposed_Asymptotic_CAP}: $\mathcal{O}(1)$           & Satisfied (Always)                      \\ \Xhline{1pt} 
\end{tabular}
\end{table*}

\textbf{Numerical Results:} To further verify our results, we show the SNRs obtained for the considered channel models versus $S_{\rm{CAP}}$ in {\figurename} {\ref{Performance Analysis Figure: SNR_CAP}}. Particularly, $\gamma_{\rm{USW}}^{\rm{CAP}}$, $\gamma_{\rm{NUSW}}^{\rm{CAP}}$, and $\gamma_{\rm{General}}^{\rm{CAP}}$ are calculated based on \eqref{Section_Performance_Analysis_SNR_USW_Expression_CAP}, \eqref{Section_Performance_Analysis_SNR_NUSW_Expression_CAP}, and \eqref{Section_Performance_Analysis_SNR_Proposed_Vert_General_Approximation_Expression_Y_Polar_CAP}, respectively. The asymptotic SNR limit given in \eqref{Section_Performance_Analysis_SNR_Proposed_Asymptotic_CAP} is also included as a baseline. As can be observed from {\figurename} {\ref{Performance Analysis Figure: SNR_CAP}}, for small and moderate $S_{\rm{CAP}}$, the SNRs obtained for all considered channel models increase linearly with $S_{\rm{CAP}}$. This is because the user is located in the far field, where all considered models are accurate. However, since the impact of the varying projected apertures and polarization losses is ignored, the asymptotic values of $\gamma_{\rm{USW}}^{\rm{CAP}}$ and $\gamma_{\rm{NUSW}}^{\rm{CAP}}$ exceed the transmit SNR $\frac{p}{\sigma^2}$, therefore breaking the law of conservation of energy. Our observations from Figs. {\ref{Performance Analysis Figure: SNR_Discrete}} and {\figurename} {\ref{Performance Analysis Figure: SNR_CAP}} highlight the importance of correctly modelling the variations of the free-space path losses, projected apertures, and polarization losses across the antenna array elements and CAP surface, respectively, when studying the asymptotic limits of the SNR.

\subsubsection{Summary of the Analytical Results}
In Table \ref{tab:Section_Performance_Analysis_SNR_Table}, we summarize our analytical results for the SNR and scaling law. In the third column of the table, the notation ${\mathcal{O}}(1)$ is used to indicate that the asymptotic SNR is a constant. Although this subsection has focused on MISO transmission, the developed results can be extended to MIMO and multiuser systems. Following a similar approach as for obtaining the results in Table \ref{tab:Section_Performance_Analysis_SNR_Table}, the SINR and power scaling law for MIMO and multiuser transmission can be derived.
\begin{figure}[!t]
\setlength{\abovecaptionskip}{0pt}
\centering
\includegraphics[width=0.4\textwidth]{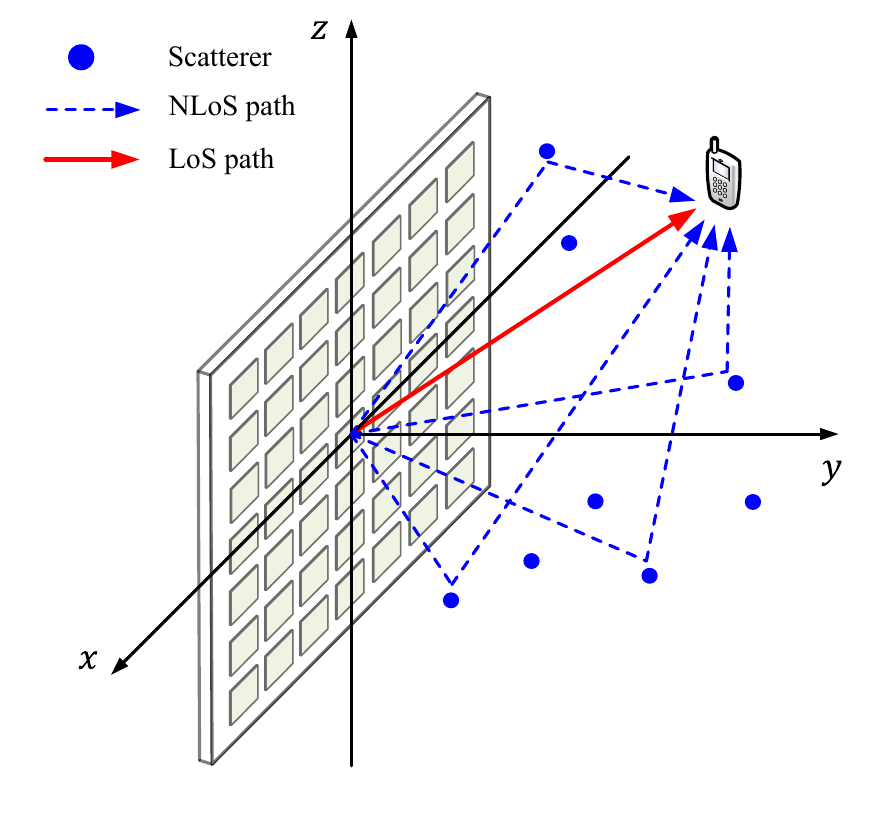}
\caption{System model for NFC for multipath channels.}
\label{Multipath_3D_Model}
\end{figure}
\subsection{Performance Analysis for Statistical Multipath Near-Field Channels}
Having analyzed the performance for the LoS near-field channel, we now shift our attention to statistical multipath near-field channels, as depicted in {\figurename} {\ref{Multipath_3D_Model}}.
\subsubsection{Channel Statistics}
If both LoS and NLOS components are present, the multipath channel coefficients can be modelled as in \eqref{H_NFC_MISO}:
\vspace{-0.2cm}
\begin{align}
\mathbf{h} = {\beta \mathbf{a}(\mathbf{r})} + \sum_{\ell=1}^L {\tilde{\beta}_{\ell} \mathbf{a}(\tilde{\mathbf{r}}_{\ell})},
\end{align}
where ${\overline{\mathbf{h}}}\triangleq{\beta \mathbf{a}(\mathbf{r})}$ is the LoS component and ${\tilde{\mathbf{h}}}_{\ell}\triangleq{\tilde{\beta}_{\ell} \mathbf{a}(\tilde{\mathbf{r}}_{\ell})}$ is the NLoS component generated by the $\ell$-th scatterer. By referring to {\figurename} {\ref{Multipath_3D_Model}}, we rewrite the NLoS component ${\tilde{\mathbf{h}}}_{\ell}$ as follows:
\begin{equation}
{\tilde{\mathbf{h}}}_{\ell} = \alpha_{\ell}{\mathbf{h}}_{\ell}(\mathbf{r}_{\ell})h_{\ell}({\mathbf{r}}_{\ell},\mathbf{r}),
\end{equation}
where ${\mathbf{h}}_{\ell}(\mathbf{r}_{\ell})\in{\mathbb C}^{N\times1}$ is the channel vector between the $\ell$-th scatterer and the BS, $h_{\ell}({\mathbf{r}}_{\ell},\mathbf{r})\in{\mathbb C}$ is the channel coefficient between the user and the $\ell$-th scatterer, and $\alpha_{\ell}$ models the complex random reflection coefficient of the $\ell$-th scatterer. The complex gains $\{\alpha_{\ell}\}_{\ell=1}^{L}$ are generally modelled as independently complex Gaussian distributed variables with $\alpha_{\ell}\sim{\mathcal{CN}}(0,\sigma_{\ell}^2)$, where $\sigma_{\ell}^2$ represents the intensity attenuation caused by the $\ell$-th scatterer \cite{Emil2017massive,Saleh1987communicating}. On this basis, the multipath channel can be modelled as follows
\begin{equation}\label{Section_Performance_Analysis_Statistical_Channel}
\begin{split}
{\mathbf{h}}\sim{\mathcal{CN}}(\overline{\mathbf{h}},\mathbf{R}),
\end{split}
\end{equation}
where ${\overline{\mathbf{h}}}=\mathbb{E}\{\mathbf{h}\}$ denotes the channel mean and
\begin{equation}
\begin{split}
{\mathbf{R}}&={\mathbb{E}}\{({\mathbf{h}}-{\overline{\mathbf{h}}})({\mathbf{h}}-{\overline{\mathbf{h}}})^{\mathsf{H}}\}\\
&=\sum_{{\ell}=1}^{L}\sigma_{\ell}^2|h_{\ell}({\mathbf{r}}_{\ell},\mathbf{r})|^2{\mathbf{h}}_{\ell}(\mathbf{r}_{\ell}){\mathbf{h}}_{\ell}^{\mathsf{H}}(\mathbf{r}_{\ell})
\end{split}
\end{equation}
is the correlation matrix. Eqn. \eqref{Section_Performance_Analysis_Statistical_Channel} corresponds to the Rician fading model. As shown in {\figurename} {\ref{Multipath_3D_Model}}, the scatterers are located in the near field of the BS \cite{dong2022near}. As a result, the LoS channels $\overline{\mathbf{h}}$, $\{h_{\ell}({\mathbf{r}}_{\ell},\mathbf{r})\}_{\ell=1}^{L}$, and $\{{\mathbf{h}}_{\ell}(\mathbf{r}_{\ell})\}_{\ell=1}^{L}$ can be described by the near-field LoS models presented in Section \ref{NFC_Tutorial_Channel_Model}. Different LoS channel models yield different correlation matrices $\mathbf{R}$, but the statistics of the resulting multipath MISO channel always follow \eqref{Section_Performance_Analysis_Statistical_Channel}.

NFC generally occurs in mmWave and sub-THz bands; therefore, the resulting channels are sparsely-scattered ($N\gg L$) and dominated by LoS propagation, i.e., $\lVert{\overline{\mathbf{h}}}\rVert^2\gg \sum_{\ell=1}^{L}\sigma_{\ell}^2|h_{\ell}({\mathbf{r}}_{\ell},\mathbf{r})|^2\lVert{\mathbf{h}}_{\ell}(\mathbf{r}_{\ell})\rVert^2$. Consequently, matrix $\mathbf{R}$ is generally rank-deficient. In practical wireless propagation environments, the LoS path might be blocked by obstacles. In this case, the mean of $\mathbf{h}$ equals zero, and \eqref{Section_Performance_Analysis_Statistical_Channel} degrades to a Rayleigh fading channel model with ${\mathbf{R}}={\mathbb{E}}\{{\mathbf{h}}{\mathbf{h}}^{\mathsf{H}}\}$. Although the NFC channel is generally dominated by its LoS component, considering the scenario with the LoS path blocked is also of interest and has theoretical significance as a limiting case. Hence, besides the Rician distribution, the Rayleigh distribution has also become one of the commonly accepted models for NFC channel modeling; see \cite{dong2022near,pizzo2020spatially,pizzo2022fourier,ji2023extra} and the references therein. The statistical model in \eqref{Section_Performance_Analysis_Statistical_Channel} is not only applicable for SPD antennas but also for CAP antennas if the CAP surface is sampled into an SPD antenna array \cite{pizzo2020spatially,pizzo2022fourier,ji2023extra}. For CAP surfaces without spatial sampling, the modelling of the resulting multipath channels is an open problem. Hence, we will focus our efforts on SPD antennas.

{\textbf{Near-Field Channel Statistics vs. Far-Field Channel Statistics:}} In NFC, different antenna elements of the same array suffer from significantly different path lengths and non-uniform channel gains. Due to this property, correlation matrix $\mathbf{R}$ cannot be simplified to an identity matrix even when the antenna array is half-wavelength-spaced and deployed in an isotropic scattering environment \cite{dong2022near}. This fact means that correlated fading models are physically more realistic for NFC. In other words, adopting a correlated fading model where ${\mathbf{R}}\ne{\mathbf{I}}$ to evaluate NFC performance is required. On the other hand, for performance analysis of conventional far-field communications, usually the independent and identically distributed (i.i.d.) Rayleigh fading channel model with ${\mathbf{R}}={\mathbf{I}}$ has been adopted in most theoretical research.

Considering the above facts, we adopt the correlated fading model for analyzing the statistical near-field performance in the presence of NLoS channels. In this case, channel vector ${\mathbf{h}}$ can be statistically described as follows
\begin{equation}
\mathbf{h}=\overline{\mathbf{h}}+\mathbf{R}^{\frac{1}{2}}\tilde{\mathbf{h}},
\end{equation}
where $\tilde{\mathbf{h}}\sim{\mathcal{CN}}(\mathbf{0},\mathbf{I})$. We next evaluate NFC performance for this statistical channel model by analyzing the OP, ECC, and EMI. For ease of understanding, we first analyze these metrics for correlated MISO Rayleigh channels with ${\overline{\mathbf{h}}}=\mathbf{0}$ and then extend the derived results to correlated MISO Rician channels with ${\overline{\mathbf{h}}}\ne{\mathbf{0}}$. Our aim is to provide a general analytical framework for analyzing the three metrics, i.e., OP, ECC, and EMI, without knowledge of the structure of correlation matrix ${\mathbf{R}}$. In other words, our proposed analytical framework is applicable for any $\mathbf{R}\succeq {\mathbf{0}}$.

\subsubsection{Analysis of the OP for Rayleigh Channels}\label{Section_Performance_Analysis_OP_Rayleigh_Part}
Let $\mathcal{R}$ denote the required communication rate. An outage occurs when the actual communication rate $\log_2(1+\gamma)$ is less than $\mathcal{R}$. Accordingly, the OP for Rayleigh channels can be written as follows:
\begin{equation}\label{Section_Performance_Analysis_OP_Definition_MISO}
\begin{split}
\mathcal{P}_{\rm{rayleigh}}&=\Pr(\log_2(1+\gamma)<\mathcal{R}).
\end{split}
\end{equation}
Inserting \eqref{Section_Performance_Analysis_Statistical_Channel_MISO_SNR_Expression} into \eqref{Section_Performance_Analysis_OP_Definition_MISO} yields
\begin{equation}\label{Section_Performance_Analysis_OP_Calculation_MISO}
\mathcal{P}_{\rm{rayleigh}}=\Pr\left(\lVert{\mathbf{h}}\rVert^2<\frac{2^{\mathcal{R}}-1}{p/\sigma^2}\right)=
F_{\lVert{\mathbf{h}}\rVert^2}\left(\frac{2^{\mathcal{R}}-1}{p/\sigma^2}\right),
\end{equation}
where $F_{\lVert{\mathbf{h}}\rVert^2}\left(\cdot\right)$ is the cumulative distribution function (CDF) of $\lVert{\mathbf{h}}\rVert^2$. The OP can be analyzed using the following basic analytical framework, which involves three steps.

$\bullet$ \emph{\textbf{Step 1 - Analyzing the Statistics of the Channel Gain: }}
In the first step, we calculate the probability density function (PDF) and CDF of $\lVert{\mathbf{h}}\rVert^2$ to obtain a closed-form expression for $\mathcal{P}$. The main results are summarized as follows.

\begin{lemma}\label{Section_Performance_Analysis_OP_SNR_CDF_PDF_Lemma}
For correlated MISO Rayleigh channel ${\mathbf{h}}\sim{\mathcal{CN}}({\mathbf{0}},{\mathbf{R}})$, the PDF and CDF of $\lVert{\mathbf{h}}\rVert^2$ are respectively given by
\begin{align}
&f_{\lVert{\mathbf{h}}\rVert^2}\left(x\right)=
\frac{\lambda_{\min}^{r_{\mathbf{R}}}}{\prod_{i=1}^{r_{\mathbf{R}}}{\lambda}_i}\sum_{k=0}^{\infty}\frac{\psi_k x^{r_{\mathbf{R}}+k-1}}{{\lambda}_{\min}^{r_{\mathbf{R}}+k}\Gamma\left(r_{\mathbf{R}}+k\right)}e^{-\frac{x}{{\lambda}_{\min}}},\label{Section_Performance_Analysis_OP_SNR_PDF}\\
&F_{\lVert{\mathbf{h}}\rVert^2}\left(x\right)
=\frac{\lambda_{\min}^{r_{\mathbf{R}}}}{\prod_{i=1}^{r_{\mathbf{R}}}{\lambda}_i}\sum_{k=0}^{\infty}\frac{\psi_k \Upsilon\left(k+r_{\mathbf{R}},{x}/{{\lambda}_{\min}}\right)}{\Gamma\left(r_{\mathbf{R}}+k\right)},\label{Section_Performance_Analysis_OP_SNR_CDF}
\end{align}
where $r_{\mathbf{R}}$ is the rank of matrix $\mathbf{R}$, $\{\lambda_i>0\}_{i=1}^{r_{\mathbf{R}}}$ are the positive eigenvalues of matrix $\mathbf{R}$, ${\lambda}_{\min}=\min_{i=1,\ldots,\lambda_{\mathbf{R}}}{\lambda}_i$, $\Gamma\left(z\right)=\int_{0}^{\infty}e^{-t}t^{z-1}{\rm d}t$ is the Gamma function, $\Upsilon\left(s,x\right)=\int_{0}^{x}t^{s-1}e^{-t}{\rm d}t$ is the lower incomplete Gamma function, ${\psi _0} = 1$, and the $\psi_k$ ($k\geq1$) can be calculated recursively as follows
\begin{align}
{\psi _{k + 1}} = \sum\nolimits_{i = 1}^{k + 1} {\left[ {\sum\nolimits_{j = 1}^{r_{\mathbf{R}}} {{{\left( {1 - {{{{\lambda}_{\min }}}}/{{{{\lambda}_{j}}}}} \right)}^i}} } \right]} \frac{\psi _{k + 1 - i}}{{k + 1}}.
\end{align}
\end{lemma}

\begin{proof}
Please refer to Appendix \ref{Section_Performance_Analysis_OP_SNR_CDF_PDF_Lemma_Proof}.
\end{proof}

$\bullet$ \emph{\textbf{Step 2 - Deriving a Closed-Form Expression of the OP: }}
In the second step, we exploit the CDF of $\lVert{\mathbf{h}}\rVert^2$ to calculate the OP, which yields the following theorem.

\begin{theorem}
The OP of the considered system is given by
\begin{align}\label{Section_Performance_Analysis_OP_Explicit_Expression}
\mathcal{P}_{\rm{rayleigh}}=
\frac{\lambda_{\min}^{r_{\mathbf{R}}}}{\prod_{i=1}^{r_{\mathbf{R}}}{\lambda}_i}\sum_{k=0}^{\infty}\frac{\psi_k \Upsilon\left(k+r_{\mathbf{R}},\frac{2^{\mathcal{R}}-1}{p/\sigma^2\lambda_{\min}}\right)}{\Gamma\left(r_{\mathbf{R}}+k\right)}.
\end{align}
\end{theorem}

\begin{proof}
This theorem can be directly proved by substituting \eqref{Section_Performance_Analysis_OP_SNR_CDF} into \eqref{Section_Performance_Analysis_OP_Calculation_MISO}.
\end{proof}

$\bullet$ \emph{\textbf{Step 3 - Deriving a High-SNR Approximation of the OP: }}
In the last step, we investigate the asymptotic behaviour of the OP in the high-SNR regime, i.e., $p\rightarrow\infty$, in order to provide more insights for system design. The main results are summarized in the following corollary.
\begin{corollary}\label{Section_Performance_Analysis_OP_Asymptotic_Expression_Standard_Theorem}
The asymptotic OP in the high-SNR regime can be expressed in the following form:
\begin{align}\label{Section_Performance_Analysis_OP_Asymptotic_Expression_Standard}
\lim\limits_{p\rightarrow\infty}{\mathcal{P}}\simeq({\mathcal{G}}_{\rm{a}}\cdot p)^{-{\mathcal{G}}_{\rm{d}}},
\end{align}
where ${\mathcal{G}}_{\rm{a}}=\frac{\left(r_{\mathbf{R}}!\prod_{i=1}^{r_{\mathbf{R}}}{\lambda}_i\right)^{{1}/{r_{\mathbf{R}}}}}{\sigma^2(2^{\mathcal{R}}-1)}$ and ${\mathcal{G}}_{\rm{d}}=r_{\mathbf{R}}$.
\end{corollary}

\begin{proof}
Please refer to Appendix \ref{Section_Performance_Analysis_OP_Asymptotic_Expression_Standard_Theorem_Proof}.
\end{proof}
In \eqref{Section_Performance_Analysis_OP_Asymptotic_Expression_Standard}, ${\mathcal{G}}_{\rm{a}}$ is referred to as the \emph{array gain}, and ${\mathcal{G}}_{\rm{d}}$ is referred to as the diversity gain or diversity order \cite{Wang2003_TCOM}. The diversity order ${\mathcal{G}}_{\rm{d}}$ determines the slope of the OP as a function of the transmit power, at high SNR, depicted in a log-log scale. On the other hand, ${\mathcal{G}}_{\rm{a}}$ (in decibels) specifies the power gain of the actual OP compared to a benchmark-OP of $p^{-{\mathcal{G}}_{\rm{d}}}$. Note that the OP can be improved by increasing ${\mathcal{G}}_{\rm{a}}$ or ${\mathcal{G}}_{\rm{d}}$.
\begin{remark}
The results in \textbf{Corollary} \ref{Section_Performance_Analysis_OP_Asymptotic_Expression_Standard_Theorem} indicate that in the high-SNR regime, the slope and power gain of the OP is given by $r_{\mathbf{R}}$ and $\frac{\left(r_{\mathbf{R}}!\prod_{i=1}^{r_{\mathbf{R}}}{\lambda}_i\right)^{{1}/{r_{\mathbf{R}}}}}{\sigma^22^{\mathcal{R}}-1}$, respectively.
\end{remark}

\begin{figure}[!t]
\setlength{\abovecaptionskip}{0pt}
\centering
\includegraphics[width=0.4\textwidth]{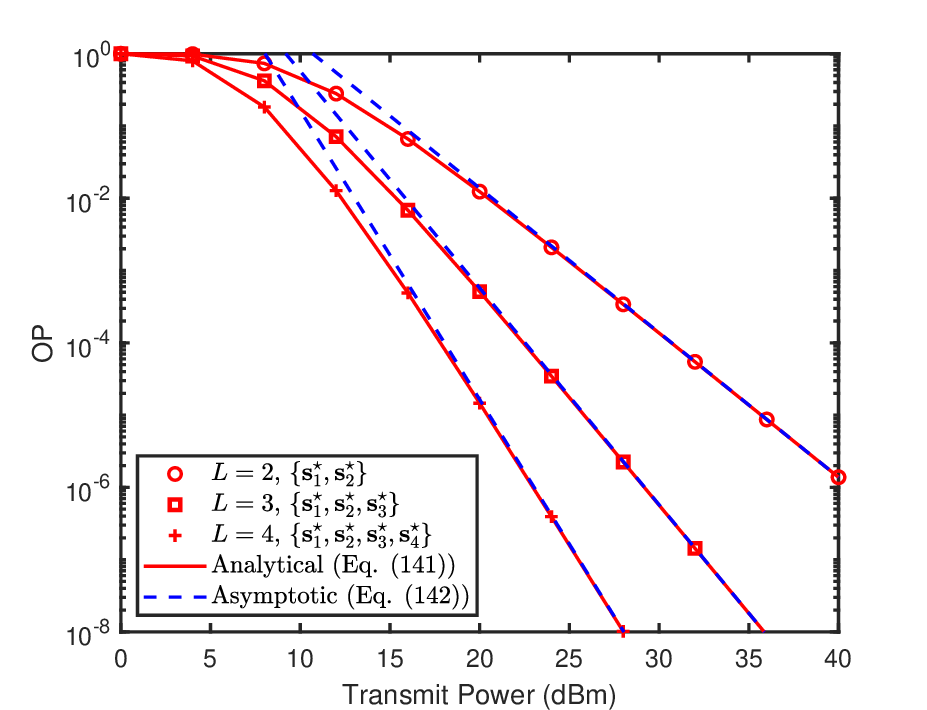}
\caption{Outage probability for correlated MISO Rayleigh channels and data rate $\mathcal{R}=1$ bps/Hz. The system is operating at $28$ GHz. $N_x=N_z=33$, $d=\frac{\lambda}{2}$, $(\theta,\phi)=(\frac{\pi}{2},\frac{\pi}{2})$, $r=4$ m, $A=\frac{\lambda^2}{4\pi}$, $e_a=1$, ${\bm\rho}=\hat{\mathbf J}({\mathbf{s}}_{m,n})=[1,0,0]^{\mathsf{T}}$, $\forall m,n$, $\sigma^2=-90$ dBm, and $\sigma_{\ell}^2=1$, $\forall \ell$. The coordinates (in meters) of the scatterers are given by $
\mathbf{s}_{1}^{\star}=[\frac{7}{10}\sin{\frac{\pi}{4}},r,\frac{7}{10}\cos{\frac{\pi}{4}}]^{\mathsf{T}}$, $\mathbf{s}_{2}^{\star}=[\frac{4}{5}\sin{\frac{\pi}{8}},r,\frac{4}{5}\cos{\frac{\pi}{8}}]^{\mathsf{T}}$, $
\mathbf{s}_{3}^{\star}=[\frac{13}{20}\sin{\frac{47\pi}{64}},r,\frac{13}{20}\cos{\frac{47\pi}{64}}]^{\mathsf{T}}$, and
$\mathbf{s}_{4}^{\star}=[\frac{2}{5}\sin{\frac{\pi}{7}},r,\frac{2}{5}\cos{\frac{\pi}{7}}]^{\mathsf{T}}$. The LoS channels $\overline{\mathbf{h}}$, $\{h_{\ell}({\mathbf{r}}_{\ell},\mathbf{r})\}_{\ell=1}^{L}$, and $\{{\mathbf{h}}_{\ell}(\mathbf{r}_{\ell})\}_{\ell=1}^{L}$ follow the USW channel model.}
\label{Performance Analysis Figure: OP_Discrete}
\end{figure}

$\bullet$ \emph{\textbf{Numerical Results:}} To illustrate the above derivations, we show the OP as a function of the transmit power, $p$, in {\figurename} {\ref{Performance Analysis Figure: OP_Discrete}} for the USW channel model and various values of $L$. The simulation results are denoted by markers. The analytical and asymptotic results are calculated using \eqref{Section_Performance_Analysis_OP_Explicit_Expression} and \eqref{Section_Performance_Analysis_OP_Asymptotic_Expression_Standard}, respectively. {\figurename} {\ref{Performance Analysis Figure: OP_Discrete}} reveals that the analytical results are in good agreement with the simulated results, and the derived asymptotic results approach the numerical results in the high-SNR regime ($p\rightarrow\infty$). Furthermore, it can be observed that larger values of $L$ yield higher diversity orders.

$\bullet$ \emph{\textbf{Extension to MIMO Case:}} Note that the definition presented in \eqref{Section_Performance_Analysis_OP_Definition_MISO} only applies to MISO channels. The extension to the single-user MIMO case with isotropic inputs is given by
\begin{equation}\label{Section_Performance_Analysis_OP_Definition_SU_MIMO}
\mathcal{P}=\Pr(\log_2\det({\mathbf{I}}+p/\sigma^2{\mathbf{H}}{\mathbf{H}}^{\mathsf{H}})<\mathcal{R}),
\end{equation}
where ${\mathbf{H}}\in{\mathbb{C}}^{N_{R}\times N_T}$ is the channel matrix with $N_R$ and $N_T$ denoting the numbers of receive and transmit antennas, respectively. The evaluation of the OP in \eqref{Section_Performance_Analysis_OP_Definition_SU_MIMO} requires the application of tools from random matrix theory; please refer to \cite{Tulino2004_FTC, Muller2013_Arxiv,Couillet2011} and the references therein for more details. The existing literature shows that the asymptotic OP for MIMO channels in the high-SNR regime also follows the standard form given in \eqref{Section_Performance_Analysis_OP_Asymptotic_Expression_Standard} (see, e.g., \cite{Yang2021_TBC}).
\subsubsection{Analysis of the ECC for Rayleigh Channels}
Having analyzed the OP, we turn our attention to the ECC. Achieving the ECC requires the BS to adaptively adjust its coding rate to the channel capacity at the beginning of each coherence interval. The ECC mathematically equals the mean of the instantaneous channel capacity $\log_2(1+p/\sigma^2\lVert{\mathbf{h}}\rVert^2)$, which can be expressed as follows:
\begin{equation}\label{Section_Performance_Analysis_ECC_Definition_MISO}
\begin{split}
\bar{\mathcal{C}}_{\rm{rayleigh}}&={\mathbb{E}}\{\log_2(1+p/\sigma^2\lVert{\mathbf{h}}\rVert^2)\}\\
&=\int_0^{\infty}\log_2(1+p/\sigma^2 x)f_{\lVert{\mathbf{h}}\rVert^2}(x){\rm{d}}x.
\end{split}
\end{equation}
The analysis of the ECC also comprises three steps.

$\bullet$ \emph{\textbf{Step 1 - Analyzing the Statistics of the Channel Gain:}} 
In the first step, we analyze the statistics of the channel gain $\lVert{\mathbf{h}}\rVert^2$. The results are provided in \textbf{Lemma} \ref{Section_Performance_Analysis_OP_SNR_CDF_PDF_Lemma}.

$\bullet$ \emph{\textbf{Step 2 - Deriving a Closed-Form Expression for the ECC:}} 
In the second step, we exploit the PDF of $\lVert{\mathbf{h}}\rVert^2$ to derive a closed-form expression for $\bar{\mathcal{C}}_{\rm{rayleigh}}$, which yields the following theorem.

\begin{theorem}
The ECC can be expressed in closed form as follows:
\begin{equation}\label{Section_Performance_Analysis_ECC_Calculation_MISO}
\begin{split}
\bar{\mathcal{C}}_{\rm{rayleigh}}&=\frac{\lambda_{\min}^{r_{\mathbf{R}}}}{\prod_{i=1}^{r_{\mathbf{R}}}{\lambda}_i}\sum_{k=0}^{\infty}\sum_{\mu=0}^{r_{\mathbf{R}}+k-1}
\frac{\psi_k/\ln{2}}{(r_{\mathbf{R}}+k-1-\mu)!}\\
&\times\left[\frac{(-1)^{r_{\mathbf{R}}+k-\mu}e^{\frac{1}{{p/\sigma^2}{\lambda_{\min}}}}}{({p/\sigma^2}{\lambda_{\min}})^{r_{\mathbf{R}}+k-1-\mu}}{\rm{Ei}}\left({\frac{-1}{{p/\sigma^2}\lambda_{\min}}}\right)\right.\\
&+\left.\sum_{u=1}^{r_{\mathbf{R}}+k-1-\mu}(u-1)!\left(\frac{-1}{p/\sigma^2\lambda_{\min}}\right)^{r_{\mathbf{R}}+k-1-\mu-u}\right].
\end{split}
\end{equation}
where ${\rm{Ei}}\left(x\right)=-\int_{-x}^{\infty}{e}^{-t}t^{-1}{\rm{d}}t$ denotes the exponential integral function.
\end{theorem}

\begin{proof}
This theorem is proved by substituting \eqref{Section_Performance_Analysis_OP_SNR_PDF} into \eqref{Section_Performance_Analysis_ECC_Definition_MISO} and solving the resulting integral with the aid of \cite[Eq. (4.337.5)]{Ryzhik2007}.
\end{proof}

$\bullet$ \emph{\textbf{Step 3 - Deriving a High-SNR Approximation for the ECC:}} 
In the third step, we perform asymptotic analysis for the ECC assuming a sufficiently high SNR ($p\rightarrow\infty$) in order to provide additional insights for system design. The asymptotic ECC in the high-SNR regime is presented in the following corollary.

\begin{corollary}\label{Section_Performance_Analysis_ECC_Asymptotic_Expression_Standard_Theorem}
The asymptotic ECC in the high-SNR regime can be expressed in the following form:
\begin{align}\label{Section_Performance_Analysis_ECC_Asymptotic_Expression_Standard}
\lim_{p\rightarrow\infty}{\bar{\mathcal{C}}}_{\rm{rayleigh}}\simeq {\mathcal{S}}_{\infty}\left(\log_2(p)-{\mathcal{L}}_{\infty}\right),
\end{align}
where ${\mathcal{S}}_{\infty}=1$ and
\begin{equation}\label{Section_Performance_Analysis_ECC_Calculation_Asymptotic_MISO}
\begin{split}
{\mathcal{L}}_{\infty}&=-\mathbb{E}\{\log_2({\lVert{\mathbf{h}}\rVert^2}/\sigma^2)\}=\log_2{\sigma^2}\\
&-\frac{\lambda_{\min}^{r_{\mathbf{R}}}}{\prod_{i=1}^{r_{\mathbf{R}}}{\lambda}_i}\sum_{k=0}^{\infty}\frac{\psi_k(\psi(r_{\mathbf{R}}+k)+\ln(\lambda_{\min}))}{\ln2},
\end{split}
\end{equation}
where $\psi\left(x\right)\!=\!\frac{{\rm d}}{{\rm d}x}\ln{\Gamma\left(x\right)}$ is the Digamma function.
\end{corollary}

\begin{proof}
Please refer to Appendix \ref{Section_Performance_Analysis_ECC_Asymptotic_Expression_Standard_Theorem_Proof}.
\end{proof}
It is worth noting that $\log_2({p})$ in \eqref{Section_Performance_Analysis_ECC_Asymptotic_Expression_Standard} can be rewritten as $\log_2({p})=\frac{10\log_{10}({p})}{10\log_{10}2}=\frac{p|_{\rm{dB}}}{3~{\rm{dB}}}$. Hence, ${\mathcal{L}}_{\infty}$ is termed the \emph{high-SNR power offset} in 3-dB units \cite{Lozano2005_TIT}, and ${\mathcal{S}}_{\infty}$ is referred to as the \emph{high-SNR slope}, the \emph{number of DoFs}, the \emph{maximum multiplexing gain}, or the \emph{pre-log factor} in bits/s/Hz/(3 dB). The high-SNR slope ${\mathcal{S}}_{\infty}$ characterizes the ECC as a function of the transmit power, at high SNR, on a log scale. In multi-antenna communications, ${\mathcal{S}}_{\infty}$ quantifies the number of spatial DoFs, which determines the number of spatial dimensions available for communications. On the other hand, ${\mathcal{L}}_{\infty}$ is the power shift with respect to the baseline ECC curve of $S_{\infty}\log_2(p)$. Note that the ECC can be improved by increasing ${\mathcal{S}}_{\infty}$ or decreasing ${\mathcal{L}}_{\infty}$.

\begin{remark}
The results in \textbf{Corollary} \ref{Section_Performance_Analysis_ECC_Asymptotic_Expression_Standard_Theorem} reveal that the high-SNR slope and the high-SNR power offset of the ECC are given by $1$ and $-\mathbb{E}\{\log_2({\lVert{\mathbf{h}}/\sigma^2\rVert^2})\}$, respectively.
\end{remark}

\begin{figure}[!t]
\setlength{\abovecaptionskip}{0pt}
\centering
\includegraphics[width=0.4\textwidth]{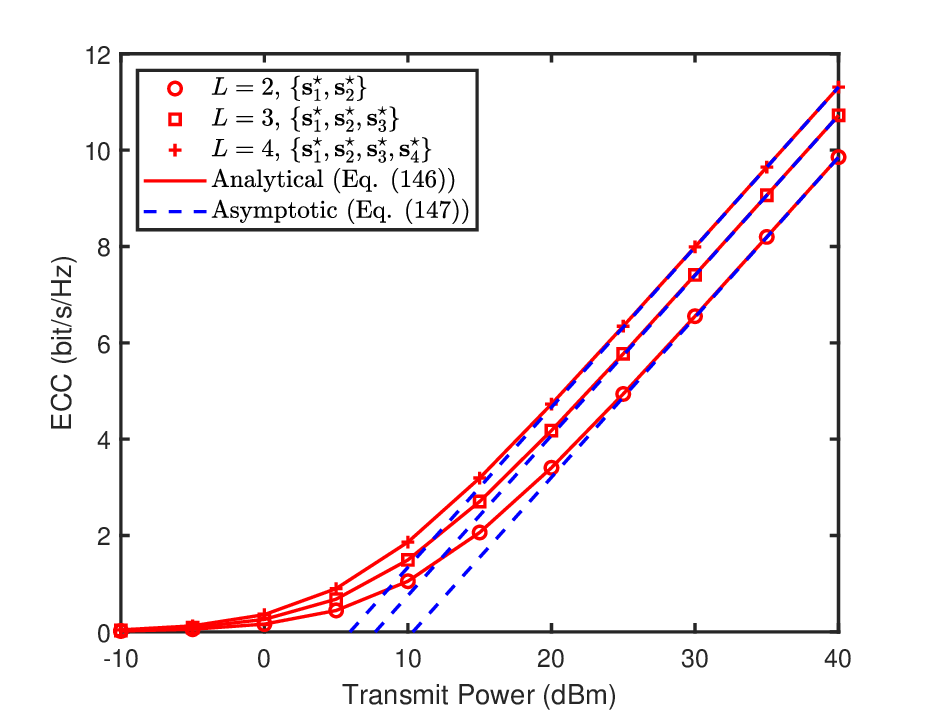}
\caption{Ergodic channel capacity for correlated MISO Rayleigh channels. The system is operating at $28$ GHz. $N_x=N_z=33$, $d=\frac{\lambda}{2}$, $(\theta,\phi)=(\frac{\pi}{2},\frac{\pi}{2})$, $r=4$ m, $A=\frac{\lambda^2}{4\pi}$, $e_a=1$, ${\bm\rho}=\hat{\mathbf J}({\mathbf{s}}_{m,n})=[1,0,0]^{\mathsf{T}}$, $\forall m,n$, $\sigma^2=-90$ dBm, and $\sigma_{\ell}^2=1$, $\forall \ell$. The coordinates (in meters) of the scatterers are given by $
\mathbf{s}_{1}^{\star}=[\frac{7}{10}\sin{\frac{\pi}{4}},r,\frac{7}{10}\cos{\frac{\pi}{4}}]^{\mathsf{T}}$, $\mathbf{s}_{2}^{\star}=[\frac{4}{5}\sin{\frac{\pi}{8}},r,\frac{4}{5}\cos{\frac{\pi}{8}}]^{\mathsf{T}}$, $
\mathbf{s}_{3}^{\star}=[\frac{13}{20}\sin{\frac{47\pi}{64}},r,\frac{13}{20}\cos{\frac{47\pi}{64}}]^{\mathsf{T}}$, and
$\mathbf{s}_{4}^{\star}=[\frac{2}{5}\sin{\frac{\pi}{7}},r,\frac{2}{5}\cos{\frac{\pi}{7}}]^{\mathsf{T}}$. The LoS channels $\overline{\mathbf{h}}$, $\{h_{\ell}({\mathbf{r}}_{\ell},\mathbf{r})\}_{\ell=1}^{L}$, and $\{{\mathbf{h}}_{\ell}(\mathbf{r}_{\ell})\}_{\ell=1}^{L}$ follow the USW channel model.}
\label{Performance Analysis Figure: ECC_Discrete}
\end{figure}

$\bullet$ \emph{\textbf{Numerical Results:}} To illustrate the above results, we show the ECC versus the transmit power, $p$, in {\figurename} {\ref{Performance Analysis Figure: ECC_Discrete}} for the USW channel model and various values of $L$. The analytical and asymptotic results are calculated using \eqref{Section_Performance_Analysis_ECC_Calculation_MISO} and \eqref{Section_Performance_Analysis_ECC_Asymptotic_Expression_Standard}, respectively. Simulation results are plotted using markers. As {\figurename} {\ref{Performance Analysis Figure: ECC_Discrete}} shows, the analytical results are in excellent agreement with the simulated results, and the derived asymptotic results approach the numerical results in the high-SNR regime ($p\rightarrow\infty$). For the considered case, the high-SNR slope is independent of $L$, which causes the ECC curves in {\figurename} {\ref{Performance Analysis Figure: ECC_Discrete}} to be parallel to each other for high transmit powers.

$\bullet$ \emph{\textbf{Extension to MIMO Case:}} The definition presented in \eqref{Section_Performance_Analysis_ECC_Definition_MISO} applies to MISO systems. The extension to single-user MIMO systems with isotropic inputs is given by
\begin{equation}\label{Section_Performance_Analysis_ECC_Definition_SU_MIMO}
\bar{\mathcal{C}}=\mathbb{E}\{\log_2\det({\mathbf{I}}+p/\sigma^2{\mathbf{H}}{\mathbf{H}}^{\mathsf{H}})\}.
\end{equation}
The evaluation of the ECC in \eqref{Section_Performance_Analysis_ECC_Definition_SU_MIMO} requires the application of random matrix theory; see \cite{Tulino2004_FTC,Muller2013_Arxiv,Couillet2011,Akemann2013} for more details. Furthermore, the existing literature has shown that the asymptotic ECC for MIMO channels in the high-SNR regime can be also expressed in the standard form given in \eqref{Section_Performance_Analysis_ECC_Asymptotic_Expression_Standard} (see, e.g., \cite{Lozano2005_TIT}). Besides the ECC and OP, the diversity-multiplexing tradeoff (DMT) is another important performance metric for statistical NFC MIMO channels. Based on the system model sketched in {\figurename} {\ref{Multipath_3D_Model}}, the statistical MIMO channel under Rayleigh fading can be characterized as ${\mathbf{H}}={\mathbf{A}}_{\rm{r}}{\bm\Lambda}_{\mathbf{H}}{\mathbf{A}}_t$, where ${\mathbf{A}}_{\rm{r}}$ and ${\mathbf{A}}_{\rm{r}}$ are two deterministic matrices containing the array steering vectors, and ${\bm\Lambda}_{\mathbf{H}}$ is a diagonal random matrix with complex Gaussian distributed diagonal elements. Since ${\mathbf{A}}_{\rm{r}}{\bm\Lambda}_{\mathbf{H}}{\mathbf{A}}_t$ corresponds to a finite-dimensional channel model \cite{Sayeed2002deconstructing}, conventional random matrix theory tools are difficult to apply. Therefore, a theoretical analysis of the DMT for this channel is an open problem, which deserves further research attention.

\vspace{-0.05cm}
\subsubsection{Analysis of the EMI for Rayleigh Channels}
For our analysis of the ECC, we have utilized Shannon's formula to calculate the channel capacity, i.e., $\log_2(1+\gamma)$. From the information-theoretic point of view, the channel capacity $\log_2(1+\gamma)$ is achievable when the transmitter sends Gaussian distributed source signals \cite{telatar1999capacity}. However, although Gaussian signals are capacity-achieving, practical systems transmit signals belong to finite and discrete constellations, such as quadrature amplitude modulation (QAM) \cite{Wu2018,Lozano2018,Rodrigues2010,Alvarado2014}. The analysis of the mutual information (MI) for finite-alphabet input signals (finite-alphabet inputs) is very different from that for Gaussian distributed input signals (Gaussian inputs). Motivated by this, we establish the basic analytical framework for the EMI achieved by finite-alphabet inputs for correlated fading channels.

The EMI for fading channels is best understood by considering the MI of a scalar Gaussian channel with finite-alphabet inputs. To this end, consider the scalar AWGN channel
\begin{align}\label{Section_Performance_Analysis_AWGN_Channel}
Y=\sqrt{\gamma}X+Z,
\end{align}
where $Z\sim{\mathcal {CN}}\left(0,1\right)$ is the AWGN, $\gamma$ is the SNR, and $X\in{\mathbb{C}}$ is the transmitted symbol. We assume that $X$ satisfies the power constraint ${\mathbb{E}}\{\left|X\right|^2\}=1$ and is taken from a finite constellation alphabet $\mathcal X$ consisting of $Q$ points, i.e., ${\mathcal{X}}=\left\{\mathsf{x}_q\right\}_{q=1}^{Q}$. The $q$th symbol in $\mathcal{X}$, i.e, $\mathsf{x}_q$, is transmitted with probability $p_q$, $0 < p_q < 1$, and the vector of probabilities ${\mathbf{p}}_{\mathcal{X}}\triangleq[p_1,\cdots,p_Q]^{\mathsf{T}}$ is called the input distribution with $\sum_{q=1}^{Q}p_q=1$. For this AWGN channel, the MI is given by \cite{Lozano2018}
\vspace{-0.05cm}
\begin{equation}\label{Section_Performance_Analysis_MI_AWGN_Definition}
\begin{split}
I_{\mathcal X}\left(\gamma\right)&=H_{{\mathbf{p}}_{\mathcal{X}}}-\frac{1}{\pi}\sum_{q=1}^{Q}\int_{\mathbb C}p_qe^{-\left|u-\sqrt{\gamma}{\mathsf{x}}_q\right|^2}\\
&\times\log_2{\left(\sum_{{q'}=1}^{Q}\frac{p_{q'}}{p_q}e^{\left|u-\sqrt{\gamma}{\mathsf{x}}_q\right|^2-\left|u-\sqrt{\gamma}{\mathsf{x}}_{q'}\right|^2}\right)}{\rm d}u,
\end{split}
\end{equation}
where $H_{{\mathbf{p}}_{\mathcal{X}}}=\sum_{q=1}^{Q}p_q\log_2\left(\frac{1}{p_q}\right)$ is the entropy of the input distribution ${\mathbf{p}}_{\mathcal{X}}$ in bits. When $\mathcal{X}$ is a Gaussian constellation, then $I_{\mathcal X}\left(\gamma\right)$ degenerates to $I_{\mathcal X}\left(\gamma\right)=\log_2(1+\gamma)$. In contrast to Shannon's formula, the MI in \eqref{Section_Performance_Analysis_MI_AWGN_Definition} generally cannot be simplified to a closed-form expression, which complicates the subsequent analysis.

By a straightforward extension of \eqref{Section_Performance_Analysis_MI_AWGN_Definition} to a single-input vector channel, the EMI achieved in the considered MISO Rayleigh channel can be expressed as follows \cite{Ouyang2020_CL}:
\begin{equation}\label{Section_Performance_Analysis_EMI_Definition_MISO}
\begin{split}
\bar{\mathcal{I}}_{\mathcal{X}}^{\rm{rayleigh}}&={\mathbb{E}}\{I_{\mathcal X}(p/\sigma^2\lVert{\mathbf{h}}\rVert^2)\}\\
&=\int_{0}^{\infty}I_{\mathcal X}(p/\sigma^2 x)f_{\lVert{\mathbf{h}}\rVert^2}(x){\rm{d}}x.
\end{split}
\end{equation}
To characterize the EMI, we follow three main steps which are detailed in the sequel.

$\bullet$ \emph{\textbf{Step 1 - Analyzing the Statistics of the Channel Gain:}} 
Similar to the analyses of the OP and the ECC, the statistics of $\lVert{\mathbf{h}}\rVert^2$ are needed. The results are given in \eqref{Section_Performance_Analysis_OP_SNR_PDF} and \eqref{Section_Performance_Analysis_OP_SNR_CDF}.

\begin{table*}[!t]
\caption{Fitting Coefficients for 4/16/64/256-QAM in \eqref{Section_Performance_Analysis_MI_AWGN_Approximation}.}
\label{tab:Section_Performance_Analysis_EDCF_Table}
\centering
\resizebox{1\textwidth}{!}{
\begin{tabular}{!{\vrule width1pt}c!{\vrule width1pt}c!{\vrule width1pt}c!{\vrule width1pt}c!{\vrule width1pt}c!{\vrule width1pt}c!{\vrule width1pt}c!{\vrule width1pt}c!{\vrule width1pt}c!{\vrule width1pt}c!{\vrule width1pt}c!{\vrule width1pt}c!{\vrule width1pt}c!{\vrule width1pt}c!{\vrule width1pt}c!{\vrule width1pt}}
\Xhline{1pt} 
$\mathcal{X}$ & $k_{\mathcal{X}}$ & $\zeta^{(\mathcal{X})}_1$ & $\zeta^{(\mathcal{X})}_2$      & $\zeta^{(\mathcal{X})}_3$      & $\zeta^{(\mathcal{X})}_4$      & $\zeta^{(\mathcal{X})}_5$      & $\zeta^{(\mathcal{X})}_6$      & $\vartheta^{({\mathcal{X}})}_1$      & $\vartheta^{({\mathcal{X}})}_2$      & $\vartheta^{({\mathcal{X}})}_3$      & $\vartheta^{({\mathcal{X}})}_4$      & $\vartheta^{({\mathcal{X}})}_5$      & $\vartheta^{({\mathcal{X}})}_6$      & RMSE    \\ \Xhline{1pt} 
4-QAM         & 6                 & 0.1634                    & 0.1161 & 0.0364 & 0.3830 & 0.2004 & 0.1007 & 0.1811 & 0.6005 & 1.8110 & 0.0066 & 0.0526 & 0.0168 & $3.50\times10^{-4}$ \\ \Xhline{1pt} 
16-QAM        & 5                 & 0.0450                    & 0.3715 & 0.1480 & 0.2325 & 0.2030 & ---    & 1.8404 & 0.0245 & 0.6305 & 0.1876 & 0.0447 & ---    & $2.62\times10^{-4}$ \\ \Xhline{1pt} 
64-QAM        & 4                 & 0.1904                    & 0.5514 & 0.0479 & 0.2103 & ---    & ---    & 0.2212 & 0.1077 & 2.0653 & 0.7627 & ---    & ---    & $1.18\times10^{-4}$ \\ \Xhline{1pt} 
256-QAM       & 3                 & 0.6603                    & 0.3127 & 0.0270 & ---    & ---    & ---    & 0.5343 & 0.9290 & 3.0391 & ---    & ---    & ---    & $5.47\times10^{-5}$ \\ \Xhline{1pt} 
\multicolumn{15}{l}{${{{}^{\ddag}}}$RMSE denotes the root mean square error caused by approximating the exact MI.}\\
\end{tabular}}
\end{table*}

$\bullet$ \emph{\textbf{Step 2 - Deriving a Closed-Form Expression for the EMI:}} 
In the second step, we leverage the PDF of $\lVert{\mathbf{h}}\rVert^2$ to derive a closed-form expression for the EMI.

As stated before, there is no closed-form expression for the MI, which makes the derivation of $\bar{\mathcal{I}}_{\mathcal{X}}$ a challenging task. As a compromise, we resort to developing accurate approximations of $\bar{\mathcal{I}}_{\mathcal{X}}$. As suggested in \cite{Ouyang2020_CL, Ouyang2020_TCOM2, Ouyang2020_TCOM1, Ouyang2020_WCL}, by using multi-exponential decay curve fitting (M-EDCF), the MI given in \eqref{Section_Performance_Analysis_OP_SNR_CDF} can be approximated as follows:
\begin{equation}\label{Section_Performance_Analysis_MI_AWGN_Approximation}
I_{\mathcal X}\left(\gamma\right)\approx \hat{I}_{\mathcal X}\left(\gamma\right)=H_{{\mathbf{p}}_{\mathcal{X}}}\left(1-\sum_{j=1}^{k_{\mathcal{X}}}\zeta^{({\mathcal{X}})}_je^{-\vartheta^{({\mathcal{X}})}_j\gamma}\right),
\end{equation}
where $\sum_{j=1}^{k_{\mathcal{X}}}\zeta^{({\mathcal{X}})}_j=1$. The fitting parameters $k_{\mathcal{X}}$, $\left\{\zeta^{({\mathcal{X}})}_j\right\}_{j=1}^{k_{\mathcal{X}}}$, and $\left\{\vartheta^{({\mathcal{X}})}_j\right\}_{j=1}^{k_{\mathcal{X}}}$ can be found by using the open-source fitting software package 1stOpt\footnote{This package is online available: \url{http://www.7d-soft.com/en/}.} \cite{Ouyang2020_TCOM2}. Table \ref{tab:Section_Performance_Analysis_EDCF_Table} lists the fitting parameters of the most commonly used equiprobable square QAM constellations, i.e., 4-QAM (or quadrature phase shift keying, QPSK), 16-QAM, 64-QAM, and 256-QAM. Based on the data in Table \ref{tab:Section_Performance_Analysis_EDCF_Table} and the discussions in \cite{Ouyang2020_TCOM2}, using $\hat{I}_{\mathcal X}\left(\cdot\right)$ to approximate ${I}_{\mathcal X}\left(\cdot\right)$ yields an absolute error of ${\mathcal{O}}(10^{-4})$ and a relative error of ${\mathcal{O}}(10^{-3})$. Considering this excellent approximation quality, it is reasonable to employ $\hat{I}_{\mathcal X}\left(\cdot\right)$ on approximating the EMI, which yields the following theorem.

\begin{theorem}\label{Section_Performance_Analysis_EMI_Approximation_Explicit_MISO_Theorem}
The EMI achieved by finite-alphabet inputs can be approximated as follows:
\begin{equation}\label{Section_Performance_Analysis_EMI_Approximation_Explicit_MISO}
\bar{\mathcal{I}}_{\mathcal{X}}^{\rm{rayleigh}}\approx H_{{\mathbf{p}}_{\mathcal{X}}}-\sum_{j=1}^{k_{\mathcal{X}}}
\sum_{k=0}^{\infty}\frac{\frac{\lambda_{\min}^{r_{\mathbf{R}}}H_{{\mathbf{p}}_{\mathcal{X}}}}{\prod_{i=1}^{r_{\mathbf{R}}}{\lambda}_i}\zeta^{({\mathcal{X}})}_j\psi_k}{\left(1+\lambda_{\min}\frac{p}{\sigma^2}\vartheta^{({\mathcal{X}})}_j\right)^{r_{\mathbf{R}}+k}}.
\end{equation}
\end{theorem}

\begin{proof}
This theorem can be proved by substituting \eqref{Section_Performance_Analysis_OP_SNR_PDF} and \eqref{Section_Performance_Analysis_MI_AWGN_Approximation} into
\eqref{Section_Performance_Analysis_EMI_Definition_MISO} and calculating the resulting integral with the aid of \cite[Eq. (3.326.2)]{Ryzhik2007}.
\end{proof}
Given the closed-form expression of the EMI, the energy efficiency (EE), i.e., the total energy consumption per bit (in bit/Joule), can also be obtained as follows:
{
\begin{equation}\label{Section_Performance_Analysis_EE_Approximation_Explicit_MISO}
{\rm{EE}}=\frac{{\rm{EMI}}\times W}{P_{\rm{tot}}},
\end{equation}
}where $P_{\rm{tot}}$ denotes the total consumed power (in Watt), $W$ is the bandwidth (in Hz), and ${\rm{EMI}}$ represents the spectral efficiency (SE) or EMI (in bit/s/Hz). Particularly, we have ${\rm{EMI}}=\bar{\mathcal{C}}$ for Gaussian inputs and ${\rm{EMI}}=\bar{\mathcal{I}}_{\mathcal{X}}$ for finite-alphabet inputs.  According to the circuit power consumption model in \cite{cui2004energy,cui2005energy}, the total consumed power is calculated as follows:
{
\begin{equation}\label{Section_Performance_Analysis_Energy_Efficiency_Total_Power}
P_{\rm{tot}}=\zeta_{\rm{eff}}^{-1}p+P_{\rm{circ}},
\end{equation}
}where $\zeta_{\rm{eff}}<1$ is the efficiency of the power amplifier, $p$ denotes the transmit power, and $P_{\rm{circ}}$ is given as follows:
{
\begin{equation}\label{Section_Performance_Analysis_Energy_Efficiency_Circuit}
P_{\rm{circ}}=2P_{\rm{syn}}+N P_{\rm{TR}} + P_{\rm{RR}},
\end{equation}
}with $P_{\rm{syn}}$, $P_{\rm{TR}}$, and $P_{\rm{RR}}$ representing the circuit power consumptions of the frequency synthesizer, the transmit RF chain, and the receive RF chain, respectively. Inserting \eqref{Section_Performance_Analysis_Energy_Efficiency_Total_Power} and \eqref{Section_Performance_Analysis_Energy_Efficiency_Circuit} into \eqref{Section_Performance_Analysis_EE_Approximation_Explicit_MISO} yields
\begin{equation}\label{Section_Performance_Analysis_EE_Approximation_Explicit_MISO_Total}
{\rm{EE}}=\frac{{\rm{EMI}}\times W}{\zeta_{\rm{eff}}^{-1}p+2P_{\rm{syn}}+N P_{\rm{TR}} + P_{\rm{RR}}}.
\end{equation}

$\bullet$ \emph{\textbf{Step 3 - Deriving a High-SNR Approximation for the EMI:}} 
In the last step, we investigate the asymptotic behaviour of the EMI in the high-SNR regime, i.e., $p\rightarrow\infty$. The main results are summarized as follows.
\begin{corollary}\label{Section_Performance_Analysis_EMI_Asymptotic_Expression_Standard_Theorem}
The asymptotic EMI in the high-SNR regime can be expressed as follows:
\begin{align}\label{Section_Performance_Analysis_EMI_Asymptotic_Expression_Standard}
\lim_{p\rightarrow\infty}{\bar{\mathcal{I}}_{\mathcal{X}}}^{\rm{rayleigh}}\simeq H_{{\mathbf{p}}_{\mathcal{X}}}-({\mathcal{A}}_{\rm{a}}\cdot p)^{-{\mathcal{A}}_{\rm{d}}},
\end{align}
where ${\mathcal{A}}_{\rm{a}}=\frac{1}{\sigma^2}\left(\frac{1}{\ln{2}}\frac{{\mathcal M}\left[{{\rm{MMSE}}_{\mathcal{X}}(t)};r_{\mathbf{R}}+1\right]}{r_{\mathbf{R}}!\prod_{i=1}^{r_{\mathbf{R}}}{\lambda}_i}\right)^{-1/r_{\mathbf{R}}}\in(0,\infty)$ and ${\mathcal{A}}_{\rm{d}}=r_{\mathbf{R}}$. Here, ${\rm{MMSE}}_{\mathcal X}\left(t\right)$ denotes the minimum mean square error (MMSE) in estimating $X$ in \eqref{Section_Performance_Analysis_AWGN_Channel} from $Y$ when $\gamma=t$, and ${\mathcal M}\left[\varrho\left(t\right);z\right]\triangleq\int_{0}^{\infty}t^{z-1}\varrho\left(t\right){\rm d}t$ denotes the Mellin transform of $\varrho\left(t\right)$ \cite{flajolet1995mellin}.
\end{corollary}

\begin{proof}
Please refer to Appendix \ref{Section_Performance_Analysis_EMI_Asymptotic_Expression_Standard_Theorem_Proof}.
\end{proof}

\begin{remark}
The result in \eqref{Section_Performance_Analysis_EMI_Asymptotic_Expression_Standard} suggests that the EMI achieved by finite-alphabet inputs converges to $H_{{\mathbf{p}}_{\mathcal{X}}}$ in the limit of large $p$, and its rate of convergence (ROC) equals the rate of $({\mathcal{A}}_{\rm{a}}\cdot p)^{-{\mathcal{A}}_{\rm{d}}}$ converging to $0$.
\end{remark}

In \eqref{Section_Performance_Analysis_EMI_Asymptotic_Expression_Standard}, ${\mathcal{A}}_{\rm{a}}$ is referred to as the \emph{array gain}, and ${\mathcal{A}}_{\rm{d}}$ is referred to as the diversity order. The diversity order ${\mathcal{A}}_{\rm{d}}$ determines the slope of the ROC, i.e., $({\mathcal{A}}_{\rm{a}}\bar\gamma)^{-{\mathcal{A}}_{\rm{d}}}$, as a function of the transmit power, at high SNR, in a log-log scale. On the other hand, ${\mathcal{A}}_{\rm{a}}$ (in decibels) determines the power gain relative to a benchmark ROC curve of $(\bar\gamma)^{-{\mathcal{A}}_{\rm{d}}}$. It is noteworthy that the ROC is accelerated by increasing ${\mathcal{A}}_{\rm{a}}$ or ${\mathcal{A}}_{\rm{d}}$, and a faster ROC yields a larger EMI.

\begin{remark}
The results in \textbf{Corollary} \ref{Section_Performance_Analysis_EMI_Asymptotic_Expression_Standard_Theorem} reveal that the diversity order and the array gain of the EMI are given by $r_{\mathbf{R}}$ and $\frac{1}{\sigma^2}\left(\frac{1}{\ln{2}}\frac{{\mathcal M}\left[{{\rm{MMSE}}_{\mathcal{X}}(x)};r_{\mathbf{R}}+1\right]}{r_{\mathbf{R}}!\prod_{i=1}^{r_{\mathbf{R}}}{\lambda}_i}\right)^{-1/r_{\mathbf{R}}}$, respectively.
\end{remark}

\begin{remark}\label{Remark_Gaussian_Finite_Alphabet}
By comparing \eqref{Section_Performance_Analysis_EMI_Asymptotic_Expression_Standard} with \eqref{Section_Performance_Analysis_ECC_Asymptotic_Expression_Standard}, one can see a significant difference between the EMI achieved by finite-alphabet inputs and Gaussian inputs. The EMI achieved by Gaussian inputs, also referred to as the ECC, grows with $\bar\gamma$ unboundedly. In contrast, the EMI achieved with finite-alphabet inputs converges to a constant value.
\end{remark}

$\bullet$ \emph{\textbf{Numerical Results:}} To further illustrate our derived results, in {\figurename} {\ref{Performance Analysis Figure: AMI_Discrete Explicit}}, we plot the EMI for USW LoS channels achieved by square QAM constellations versus the transmit power, $p$. The simulation results are denoted by markers. The approximated results are obtained based on \eqref{Section_Performance_Analysis_EMI_Approximation_Explicit_MISO}. As can be observed in {\figurename} {\ref{Performance Analysis Figure: AMI_Discrete Explicit}}, the approximated results are in excellent agreement with the simulation results. This verifies the accuracy of the approximation in \eqref{Section_Performance_Analysis_MI_AWGN_Approximation} and \eqref{Section_Performance_Analysis_EMI_Approximation_Explicit_MISO}. The EMI achieved by Gaussian inputs is also shown as a baseline. As {\figurename} {\ref{Performance Analysis Figure: AMI_Discrete Explicit}} shows, the EMI for Gaussian inputs grows unboundedly as $p$ increases, whereas the EMI for finite-alphabet inputs converges to the entropy of the input distribution $H_{{\mathbf{p}}_{\mathcal{X}}}$, in the limit of large $p$. This observation validates our discussions in Remark \ref{Remark_Gaussian_Finite_Alphabet}. Moreover, we observe that in the low-SNR regime, the EMI achieved by finite-alphabet inputs is close to that achieved by Gaussian inputs, which is consistent with the results in \cite{Guo2005}. To illustrate the ROC of the EMI, we depict $H_{{\mathbf{p}}_{\mathcal{X}}}-{\bar{\mathcal{I}}_{\mathcal{X}}}$ versus $p$ in {\figurename} {\ref{Performance Analysis Figure: AMI_Discrete Asymptotic}}. As can be observed, in the high-SNR regime, the derived asymptotic results approach the numerical results. Besides, it can be observed that lower modulation orders yield faster ROCs. The numerical results presented in {\figurename} {\ref{Performance Analysis Figure: AMI_Discrete Explicit}} and {\figurename} {\ref{Performance Analysis Figure: AMI_Discrete Asymptotic}} are based on square $M$-QAM constellations, which are the most widely used modulation schemes in practical communication systems, supported, e.g., in 5G new radio (NR) \cite{3gpp2018nr}. For a thorough study, we compare the EMI achieved by QAM and phase shift keying (PSK) modulation in {\figurename} {\ref{Performance Analysis Figure: AMI_Different Modulation}. PSK is used in wireless local area network applications such as Bluetooth and radio frequency identification (RFID). As observed from {\figurename} {\ref{Performance Analysis Figure: AMI_Different Modulation}}, for a given modulation order $M>4$, PSK's EMI is smaller than QAM's. This means that PSK has a lower spectral efficiency than QAM. Moreover, PSK has an inferior BER performance compared to QAM at high SNRs, as the phase transitions become more difficult to detect. PSK may also require more complex phase synchronization and demodulation techniques, especially for higher-order PSK \cite{Proakis2008}. The above arguments imply that QAM is preferred over PSK for application in future NFC networks.}

\begin{figure}[!t]
\setlength{\abovecaptionskip}{0pt}
\centering
\includegraphics[width=0.4\textwidth]{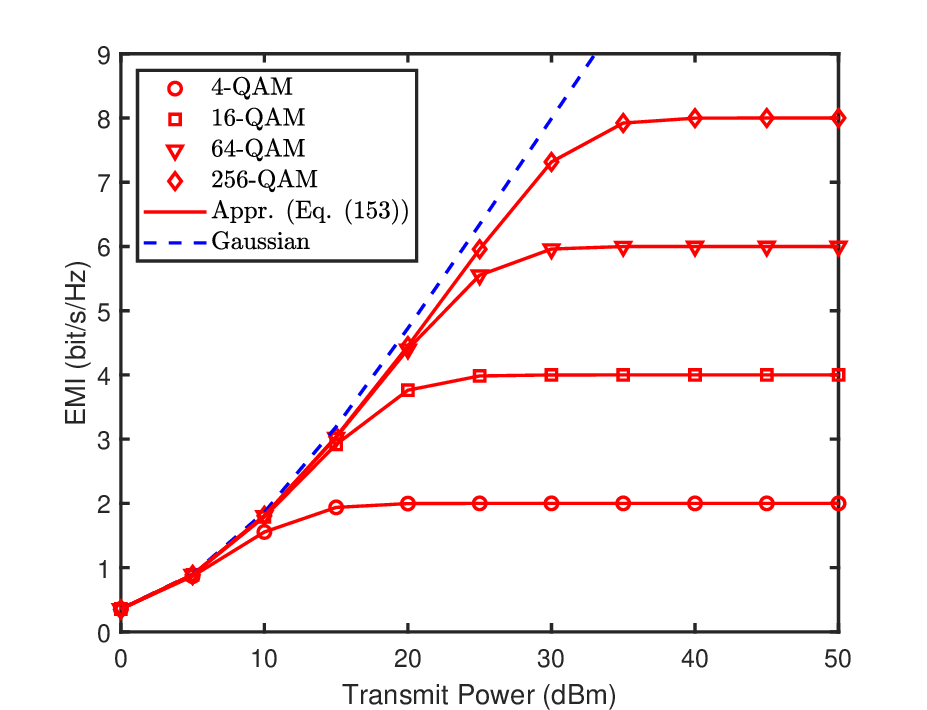}
\caption{Ergodic mutual information for correlated MISO Rayleigh channels. The system is operating at $28$ GHz. $N_x=N_z=33$, $d=\frac{\lambda}{2}$, $(\theta,\phi)=(\frac{\pi}{2},\frac{\pi}{2})$, $r=4$ m, $A=\frac{\lambda^2}{4\pi}$, $e_a=1$, ${\bm\rho}=\hat{\mathbf J}({\mathbf{s}}_{m,n})=[1,0,0]^{\mathsf{T}}$, $\forall m,n$, $\sigma^2=-90$ dBm, and $\sigma_{\ell}^2=1$, $\forall \ell$. The coordinates (in meters) of the scatterers are given by $
\mathbf{s}_{1}^{\star}=[\frac{7}{10}\sin{\frac{\pi}{4}},r,\frac{7}{10}\cos{\frac{\pi}{4}}]^{\mathsf{T}}$, $\mathbf{s}_{2}^{\star}=[\frac{4}{5}\sin{\frac{\pi}{8}},r,\frac{4}{5}\cos{\frac{\pi}{8}}]^{\mathsf{T}}$, $
\mathbf{s}_{3}^{\star}=[\frac{13}{20}\sin{\frac{47\pi}{64}},r,\frac{13}{20}\cos{\frac{47\pi}{64}}]^{\mathsf{T}}$, and
$\mathbf{s}_{4}^{\star}=[\frac{2}{5}\sin{\frac{\pi}{7}},r,\frac{2}{5}\cos{\frac{\pi}{7}}]^{\mathsf{T}}$. The LoS channels $\overline{\mathbf{h}}$, $\{h_{\ell}({\mathbf{r}}_{\ell},\mathbf{r})\}_{\ell=1}^{L}$, and $\{{\mathbf{h}}_{\ell}(\mathbf{r}_{\ell})\}_{\ell=1}^{L}$ follow the USW channel model.}
\label{Performance Analysis Figure: AMI_Discrete Explicit}
\end{figure}

\begin{figure}[!t]
\setlength{\abovecaptionskip}{0pt}
\centering
\includegraphics[width=0.4\textwidth]{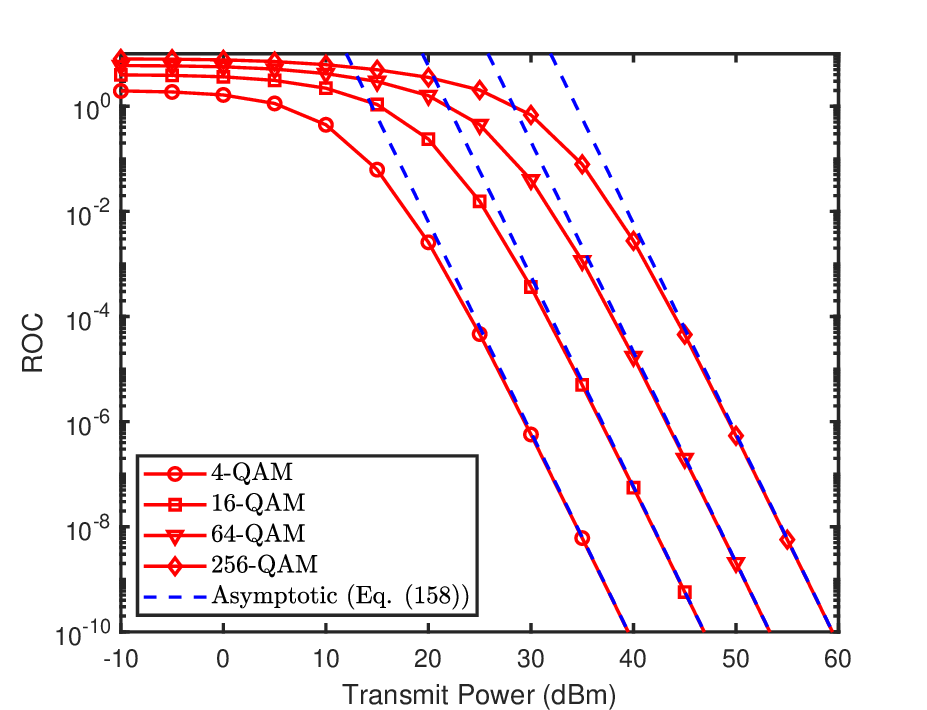}
\caption{Rate of convergence of the EMI for correlated MISO Rayleigh channels. The system is operating at $28$ GHz. $N_x=N_z=33$, $d=\frac{\lambda}{2}$, $(\theta,\phi)=(\frac{\pi}{2},\frac{\pi}{2})$, $r=4$ m, $A=\frac{\lambda^2}{4\pi}$, $e_a=1$, ${\bm\rho}=\hat{\mathbf J}({\mathbf{s}}_{m,n})=[1,0,0]^{\mathsf{T}}$, $\forall m,n$, $\sigma^2=-90$ dBm, and $\sigma_{\ell}^2=1$, $\forall \ell$. The coordinates (in meters) of the scatterers are given by $
\mathbf{s}_{1}^{\star}=[\frac{7}{10}\sin{\frac{\pi}{4}},r,\frac{7}{10}\cos{\frac{\pi}{4}}]^{\mathsf{T}}$, $\mathbf{s}_{2}^{\star}=[\frac{4}{5}\sin{\frac{\pi}{8}},r,\frac{4}{5}\cos{\frac{\pi}{8}}]^{\mathsf{T}}$, $
\mathbf{s}_{3}^{\star}=[\frac{13}{20}\sin{\frac{47\pi}{64}},r,\frac{13}{20}\cos{\frac{47\pi}{64}}]^{\mathsf{T}}$, and
$\mathbf{s}_{4}^{\star}=[\frac{2}{5}\sin{\frac{\pi}{7}},r,\frac{2}{5}\cos{\frac{\pi}{7}}]^{\mathsf{T}}$. The LoS channels $\overline{\mathbf{h}}$, $\{h_{\ell}({\mathbf{r}}_{\ell},\mathbf{r})\}_{\ell=1}^{L}$, and $\{{\mathbf{h}}_{\ell}(\mathbf{r}_{\ell})\}_{\ell=1}^{L}$ follow the USW channel model.}
\label{Performance Analysis Figure: AMI_Discrete Asymptotic}
\end{figure}

\begin{figure}[!t]
\setlength{\abovecaptionskip}{0pt}
\centering
\includegraphics[width=0.4\textwidth]{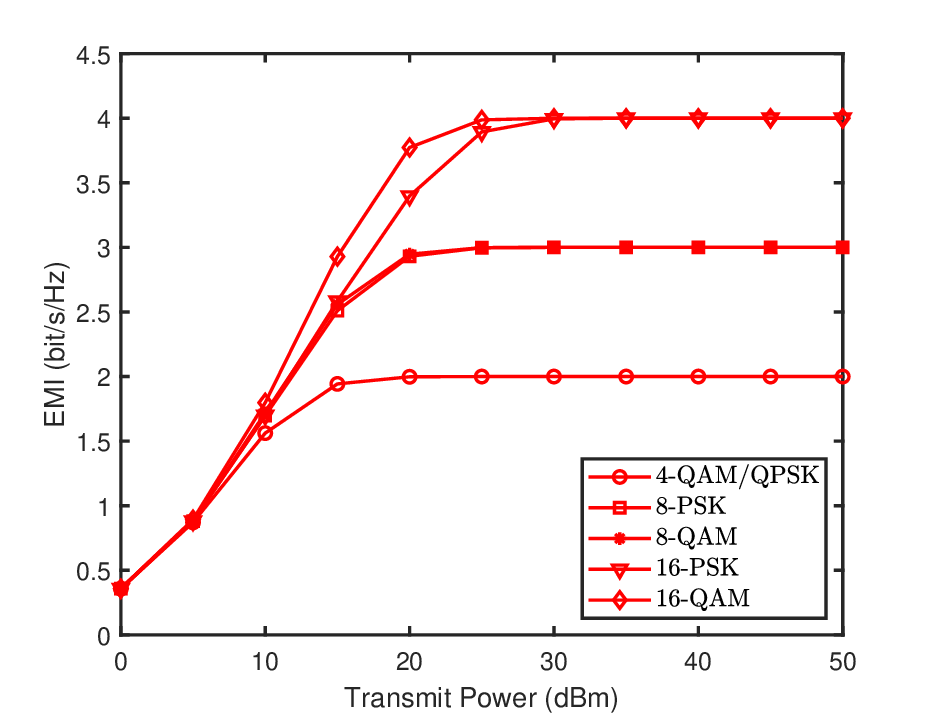}
\caption{Ergodic mutual information achieved by different modulation schemes for correlated MISO Rayleigh channels. The system is operating at frequency $28$ GHz. $N_x=N_z=33$, $d=\frac{\lambda}{2}$, $(\theta,\phi)=(\frac{\pi}{2},\frac{\pi}{2})$, $r=4$ m, $A=\frac{\lambda^2}{4\pi}$, $e_a=1$, ${\bm\rho}=\hat{\mathbf J}({\mathbf{s}}_{m,n})=[1,0,0]^{\mathsf{T}}$, $\forall m,n$, $\sigma^2=-90$ dBm, and $\sigma_{\ell}^2=1$, $\forall \ell$. The coordinates (in meters) of the scatterers are given by $
\mathbf{s}_{1}^{\star}=[\frac{7}{10}\sin{\frac{\pi}{4}},r,\frac{7}{10}\cos{\frac{\pi}{4}}]^{\mathsf{T}}$, $\mathbf{s}_{2}^{\star}=[\frac{4}{5}\sin{\frac{\pi}{8}},r,\frac{4}{5}\cos{\frac{\pi}{8}}]^{\mathsf{T}}$, $
\mathbf{s}_{3}^{\star}=[\frac{13}{20}\sin{\frac{47\pi}{64}},r,\frac{13}{20}\cos{\frac{47\pi}{64}}]^{\mathsf{T}}$, and
$\mathbf{s}_{4}^{\star}=[\frac{2}{5}\sin{\frac{\pi}{7}},r,\frac{2}{5}\cos{\frac{\pi}{7}}]^{\mathsf{T}}$. The LoS channels $\overline{\mathbf{h}}$, $\{h_{\ell}({\mathbf{r}}_{\ell},\mathbf{r})\}_{\ell=1}^{L}$, and $\{{\mathbf{h}}_{\ell}(\mathbf{r}_{\ell})\}_{\ell=1}^{L}$ are described by the USW channel model.}
\label{Performance Analysis Figure: AMI_Different Modulation}
\end{figure}

\begin{figure}[!t]
\setlength{\abovecaptionskip}{0pt}
\centering
\includegraphics[width=0.4\textwidth]{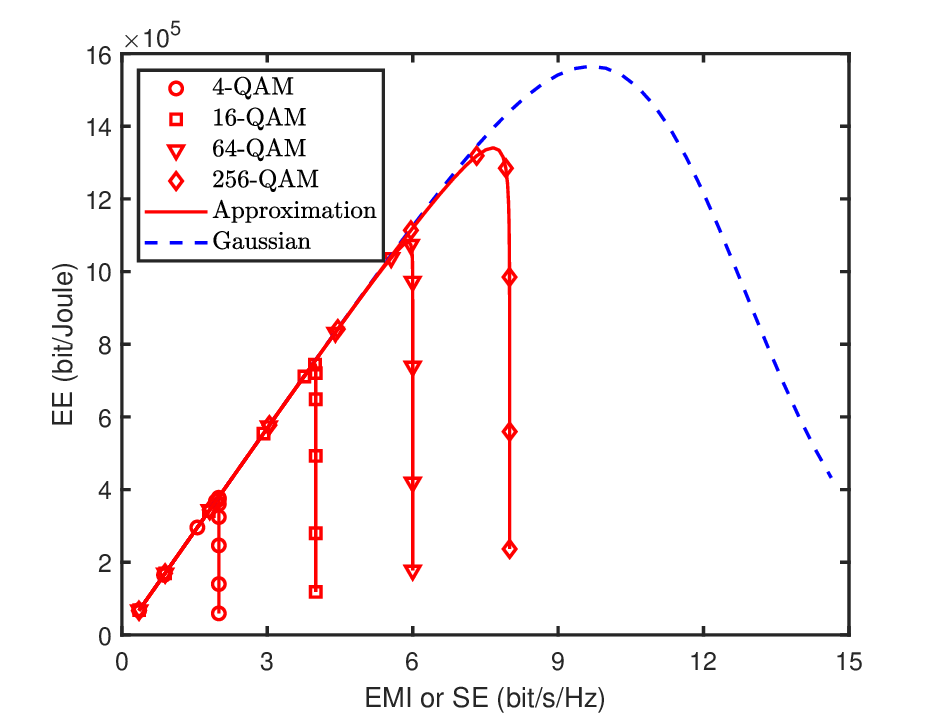}
\caption{EE-SE tradeoff for correlated MISO Rayleigh channels. The system is operating at frequency $28$ GHz with $W=10$ MHz. $N_x=N_z=33$, $d=\frac{\lambda}{2}$, $(\theta,\phi)=(\frac{\pi}{2},\frac{\pi}{2})$, $r=4$ m, $A=\frac{\lambda^2}{4\pi}$, $e_a=1$, ${\bm\rho}=\hat{\mathbf J}({\mathbf{s}}_{m,n})=[1,0,0]^{\mathsf{T}}$, $\forall m,n$, $\sigma^2=-90$ dBm, and $\sigma_{\ell}^2=1$, $\forall \ell$. The coordinates (in meters) of the scatterers are given by $
\mathbf{s}_{1}^{\star}=[\frac{7}{10}\sin{\frac{\pi}{4}},r,\frac{7}{10}\cos{\frac{\pi}{4}}]^{\mathsf{T}}$, $\mathbf{s}_{2}^{\star}=[\frac{4}{5}\sin{\frac{\pi}{8}},r,\frac{4}{5}\cos{\frac{\pi}{8}}]^{\mathsf{T}}$, $
\mathbf{s}_{3}^{\star}=[\frac{13}{20}\sin{\frac{47\pi}{64}},r,\frac{13}{20}\cos{\frac{47\pi}{64}}]^{\mathsf{T}}$, and
$\mathbf{s}_{4}^{\star}=[\frac{2}{5}\sin{\frac{\pi}{7}},r,\frac{2}{5}\cos{\frac{\pi}{7}}]^{\mathsf{T}}$. The LoS channels $\overline{\mathbf{h}}$, $\{h_{\ell}({\mathbf{r}}_{\ell},\mathbf{r})\}_{\ell=1}^{L}$, and $\{{\mathbf{h}}_{\ell}(\mathbf{r}_{\ell})\}_{\ell=1}^{L}$ are described by the USW channel model. In addition, a practical class-A RF power amplifier is considered for which $\zeta_{\rm{eff}}=0.35$, $P_{\rm{syn}}=50$ mW, $P_{\rm{TR}}=48.2$ mW, and $P_{\rm{RR}} = 62.5$ mW \cite[Table 1]{cui2005energy} hold. The values of $\zeta_{\rm{eff}}$ and $P_{\rm{syn}}$ can be directly found in \cite[Table 1]{cui2005energy}. However, the values of $P_{\rm{TR}}$ and $P_{\rm{RR}}$ are calculated with the aid of \cite[Eqn. (32)]{cui2005energy} and \cite[Eqn. (33)]{cui2005energy}, respectively. Since the computation is trivial, we omit the detailed steps here.}
\label{Performance Analysis Figure: EE_SE_TF_Discrete Explicit}
\end{figure}

{\figurename} {\ref{Performance Analysis Figure: EE_SE_TF_Discrete Explicit}} illustrates the EE-SE tradeoff obtained by \eqref{Section_Performance_Analysis_EE_Approximation_Explicit_MISO_Total} when the transmit power budget $p$ varies from $0$ dBm to $50$ dBm. It can be observed that for increasing EMI or SE, the EE first increases and then decreases. This suggests that there exists an optimal EMI that maximizes the EE, which is denoted as ${\rm{EMI}}^{\star}$. The reason for this behaviour is that the total consumed power $P_{\rm{tot}}$ increases linearly with the power budget $p$, whereas the EMI and SE increase sub-linearly with $p$. Particularly, for finite-alphabet inputs, the achieved EMI converges to some constant as $p$ increases. Therefore, the corresponding EE decreases rapidly when the EMI exceeds ${\rm{EMI}}^{\star}$. Additionally, it can be seen that a larger modulation order yields a larger value of ${\rm{EMI}}^{\star}$ and a better SE-EE tradeoff. Gaussian input signals provide an EE performance upper bound. In the low-EMI or low-SNR regime, we observe that the EE achieved by finite-alphabet inputs is close to that achieved by Gaussian inputs, which is consistent with the results shown in {\figurename} {\ref{Performance Analysis Figure: AMI_Discrete Explicit}}. It, thus, is advisable to apply lower-order modulation types in the low-power regime \cite{verdu2002spectral}.

$\bullet$ \emph{\textbf{Extension to MIMO Case:}} Compared with the EMI achieved in MISO channels (as formulated in \eqref{Section_Performance_Analysis_EMI_Definition_MISO}), the analysis of the EMI achieved in MIMO channels is much more challenging. In a multiple-stream MIMO channel, the received signal vector is given by
\begin{align}\label{Section_Performance_Analysis_AWGN_MIMO_Channel}
{\mathbf{y}}=\sqrt{\bar{\gamma}}{\mathbf{H}}{\mathbf{P}}{\mathbf{x}}+{\mathbf{n}},
\end{align}
where ${\mathbf{n}}\sim{\mathcal{CN}}(0,\mathbf{I})$ is Gaussian noise, ${\mathbf{P}}\in{\mathbb{C}}^{N_T\times N_D}$ denotes the precoding matrix satisfying ${\mathsf{tr}}\left\{{\mathbf{P}}{\mathbf{P}}^{\mathsf{H}}\right\}=1$ with $N_D$
being the number of data streams, and ${\mathbf{x}}\in{\mathbb{C}}^{N_D\times1}$ is the data vector with i.i.d. elements drawn from constellation $\mathcal X$. Hence, input signal $\mathbf{x}$ is taken from a multi-dimensional constellation $\hat{\mathcal{X}}$ comprising $Q^{N_D}$ points, i.e., $\mathbf{x}\in{\hat{\mathcal{X}}}=\left\{{\bm{\mathsf{x}}}_q\in{\mathbb{C}}^{N_D}\right\}_{q=1}^{Q^{N_D}}$, with $\mathbb{E}\left\{{\mathbf{x}}{\mathbf{x}}^{\mathsf{H}}\right\}={\mathbf{I}}$. Assume ${\bm{\mathsf{x}}}_q$ is sent with probability $q_q$, $0 < q_q < 1$, and the input distribution is given by ${\mathbf{q}}_{\hat{\mathcal{X}}}\triangleq[q_1,\cdots,q_{Q^{N_D}}]^{\mathsf{T}}$ with $\sum_{g=1}^{Q^{N_D}}q_g=1$. The MI in this case can be written as follows \cite{Rodrigues2014_TIT,Ouyang2023_CL,Lozano2018}:
\begin{equation}\label{Section_Performance_Analysis_MI_AWGN_MIMO_Definition}
\begin{split}
&I_{\hat{\mathcal X}}\left(p/\sigma^2;{\mathbf{H}}{\mathbf{P}}\right)=H_{{\mathbf{q}}_{\hat{\mathcal{X}}}}-N_R\log_2{e}-
\sum\nolimits_{q=1}^{Q^{N_D}}q_q\\
&\times{\mathbb{E}}_{\mathbf{n}}\left\{
\log_2\left({\sum_{q'=1}^{Q^{N_D}}\frac{q_{q'}}{q_q}{e}^{-\left\|{\mathbf{n}}+\sqrt{{p/\sigma^2}}{\mathbf{H}}{\mathbf{P}}
({\bm{\mathsf{x}}}_q-{\bm{\mathsf{x}}}_{q'})\right\|^2}}\right)
\right\}.
\end{split}
\end{equation}
The EMI achieved in channel \eqref{Section_Performance_Analysis_AWGN_MIMO_Channel} is given by
{
\begin{equation}\label{Section_Performance_Analysis_EMI_AWGN_MIMO_Definition}
\begin{split}
&\bar{\mathcal{I}}_{\hat{\mathcal{X}}}={\mathbb E}\left\{I_{\hat{\mathcal X}}\left(\frac{p}{\sigma^2};{\mathbf{H}}{\mathbf{P}}\right)\right\}=H_{{\mathbf{q}}_{\hat{\mathcal{X}}}}-N_R\log_2{e}-
\sum\nolimits_{q=1}^{Q^{N_D}}q_q\\
&\times{\mathbb{E}}_{{\mathbf{H}},{\mathbf{n}}}\left\{
\log_2\left({\sum_{q'=1}^{Q^{N_D}}\frac{q_{q'}}{q_q}{e}^{-\left\|{\mathbf{n}}+\sqrt{p/\sigma^2}{\mathbf{H}}{\mathbf{P}}
({\bm{\mathsf{x}}}_q-{\bm{\mathsf{x}}}_{q'})\right\|^2}}\right)
\right\}.
\end{split}
\end{equation}
}By comparing $I_{\hat{\mathcal X}}\left({p/\sigma^2};{\mathbf{H}}{\mathbf{P}}\right)$ in \eqref{Section_Performance_Analysis_MI_AWGN_MIMO_Definition} with $I_{\mathcal X}\left({p/\sigma^2}\right)$ in \eqref{Section_Performance_Analysis_MI_AWGN_Definition}, we observe that $I_{\hat{\mathcal X}}\left({{\bar{\gamma}}};{\mathbf{H}}{\mathbf{P}}\right)$ presents an even more intractable form than $I_{\mathcal X}\left({p/\sigma^2}\right)$. Therefore, our developed methodology for analyzing ${\mathbb{E}}\{I_{\mathcal X}\left({p/\sigma^2}\right)\}$ cannot be straightforwardly applied to analyzing ${\mathbb E}\{I_{\hat{\mathcal X}}\left({{p/\sigma^2}};{\mathbf{H}}{\mathbf{P}}\right)\}$. In the past years, the EMI achieved by finite-alphabet inputs in MIMO channels has been studied extensively and many approximated expressions for the EMI were derived under different fading models; see \cite{wu2014linear,wu2017low, Wu2018} and the references therein. However, it should be noted that the problem of characterizing the high-SNR asymptotic EMI for MIMO transmission, i.e., $\lim_{p\rightarrow\infty}{\mathbb E}\{I_{\hat{\mathcal X}}\left({{p/\sigma^2}};{\mathbf{H}}{\mathbf{P}}\right)\}$ has been open for years, and only a couple of works appeared recently. The author in \cite{Rodrigues2014_TIT} discussed the high-SNR asymptotic behaviour of ${\mathbb E}\{I_{\hat{\mathcal X}}\left({{\bar{\gamma}}};{\mathbf{H}}{\mathbf{P}}\right)\}$ for isotropic inputs and correlated Rician channels. The authors in \cite{Ouyang2023_CL} further characterized the high-SNR EMI by considering a double-scattering fading model and non-isotropic precoding. Most interestingly, as shown in these two works, the asymptotic EMI for MIMO channels in the high-SNR regime also follows the standard form given in \eqref{Section_Performance_Analysis_EMI_Asymptotic_Expression_Standard}.

By now, we have established a framework for analyzing the OP, ECC, and EMI for the correlated MISO Rayleigh fading model. We next exploit this framework to analyze the NFC performance for correlated Rician fading.
\subsubsection{Analysis of the OP for Rician Channels}
Based on \eqref{Section_Performance_Analysis_OP_Definition_MISO} and \eqref{Section_Performance_Analysis_OP_Calculation_MISO}, the OP can be expressed as follows:
\begin{equation}\label{Section_Performance_Analysis_OP_Calculation_MISO_Rician}
\mathcal{P}_{\rm{rician}}=\Pr\left(\lVert\overline{\mathbf{h}}+\mathbf{R}^{\frac{1}{2}}\tilde{\mathbf{h}}\rVert^2<\frac{2^{\mathcal{R}}-1}{p/\sigma^2}\right),
\end{equation}
which is analyzed in the following three steps.

$\bullet$ \emph{\textbf{Step 1 - Analyzing the Statistics of the Channel Gain:}} 
In the first step, we analyze the statistics of ${\lVert{\mathbf{h}}\rVert^2}=\lVert\overline{\mathbf{h}}+\mathbf{R}^{\frac{1}{2}}\tilde{\mathbf{h}}\rVert^2$ to obtain a closed-form expression for $\mathcal{P}_{\rm{rician}}$. We commence with the following lemma.
\begin{lemma}\label{Section_Performance_Analysis_Rician_OP_SNR_Trans_Lemma}
The variate $\lVert\overline{\mathbf{h}}+\mathbf{R}^{\frac{1}{2}}\tilde{\mathbf{h}}\rVert^2$ has the same statistics as the following variate:
\begin{align}\label{Section_Performance_Analysis_Rician_OP_SNR_Trans_Variable}
\underbrace{\sum_{i=1}^{r_{\mathbf{R}}}\lambda_i\lVert\overline{{h}}_i\lambda_i^{-\frac{1}{2}}+\tilde{{h}}_i\rVert^2}_{\tilde{a}}+
\underbrace{\sum_{i=r_{\mathbf{R}}+1}^{N}\lVert\overline{{h}}_i\rVert^2}_{\tilde{a}_0},
\end{align}
where ${\mathbf{U}}^{\mathsf{H}}{\bm\Lambda}{\mathbf{U}}$ is the eigenvalue decomposition of $\mathbf{R}$ with $\mathbf{U}{\mathbf{U}}^{\mathsf{H}}=\mathbf{I}$ and ${\bm\Lambda}={\mathsf{diag}}\{\lambda_1,\ldots,\lambda_{r_{\mathbf{R}}},0,\ldots,0\}$, $\{\lambda_i>0\}_{i=1}^{r_{\mathbf{R}}}$ are the positive eigenvalues of matrix $\mathbf{R}$, $\overline{{h}}_i$ is the $i$th element of vector $\mathbf{U}\overline{\mathbf{h}}\in{\mathbb{C}}^{N}$, and $[\tilde{{h}}_1,\ldots,\tilde{{h}}_{r_{\mathbf{R}}}]^{\mathsf{T}}\sim{\mathcal{CN}}(\mathbf{0},\mathbf{I})$.
\end{lemma}

\begin{proof}
Please refer to Appendix \ref{Section_Performance_Analysis_Rician_OP_SNR_Trans_Lemma_Proof}.
\end{proof}
Note that $\tilde{a}_0$ is a deterministic constant and the randomness of \eqref{Section_Performance_Analysis_Rician_OP_SNR_Trans_Variable} originates from $\tilde{a}$. Moreover, since $[\tilde{{h}}_1,\ldots,\tilde{{h}}_{r_{\mathbf{R}}}]^{\mathsf{T}}\sim{\mathcal{CN}}(\mathbf{0},\mathbf{I})$, we have $\lambda_i^{-\frac{1}{2}}\overline{{h}}_i+\tilde{{h}}_i\sim{\mathcal{CN}}(\lambda_i^{-\frac{1}{2}}\overline{{h}}_i,1)$, which means that $2\lVert\overline{{h}}_i\lambda_i^{-\frac{1}{2}}+\tilde{{h}}_i\rVert^2$ follows a noncentral chi-square distribution with $2$ degrees of freedom and noncentrality parameter $\kappa_i=2\lambda_i^{-1}\lvert\overline{h}_i\rvert^2$ \cite{johnson1995continuous}. Consequently, $\tilde{a}$ can be expressed as a weighted sum of noncentral chi-square variates with $2$ degrees of freedom, noncentrality parameter equaling $2\lambda_i^{-1}\lvert\overline{h}_i\rvert^2$, and weight coefficients $a_i=\frac{\lambda_i}{2}$. On this basis, the PDF and CDF of $\tilde{a}$ are presented in the following lemmas.

\begin{lemma}\label{Section_Performance_Analysis_OP_Rician_SNR_PDF_Help_Lemma}
The PDF of $\tilde{a}$ is given by
\begin{align}\label{Section_Performance_Analysis_OP_Rician_SNR_PDF_Help_Expression}
f_{\tilde{a}}\left(x\right)=
\frac{e^{-\frac{x}{2\varpi}}x^{r_{\mathbf{R}}-1}}{(2\varpi)^{r_{\mathbf{R}}}\Gamma(r_{\mathbf{R}})}\sum_{k=0}^{\infty}\frac{k!c_k}{(r_{\mathbf{R}})_k}L_k^{(r_{\mathbf{R}}-1)}\left(\frac{r_{\mathbf{R}}x}{2\varpi\xi_0}\right),
\end{align}
where $L_n^{(\alpha)}(\cdot)$ is the generalized Laguerre polynomial \cite[Eq. (8.970.1)]{Ryzhik2007} and $(x)_n$ is the Pochhammer symbol \cite{olver2010nist}. Based on \cite{abramowitz1988handbook}, we have
\begin{align}
L_n^{(\alpha)}(x)&=\frac{\Gamma(n+\alpha+1)}{n!}\sum_{k=0}^{n}\frac{(-n)_k x^k}{k!\Gamma(\alpha+k+1)},\\
(x)_n&=\frac{\Gamma(x+n)}{\Gamma(x)},(-x)_n=(-1)^n(x-n+1)_n.
\end{align}
The coefficients $c_k$ are recursively obtained using the following formulas:
\begin{subequations}\label{Section_Performance_Analysis_OP_Rician_SNR_PDF_Help_Coefficient}
\begin{align}
c_k&=\frac{1}{k}\sum_{j=0}^{k-1}c_jd_{k-j},~k\geq1,\\
c_0&=\left(\frac{r_{\mathbf{R}}}{\xi_0}\right)^{r_{\mathbf{R}}}{e}^{-\frac{1}{2}\sum\limits_{i=1}^{r_{\mathbf{R}}}\frac{\kappa_ia_i({r_{\mathbf{R}}}-\xi_0)}{\varpi\xi_0+a_i({r_{\mathbf{R}}}-\xi_0)}}\nonumber\\
&\times\prod_{i=1}^{r_{\mathbf{R}}}\left(1+\frac{a_i}{\varpi}\left({{r_{\mathbf{R}}}}/{\xi_0}-1\right)\right)^{-1},\\
d_j&=-\frac{j\varpi {r_{\mathbf{R}}}}{2\xi_0}\sum_{i=1}^{r_{\mathbf{R}}}\kappa_ia_i(\varpi-a_i)^{j-1}\left(\frac{\xi_0}{\varpi\xi_0+a_i({r_{\mathbf{R}}}-\xi_0)}\right)^{j+1}\nonumber\\
&+\sum_{i=1}^{r_{\mathbf{R}}}\left(\frac{1-a_i/\varpi}{1+(a_i/\varpi)({r_{\mathbf{R}}}/\xi_0-1)}\right)^{j},j\geq1.
\end{align}
\end{subequations}
In \eqref{Section_Performance_Analysis_OP_Rician_SNR_PDF_Help_Coefficient}, the parameters $\xi_0$ and $\varpi$ are selected in a suitable manner to guarantee the uniform convergence of \eqref{Section_Performance_Analysis_OP_Rician_SNR_PDF_Help_Expression}. More specifically, if $\xi_0<{r_{\mathbf{R}}}/2$, then the series representation in \eqref{Section_Performance_Analysis_OP_Rician_SNR_PDF_Help_Expression} uniformly converges in any finite interval for all $\varpi>0$. If $\xi_0\geq {r_{\mathbf{R}}}/2$, then \eqref{Section_Performance_Analysis_OP_Rician_SNR_PDF_Help_Expression} uniformly converges in any finite interval for
$\varpi>(2-{r_{\mathbf{R}}}/\xi_0)a_{\max}/2$, where $a_{\max}=\max_{i=1,\ldots,r_{\mathbf{R}}}{a_i}$. In this paper, we set $\xi_0={r_{\mathbf{R}}}/3$ and $\varpi=\frac{1}{r_{\mathbf{R}}}\sum_{i=1}^{r_{\mathbf{R}}}a_i$.
\end{lemma}

\begin{proof}
Please refer to \cite[Section 3]{castano2005distribution}.
\end{proof}

\begin{lemma}\label{Section_Performance_Analysis_OP_Rician_SNR_CDF_Help_Lemma}
The CDF of $\tilde{a}$ is given by
\begin{equation}\label{Section_Performance_Analysis_OP_Rician_SNR_CDF_Help_Expression}
\begin{split}
F_{\tilde{a}}\left(x\right)&=
\frac{e^{-\frac{x}{2\varpi}}x^{{r_{\mathbf{R}}}}}{(2\varpi)^{{r_{\mathbf{R}}}+1}\Gamma({r_{\mathbf{R}}}+1)}\\
&\times\sum_{k=0}^{\infty}\frac{k!m_k}{({r_{\mathbf{R}}}+1)_k}L_k^{({r_{\mathbf{R}}})}\left(\frac{{r_{\mathbf{R}}}+1}{2\varpi\xi_0}x\right).
\end{split}
\end{equation}
The coefficients $m_k$ are recursively obtained using the following formulas:
{
\begin{subequations}\label{Section_Performance_Analysis_OP_Rician_SNR_CDF_Help_Coefficient}
\begin{align}
m_k&=\frac{1}{k}\sum_{j=0}^{k-1}m_jq_{k-j},~k\geq1,\\
m_0&=2\left({r_{\mathbf{R}}}+1\right)^{{r_{\mathbf{R}}}+1}{e}^{-\frac{1}{2}\sum\limits_{i=1}^{r_{\mathbf{R}}}\frac{\kappa_ia_i({r_{\mathbf{R}}}+1-\xi_0)}{\varpi\xi_0+a_i({r_{\mathbf{R}}}+1-\xi_0)}}\nonumber\\
&\times\frac{\varpi^{{r_{\mathbf{R}}}+1}}{{r_{\mathbf{R}}}+1-\xi_0}\prod_{i=1}^{r_{\mathbf{R}}}\left(\varpi\xi_0+a_i({r_{\mathbf{R}}}+1-\xi_0)\right)^{-1},\\
q_j&=-\frac{j\varpi({r_{\mathbf{R}}}+1)}{2\xi_0}\sum_{i=1}^{r_{\mathbf{R}}}\kappa_ia_i(\varpi-a_i)^{j-1}\\
&\times\left(\frac{\xi_0}{\varpi\xi_0+a_i({r_{\mathbf{R}}}+1-\xi_0)}\right)^{j+1}+\left(\frac{-\xi_0}{{r_{\mathbf{R}}}+1-\xi_0}\right)^{j}\nonumber\\
&+\sum_{i=1}^{r_{\mathbf{R}}}\left(\frac{\xi_0(\varpi-a_i)}{\varpi\xi_0+a_i({r_{\mathbf{R}}}+1-\xi_0)}\right)^{j},j\geq1.
\end{align}
\end{subequations}
}In \eqref{Section_Performance_Analysis_OP_Rician_SNR_CDF_Help_Coefficient}, parameters $\xi_0$ and $\varpi$ are selected in a similar manner as for the PDF to ensure the uniform convergence of \eqref{Section_Performance_Analysis_OP_Rician_SNR_CDF_Help_Expression}.
\end{lemma}

\begin{proof}
Please refer to \cite[Section 3]{castano2005distribution}.
\end{proof}
Based on \eqref{Section_Performance_Analysis_OP_Rician_SNR_CDF_Help_Expression} and \eqref{Section_Performance_Analysis_OP_Rician_SNR_PDF_Help_Expression}, we obtain the CDF and PDF of ${\lVert{\mathbf{h}}\rVert^2}=\lVert\overline{\mathbf{h}}+\mathbf{R}^{\frac{1}{2}}\tilde{\mathbf{h}}\rVert^2$ as follows:
\begin{align}
F_{\lVert{\mathbf{h}}\rVert^2}(x)&=F_{\tilde{a}}\left(x-\tilde{a}_0\right), x\geq\tilde{a}_0 \label{Section_Performance_Analysis_OP_Rician_SNR_CDF_Master_Expression},\\
f_{\lVert{\mathbf{h}}\rVert^2}(x)&=f_{\tilde{a}}\left(x-\tilde{a}_0\right), x\geq\tilde{a}_0,
\label{Section_Performance_Analysis_OP_Rician_SNR_PDF_Master_Expression}
\end{align}
where $\tilde{a}_0\geq0$.

$\bullet$ \emph{\textbf{Step 2 -  Deriving a Closed-Form Expression for the OP:}}
In the second step, we exploit the CDF of $\lVert{\mathbf{h}}\rVert^2$ to calculate the OP, which yields the following theorem.

\begin{theorem}\label{Section_Performance_Analysis_OP_Rician_Explicit_Expression_Theorem}
The OP of the considered system is given by
\begin{align}\label{Section_Performance_Analysis_OP_Rician_Explicit_Expression}
\mathcal{P}_{\rm{rician}}
=\left\{
\begin{array}{rl}
F_{\tilde{a}}\left(\frac{2^{\mathcal{R}}-1}{p/\sigma^2}-\tilde{a}_0\right)           & {\frac{(2^{\mathcal{R}}-1)\sigma^2}{\tilde{a}_0}>p\geq0}\\
0         & {\frac{(2^{\mathcal{R}}-1)\sigma^2}{\tilde{a}_0}<p}
\end{array} \right..
\end{align}
\end{theorem}

\begin{proof}
This theorem can be directly proved by substituting \eqref{Section_Performance_Analysis_OP_Rician_SNR_CDF_Master_Expression} into \eqref{Section_Performance_Analysis_OP_Calculation_MISO_Rician}.
\end{proof}

\begin{remark}
The result in \eqref{Section_Performance_Analysis_OP_Rician_Explicit_Expression} suggests that the OP for NFC Rician fading channels is a piecewise function of $p$, where the two pieces intersect at $p_0={\frac{(2^{\mathcal{R}}-1)\sigma^2}{\tilde{a}_0}}$. This means that for MISO Rician channels, achieving zero OP only requires a finite transmit power. By contrast, for MISO Rayleigh channels, we have $\overline{\mathbf{h}}=\mathbf{0}$ and thus $\tilde{a}_0=0$, which means that achieving a zero OP requires an infinitely high transmit power. The above discussion reveals that the LoS component can improve the outage performance of NFC.
\end{remark}

$\bullet$ \emph{\textbf{Step 3 - Deriving the Rate of Convergence of the OP:}}
In Step 3 of Section \ref{Section_Performance_Analysis_OP_Rayleigh_Part}, we investigate the high-SNR asymptotic behaviour of $\mathcal{P}_{\rm{rayleigh}}$ and formulate it into the standard form $({\mathcal{G}}_{\rm{a}}\cdot p)^{-{\mathcal{G}}_{\rm{d}}}$ (given by \eqref{Section_Performance_Analysis_OP_Asymptotic_Expression_Standard}). Considering \eqref{Section_Performance_Analysis_OP_Rician_Explicit_Expression}, $\mathcal{P}_{\rm{rician}}$ satisfies $\lim_{p\rightarrow\infty}\mathcal{P}_{\rm{rician}}=0$. However, since $\mathcal{P}_{\rm{rician}}$ is a piecewise function of $p$, we cannot formulate its high-SNR approximation in the standard form as given by \eqref{Section_Performance_Analysis_OP_Asymptotic_Expression_Standard}. It is worth noting that the standard form $({\mathcal{G}}_{\rm{a}}\cdot p)^{-{\mathcal{G}}_{\rm{d}}}$ essentially characterizes the rate of $\mathcal{P}_{\rm{rayleigh}}$ converging to zero. Motivated by this, in the last step, we investigate the rate of $\mathcal{P}_{\rm{rician}}$ converging to zero.

Based on \eqref{Section_Performance_Analysis_OP_Rician_Explicit_Expression}, as $p$ increases, $\mathcal{P}_{\rm{rician}}$ becomes zero when $p=p_0={\frac{(2^{\mathcal{R}}-1)\sigma^2}{\tilde{a}_0}}$, i.e.,
\begin{align}
\lim\nolimits_{p\rightarrow{\frac{(2^{\mathcal{R}}-1)\sigma^2}{\tilde{a}_0}}}\mathcal{P}_{\rm{rician}}=0.
\end{align}
The following corollary characterizes the asymptotic behaviour of $\mathcal{P}_{\rm{rician}}$ as $p\rightarrow{\frac{(2^{\mathcal{R}}-1)\sigma^2}{\tilde{a}_0}}$.

\begin{corollary}\label{Section_Performance_Analysis_OP_Asymptotic_Expression_Standard_Rician_Theorem}
When $p\rightarrow{\frac{(2^{\mathcal{R}}-1)\sigma^2}{\tilde{a}_0}}$, $\mathcal{P}_{\rm{rician}}$ satisfies
{
\begin{equation}\label{Section_Performance_Analysis_OP_Asymptotic_Expression_Standard_Rician}
\begin{split}
&\lim\nolimits_{p\rightarrow{\frac{(2^{\mathcal{R}}-1)\sigma^2}{\tilde{a}_0}}}{\mathcal{P}}_{\rm{rician}}\\
&\simeq\left({\mathfrak{G}}_{\rm{a}}^{-1}\cdot\left(\frac{1}{p}-\frac{1}{p_0}
\right)\right)^{{\mathfrak{G}}_{\rm{d}}},
\end{split}
\end{equation}
}where ${\mathfrak{G}}_{\rm{a}}=\frac{\left(r_{\mathbf{R}}!\prod_{i=1}^{r_{\mathbf{R}}}{\lambda}_i{e}^{\lambda_i^{-1}\lvert\overline{h}_i\rvert^2}\right)^{{1}/{r_{\mathbf{R}}}}}
{\sigma^2(2^{\mathcal{R}}-1)}$ and ${\mathfrak{G}}_{\rm{d}}=r_{\mathbf{R}}$.
\end{corollary}

\begin{proof}
Please refer to Appendix \ref{Section_Performance_Analysis_OP_Asymptotic_Expression_Standard_Rician_Theorem_Proof}.
\end{proof}

\begin{figure}[!t]
    \setlength{\abovecaptionskip}{0pt}
    \centering
    \includegraphics[width=0.4\textwidth]{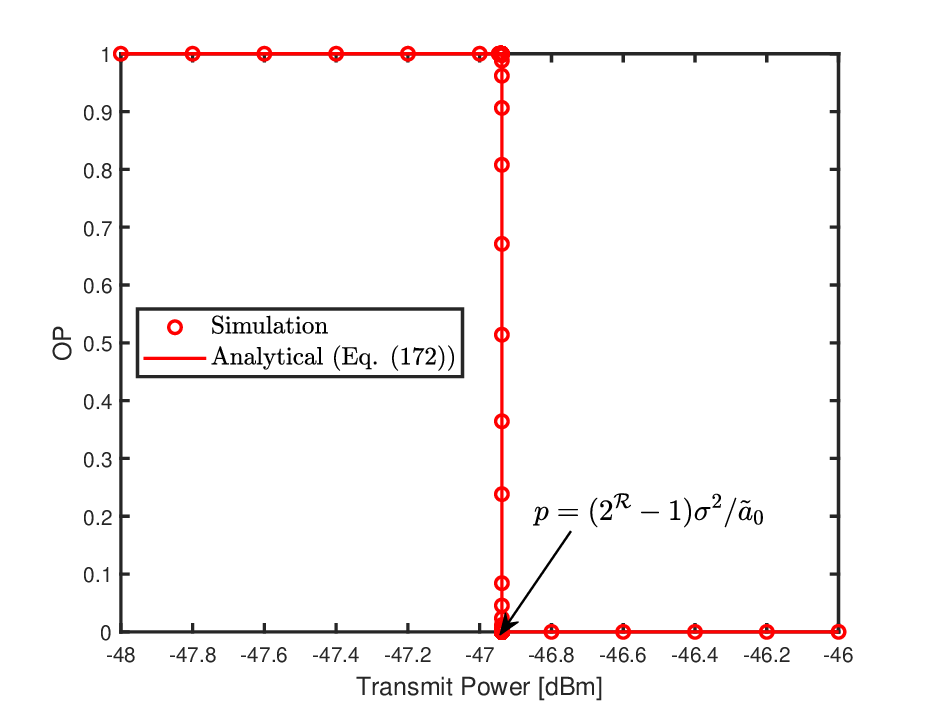}
    \caption{Outage probability for correlated MISO Rician channels and data rate $\mathcal{R}=1$ bps/Hz. The system is operating at $28$ GHz. $N_x=N_z=33$, $d=\frac{\lambda}{2}$, $(\theta,\phi)=(\frac{\pi}{2},\frac{\pi}{2})$, $r=4$ m, $A=\frac{\lambda^2}{4\pi}$, $e_a=1$, ${\bm\rho}=\hat{\mathbf J}({\mathbf{s}}_{m,n})=[1,0,0]^{\mathsf{T}}$, $\forall m,n$, $\sigma^2=-90$ dBm, $L=4$, and $\sigma_{\ell}^2=1$, $\forall \ell$. The coordinates (in meters) of the scatterers are given by $
    \mathbf{s}_{1}^{\star}=[\frac{7}{10}\sin{\frac{\pi}{4}},r,\frac{7}{10}\cos{\frac{\pi}{4}}]^{\mathsf{T}}$, $\mathbf{s}_{2}^{\star}=[\frac{4}{5}\sin{\frac{\pi}{8}},r,\frac{4}{5}\cos{\frac{\pi}{8}}]^{\mathsf{T}}$, $
    \mathbf{s}_{3}^{\star}=[\frac{13}{20}\sin{\frac{47\pi}{64}},r,\frac{13}{20}\cos{\frac{47\pi}{64}}]^{\mathsf{T}}$, and
    $\mathbf{s}_{4}^{\star}=[\frac{2}{5}\sin{\frac{\pi}{7}},r,\frac{2}{5}\cos{\frac{\pi}{7}}]^{\mathsf{T}}$. The LoS channels $\overline{\mathbf{h}}$, $\{h_{\ell}({\mathbf{r}}_{\ell},\mathbf{r})\}_{\ell=1}^{L}$, and $\{{\mathbf{h}}_{\ell}(\mathbf{r}_{\ell})\}_{\ell=1}^{L}$ follow the USW channel model.}
    \label{Performance Analysis Figure: Rician_OP_Discrete_Explicit}
\end{figure}

\begin{figure}[!t]
    \setlength{\abovecaptionskip}{0pt}
    \centering
    \includegraphics[width=0.4\textwidth]{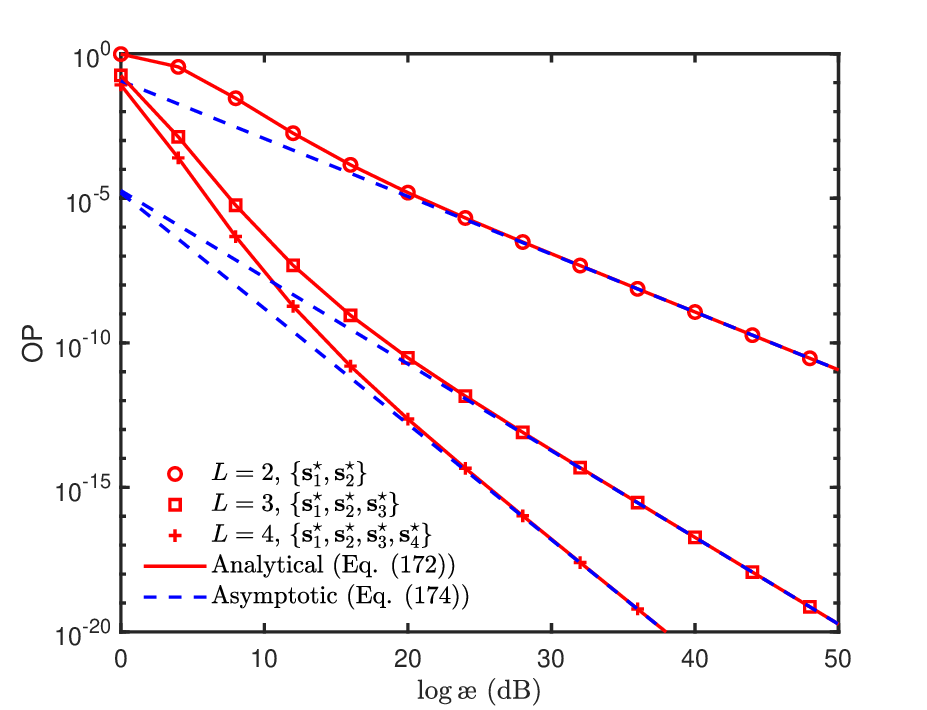}
    \caption{Outage probability versus $\ae$ for correlated MISO Rician channels and data rate $\mathcal{R}=1$ bps/Hz. The system is operating at $28$ GHz. $N_x=N_z=33$, $d=\frac{\lambda}{2}$, $(\theta,\phi)=(\frac{\pi}{2},\frac{\pi}{2})$, $r=4$ m, $A=\frac{\lambda^2}{4\pi}$, $e_a=1$, ${\bm\rho}=\hat{\mathbf J}({\mathbf{s}}_{m,n})=[1,0,0]^{\mathsf{T}}$, $\forall m,n$, $\sigma^2=-90$ dBm, $L=4$, and $\sigma_{\ell}^2=1$, $\forall \ell$. The coordinates (in meters) of the scatterers are given by $
    \mathbf{s}_{1}^{\star}=[\frac{7}{10}\sin{\frac{\pi}{4}},r,\frac{7}{10}\cos{\frac{\pi}{4}}]^{\mathsf{T}}$, $\mathbf{s}_{2}^{\star}=[\frac{4}{5}\sin{\frac{\pi}{8}},r,\frac{4}{5}\cos{\frac{\pi}{8}}]^{\mathsf{T}}$, $
    \mathbf{s}_{3}^{\star}=[\frac{13}{20}\sin{\frac{47\pi}{64}},r,\frac{13}{20}\cos{\frac{47\pi}{64}}]^{\mathsf{T}}$, and
    $\mathbf{s}_{4}^{\star}=[\frac{2}{5}\sin{\frac{\pi}{7}},r,\frac{2}{5}\cos{\frac{\pi}{7}}]^{\mathsf{T}}$. The LoS channels $\overline{\mathbf{h}}$, $\{h_{\ell}({\mathbf{r}}_{\ell},\mathbf{r})\}_{\ell=1}^{L}$, and $\{{\mathbf{h}}_{\ell}(\mathbf{r}_{\ell})\}_{\ell=1}^{L}$ follow the USW channel model.}
    \label{Performance Analysis Figure: Rician_OP_Discrete_Asymp}
\end{figure}

\begin{remark}
The results in \textbf{Corollary} \ref{Section_Performance_Analysis_OP_Asymptotic_Expression_Standard_Rician_Theorem} indicate that as $p\rightarrow{\frac{(2^{\mathcal{R}}-1)\sigma^2}{\tilde{a}_0}}$, the rate at which ${\mathcal{P}}_{\rm{rician}}$ converges to zero equals the rate of $\left({\mathfrak{G}}_{\rm{a}}^{-1}\cdot\left(\frac{1}{p}-
\frac{\tilde{a}_0}{\left(2^{\mathcal{R}}-1\right)\sigma^2}\right)\right)^{{\mathfrak{G}}_{\rm{d}}}$ converging to zero. Based on \eqref{Section_Performance_Analysis_OP_Asymptotic_Expression_Standard_Rician}, the diversity order and the array gain of ${\mathcal{P}}_{\rm{rician}}$ are given by $r_{\mathbf{R}}$ and $\frac{\left(r_{\mathbf{R}}!\prod_{i=1}^{r_{\mathbf{R}}}{\lambda}_i{e}^{\lambda_i^{-1}\lvert\overline{h}_i\rvert^2}\right)^{{1}/{r_{\mathbf{R}}}}}
{\sigma^2(2^{\mathcal{R}}-1)}$, respectively.
\end{remark}

\begin{remark}
When $\overline{\mathbf{h}}=\mathbf{0}$, the Rician model degenerates to the Rayleigh model as $\tilde{a}_0=0$. In this case, \eqref{Section_Performance_Analysis_OP_Asymptotic_Expression_Standard_Rician} degenerates into the standard form given by \eqref{Section_Performance_Analysis_OP_Asymptotic_Expression_Standard}. By comparing \eqref{Section_Performance_Analysis_OP_Asymptotic_Expression_Standard} with \eqref{Section_Performance_Analysis_OP_Asymptotic_Expression_Standard_Rician}, we find that ${\mathcal{G}}_{\rm{d}}={\mathfrak{G}}_{\rm{d}}$ and ${\mathcal{G}}_{\rm{a}}<{\mathfrak{G}}_{\rm{a}}$. This suggests that the OP for Rician fading yields the same diversity order as that for Rayleigh fading, yet has a larger array gain than the latter one.
\end{remark}

$\bullet$ \emph{\textbf{Numerical Results:}} To further illustrate the derived results, in {\figurename} {\ref{Performance Analysis Figure: Rician_OP_Discrete_Explicit}}, we show the OP for USW LoS channels as a function of the transmit power, $p$. The analytical results are calculated using \eqref{Section_Performance_Analysis_OP_Rician_Explicit_Expression}. As can be observed in {\figurename} {\ref{Performance Analysis Figure: Rician_OP_Discrete_Explicit}}, the analytical results are in good agreement with the simulated results, and the OP collapses to zero when $p={\frac{(2^{\mathcal{R}}-1)\sigma^2}{\tilde{a}_0}}$. This verifies the correctness of Theorem \ref{Section_Performance_Analysis_OP_Rician_Explicit_Expression_Theorem}. The simulation parameters used to generate {\figurename} {\ref{Performance Analysis Figure: Rician_OP_Discrete_Explicit}} are the same as those used to generate {\figurename} {\ref{Performance Analysis Figure: OP_Discrete}}. For this simulation setting, we have $\tilde{a}_0\gg {\mathbb E}\{\tilde{a}\}$, which means that the considered BS-to-user channel is LoS-dominated. By comparing the results in {\figurename} {\ref{Performance Analysis Figure: Rician_OP_Discrete_Explicit}} and {\figurename} {\ref{Performance Analysis Figure: OP_Discrete}}, we find that to achieve the same OP, the Rician fading channel requires much less power resources than the Rayleigh fading channel. This performance gain mainly originates from the strong LoS component for Rician fading. To illustrate the ROC of the OP, we depict ${\mathcal{P}}_{\rm{rician}}$ versus $\ae\triangleq\left(\frac{2^{\mathcal{R}}-1}{p/\sigma^2}-\tilde{a}_0\right)^{-1}$ in {\figurename} {\ref{Performance Analysis Figure: Rician_OP_Discrete_Asymp}}. As can be observed, when $\ae\rightarrow\infty$, i.e., $p\rightarrow{\frac{(2^{\mathcal{R}}-1)\sigma^2}{\tilde{a}_0}}$, the derived asymptotic results approach the analytical results. Besides, it can be observed that a higher diversity order is achievable when the channel contains more scatterers.

\subsubsection{Analysis of the ECC for Rician Channels}
Having analyzed the OP, we turn our attention to the ECC, which is given as follows:
\begin{equation}\label{Section_Performance_Analysis_ECC_Definition_MISO_Rician}
\begin{split}
\bar{\mathcal{C}}_{\rm{rician}}&={\mathbb{E}}\{\log_2(1+\bar\gamma\lVert{\mathbf{h}}\rVert^2)\}\\
&=\int_{{\tilde{a}}_0}^{\infty}\log_2(1+p/\sigma^2 x)f_{\lVert{\mathbf{h}}\rVert^2}(x){\rm{d}}x\\
&=\int_0^{\infty}\log_2(1+p/\sigma^2 ({\tilde{a}}_0+x))f_{\tilde{a}}\left(x\right){\rm{d}}x.
\end{split}
\end{equation}
The analysis of the ECC also comprises three steps.

$\bullet$ \emph{\textbf{Step 1 - Analyzing the Statistics of the Channel Gain:}} 
In the first step, we analyze the statistics of the channel gain $\lVert{\mathbf{h}}\rVert^2$. The results can be found in \textbf{Lemma} \ref{Section_Performance_Analysis_OP_Rician_SNR_PDF_Help_Lemma}.

$\bullet$ \emph{\textbf{Step 2 - Deriving a Closed-Form Expression for the ECC:}}
In the second step, we exploit the CDF of $\lVert{\mathbf{h}}\rVert^2$ to derive a closed-form expression for $\bar{\mathcal{C}}_{\rm{rician}}$, which leads to the following theorem.
\vspace{-5pt}
\begin{theorem}
The ECC can be expressed in closed form as follows:
\vspace{-0.1cm}
\begin{equation}\label{Section_Performance_Analysis_ECC_Calculation_MISO_Rician}
\begin{split}
&\bar{\mathcal{C}}_{\rm{rician}}=\log_2(1+p/\sigma^2{\tilde{a}}_0)+\sum_{k=0}^{\infty}\sum_{t=0}^{k}
\frac{c_k(-k)_t({r_{\mathbf{R}}}/\xi_0)^t}{t!\Gamma({r_{\mathbf{R}}}+t)\ln{2}}\\
&\quad\times\sum_{\mu=0}^{{r_{\mathbf{R}}}+t-1}\frac{\Gamma({r_{\mathbf{R}}}+t)}{\Gamma({r_{\mathbf{R}}}+t-\mu)}\left[\frac{(-1)^{{r_{\mathbf{R}}}+t-\mu}e^{\frac{1}{\mathfrak{a}}}}{{\mathfrak{a}}^{{r_{\mathbf{R}}}+t-1-\mu}}{\rm{Ei}}\left({\frac{-1}{\mathfrak{a}}}\right)\right.\\
&\quad+\left.\sum_{u=1}^{{r_{\mathbf{R}}}+t-1-\mu}(u-1)!\left(\frac{-1}{\mathfrak{a}}\right)^{{r_{\mathbf{R}}}+t-1-\mu-u}\right]
\end{split}
\vspace{-0.1cm}
\end{equation}
with ${\mathfrak{a}}\triangleq2\varpi p/(\sigma^2+p {\tilde{a}}_0)$.
\end{theorem}

\begin{proof}
This theorem is proved by substituting \eqref{Section_Performance_Analysis_OP_Rician_SNR_PDF_Master_Expression} into \eqref{Section_Performance_Analysis_ECC_Definition_MISO_Rician} and solving the resulting integral with the aid of \cite[Eq. (4.337.5)]{Ryzhik2007}.
\end{proof}

\begin{figure}[!t]
    \setlength{\abovecaptionskip}{0pt}
    \centering
    \includegraphics[width=0.4\textwidth]{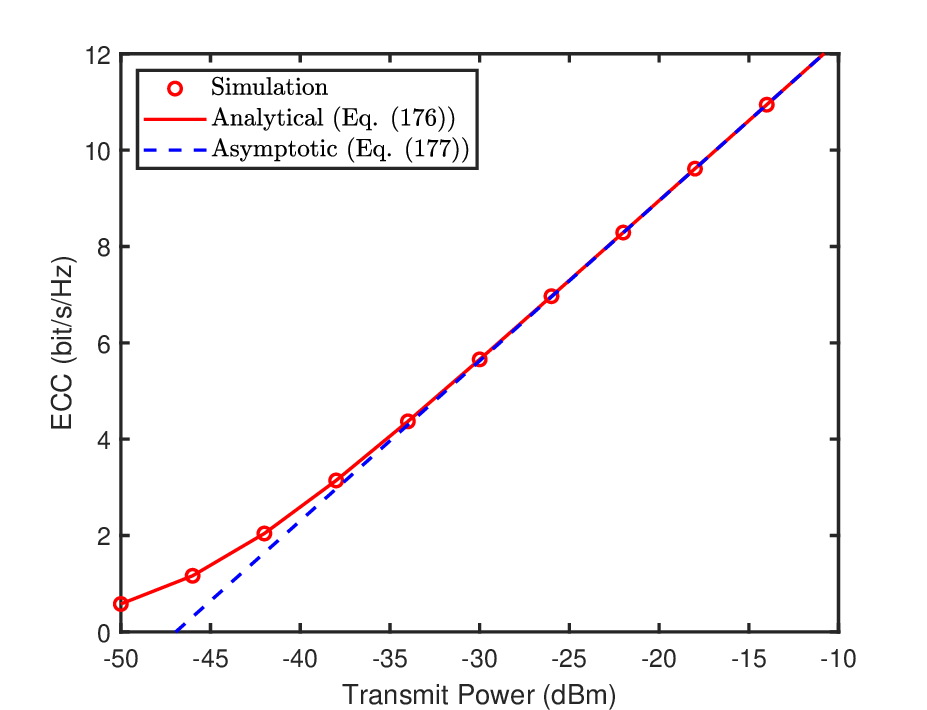}
    \caption{Ergodic channel capacity for correlated MISO Rician channels. The system is operating at $28$ GHz. $N_x=N_z=33$, $d=\frac{\lambda}{2}$, $(\theta,\phi)=(\frac{\pi}{2},\frac{\pi}{2})$, $r=4$ m, $A=\frac{\lambda^2}{4\pi}$, $e_a=1$, ${\bm\rho}=\hat{\mathbf J}({\mathbf{s}}_{m,n})=[1,0,0]^{\mathsf{T}}$, $\forall m,n$, $\sigma^2=-90$ dBm, $L=4$, and $\sigma_{\ell}^2=1$, $\forall \ell$. The coordinates (in meters) of the scatterers are given by $
    \mathbf{s}_{1}^{\star}=[\frac{7}{10}\sin{\frac{\pi}{4}},r,\frac{7}{10}\cos{\frac{\pi}{4}}]^{\mathsf{T}}$, $\mathbf{s}_{2}^{\star}=[\frac{4}{5}\sin{\frac{\pi}{8}},r,\frac{4}{5}\cos{\frac{\pi}{8}}]^{\mathsf{T}}$, $
    \mathbf{s}_{3}^{\star}=[\frac{13}{20}\sin{\frac{47\pi}{64}},r,\frac{13}{20}\cos{\frac{47\pi}{64}}]^{\mathsf{T}}$, and
    $\mathbf{s}_{4}^{\star}=[\frac{2}{5}\sin{\frac{\pi}{7}},r,\frac{2}{5}\cos{\frac{\pi}{7}}]^{\mathsf{T}}$. The LoS channels $\overline{\mathbf{h}}$, $\{h_{\ell}({\mathbf{r}}_{\ell},\mathbf{r})\}_{\ell=1}^{L}$, and $\{{\mathbf{h}}_{\ell}(\mathbf{r}_{\ell})\}_{\ell=1}^{L}$ follow the USW channel model.}
    \label{Performance Analysis Figure: ECC_Discrete_Rician}
\end{figure}

$\bullet$ \emph{\textbf{Step 3 - Deriving a High-SNR Approximation for the ECC:}}
In the third step, we perform an asymptotic analysis of the ECC for a sufficiently large transmit power, i.e., $p\rightarrow\infty$ in order to obtain more insights for system design. The asymptotic ECC in the high-SNR regime is presented in the following corollary.

\begin{corollary}\label{Section_Performance_Analysis_ECC_Asymptotic_Expression_Standard_Rician_Theorem}
The asymptotic ECC in the high-SNR regime can be expressed in the following form:
\begin{align}\label{Section_Performance_Analysis_ECC_Asymptotic_Expression_Standard_Rician}
\lim_{p\rightarrow\infty}{\bar{\mathcal{C}}}_{\rm{rician}}\simeq {\mathfrak{S}}_{\infty}\left(\log_2(p)-{\mathfrak{L}}_{\infty}\right),
\end{align}
where ${\mathfrak{S}}_{\infty}=1$ and
\begin{equation}\label{Section_Performance_Analysis_ECC_Calculation_Asymptotic_MISO_Rician}
\begin{split}
{\mathfrak{L}}_{\infty}&=-\mathbb{E}\{\log_2({\lVert{\mathbf{h}}\rVert^2/\sigma^2})\}\\
&=-\log_2({\tilde{a}}_0)-\sum_{k=0}^{\infty}\sum_{t=0}^{k}
\frac{c_k(-k)_t({r_{\mathbf{R}}}/\xi_0)^t}{t!\Gamma({r_{\mathbf{R}}}+t)\ln{2}}\\
&\times\sum_{\mu=0}^{{r_{\mathbf{R}}}+t-1}\frac{\Gamma({r_{\mathbf{R}}}+t)}{\Gamma({r_{\mathbf{R}}}+t-\mu)}\left[\frac{(-1)^{{r_{\mathbf{R}}}+t-\mu}e^{{\tilde{a}}_0}}{{{\tilde{a}}_0}^{-{r_{\mathbf{R}}}-t+1+\mu}}{\rm{Ei}}\left(-{\tilde{a}}_0\right)\right.\\
&+\left.\sum_{u=1}^{{r_{\mathbf{R}}}+t-1-\mu}(u-1)!\left(-{\tilde{a}}_0\right)^{{r_{\mathbf{R}}}+t-1-\mu-u}\right].
\end{split}
\end{equation}
\end{corollary}

\begin{proof}
The proof closely follows \textbf{Corollary} \ref{Section_Performance_Analysis_ECC_Asymptotic_Expression_Standard_Theorem}.
\end{proof}

\begin{remark}\label{Section_Performance_Analysis_ECC_Asymptotic_Expression_Standard_Rician_Remark}
The results in \textbf{Corollary} \ref{Section_Performance_Analysis_ECC_Asymptotic_Expression_Standard_Rician_Theorem} suggest that the high-SNR slope and the high-SNR power offset of ${\bar{\mathcal{C}}}_{\rm{rician}}$ are given by $1$ and $-\mathbb{E}\{\log_2({\lVert{\mathbf{h}}\rVert^2/\sigma^2})\}$, respectively. By comparing \eqref{Section_Performance_Analysis_ECC_Asymptotic_Expression_Standard_Rician} with \eqref{Section_Performance_Analysis_ECC_Asymptotic_Expression_Standard}, we find that ${\mathcal{S}}_{\infty}={\mathfrak{S}}_{\infty}$, which suggests that the ECC for Rician fading yields the same high-SNR slope as that for Rayleigh fading. However, providing a theoretical comparison between ${\mathcal{L}}_{\infty}$ and ${\mathfrak{L}}_{\infty}$ is challenging. Thus, we will use numerical results to compare these two metrics.
\end{remark}

$\bullet$ \emph{\textbf{Numerical Results:}} To illustrate the above results, in {\figurename} {\ref{Performance Analysis Figure: ECC_Discrete_Rician}}, we show the ECC for USW LoS channels versus the transmit power, $p$. As {\figurename} {\ref{Performance Analysis Figure: ECC_Discrete_Rician}} shows, the analytical results are in excellent agreement with the simulated results, and the derived asymptotic results approach the analytical results in the high-SNR regime. The simulation parameters used to generate {\figurename} {\ref{Performance Analysis Figure: ECC_Discrete_Rician}} are the same as those used to generate {\figurename} {\ref{Performance Analysis Figure: ECC_Discrete}}. By comparing the results in these two figures, we find that to achieve the same ECC, the Rician fading channel requires much less power than the Rayleigh fading channel. Since the ECCs achieved for these two types of fading have the same high-SNR slope, we conclude that the ECC for Rician fading yields a smaller high-SNR power offset than that for Rayleigh fading, i.e., ${\mathfrak{L}}_{\infty}<{\mathcal{L}}_{\infty}$. This performance gain mainly originates from the strong LoS component for Rician fading.

\subsubsection{Analysis of the EMI for Rician Channels}
The EMI can be written as follows \cite{Ouyang2020_CL}:
\begin{equation}\label{Section_Performance_Analysis_EMI_Definition_MISO_Rician}
\begin{split}
\bar{\mathcal{I}}_{\mathcal{X}}^{\rm{rician}}&={\mathbb{E}}\{I_{\mathcal X}(p/\sigma^2\lVert{\mathbf{h}}\rVert^2)\}\\
&=\int_{{\tilde{a}}_0}^{\infty}I_{\mathcal X}(p/\sigma^2 x)f_{\lVert{\mathbf{h}}\rVert^2}(x){\rm{d}}x\\
&=\int_{0}^{\infty}I_{\mathcal X}(p/\sigma^2 ({\tilde{a}}_0+x))f_{{\tilde{a}}}(x){\rm{d}}x.
\end{split}
\end{equation}
To characterize the EMI, we follow three main steps which are detailed in the sequel.

$\bullet$ \emph{\textbf{Step 1 - Analyzing the Statistics of the Channel Gain:}}
Similar to the analyses of the OP and the ECC, we discuss the statistics of $\lVert{\mathbf{h}}\rVert^2$ in the first step. The results are provided in \eqref{Section_Performance_Analysis_OP_Rician_SNR_CDF_Master_Expression} and \eqref{Section_Performance_Analysis_OP_Rician_SNR_PDF_Master_Expression}.

\begin{figure}[!t]
    \setlength{\abovecaptionskip}{0pt}
    \centering
    \includegraphics[width=0.4\textwidth]{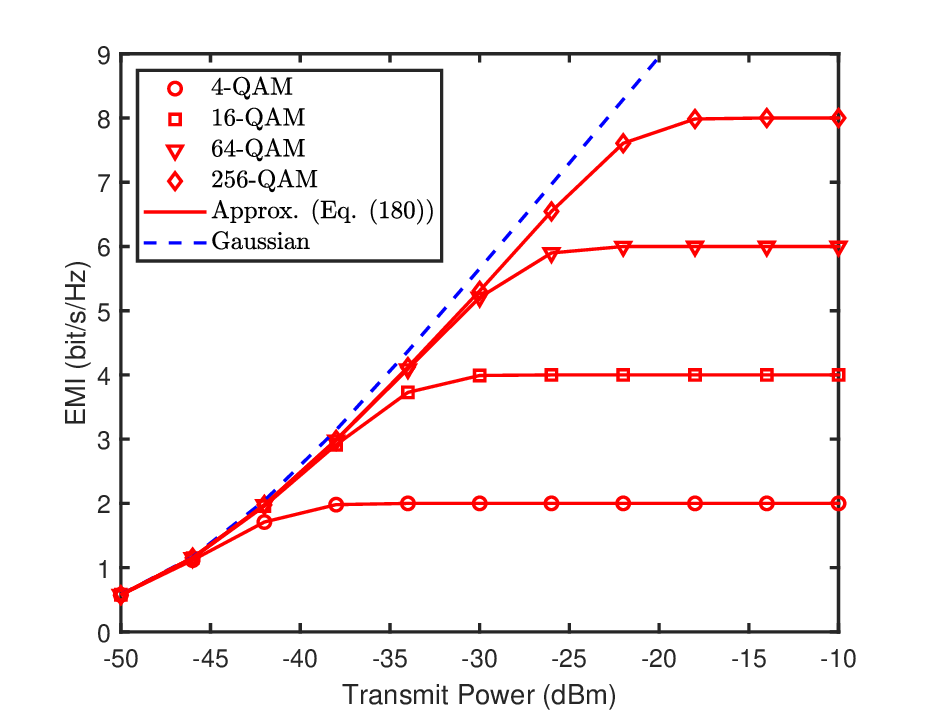}
    \caption{Ergodic mutual information for correlated MISO Rician channels. The system is operating at $28$ GHz. $N_x=N_z=33$, $d=\frac{\lambda}{2}$, $(\theta,\phi)=(\frac{\pi}{2},\frac{\pi}{2})$, $r=4$ m, $A=\frac{\lambda^2}{4\pi}$, $e_a=1$, ${\bm\rho}=\hat{\mathbf J}({\mathbf{s}}_{m,n})=[1,0,0]^{\mathsf{T}}$, $\forall m,n$, $\sigma^2=-90$ dBm, $L=4$, and $\sigma_{\ell}^2=1$, $\forall \ell$. The coordinates (in meters) of the scatterers are given by $
    \mathbf{s}_{1}^{\star}=[\frac{7}{10}\sin{\frac{\pi}{4}},r,\frac{7}{10}\cos{\frac{\pi}{4}}]^{\mathsf{T}}$, $\mathbf{s}_{2}^{\star}=[\frac{4}{5}\sin{\frac{\pi}{8}},r,\frac{4}{5}\cos{\frac{\pi}{8}}]^{\mathsf{T}}$, $
    \mathbf{s}_{3}^{\star}=[\frac{13}{20}\sin{\frac{47\pi}{64}},r,\frac{13}{20}\cos{\frac{47\pi}{64}}]^{\mathsf{T}}$, and
    $\mathbf{s}_{4}^{\star}=[\frac{2}{5}\sin{\frac{\pi}{7}},r,\frac{2}{5}\cos{\frac{\pi}{7}}]^{\mathsf{T}}$. The LoS channels $\overline{\mathbf{h}}$, $\{h_{\ell}({\mathbf{r}}_{\ell},\mathbf{r})\}_{\ell=1}^{L}$, and $\{{\mathbf{h}}_{\ell}(\mathbf{r}_{\ell})\}_{\ell=1}^{L}$ follow the USW channel model.}
    \label{Performance Analysis Figure: AMI_Discrete Explicit_Rician}
\end{figure}

\begin{figure}[!t]
    \setlength{\abovecaptionskip}{0pt}
    \centering
    \includegraphics[width=0.4\textwidth]{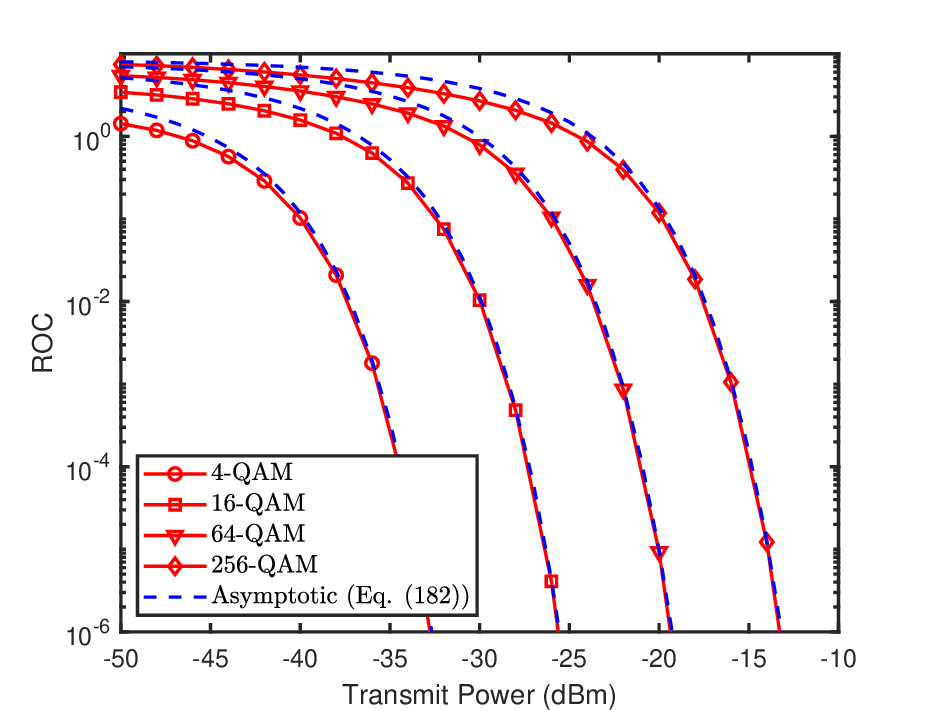}
    \caption{Rate of convergence of the EMI for correlated MISO Rician channels. The system is operating at $28$ GHz. $N_x=N_z=33$, $d=\frac{\lambda}{2}$, $(\theta,\phi)=(\frac{\pi}{2},\frac{\pi}{2})$, $r=4$ m, $A=\frac{\lambda^2}{4\pi}$, $e_a=1$, ${\bm\rho}=\hat{\mathbf J}({\mathbf{s}}_{m,n})=[1,0,0]^{\mathsf{T}}$, $\forall m,n$, $\sigma^2=-90$ dBm, $L=4$, and $\sigma_{\ell}^2=1$, $\forall \ell$. The coordinates (in meters) of the scatterers are given by $
    \mathbf{s}_{1}^{\star}=[\frac{7}{10}\sin{\frac{\pi}{4}},r,\frac{7}{10}\cos{\frac{\pi}{4}}]^{\mathsf{T}}$, $\mathbf{s}_{2}^{\star}=[\frac{4}{5}\sin{\frac{\pi}{8}},r,\frac{4}{5}\cos{\frac{\pi}{8}}]^{\mathsf{T}}$, $
    \mathbf{s}_{3}^{\star}=[\frac{13}{20}\sin{\frac{47\pi}{64}},r,\frac{13}{20}\cos{\frac{47\pi}{64}}]^{\mathsf{T}}$, and
    $\mathbf{s}_{4}^{\star}=[\frac{2}{5}\sin{\frac{\pi}{7}},r,\frac{2}{5}\cos{\frac{\pi}{7}}]^{\mathsf{T}}$. The LoS channels $\overline{\mathbf{h}}$, $\{h_{\ell}({\mathbf{r}}_{\ell},\mathbf{r})\}_{\ell=1}^{L}$, and $\{{\mathbf{h}}_{\ell}(\mathbf{r}_{\ell})\}_{\ell=1}^{L}$ follow the USW channel model.}
    \label{Performance Analysis Figure: AMI_Discrete Asymptotic_Rician}
\end{figure}

$\bullet$ \emph{\textbf{Step 2 - Deriving a Closed-Form Expressions for the EMI:}}
In the second step, we leverage the PDF of $\lVert{\mathbf{h}}\rVert^2$ to derive a closed-form expression for the EMI. The main results are summarized as follows.

\begin{theorem}\label{Section_Performance_Analysis_EMI_Approximation_Explicit_MISO_Rician_Theorem}
The EMI achieved by finite-alphabet inputs can be approximated as follows
\begin{equation}\label{Section_Performance_Analysis_EMI_Approximation_Explicit_MISO_Rician}
\bar{\mathcal{I}}_{\mathcal{X}}^{\rm{rician}}\approx H_{{\mathbf{p}}_{\mathcal{X}}}-\sum_{k=0}^{\infty}c_k\sum_{t=0}^{k}\frac{(-k)_t}{t!}\sum_{j=1}^{k_{\mathcal{X}}}
\frac{H_{{\mathbf{p}}_{\mathcal{X}}}\zeta^{({\mathcal{X}})}_je^{-\vartheta^{({\mathcal{X}})}_j{\tilde{a}}_0p/\sigma^2}}{(1+2\varpi\vartheta^{({\mathcal{X}})}_jp/\sigma^2)^{U+t}}.
\end{equation}
\end{theorem}

\begin{proof}
This theorem can be proved by substituting \eqref{Section_Performance_Analysis_MI_AWGN_Approximation} and \eqref{Section_Performance_Analysis_OP_Rician_SNR_CDF_Master_Expression} into
\eqref{Section_Performance_Analysis_EMI_Definition_MISO_Rician} and calculating the resulting integral with the aid of \cite[Eq. (3.326.2)]{Ryzhik2007}.
\end{proof}

$\bullet$ \emph{\textbf{Step 3 - Deriving a High-SNR Approximation for the EMI:}}
In the last step, we investigate the asymptotic behaviour of the EMI in the high-SNR regime, i.e., $p\rightarrow\infty$. The main results are summarized as follows.

\begin{corollary}\label{Section_Performance_Analysis_EMI_Asymptotic_Expression_Standard_Rician_Theorem}
The asymptotic EMI in the high-SNR regime can be expressed in the following form:
\begin{align}\label{Section_Performance_Analysis_EMI_Asymptotic_Expression_Standard_Rician}
\lim_{p\rightarrow\infty}\bar{\mathcal{I}}_{\mathcal{X}}^{\rm{rician}}\simeq H_{{\mathbf{p}}_{\mathcal{X}}}-{\mathcal{O}}(p^{-r_{\mathbf{R}}-\frac{1}{2}}e^{-\frac{d_{\mathcal X,\min}^2}{8}\frac{p}{\sigma^2}{\tilde{a}}_0}),
\end{align}
where $d_{\mathcal X,\min}\triangleq\min_{q\neq q'}\left|{\mathsf{x}_q}-{\mathsf{x}_{q'}}\right|$. For an equiprobable square $M$-QAM constellation, the asymptotic EMI in the high-SNR regime can be expressed as follows:
\begin{align}\label{Section_Performance_Analysis_EMI_Asymptotic_Expression_Standard_Rician_MQAM}
\lim_{p\rightarrow\infty}\bar{\mathcal{I}}_{\mathcal{X}}^{\rm{rician}}\simeq \log_2{M}-{\mathcal{O}}(p^{-r_{\mathbf{R}}-\frac{1}{2}}e^{-\frac{d_{{\mathcal{X}},{\min}}^2}{8}\frac{p}{\sigma^2}{\tilde{a}}_0}{\mathfrak{A}}_{\rm{a}}),
\end{align}
where
\begin{align}\label{Section_Performance_Analysis_EMI_Asymptotic_Expression_Standard_Rician_MQAM_Gain}
{\mathfrak{A}}_{\rm{a}}=\frac{(2\sqrt{M}-1)\sqrt{\pi}d_{{\mathcal{X}},\min}/\ln(2)}
{2\sqrt{2}\sqrt{{\tilde{a}}_0}\sqrt{M}(d_{{\mathcal{X}},\min}^2/8)^{r_{\mathbf{R}}+1}}\prod_{i=1}^{r_{\mathbf{R}}}\frac{1}{\lambda_i}
{{e}^{-{\lvert\overline{h}_i\rvert^2}/{\lambda_i}}}.
\end{align}
\end{corollary}

\begin{proof}
Please refer to Appendix \ref{Section_Performance_Analysis_EMI_Asymptotic_Expression_Standard_Rician_Theorem_Proof}.
\end{proof}

\begin{remark}\label{Section_Performance_Analysis_EMI_Asymptotic_Expression_Standard_Rician_Remark}
The results in \eqref{Section_Performance_Analysis_EMI_Asymptotic_Expression_Standard_Rician} suggest that the EMI achieved by finite-alphabet inputs converges to $H_{{\mathbf{p}}_{\mathcal{X}}}$ in the limit of large $p$, and its ROC is proportional to the rate of $p^{-r_{\mathbf{R}}-\frac{1}{2}}e^{-\frac{d_{\mathcal X,\min}^2}{8}\frac{p}{\sigma^2}{\tilde{a}}_0}$ converging to $0$. This means that $\lim_{p\rightarrow\infty}\bar{\mathcal{I}}_{\mathcal{X}}^{\rm{rician}}\simeq H_{{\mathbf{p}}_{\mathcal{X}}}-{\mathcal{O}}(p^{-\infty})$. By comparing \eqref{Section_Performance_Analysis_EMI_Asymptotic_Expression_Standard_Rician} with \eqref{Section_Performance_Analysis_EMI_Asymptotic_Expression_Standard}, we find that $\bar{\mathcal{I}}_{\mathcal{X}}^{\rm{rician}}$ yields a much faster ROC than $\bar{\mathcal{I}}_{\mathcal{X}}^{\rm{rayleigh}}$. This fact means that to achieve the same EMI, in Rician fading much less power resources are required than in Rayleigh fading.
\end{remark}

\begin{table*}[!t]

\caption{Summary of the Analytical Results for Statistical Multipath Channels.}
\label{tab:Section_Performance_Analysis_Fading_Table}
\centering
\resizebox{0.85\textwidth}{!}{
\begin{tabular}{!{\vrule width1pt}c!{\vrule width1pt}c!{\vrule width1pt}c!{\vrule width1pt}c!{\vrule width1pt}c!{\vrule width1pt}c!{\vrule width1pt}}
\Xhline{1pt} 
Fading                    & Metric               & Definition        & Calculation       & High-SNR Approximation & Characteristics              \\ \Xhline{1pt} 
\multirow{6}{*}{Rayleigh} & \multirow{2}{*}{OP}  & \multirow{2}{*}{\eqref{Section_Performance_Analysis_OP_Definition_MISO}} & \multirow{2}{*}{\eqref{Section_Performance_Analysis_OP_Explicit_Expression}} & \multirow{2}{*}{\eqref{Section_Performance_Analysis_OP_Asymptotic_Expression_Standard}: $\lim_{p\rightarrow\infty}{\mathcal{P}}_{\rm{rayleigh}}\simeq({\mathcal{G}}_{\rm{a}}p)^{-{\mathcal{G}}_{\rm{d}}}$}      & Array Gain (${\mathcal{G}}_{\rm{a}}$)            \\
                          &                      &                   &                   &                        & Diversity Order (${\mathcal{G}}_{\rm{d}}$)       \\ \cline{2-6}
                          & \multirow{2}{*}{ECC} & \multirow{2}{*}{\eqref{Section_Performance_Analysis_ECC_Definition_MISO}} & \multirow{2}{*}{\eqref{Section_Performance_Analysis_ECC_Calculation_MISO}} & \multirow{2}{*}{\eqref{Section_Performance_Analysis_ECC_Asymptotic_Expression_Standard}: $\lim_{p\rightarrow\infty}{\bar{\mathcal{C}}}_{\rm{rayleigh}}\simeq {\mathcal{S}}_{\infty}\left(\log_2(p)-{\mathcal{L}}_{\infty}\right)$}      & High-SNR Slope (${\mathcal{S}}_{\infty}$)        \\
                          &                      &                   &                   &                        & High-SNR Power Offset (${\mathcal{L}}_{\infty}$) \\ \cline{2-6}
                          & \multirow{2}{*}{EMI} & \multirow{2}{*}{\eqref{Section_Performance_Analysis_EMI_Definition_MISO}} & \multirow{2}{*}{\eqref{Section_Performance_Analysis_EMI_Approximation_Explicit_MISO}} & \multirow{2}{*}{\eqref{Section_Performance_Analysis_EMI_Asymptotic_Expression_Standard}: $\lim_{p\rightarrow\infty}{\bar{\mathcal{I}}_{\mathcal{X}}^{\rm{rayleigh}}}\simeq H_{{\mathbf{p}}_{\mathcal{X}}}-({\mathcal{A}}_{\rm{a}}p)^{-{\mathcal{A}}_{\rm{d}}}$}      & Array Gain (${\mathcal{A}}_{\rm{a}}$)           \\
                          &                      &                   &                   &                        & Diversity Order (${\mathcal{A}}_{\rm{d}}$)     \\ \Xhline{1pt} 
\multirow{6}{*}{Rician}   & \multirow{2}{*}{OP}  & \multirow{2}{*}{\eqref{Section_Performance_Analysis_OP_Calculation_MISO_Rician}} & \multirow{2}{*}{\eqref{Section_Performance_Analysis_OP_Rician_Explicit_Expression}} & \multirow{2}{*}{\eqref{Section_Performance_Analysis_OP_Asymptotic_Expression_Standard_Rician}: $\lim_{p\rightarrow p_0}{\mathcal{P}}_{\rm{rician}}\simeq
\left({\mathfrak{G}}_{\rm{a}}^{-1}(p^{-1}-p_0^{-1})\right)^{{\mathfrak{G}}_{\rm{d}}}$}      & Array Gain (${\mathfrak{G}}_{\rm{a}}$)            \\
                          &                      &                   &                   &                        & Diversity Order (${\mathfrak{G}}_{\rm{d}}$)      \\ \cline{2-6}
                          & \multirow{2}{*}{ECC} & \multirow{2}{*}{\eqref{Section_Performance_Analysis_ECC_Definition_MISO_Rician}} & \multirow{2}{*}{\eqref{Section_Performance_Analysis_ECC_Calculation_MISO_Rician}} & \multirow{2}{*}{\eqref{Section_Performance_Analysis_ECC_Asymptotic_Expression_Standard_Rician}: $\lim_{p\rightarrow\infty}{\bar{\mathcal{C}}}_{\rm{rician}}\simeq {\mathfrak{S}}_{\infty}\left(\log_2(p)-{\mathfrak{L}}_{\infty}\right)$}      & High-SNR Slope (${\mathfrak{S}}_{\infty}$)        \\
                          &                      &                   &                   &                        & High-SNR Power Offset (${\mathfrak{L}}_{\infty}$)\\ \cline{2-6}
                          & \multirow{2}{*}{EMI} & \multirow{2}{*}{\eqref{Section_Performance_Analysis_EMI_Definition_MISO_Rician}} & \multirow{2}{*}{\eqref{Section_Performance_Analysis_EMI_Approximation_Explicit_MISO_Rician}} & \multirow{2}{*}{\eqref{Section_Performance_Analysis_EMI_Asymptotic_Expression_Standard_Rician}: $\lim_{p\rightarrow\infty}\bar{\mathcal{I}}_{\mathcal{X}}^{\rm{rician}}\simeq H_{{\mathbf{p}}_{\mathcal{X}}}-{\mathcal{O}}(p^{-\infty})$}      & ---                     \\
                          &                      &                   &                   &                        & ---                    \\ \Xhline{1pt} 
\end{tabular}}
\end{table*}

$\bullet$ \emph{\textbf{Numerical Results:}} To further illustrate the obtained results, in {\figurename} {\ref{Performance Analysis Figure: AMI_Discrete Explicit_Rician}}, we plot the EMI for USW LoS channels achieved by equiprobable square $M$-QAM constellations versus the transmit power, $p$. The simulation results are denoted by markers. As can be observed in {\figurename} {\ref{Performance Analysis Figure: AMI_Discrete Explicit_Rician}}, the approximated results are in excellent agreement with the simulated results. The simulation parameters used to generate {\figurename} {\ref{Performance Analysis Figure: AMI_Discrete Explicit}} are the same as those used to generate {\figurename} {\ref{Performance Analysis Figure: AMI_Discrete Explicit}}. By comparing the results in both figures, we find that to achieve the same EMI, in Rician fading much less power is required than in Rayleigh fading, which supports the discussion in \textbf{Remark} \ref{Section_Performance_Analysis_EMI_Asymptotic_Expression_Standard_Rician_Remark}. To illustrate the ROC of the EMI, we plot $H_{{\mathbf{p}}_{\mathcal{X}}}-{\bar{\mathcal{I}}_{\mathcal{X}}^{\rm{rician}}}$ versus $p$ in {\figurename} {\ref{Performance Analysis Figure: AMI_Discrete Asymptotic_Rician}}. As can be observed, in the high-SNR regime, the derived asymptotic results approach the numerical results. Besides, it can be observed that lower modulation orders yield faster ROCs\footnote{By comparing {\figurename} \ref{Performance Analysis Figure: AMI_Discrete Explicit} and {\figurename} {\ref{Performance Analysis Figure: AMI_Discrete Explicit_Rician}}, we find that the EMI curves for Rician fading are similar to those for Rayleigh fading. Hence, we omit corresponding numerical results for the EE-SE tradeoff and a PSK-QAM comparison for brevity.}.

\subsubsection{Summary of the Analytical Results}
For convenience, we summarize the analytical results for the OP, ECC, and EMI in Table \ref{tab:Section_Performance_Analysis_Fading_Table}. Despite being developed for MISO channels, the expressions given in Table \ref{tab:Section_Performance_Analysis_Fading_Table} also apply to other types of channels subject to single-stream transmission, such as single-input multiple-output (SIMO) channels, single-stream MIMO channels, and multicast channels. The only difference lies in the statistics of the channel gain. For example, the received SNR for single-stream MIMO channels can be written as $\gamma_{\rm{St-MIMO}}={\bar\gamma}a_{\rm{St-MIMO}}$ with the channel gain given by \cite{Ouyang2023_CL}
\begin{align}
a_{\rm{St-MIMO}}=\lvert{\mathbf{v}}^{\mathsf{H}}{\mathbf{H}}{\mathbf{w}}\rvert^2,
\end{align}
where ${\mathbf{w}}$ and ${\mathbf{v}}$ are the beamforming vectors utilized at transmitter and receiver, respectively. As another example, the received SNR for a $K$-user MISO multicast channel can be written as $\gamma_{\rm{Multicast}}=\bar\gamma a_{\rm{Multicast}}$ with the channel gain given by \cite{jindal2006capacity}
\begin{align}
a_{\rm{Multicast}}=\min\nolimits_{k\in\{1,\ldots,K\}}\lvert{\mathbf{h}}_k^{\mathsf{H}}{\mathbf{w}}\rvert^2,
\end{align}
where $\mathbf{w}$ is the transmit beamforming vector, and ${\mathbf{h}}_k$ is the channel vector of user $k$. After obtaining the PDF and CDF of $a_{\rm{St-MIMO}}$ or $a_{\rm{Multicast}}$, one can directly leverage the expressions in Table \ref{tab:Section_Performance_Analysis_Fading_Table} to analyze the OP, ECC, and EMI for the corresponding channels. We also note that despite requiring a more complicated analysis than MISO channels, the OP, ECC, and EMI achieved in MIMO channels still follow the standard high-SNR forms given in \eqref{Section_Performance_Analysis_OP_Asymptotic_Expression_Standard} (or \eqref{Section_Performance_Analysis_OP_Asymptotic_Expression_Standard_Rician}), \eqref{Section_Performance_Analysis_ECC_Asymptotic_Expression_Standard}, and \eqref{Section_Performance_Analysis_EMI_Asymptotic_Expression_Standard} (or \eqref{Section_Performance_Analysis_EMI_Asymptotic_Expression_Standard_Rician}), respectively.

\subsection{Discussion and Open Research Problems}
We have analyzed several fundamental performance evaluation metrics for NFC for both deterministic and statistical near-field channel models. It is hoped that our established analytical framework and derived results will provide in-depth insight into the design of NFC systems. However, there are still numerous open research problems in this area, some of which are summarized in the following.
\begin{itemize}
  \item \textbf{Information-Theoretic Limit Characterization}: Understanding the information-theoretic aspects of NFC is vital for practical implementation. NFC differs from conventional FFC with regard to the channel and signal models, and thus further efforts are required to explore the information-theoretic limits of NFC. For SPD antennas, most information-theoretic results developed for FFC also apply to NFC if the channel model is adjusted accordingly. However, this is different for CAP antenna-based NFC. In an NFC channel established by CAP antennas, determining the information-theoretic limits and design principles has to be based on continuous electromagnetic models, which gives rise to the interdisciplinary problem of integrating information theory and electromagnetic theory. Fundamental research on this topic deserves in-depth study.
  \item \textbf{System-Level Performance Analysis}: Although we have provided a comprehensive performance evaluation framework for NFC, our results are limited to the simple MISO case. More research is needed for the MIMO and multiuser scenarios. Leveraging the performance metrics adopted in this section to evaluate the performance gap between NFC and FFC in more complicated scenarios is a promising research direction that can unveil important system design insights. Furthermore, the fading performance for CAP antenna-based NFC has received limited attention due to the absence of an analytically tractable statistical model. Last but not least, stochastic geometry (SG) tools can capture the randomness of the locations of the users. Incorporating NFC’s physical properties into the SG tools may contribute to new spatial models and channel statistics, facilitating the derivation of computable expressions of further key performance metrics.
  \item \textbf{Network-Level Performance Analysis}: In practice, NFC will be deployed in multi-cell environments. As the density of wireless networks increases, inter-cell interference becomes a major obstacle to realizing the benefits of NFC. As such, analyzing NFC performance at the network level and unveiling system design insights with respect to interference management is crucial. Multi-cell settings yield more complicated wireless propagation environments. For example, the near field of one BS may overlap with another BS's far field or near field. Analyzing the NFC performance in such a complex communication scenario is challenging, and constitutes an important direction for future research.
\end{itemize}

\section{Conclusions}


This paper has presented a comprehensive tutorial on the emerging NFC technology, focusing on three fundamental aspects: near-field channel modelling, beamforming and antenna architectures, and performance analysis. 1) For near-field channel modelling, various models for SPD antennas were introduced, providing different levels of accuracy and complexity. Additionally, a Green's function method-based model was presented for CAP antennas. 2) For beamforming and antenna architectures, the unique beamfocusing property in NFC was highlighted and practical antenna structures for achieving beamfocusing in narrowband and wideband NFC were highlighted, along with practical beam training techniques. 3) For performance analysis, the received SNR and power scaling law under deterministic LoS channels for both SPD and CAP antennas were derived, and a general analytical framework was proposed for NFC performance analysis in statistical multipath channels, yielding valuable insights for practical system design. Throughout this tutorial, we have identified several open problems and research directions to inspire and guide future work in the nascent field of NFC. As NFC is still in its infancy, we hope that this tutorial will serve as a valuable tool for researchers, enabling them to explore the vast potential of the “NFC Golden Mine”.

\begin{appendices}
\section{Proof of Lemma 1} \label{proof_polarization}
The transmit-mode polarization vector at the transmit antenna essentially represents the normalized electric field. The electric field ${\mathbf E}({\mathbf{r}},{\mathbf{s}})\in{\mathbb C}^{3}$ generated in $\mathbf{r}$ from $\mathbf{s}$ is
\begin{align}
{\mathbf E}({\mathbf{r}},{\mathbf{s}})={\mathbf G}({\mathbf{r}}-{\mathbf{s}}){\mathbf J}({\mathbf{s}}),
\end{align}
where ${\mathbf J}({\mathbf{s}})=J_x({\mathbf{s}}){\hat{\mathbf u}}_{x}+J_y({\mathbf{s}}){\hat{\mathbf u}}_{y}+J_z({\mathbf{s}}){\hat{\mathbf u}}_{z}$ is the electric current vector (with ${\hat{\mathbf u}}_{x}$, ${\hat{\mathbf u}}_{y}$, and ${\hat{\mathbf u}}_{z}$ representing the unit vectors in the $x$, $y$, $z$ directions) and ${\mathbf G}({\mathbf{r}}-{\mathbf{s}})\in{\mathbb C}^{3\times3}$ is the Green function which 
is given as
\begin{align}
{\mathbf G}(\mathbf{x})=-\frac{j\omega\mu_0}{4\pi}\left[{\mathbf I}+\frac{1}{k_0}\nabla\nabla\right]\frac{e^{-jk_0\lVert{\mathbf x}\rVert}}{\lVert{\mathbf x}\rVert}.
\end{align}
where $\mu_0$ is the free space permeability, $\omega$ is the angular frequency of the signal, and $k_0$ is the wave number. Mathematically, the Green function can be further expanded as follows
\begin{equation}
\begin{split}
{\mathbf G}(\mathbf{x})&=-\frac{j\eta_0 e^{-jk_0\lVert{\mathbf x}\rVert}}{2\lambda\lVert{\mathbf x}\rVert}\left[\left({\mathbf{I}}-\hat{\mathbf x}\hat{\mathbf x}^{\mathsf H}\right)+\frac{j\lambda}{2\pi\lVert{\mathbf x}\rVert}\right.\\
&\times\left.\left({\mathbf{I}}-3\hat{\mathbf x}\hat{\mathbf x}^{\mathsf H}\right)-\frac{\lambda^2}{(2\pi\lVert{\mathbf x}\rVert)^2}\left({\mathbf{I}}-3\hat{\mathbf x}\hat{\mathbf x}^{\mathsf H}\right)\right],
\end{split}
\end{equation}
where $\hat{\mathbf x}=\frac{\mathbf{x}}{\lVert{\mathbf x}\rVert}$, $\eta_0=\sqrt{{\mu_0}/{\epsilon_0}}$, and $\epsilon_0$ is the free space permittivity. When ${\lVert{\mathbf x}\rVert}\gg\lambda$, the Green function 
is well approximated as 
\begin{equation}\label{Section_Performance_Analysis_Simplified_Green_Function}
{\mathbf G}(\mathbf{x})\simeq-\frac{j\eta_0 e^{-jk_0\lVert{\mathbf x}\rVert}}{2\lambda\lVert{\mathbf x}\rVert}\left({\mathbf{I}}-\hat{\mathbf x}\hat{\mathbf x}^{\mathsf H}\right).
\end{equation}
To guarantee ${\lVert{\mathbf x}\rVert}\gg\lambda$, the electric field should not be observed in the reactive near-field, which is a mild condition for practical systems. Henceforth, we assume ${\lVert{\mathbf x}\rVert}\gg\lambda$ and directly exploit \eqref{Section_Performance_Analysis_Simplified_Green_Function}. In this case, the normalized electric field generated in $\mathbf{r}$ from $\mathbf{s}$ is given by
\begin{equation}\label{Section_Performance_Analysis_Normalized_Electric Field}
{\bm\rho}_{\mathrm a}({\mathbf{r}},{\mathbf{s}})=\frac{{\mathbf E}({\mathbf{r}},{\mathbf{s}})}{\lVert{\mathbf E}({\mathbf{r}},{\mathbf{s}})\rVert}
=\frac{{\mathbf G}({\mathbf{r}}-{\mathbf{s}})\hat{\mathbf J}({\mathbf{s}})}{\lVert{\mathbf G}({\mathbf{r}}-{\mathbf{s}})\hat{\mathbf J}({\mathbf{s}})\rVert},
\end{equation}
where $\hat{\mathbf J}({\mathbf{s}})=\frac{{\mathbf J}({\mathbf{s}})}{\lVert\hat{\mathbf J}({\mathbf{s}})\rVert}$ is the normalized electric current vector given as follows
\begin{equation}\label{Section_Performance_Analysis_Normalized_Current_Vector}
\hat{\mathbf J}({\mathbf{s}})=\frac{[J_x({\mathbf{s}}),J_y({\mathbf{s}}),J_z({\mathbf{s}})]^{\mathsf{T}}}{\sqrt{\lvert J_x({\mathbf{s}})\rvert^2+\lvert J_y({\mathbf{s}})\rvert^2+\lvert J_z({\mathbf{s}})\rvert^2}}.
\end{equation}
Substituting \eqref{Section_Performance_Analysis_Simplified_Green_Function} into \eqref{Section_Performance_Analysis_Normalized_Electric Field} yields
\begin{equation}
\begin{split}
{\bm\rho}_{\mathrm a}({\mathbf{r}},{\mathbf{s}})=\frac{-{j e^{-jk_0\lVert{\mathbf r}-{\mathbf s}\rVert}}\left({\mathbf{I}}-\frac{({\mathbf r}-{\mathbf s})({\mathbf r}-{\mathbf s})^{\mathsf T}}{\lVert{\mathbf r}-{\mathbf s}\rVert^2}\right)\hat{\mathbf J}({\mathbf{s}})}{\left\lVert\left({\mathbf{I}}-\frac{({\mathbf r}-{\mathbf s})({\mathbf r}-{\mathbf s})^{\mathsf T}}{\lVert{\mathbf r}-{\mathbf s}\rVert^2}\right)\hat{\mathbf J}({\mathbf{s}})\right\rVert}.
\end{split}
\end{equation}
Let ${\bm\rho}_{\mathrm w}({\mathbf{r}})$ denote the normalized receive-mode polarization vector at receive point $\mathbf{r}$. Taken together, the  power gain due to the projected antenna aperture equals the polarization loss factor, which can be expressed as
\begin{align}
p_{\mathrm{polar}}({\mathbf{r}},{\mathbf{s}}) = \lvert{\bm\rho}_{\mathrm w}^{\mathsf T}({\mathbf{r}}){\bm\rho}_{\mathrm a}({\mathbf{r}},{\mathbf{s}})\rvert^2.
\end{align}
Noting that $\lvert-{j e^{-jk_0\lVert{\mathbf r}-{\mathbf s}\rVert}}\rvert=1$, we obtain
\begin{align}
p_{\mathrm{polar}}({\mathbf{r}},{\mathbf{s}}) = \frac{\left\lvert{\bm\rho}_{\mathrm w}^{\mathsf T}({\mathbf{r}})\left({\mathbf{I}}-\frac{({\mathbf r}-{\mathbf s})({\mathbf r}-{\mathbf s})^{\mathsf T}}{\lVert{\mathbf r}-{\mathbf s}\rVert^2}\right)\hat{\mathbf J}({\mathbf{s}})\right\rvert^2}{\left\lVert\left({\mathbf{I}}-\frac{({\mathbf r}-{\mathbf s})({\mathbf r}-{\mathbf s})^{\mathsf T}}{\lVert{\mathbf r}-{\mathbf s}\rVert^2}\right)\hat{\mathbf J}({\mathbf{s}})\right\rVert^2}.
\end{align}
The proof is thus completed.

\section{Proof of Lemma \ref{lemma_focus}} \label{proof_focus}
To determine the depth of focus, the value of $\frac{1}{N}|\mathbf{a}^T(\theta, r_0) \mathbf{a}^*(\theta, r)|$ should be calculated. According to the results of \cite[Appendix A]{cui2022channel}, by defining 
\begin{equation}
    \eta = \sqrt{ \frac{N^2 d^2 \sin^2 \theta}{2 \lambda} \left| \frac{1}{r_0} - \frac{1}{r} \right| },
\end{equation}
we can obtain the following approximation:
\begin{align}
    \frac{1}{N}|\mathbf{a}^{\mathsf{T}}(\theta, r_0) \mathbf{a}^*(\theta, r)| \approx \left| \frac{C(\eta) + j S(\eta)}{\eta}  \right|,
\end{align}
where $C(\eta) = \int_{0}^{\eta} \cos(\frac{\pi}{2}x^2)  \,dx$ and $S(\eta) = \int_{0}^{\eta} \sin(\frac{\pi}{2}x^2)  \,dx$. Based on numerical results, it is easy to show that $\left| \frac{C(\eta) + j S(\eta)}{\eta}  \right| \ge \frac{1}{2}$ if $\eta \le 1.6$. Therefore, by defining $\eta_{3 \mathrm{dB}} = 1.6$, we approximately have $\frac{1}{N}|\mathbf{a}^T(\theta, r_0) \mathbf{a}^*(\theta, r)| \ge \frac{1}{2}$ if 
\begin{align}
    \sqrt{ \frac{N^2 d^2 \sin^2 \theta}{2 \lambda} \left| \frac{1}{r_0} - \frac{1}{r} \right| } \le \eta_{3 \mathrm{dB}},
\end{align}
which is equivalent to 
\begin{equation}
    \max \left\{0, \frac{1}{r} - \frac{2 \lambda \eta_{3 \mathrm{dB}}^2 }{N^2 d^2 \sin^2 \theta} \right\} \le \frac{1}{r_0} \le \frac{1}{r} + \frac{2 \lambda \eta_{3 \mathrm{dB}}^2 }{N^2 d^2 \sin^2 \theta}.
\end{equation}
Define $r_{\mathrm{DF}} \triangleq \frac{N^2 d^2 \sin^2 \theta}{2 \lambda \eta^2_{3 \mathrm{dB}}}$. 
If $r < r_{\mathrm{DF}}$, the range of $r_0$ is given by 
\begin{equation}
    \frac{ r r_{\mathrm{DF}}}{r_{\mathrm{DF}} + r} \le r_0 \le  \frac{ r r_{\mathrm{DF}}}{r_{\mathrm{DF}} - r}.
\end{equation}
The depth of focus is thus given by 
\begin{equation}
    \mathrm{DF} = \frac{ r r_{\mathrm{DF}}}{r_{\mathrm{DF}}-r} - \frac{ r r_{\mathrm{DF}}}{r_{\mathrm{DF}} + r} = \frac{2r^2 r_{\mathrm{DF}}}{r_{\mathrm{DF}}^2 - r^2}.
\end{equation}
If $r \ge r_{\mathrm{DF}}$, the range of $r_0$ is given by 
\begin{equation}
    \frac{ r r_{\mathrm{DF}}}{r + r_{\mathrm{DF}}} \le r_0 \le  \infty.
\end{equation}
Thus, the depth of focus is given by 
\begin{equation}
    \mathrm{DF} = \infty.
\end{equation}
The proof is thus completed.

\section{Proof of Lemma \ref{lemma:beam spit}} \label{proof_split}

By defining $\tilde{N} = \frac{N-1}{2}$, the normalized array gain achieved by $\mathbf{f}_{\mathrm{RF}} = \mathbf{a}^*(f_c, \theta_c, r_c)$ at location $(\theta, r)$ at subcarrier $m$ is given by 
\begin{align}
    &\frac{1}{N} \left| \mathbf{a}^{\mathsf{T}}(f_m, \theta, r) \mathbf{f}_{\mathrm{RF}} \right| = \frac{1}{N} \left| \sum_{n = -\tilde{N}}^{\tilde{N}} e^{j \pi  \left( \delta_n(\theta_c, r_c) -\frac{f_m}{f_c} \delta_n(\theta, r) \right) } \right| \nonumber \\
    = & \frac{1}{N} \left| \sum_{n = -\tilde{N}}^{\tilde{N}} e^{j \pi \left( n \left( \cos \theta_c - \frac{f_m}{f_c} \cos \theta \right) - n^2 \left( \frac{d \sin^2\theta_c }{2r_c} - \frac{f_m}{f_c} \frac{d \sin^2\theta}{2r} \right) \right) }  \right|.
\end{align}  
To simplify the array gain, we define a function $G(x, y) = \frac{1}{N} \left|\sum_{n=1}^N e^{j \left( n x + n^2 y \right)} \right|$. Then, the array gain can be written as $G\left(\pi(\cos \theta_c - \frac{f_m}{f_c} \cos \theta), -\pi(\frac{d\sin^2 \theta_c}{2r_c} - \frac{f_m}{f_c} \frac{d \sin^2\theta}{2r})\right)$. It can be easily verified that the maximum value of $G(x, y)$ is obtained when $(x,y) = (0,0)$.   
The location $(\theta_m, r_m)$  that the beam focuses on for subcarrier $m$ is the location that has the maximum array gain, i.e., $\cos \theta_c - \frac{f_m}{f_c} \cos \theta_m = 0$ and $\frac{d\sin^2 \theta_c}{2r_c} - \frac{f_m}{f_c} \frac{d\sin^2 \theta_m}{2r_m} = 0$. Thus, we have 
\begin{align}
    \theta_m = \arccos \left( \frac{f_c}{f_m} \cos \theta_c \right), \quad r_m = \frac{f_m \sin^2\theta_m}{f_c \sin^2\theta_c} r_c.
\end{align}
The proof is thus completed.

\section{Proof of Theorem \ref{Section_Performance_Analysis_SNR_USW_Expression_Theorem}}\label{Section_Performance_Analysis_SNR_USW_Expression_Theorem_Proof}
The effective power gain from the $(m,n)$th transmit array element to the receiver satisfies
\begin{align}\label{Performance_Analysis_Proof_Effetive_Channel_General}
\left\lvert h_{m,n}^{i}({\mathbf{r}})\right\rvert^2=\int_{{\mathcal{S}}_{m,n}}\left\lvert {h}_{i}({\mathbf{s}},{\mathbf{r}})\right\rvert^2e_a{\rm{d}}\mathbf{s},
\end{align}
where ${\mathcal{S}}_{m,n}=\left[nd-\frac{\sqrt{A}}{2},nd+\frac{\sqrt{A}}{2}\right]\times\left[md-\frac{\sqrt{A}}{2},md+\frac{\sqrt{A}}{2}\right]$ denotes the surface region of the $(m,n)$-th array element, $e_a\cdot{\rm{d}}\mathbf{s}$ is the maximal value of the effective antenna area in the $x$-$z$ plane located around $\mathbf{s}$, and ${h}_i({\mathbf s},{\mathbf{r}})$ is the complex-valued channel from a point source located in $\mathbf{s}$ to the receive point $\mathbf{r}$. Recalling that the system operates in the radiating near-field region, i.e., $r\gg \lambda$, and the size of each individual element $\sqrt{A}$ is on the order of a wavelength, we have $r\gg \sqrt{A}$. Here, $r\gg \sqrt{A}$ means that the variation of the complex-value channel ${h}_i({\mathbf s},{\mathbf{r}})$ across different points ${\mathbf{s}}\in{\mathcal{S}}_{m,n}$ is negligible. Hence, we can rewrite \eqref{Performance_Analysis_Proof_Effetive_Channel_General} as follows:
\begin{equation}\label{Performance_Analysis_Proof_Effetive_Channel_General_Practical_Simplicafiction}
\begin{split}
\left\lvert h_{m,n}^{i}({\mathbf{r}})\right\rvert^2&=e_a\left\lvert {h}_i({\mathbf{s}}_{m,n},{\mathbf r})\right\rvert^2\int_{{\mathcal{S}}_{m,n}}{\rm{d}}\mathbf{s}\\
&=\left\lvert {h}_i({\mathbf{s}}_{m,n},{\mathbf r})\right\rvert^2Ae_a,
\end{split}
\end{equation}
where the integral $\int_{{\mathcal{S}}_{m,n}}{\rm{d}}\mathbf{s}$ returns the physical area of the $(m,n)$-th transmit array element. Thus, the received SNR for the USW model ($i={\rm{U}}$) is given by
\begin{align}
\gamma_{\rm{USW}}&=\frac{p}{\sigma^2}\lVert{\mathbf{h}}\rVert^2=\frac{p}{\sigma^2}\sum\nolimits_{n\in{\mathcal{N}}_x}\sum\nolimits_{m\in{\mathcal{N}}_z}\lvert h_{m,n}^{\rm{U}}({\mathbf{r}}) \rvert^2\nonumber\\
&=Ae_a\frac{p}{\sigma^2}\sum\nolimits_{n\in{\mathcal{N}}_x}\sum\nolimits_{m\in{\mathcal{N}}_z}\left\lvert {h}_{\rm{U}}({\mathbf{s}}_{m,n},{\mathbf r})\right\rvert^2.
\end{align}
Next, we note that the influence of the effective aperture loss and the polarization loss was not considered in \eqref{USW_Model_Channel_Coefficient}. Thus, to obtain a general expression of the received SNR, we should add the effective aperture loss and polarization loss terms to \eqref{USW_Model_Channel_Coefficient}, which yields
\begin{equation}\label{Performance_Analysis_Proof_Effetive_Power_Gain_USW_Trans1}
\begin{split}
\left\lvert {h}_{\rm{U}}({\mathbf{s}}_{m,n},{\mathbf r})\right\rvert^2=G_1({\mathbf{s}}_{m,n},{\mathbf r})G_2({\mathbf{s}}_{m,n},{\mathbf r})G_3({\mathbf{s}}_{m,n},{\mathbf r}).
\end{split}
\end{equation}
In \eqref{Performance_Analysis_Proof_Effetive_Power_Gain_USW_Trans1}, $G_1({\mathbf{s}}_{m,n},{\mathbf r})$ (given by Eq. \eqref{p_proj_general}), $G_2({\mathbf{s}}_{m,n},{\mathbf r})$ (given by Eq. \eqref{eq_proj}), and $G_3({\mathbf{s}}_{m,n},{\mathbf r})\triangleq\frac{1}{4\pi\lVert{\mathbf{s}}_{m,n}-{\mathbf r}\rVert^2}$ model the effective aperture loss, the polarization loss, and the free-space path loss, respectively. Under the USW channel model, the powers radiated by different transmit antenna elements are affected by identical polarization mismatches, projected antenna apertures, and free-space path losses \cite{Lu2022communicating}. Thus, it follows that
\begin{equation}
\begin{split}
\left\lvert {h}_{\rm{U}}({\mathbf{s}}_{m,n},{\mathbf r})\right\rvert^2=G_1({\mathbf{s}}_{0},{\mathbf r})G_2({\mathbf{s}}_{0},{\mathbf r})G_3({\mathbf{s}}_{0},{\mathbf r}),
\end{split}
\end{equation}
where ${\mathbf{s}}_0$ denotes the center location of the transmit UPA. In the considered system, we have ${\mathbf{s}}_0={\mathbf{0}}$, from which \eqref{Section_Performance_Analysis_SNR_USW_Expression} follows directly. The proof is thus completed.

\section{Proof of Theorem \ref{Section_Performance_Analysis_SNR_NUSW_Expression_Theorem}}
\label{Section_Performance_Analysis_SNR_NUSW_Expression_Theorem_Proof}
Note that the influence of the effective aperture loss and the polarization loss was not considered in \eqref{NUSW_Model_Channel_Coefficient}. To obtain a general expression for the received SNR, we should add the effective aperture loss and polarization loss terms back into \eqref{NUSW_Model_Channel_Coefficient}, which yields
\begin{equation}\label{Performance_Analysis_Proof_Effetive_Power_Gain_NUSW_Trans1}
\begin{split}
\left\lvert {h}_{\rm{N}}({\mathbf{s}}_{m,n},{\mathbf r})\right\rvert^2=G_1({\mathbf{s}}_{m,n},{\mathbf r})G_2({\mathbf{s}}_{m,n},{\mathbf r})G_3({\mathbf{s}}_{m,n},{\mathbf r}).
\end{split}
\end{equation}
Under the NUSW channel model, the powers radiated by different transmit antenna elements are affected by the same polarization mismatches and have the same projected antenna apertures but have different free-space path losses \cite{Lu2022communicating}. Therefore, we have
\begin{align}
&G_1({\mathbf{s}}_{m,n},{\mathbf r})=G_1({\mathbf{0}},{\mathbf r}),\\ 
&G_2({\mathbf{s}}_{m,n},{\mathbf r})=G_2({\mathbf{0}},{\mathbf r}).
\end{align}
Hence, the effective power gain satisfies
\begin{equation}\label{Performance_Analysis_Proof_Effetive_Power_Gain_NUSW_Trans2}
\begin{split}
\left\lvert {h}_{\rm{N}}({\mathbf{s}}_{m,n},{\mathbf r})\right\rvert^2=G_1({\mathbf{0}},{\mathbf r})G_2({\mathbf{0}},{\mathbf r})G_3({\mathbf{s}}_{m,n},{\mathbf r}).
\end{split}
\end{equation}
Inserting \eqref{Performance_Analysis_Proof_Effetive_Power_Gain_NUSW_Trans2} into \eqref{Section_Performance_Analysis_SNR_Most_General_Expression} yields the results in \eqref{Section_Performance_Analysis_SNR_NUSW_Expression}. The proof is thus completed.

\section{Proof of Corollary \ref{Section_Performance_Analysis_SNR_NUSW_Asymptotic_Theorem}}\label{Section_Performance_Analysis_SNR_NUSW_Asymptotic_Theorem_Proof}
Substituting ${\mathbf{s}}_{m,n}=[nd,0,md]^{\mathsf{T}}$ and ${\mathbf{r}}=[r\Phi,r\Psi,r\Omega]^{\mathsf{T}}$ into \eqref{Section_Performance_Analysis_SNR_NUSW_Expression} yields
\begin{equation}
\begin{split}
\gamma_{\rm{NUSW}}&=\frac{p\beta_0^2}{4\pi\sigma^2r^2}\sum_{n\in{\mathcal{N}}_x}\sum_{m\in{\mathcal{N}}_z}\\
&\times\frac{1}{\Psi^2+(n\epsilon-\Phi)^2+(m\epsilon-\Omega)^2},
\end{split}
\end{equation}
where $\epsilon=\frac{d}{r}$. Since the array element separation $d$ is typically on the order of a wavelength, in practice, we have $r\gg d$ and thus $\epsilon\ll1$. Furthermore, we define the function
\begin{align}
f_1(x,z)\triangleq\frac{1}{\Psi^2+(x-\Phi)^2+(z-\Omega)^2}
\end{align}
in the rectangular area ${\mathcal{H}}=\left\{(x,z)\left|-\frac{N_x\epsilon}{2}\leq x\leq\frac{N_x\epsilon}{2},-\frac{N_z\epsilon}{2}\leq z\leq\frac{N_z\epsilon}{2}\right.\right\}$. This rectangular area is further partitioned into $N_xN_z$ sub-rectangles, each with equal area $\epsilon^2$. Since $\epsilon\ll1$, we have $f_1(x,z)\approx f_1(n\epsilon,m\epsilon)$ for $\forall (x,z)\in\{(x,z)|(n-\frac{1}{2})\epsilon\leq x\leq(n+\frac{1}{2})\epsilon,(m-\frac{1}{2})\epsilon\leq z\leq(m+\frac{1}{2})\epsilon\}$. Based on the concept of double integrals, we obtain
\begin{align}
\sum_{n\in{\mathcal{N}}_x}\sum_{m\in{\mathcal{N}}_z}f_1(n\epsilon,m\epsilon)\epsilon^2\approx
\int\int_{\mathcal{H}}f_1(x,z){\rm{d}}x{\rm{d}}z,
\end{align}
which yields
\begin{align}\label{Section_Performance_Analysis_SNR_NUSW_Expression_Approximation}
\gamma_{\rm{NUSW}}\approx\int_{-\frac{N_z}{2}\epsilon}^{\frac{N_z}{2}\epsilon}\int_{-\frac{N_x}{2}\epsilon}^{\frac{N_x}{2}\epsilon}
\frac{\frac{p\beta_0^2}{\sigma^2}\frac{1}{4\pi r^2\epsilon^2}{\rm d}x{\rm d}z}{\Psi^2+(x-\Phi)^2+(z-\Omega)^2}.
\end{align}
Since $\Omega\in[0,1]$ and $\Phi\in[0,1]$, \eqref{Section_Performance_Analysis_SNR_NUSW_Expression_Approximation} reduces to the following
form as $N_x,N_z\rightarrow\infty$:
\begin{equation}\label{Section_Performance_Analysis_NUSW_SNR_Explicit_Trans}
\gamma_{\rm{NUSW}}\approx\int_{-\frac{N_z}{2}\epsilon}^{\frac{N_z}{2}\epsilon}\int_{-\frac{N_x}{2}\epsilon}^{\frac{N_x}{2}\epsilon}
\frac{\frac{pg^2}{\sigma^2}\frac{1}{4\pi r^2\epsilon^2}{\rm d}x{\rm d}z}{\Psi^2+x^2+z^2}.
\end{equation}
As shown in {\figurename} {\ref{Geometry_Region}}, the integration region in \eqref{Section_Performance_Analysis_NUSW_SNR_Explicit_Trans} is bounded by two disks with radius $R_2=\frac{\epsilon}{2}\min\left\{N_x,{N_z}\right\}$ and $R_1=\frac{\epsilon}{2}\sqrt{N_x^2+N_z^2}$, respectively. By defining function
\begin{equation}
\begin{split}
\hat{f}(R)\triangleq\int_{0}^{2\pi}\int_{0}^{R}{\frac{\rho{\rm{d}}\rho{\rm d}\theta}{\Psi^2+\rho^2}}=\pi\ln\left(1+\frac{R^2}{\Psi^2}\right),
\end{split}
\end{equation}
we have
\begin{equation}
\begin{split}
\hat{f}(R_2)<\int_{-\frac{N_z}{2}\epsilon}^{\frac{N_z}{2}\epsilon}\int_{-\frac{N_x}{2}\epsilon}^{\frac{N_x}{2}\epsilon}
{\frac{{\rm d}x{\rm d}z}{\Psi^2+x^2+z^2}}<\hat{f}(R_1).
\end{split}
\end{equation}
As $N_x,N_z\rightarrow\infty$, we have $R_1,R_2\rightarrow\infty$. By the Squeeze Theorem \cite{Rohde2011}, the asymptotic SNR satisfies $\lim_{N\rightarrow\infty}\gamma_{\rm{NUSW}}\simeq\mathcal{O}(\log{N})$. The proof is thus completed.

\begin{figure}[!t]
\setlength{\abovecaptionskip}{0pt}
\centering
\includegraphics[height=0.3\textwidth]{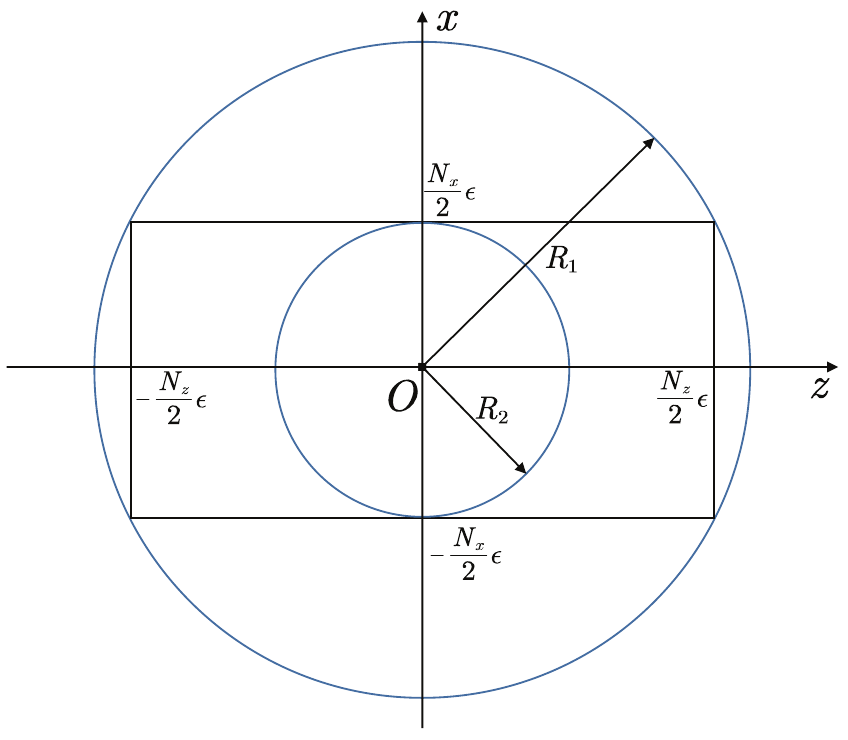}
\caption{The inscribed and circumscribed disks of the rectangular region $N_x\epsilon \times N_z\epsilon$.}
\label{Geometry_Region}
\end{figure}

\section{Proof of Theorem \ref{Section_Performance_Analysis_SNR_Proposed_Vert_General_Expression_Theorem}}\label{Section_Performance_Analysis_SNR_Proposed_Vert_General_Expression_Theorem_Proof}
Based on \eqref{Performance_Analysis_Proof_Effetive_Channel_General_Practical_Simplicafiction}, the effective power gain from the $(m,n)-$th transmit array element to the receiver satisfies
\begin{equation}
\begin{split}
\left\lvert h_{m,n}^{\rm{G}}({\mathbf{r}})\right\rvert^2=\left\lvert {h}_{\rm{G}}({\mathbf{s}}_{m,n},{\mathbf r})\right\rvert^2Ae_a,
\end{split}
\end{equation}
where
\begin{equation}
\begin{split}
\left\lvert {h}_{\rm{G}}({\mathbf{s}}_{m,n},{\mathbf r})\right\rvert^2=G_1({\mathbf{s}}_{m,n},{\mathbf r})G_2({\mathbf{s}}_{m,n},{\mathbf r})G_3({\mathbf{s}}_{m,n},{\mathbf r}).
\end{split}
\end{equation}
In the general channel model, the powers radiated by different transmit antenna elements are affected by different free-space path losses, projected antenna apertures, and polarization mismatches. Consequently, for $(m,n)\ne (m',n')$, we have
\begin{align}
&G_1({\mathbf{s}}_{m,n},{\mathbf r})\ne G_1({\mathbf{s}}_{m',n'},{\mathbf r}),\\
&G_2({\mathbf{s}}_{m,n},{\mathbf r})\ne G_2({\mathbf{s}}_{m',n'},{\mathbf r}),\\
&G_3({\mathbf{s}}_{m,n},{\mathbf r})\ne G_3({\mathbf{s}}_{m',n'},{\mathbf r}).
\end{align}
Hence, the effective power gain satisfies
\begin{equation}\label{Performance_Analysis_Proof_Effetive_Power_Gain_Proposed_Model}
\begin{split}
\left\lvert {h}_{\rm{G}}({\mathbf{s}}_{m,n},{\mathbf r})\right\rvert^2&=G_1({\mathbf{s}}_{m,n},{\mathbf r})G_2({\mathbf{s}}_{m,n},{\mathbf r})G_3({\mathbf{s}}_{m,n},{\mathbf r})\\
&\ne\left\lvert {h}_{\rm{G}}({\mathbf{s}}_{m',n'},{\mathbf r})\right\rvert^2.
\end{split}
\end{equation}
Inserting \eqref{Performance_Analysis_Proof_Effetive_Power_Gain_Proposed_Model} into \eqref{Section_Performance_Analysis_SNR_Most_General_Expression} yields \eqref{Section_Performance_Analysis_SNR_Proposed_Vert_General_Expression}. The proof is thus completed.
\section{Proof of Corollary \ref{Section_Performance_Analysis_SNR_Proposed_Vert_General_Approximation_Expression_Y_Polar_Corollary}}\label{Section_Performance_Analysis_SNR_Proposed_Vert_General_Approximation_Expression_Y_Polar_Corollary_Proof}
Based on \eqref{p_proj_general}, \eqref{eq_proj}, when ${\bm\rho}_{w}(\mathbf{r})={\bm\rho}=\hat{\mathbf J}({\mathbf{s}}_{m,n})=[1,0,0]^{\mathsf{T}}$, $\forall m,n$, we have
\begin{align}
G_1({\mathbf{s}}_{m,n},{\mathbf{r}})&=\frac{r\Psi}{\lVert {\mathbf r}-{\mathbf{s}}_{m,n}\rVert},\\
G_2({\mathbf{s}}_{m,n},{\mathbf{r}})&=\frac{r^2\Psi^2+(r\Omega-md)^2}{\lVert {\mathbf r}-{\mathbf{s}}_{m,n}\rVert^2},\\
G_3({\mathbf{s}}_{m,n},{\mathbf{r}})&=\frac{1}{4\pi\lVert {\mathbf r}-{\mathbf{s}}_{m,n}\rVert^2},
\end{align}
Therefore, the received SNR can be written as
\begin{equation}
\begin{split}
\gamma_{\rm{General}}&=
\frac{p}{\sigma^2}\sum_{n\in{\mathcal{N}}_x}\sum_{m\in{\mathcal{N}}_z}Ae_a\\
&{\frac{r^3\Psi^3+r\Psi(r\Omega-md)^2}{4\pi
((r\Psi)^2+(nd-r\Phi)^2+(md-r\Omega)^2)^{5/2}}}.
\end{split}
\end{equation}
Following the same approach as for obtaining \eqref{Section_Performance_Analysis_SNR_NUSW_Expression_Approximation}, we approximate $\gamma_{\rm{General}}$ as follows:
\begin{equation}\label{Section_Performance_Analysis_Approximation_Integral_Double_Simplified_SNR}
\begin{split}
\gamma_{\rm{General}}&\approx\frac{pAe_{{a}}}{4\pi r^2\epsilon^2\sigma^2}\int_{-\frac{N_z}{2}\epsilon}^{\frac{N_z}{2}\epsilon}\int_{-\frac{N_x}{2}\epsilon}^{\frac{N_x}{2}\epsilon}\\
&\times{\frac{\Psi^3+\Psi(\Omega-z)^2}
{
(\Psi^2+(x-\Phi)^2+(z-\Omega)^2)^{5/2}}}{\rm d}x{\rm d}z.
\end{split}
\end{equation}
The integral in \eqref{Section_Performance_Analysis_Approximation_Integral_Double_Simplified_SNR} can be solved in closed form with the aid of the following identities:
\begin{align}
&\int\frac{{\rm{d}}x}{(x^2+a)^{3/2}}=\frac{x}{a\sqrt{x^2+a}}+C\label{Section_Performance_Analysis_Integral_Helper_1}\\
&\int\frac{{\rm{d}}x}{(x^2+a)^{5/2}}=\frac{x}{3a({x^2+a})^{3/2}}+\frac{2x}{3a^2\sqrt{x^2+a}}+C\label{Section_Performance_Analysis_Integral_Helper_2}\\
&\int\frac{{\rm{d}}x}{(x^2+a)\sqrt{x^2+a+b}}\nonumber\\
&\quad=\frac{1}{\sqrt{ab}}\arctan\left(\frac{\sqrt{b}x}{\sqrt{a}\sqrt{x^2+a+b}}\right)+C\label{Section_Performance_Analysis_Integral_Helper_3}
\end{align}
where $a$, $b$ are arbitrary scalars and $C$ is an arbitrary constant. Upon calculating the double integral, we obtain \eqref{Section_Performance_Analysis_SNR_Proposed_Vert_General_Approximation_Expression_Y_Polar}. The proof is thus completed.
\section{Proof of Theorem \ref{Section_Performance_Analysis_SNR_USW_Expression_CAP_Theorem}}\label{Section_Performance_Analysis_SNR_UPW_Expression_CAP_Theorem_Proof}
The effective power gain from the transmit CAP surface to the user satisfies
\begin{equation}
\begin{split}
\lvert h_{\rm{cap}}(r,\theta,\phi)\rvert^2&=\int_{{\mathcal{S}}}\left\lvert h({\mathbf r},{\mathbf{s}})\right\rvert^2\cdot e_a{\rm{d}}\mathbf{s}\\
&=\int_{-\frac{L_z}{2}}^{\frac{L_z}{2}}\!\int_{-\frac{L_x}{2}}^{\frac{L_x}{2}}
\lvert h([x,0,z]^{\mathsf{T}},{\mathbf{r}})\rvert^2 e_a{\rm{d}}x{\rm{d}}z,
\end{split}
\end{equation}
where ${\mathcal{S}}=\left[-\frac{L_x}{2},\frac{L_x}{2}\right]\times\left[-\frac{L_z}{2},\frac{L_z}{2}\right]$ denote the surface region of the whole transmit CAP surface. For USW, we have
\begin{equation}
\begin{split}
\left\lvert h({\mathbf r},{\mathbf{s}}_{m_x,m_z})\right\rvert^2=G_1({\mathbf{s}}_0,{\mathbf{r}})G_2({\mathbf{s}}_0,{\mathbf{r}})G_3({\mathbf{s}}_0,{\mathbf{r}}),
\end{split}
\end{equation}
where ${\mathbf{s}}_0$ denotes the location of the center of the transmit CAP surface. In the considered system, we have ${\mathbf{s}}_0={\mathbf{0}}$, from which \eqref{Section_Performance_Analysis_SNR_USW_Expression_CAP} follows directly. The proof is thus completed.
\section{Proof of Lemma \ref{Section_Performance_Analysis_OP_SNR_CDF_PDF_Lemma}}\label{Section_Performance_Analysis_OP_SNR_CDF_PDF_Lemma_Proof}
The channel gain $\lVert{\mathbf{h}}\rVert^2$ satisfies
\begin{align}
\lVert{\mathbf{h}}\rVert^2={\tilde{\mathbf{h}}}^{\mathsf{H}}\mathbf{R}\tilde{\mathbf{h}}.
\end{align}
To facilitate the derivation, we perform eigenvalue decomposition of $\mathbf{R}$ and obtain ${\mathbf{R}}={\mathbf{U}}^{\mathsf{H}}{\bm\Lambda}{\mathbf{U}}$. Matrix $\mathbf{U}$ is a unitary matrix with $\mathbf{U}{\mathbf{U}}^{\mathsf{H}}=\mathbf{I}$, and ${\bm\Lambda}={\mathsf{diag}}\{\lambda_1,\ldots,\lambda_{r_{\mathbf{R}}},0,\ldots,0\}$, where $\{\lambda_i>0\}_{i=1}^{r_{\mathbf{R}}}$ are the positive eigenvalues of $\mathbf{R}$. Consequently, we obtain
\begin{equation}
\lVert{\mathbf{h}}\rVert^2={\tilde{\mathbf{h}}}^{\mathsf{H}}{\mathbf{U}}^{\mathsf{H}}{\bm\Lambda}{\mathbf{U}}\tilde{\mathbf{h}}.
\end{equation}
Since $\mathbf{U}{\mathbf{U}}^{\mathsf{H}}=\mathbf{I}$, we have
\begin{align}
&{\mathbb{E}}\{{\mathbf{U}}\tilde{\mathbf{h}}\}={\mathbf{U}}{\mathbb{E}}\{\tilde{\mathbf{h}}\}=\mathbf{0},\\
&{\mathbb{E}}\{({\mathbf{U}}\tilde{\mathbf{h}})({\mathbf{U}}\tilde{\mathbf{h}})^{\mathsf{H}}\}=\mathbf{U}{\mathbb{E}}\{\tilde{\mathbf{h}}{\tilde{\mathbf{h}}}^{\mathsf{H}}\}{\mathbf{U}}^{\mathsf{H}}
=\mathbf{U}{\mathbf{U}}^{\mathsf{H}}=\mathbf{I},
\end{align}
which yields ${\mathbf{U}}\tilde{\mathbf{h}}=[{\tilde{h}}_1,\ldots,\tilde{h}_{N}]^{\mathsf{T}}\sim{\mathcal{CN}}(\mathbf{0},{\mathbf{I}})$. Then, we have
\begin{align}\label{Correlated_Rayleigh_Channel_Gain_Proof}
\lVert{\mathbf{h}}\rVert^2=\sum_{i=1}^{r_{\mathbf{R}}}\lambda_i|\tilde{h}_i|^2.
\end{align}
Since $\{{\tilde{h}}_i\}_{i=1}^{N}$ contains ${N}$ i.i.d. complex Gaussian distributed variables, $\{|x_i|^2\}_{i=1}^{M}$ contains $N$ i.i.d. exponentially distributed variables each with CDF $1-{e}^{-x}$, $x\geq0$. Then, by virtue of \cite{Moschopoulos1985}, the PDF of $\sum_{i=1}^{r_{\mathbf{R}}}\lambda_i|x_i|^2$ is given by \eqref{Section_Performance_Analysis_OP_SNR_PDF}. The CDF of $\lVert{\mathbf{h}}\rVert^2$ can be further calculated by
\begin{equation}
\begin{split}
F_{\lVert{\mathbf{h}}\rVert^2}\left(x\right)
=\int_{0}^{x}F_{\lVert{\mathbf{h}}\rVert^2}\left(y\right){\rm{d}}y.
\end{split}
\end{equation}
The proof is thus completed.
\section{Proof of Corollary \ref{Section_Performance_Analysis_OP_Asymptotic_Expression_Standard_Theorem}}\label{Section_Performance_Analysis_OP_Asymptotic_Expression_Standard_Theorem_Proof}
As $p\rightarrow\infty$, we have
\begin{equation}
\begin{split}
\lim_{p\rightarrow\infty}\frac{2^{\mathcal{R}}-1}{p/\sigma^2\lambda_{\min}}=0.
\end{split}
\end{equation}
Using the asymptotic property of the incomplete Gamma function \cite[Eq. (8.354.1)]{Ryzhik2007}:
\begin{equation}\label{Section_Performance_Analysis_Gamma_Function_Asymptotic_Expression}
\lim_{t\rightarrow0}\Upsilon\left(s,t\right)\simeq\frac{t^s}{s},
\end{equation}
we obtain
\begin{equation}\label{Section_Performance_Analysis_CDF_Asymptotic_Expression}
\begin{split}
&\lim_{p\rightarrow\infty}\Upsilon\left(k+r_{\mathbf{R}},\frac{2^{\mathcal{R}}-1}{p/\sigma^2\lambda_{\min}}\right)\\
&\simeq\frac{1}{{(p/\sigma^2)}^{k+r_{\mathbf{R}}}}\left(\frac{2^{\mathcal{R}}-1}{\lambda_{\min}}\right)^{k+r_{\mathbf{R}}}\frac{1}{k+r_{\mathbf{R}}}.
\end{split}
\end{equation}
Inserting \eqref{Section_Performance_Analysis_CDF_Asymptotic_Expression} into \eqref{Section_Performance_Analysis_OP_Explicit_Expression} yields
\begin{equation}\label{Section_Performance_Analysis_OP_Asymptotic_Proof_Basic_Expression}
\begin{split}
\mathcal{P}_{\rm{rayleigh}}&\simeq
\frac{\lambda_{\min}^{r_{\mathbf{R}}}}{\prod_{i=1}^{r_{\mathbf{R}}}{\lambda}_i}\sum_{k=0}^{\infty}\frac{\psi_k \frac{1}{{(p/\sigma^2)}^{k+r_{\mathbf{R}}}}\left(\frac{2^{\mathcal{R}}-1}{\lambda_{\min}}\right)^{k+r_{\mathbf{R}}}}{\Gamma\left(r_{\mathbf{R}}+k+1\right)}\\
&=\frac{(2^{\mathcal{R}}-1)^{r_{\mathbf{R}}}}{r_{\mathbf{R}}!\prod_{i=1}^{r_{\mathbf{R}}}{\lambda}_i}
\frac{1}{{(p/\sigma^2)}^{r_{\mathbf{R}}}}+o\left(p^{-r_{\mathbf{R}}}\right),
\end{split}
\end{equation}
where $o(\cdot)$ denotes higher order terms. By omitting the higher order terms, we obtain \eqref{Section_Performance_Analysis_OP_Asymptotic_Expression_Standard}. The proof is thus completed.
\section{Proof of Corollary \ref{Section_Performance_Analysis_ECC_Asymptotic_Expression_Standard_Theorem}}\label{Section_Performance_Analysis_ECC_Asymptotic_Expression_Standard_Theorem_Proof}
Using the fact that $\lim_{x\rightarrow\infty}\frac{\log_2(1+ax)}{\log_2(ax)}=1$ ($a>0$), we obtain
\begin{align}
\lim_{\bar\gamma\rightarrow\infty}\frac{{\mathbb{E}}\{\log_2(1+p/\sigma^2 \lVert{\mathbf{h}}\rVert^2)\}}{{\mathbb{E}}\{\log_2(p/\sigma^2 \lVert{\mathbf{h}}\rVert^2)\}}=1,
\end{align}
which yields
\begin{equation}
\begin{split}
&\lim_{\bar\gamma\rightarrow\infty}{{\mathbb{E}}\{\log_2(1+p/\sigma^2 \lVert{\mathbf{h}}\rVert^2)\}}\\
&\simeq\log_2(p)+{{\mathbb{E}}\{\log_2( \lVert{\mathbf{h}}\rVert^2/\sigma^2)\}}.
\end{split}
\end{equation}
By leveraging \cite[Eq. (4.352.1)]{Ryzhik2007} to calculate the expectation $\mathbb{E}\{\log_2({\lVert{\mathbf{h}}\rVert^2/\sigma^2})\}$, we obtain \eqref{Section_Performance_Analysis_ECC_Calculation_Asymptotic_MISO}. The proof is thus completed.

\section{Proof of Corollary \ref{Section_Performance_Analysis_EMI_Asymptotic_Expression_Standard_Theorem}}\label{Section_Performance_Analysis_EMI_Asymptotic_Expression_Standard_Theorem_Proof}
To facilitate the discussion, we first rewrite the EMI in \eqref{Section_Performance_Analysis_EMI_Definition_MISO} as follows:
\begin{equation}
\begin{split}
\bar{\mathcal{I}}_{\mathcal{X}}^{\rm{rayleigh}}&=\int_{0}^{\infty}I_{\mathcal X}(t){\rm{d}}F_{\lVert{\mathbf{h}}\rVert^2}\left(\frac{t}{p/\sigma^2 }\right)\\
&=\left.I_{\mathcal X}( t)F_{\lVert{\mathbf{h}}\rVert^2}\left(\frac{t\sigma^2}{p }\right)\right|_{0}^{\infty}\\
&-\int_{0}^{\infty}F_{\lVert{\mathbf{h}}\rVert^2}\left(\frac{t\sigma^2}{p }\right){\rm{d}}I_{\mathcal X}(t).
\end{split}
\end{equation}
Using the properties $\lim_{t\rightarrow\infty}I_{\mathcal X}(t)=H_{{\mathbf{p}}_{\mathcal{X}}}$ and $\lim_{t\rightarrow0}I_{\mathcal X}(t)=0$ \cite{Guo2005}, we obtain
\begin{align}
&\lim_{t\rightarrow\infty}I_{\mathcal X}(t)F_{\lVert{\mathbf{h}}\rVert^2}(t\sigma^2/p)=H_{{\mathbf{p}}_{\mathcal{X}}}\cdot1=H_{{\mathbf{p}}_{\mathcal{X}}},\\
&\lim_{t\rightarrow0}I_{\mathcal X}(t)F_{\lVert{\mathbf{h}}\rVert^2}(t\sigma^2/p)=0\cdot0=0.
\end{align}
Hence, we have $\left.I_{\mathcal X}( t)F_{\lVert{\mathbf{h}}\rVert^2}\left(\frac{t\sigma^2}{p }\right)\right|_{0}^{\infty}=H_{{\mathbf{p}}_{\mathcal{X}}}$, which, together with the results in \cite{Guo2005}, yields
\begin{equation}\label{Section_Performance_Analysis_EMI_Trans_for_Asym}
\begin{split}
\bar{\mathcal{I}}_{\mathcal{X}}^{\rm{rayleigh}}&=H_{{\mathbf{p}}_{\mathcal{X}}}-\int_{0}^{\infty}F_{\lVert{\mathbf{h}}\rVert^2}
\left(\frac{t}{p/\sigma^2 }\right)\frac{{\rm{MMSE}}_{\mathcal{X}}(t)}{\ln{2}}{\rm{d}}t,
\end{split}
\end{equation}
where ${\rm{MMSE}}_{\mathcal X}\left(t\right)$ denotes the minimum mean square error (MMSE) in estimating $X$ in \eqref{Section_Performance_Analysis_AWGN_Channel} from $Y$ when $\gamma=t$. Based on \eqref{Section_Performance_Analysis_OP_Calculation_MISO} and \eqref{Section_Performance_Analysis_OP_Asymptotic_Proof_Basic_Expression}, we have
\begin{align}\label{Section_Performance_Analysis_EMI_Asymptotic_CDF}
\lim_{p\rightarrow\infty}F_{\lVert{\mathbf{h}}\rVert^2}\left(\frac{t}{p/\sigma^2 }\right)\simeq
\frac{t^{r_{\mathbf{R}}}}{r_{\mathbf{R}}!\prod_{i=1}^{r_{\mathbf{R}}}{\lambda}_i}
\frac{1}{{(p/\sigma^2)}^{r_{\mathbf{R}}}}.
\end{align}
Inserting \eqref{Section_Performance_Analysis_EMI_Asymptotic_CDF} into \eqref{Section_Performance_Analysis_EMI_Trans_for_Asym} gives
\begin{align}\label{Section_Performance_Analysis_EMI_Standard_for_Asym}
\lim_{p\rightarrow\infty}\bar{\mathcal{I}}_{\mathcal{X}}^{\rm{rayleigh}}\simeq
H_{{\mathbf{p}}_{\mathcal{X}}}-\frac{1}{\ln{2}}\frac{{\mathcal M}\left[{{\rm{MMSE}}_{\mathcal{X}}(t)};r_{\mathbf{R}}+1\right]}{r_{\mathbf{R}}!\prod_{i=1}^{r_{\mathbf{R}}}{\lambda}_i{(p/\sigma^2)}^{r_{\mathbf{R}}} },
\end{align}
where ${\mathcal M}\left[\varrho\left(t\right);z\right]\triangleq\int_{0}^{\infty}t^{z-1}\varrho\left(t\right){\rm d}t$ denotes the Mellin transform of $\varrho\left(t\right)$ \cite{flajolet1995mellin}. We next introduce the following two lemmas to facilitate the subsequent discussion.

\begin{lemma}\label{Lemma_EMI_Asy_Aux1}
  Given the finite constellation ${\mathcal X}=\left\{\mathsf{x}_q\right\}_{q=1}^{Q}$, the MMSE function satisfies $\lim_{t\rightarrow\infty}{\rm{MMSE}}_{\mathcal X}\left(t\right)={\mathcal O}(t^{-\frac{1}{2}}e^{-\frac{t}{8}d_{{\mathcal X},{\min}}^2})$, where $d_{\mathcal X,\min}\triangleq\min_{q\neq q'}\left|{\mathsf{x}_q}-{\mathsf{x}_{q'}}\right|$ \cite{Wu2011}.
\end{lemma}

\begin{lemma}\label{Lemma_EMI_Asy_Aux2}
  If $\varrho\left(t\right)$ is ${\mathcal O}\left(t^a\right)$ as $t\rightarrow0^{+}$ and ${\mathcal O}\left(t^b\right)$ as $t\rightarrow+\infty$, then $\left|{\mathcal M}\left[\varrho\left(t\right);z\right]\right|<\infty$ when $-a<z<-b$ \cite{flajolet1995mellin}.
\end{lemma}

Particularly, $\lim_{t\rightarrow0^{+}}{\rm{MMSE}}_{\mathcal X}\left(t\right)=1$ \cite{Guo2005}, which together with Lemma \ref{Lemma_EMI_Asy_Aux1}, suggests that ${\rm{MMSE}}_{\mathcal X}\left(t\right)$ is ${\mathcal O}\left(t^0\right)$ as $t\rightarrow0^{+}$ and ${\mathcal O}\left(t^{-\infty}\right)$ as $t\rightarrow\infty$. Using this fact and Lemma \ref{Lemma_EMI_Asy_Aux2}, we find that $\left|{\mathcal M}\left[{{\rm{MMSE}}_{\mathcal{X}}(t)};z\right]\right|<\infty$ holds for $0<z<\infty$, which in combination with the fact that ${\rm{MMSE}}_{\mathcal X}\left(t\right)>0$ ($t>0$), suggests that ${\mathcal M}\left[{{\rm{MMSE}}_{\mathcal{X}}(t)};z\right]\in\left(0,\infty\right)$ holds for $0<z<\infty$. Since $r_{\mathbf{R}}>0$, we conclude that ${\mathcal M}\left[{{\rm{MMSE}}_{\mathcal{X}}(t)};r_{\mathbf{R}}+1\right]$ is some positive constant. The proof is thus completed.
\section{Proof of Lemma \ref{Section_Performance_Analysis_Rician_OP_SNR_Trans_Lemma}}\label{Section_Performance_Analysis_Rician_OP_SNR_Trans_Lemma_Proof}
Using the fact that ${\mathbf{U}}{\mathbf{U}}^{\mathsf{H}}={\mathbf{I}}$, we obtain
\begin{align}\label{Section_Performance_Analysis_Rician_OP_SNR_Trans_Lemma_Proof_Trans1}
\lVert{\mathbf{h}}\rVert^2=\lVert{\mathbf{Uh}}\rVert^2=\lVert\mathbf{U}\overline{\mathbf{h}}+\mathbf{U}\mathbf{R}^{\frac{1}{2}}\tilde{\mathbf{h}}\rVert^2.
\end{align}
Plugging ${\mathbf{R}}^{\frac{1}{2}}={\mathbf{U}}^{\mathsf{H}}{\bm\Lambda}^{\frac{1}{2}}{\mathbf{U}}$ into \eqref{Section_Performance_Analysis_Rician_OP_SNR_Trans_Lemma_Proof_Trans1} yields
\begin{align}\label{Section_Performance_Analysis_Rician_OP_SNR_Trans_Lemma_Proof_Trans2}
\lVert{\mathbf{h}}\rVert^2=\lVert\mathbf{U}\overline{\mathbf{h}}+{\mathsf{diag}}\{\lambda_1^{\frac{1}{2}},\ldots,\lambda_{r_{\mathbf{R}}}^{\frac{1}{2}}
,0,\ldots,0\}{\mathbf{U}}\tilde{\mathbf{h}}\rVert^2.
\end{align}
Since ${\mathbf{U}}{\mathbf{U}}^{\mathsf{H}}={\mathbf{I}}$ and $\tilde{\mathbf{h}}\sim{\mathcal{CN}}(\mathbf{0},\mathbf{I})$, we have ${\mathbf{U}}\tilde{\mathbf{h}}\sim{\mathcal{CN}}(\mathbf{0},\mathbf{I})$. Let $\{\tilde{{h}}_i\}_{i=1}^{{r_{\mathbf{R}}}}$ denote the first ${r_{\mathbf{R}}}$ elements of vector ${\mathbf{U}}\tilde{\mathbf{h}}$. Then, we rewrite \eqref{Section_Performance_Analysis_Rician_OP_SNR_Trans_Lemma_Proof_Trans2} as follows
\begin{equation}\label{Section_Performance_Analysis_Rician_OP_SNR_Trans_Lemma_Proof_Trans3}
\begin{split}
\lVert{\mathbf{h}}\rVert^2&=\lVert[\overline{{h}}_1+\lambda_1^{\frac{1}{2}}\tilde{{h}}_1,\ldots,\overline{{h}}_{r_{\mathbf{R}}}+
\lambda_{r_{\mathbf{R}}}^{\frac{1}{2}}\tilde{{h}}_{r_{\mathbf{R}}}\\
&,\overline{{h}}_{1+r_{\mathbf{R}}},\ldots,\overline{{h}}_N\}]^{\mathsf{T}}\rVert^2\\
&=\sum\nolimits_{i=1}^{r_{\mathbf{R}}}\lambda_i\lVert\overline{{h}}_i\lambda_i^{-\frac{1}{2}}+\tilde{{h}}_i\rVert^2+
{\sum\nolimits_{i=r_{\mathbf{R}}+1}^{N}\lVert\overline{{h}}_i\rVert^2}.
\end{split}
\end{equation}
The proof is thus completed.
\section{Proof of Corollary \ref{Section_Performance_Analysis_OP_Asymptotic_Expression_Standard_Rician_Theorem}}
\label{Section_Performance_Analysis_OP_Asymptotic_Expression_Standard_Rician_Theorem_Proof}
According to \cite[Eq. (2.16)]{simon2001digital}, the PDF of $\lvert{\mathcal{CN}}(\overline{{h}}_i,\lambda_i)\rvert^2$ is given by
\begin{equation}
f_{\lvert{\mathcal{CN}}(\overline{{h}}_i,\lambda_i)\rvert^2}(x)=\frac{1}{\lambda_i}{e}^{-\frac{\lvert\overline{h}_i\rvert^2}{\lambda_i}}e^{-\frac{x}{\lambda_i}}I_0\left(2\sqrt{x\frac{\lvert\overline{h}_i\rvert^2}{\lambda_i^2}}\right),
\end{equation}
where $I_n(\cdot)$ is the $n$-th-order modified Bessel function of the first kind. By \cite[Eq. (10.25.2)]{olver2010nist}, we have
\begin{equation}\label{Section_Performance_Analysis_OP_Asymptotic_Expression_Standard_Rician_Theorem_Proof_Trans1}
I_0(z)=\sum_{k=0}^{\infty}\frac{1}{\Gamma(k+1)k!}\left(\frac{z}{2}\right)^{2k}.
\end{equation}
The Laplace transform of $f_{\lvert{\mathcal{CN}}(\overline{{h}}_i,\lambda_i)\rvert^2}(x)$ is derived as
\begin{equation}\label{Section_Performance_Analysis_OP_Asymptotic_Expression_Standard_Rician_Theorem_Proof_Trans2}
\begin{split}
{\mathcal{L}}_{f_{\lvert{\mathcal{CN}}(\overline{{h}}_i,\lambda_i)\rvert^2}}(s)&=\int_{0}^{\infty}\frac{1}{\lambda_i}
{e}^{-\frac{\lvert\overline{h}_i\rvert^2}{\lambda_i}}e^{-\frac{x}{\lambda_i}}\\
&\times I_0\left(\frac{2\lvert\overline{h}_i\rvert}{\lambda_i}\sqrt{x}\right)e^{-sx}{\rm d}x.
\end{split}
\end{equation}
By substituting \eqref{Section_Performance_Analysis_OP_Asymptotic_Expression_Standard_Rician_Theorem_Proof_Trans1} into \eqref{Section_Performance_Analysis_OP_Asymptotic_Expression_Standard_Rician_Theorem_Proof_Trans2} and calculating the resulting integral with \cite[Eq. (3.326.2)]{Ryzhik2007}, we have
\begin{equation}\label{Section_Performance_Analysis_OP_Asymptotic_Expression_Standard_Rician_Theorem_Proof_Trans3}
\begin{split}
{\mathcal{L}}_{f_{\lvert{\mathcal{CN}}(\overline{{h}}_i,\lambda_i)\rvert^2}}(s)=\sum_{k=0}^{\infty}\frac{1}{k!}\frac{\lvert\overline{h}_i\rvert^{2k}}{\lambda_i^{2k+1}}
\frac{{e}^{-\frac{\lvert\overline{h}_i\rvert^2}{\lambda_i}}}{\left(s+\frac{1}{\lambda_i}\right)^{k+1}}.
\end{split}
\end{equation}
When $s\rightarrow\infty$, we obtain
\begin{equation}\label{Section_Performance_Analysis_OP_Asymptotic_Expression_Standard_Rician_Theorem_Proof_Trans4}
\begin{split}
\lim\nolimits_{s\rightarrow\infty}{\mathcal{L}}_{f_{\lvert{\mathcal{CN}}(\overline{{h}}_i,\lambda_i)\rvert^2}}(s)\simeq\frac{1}{\lambda_i}
\frac{1}{s}{e}^{-\frac{\lvert\overline{h}_i\rvert^2}{\lambda_i}}.
\end{split}
\end{equation}
Since the $\{{\mathcal{CN}}(\overline{{h}}_i,\lambda_i)\}_{i=1}^{r_{\mathbf{R}}}$ are mutually independent, the Laplace transform of $\tilde{a}=\sum_{i=1}^{r_{\mathbf{R}}}\lvert{\mathcal{CN}}(\overline{{h}}_i,\lambda_i)\rvert^2$ satisfies
\begin{equation}\label{Section_Performance_Analysis_OP_Asymptotic_Expression_Standard_Rician_Theorem_Proof_Trans5}
\begin{split}
\lim\nolimits_{s\rightarrow\infty}{\mathcal{L}}_{f_{\tilde{a}}}(s)\simeq s^{-r_{\mathbf{R}}}\prod_{i=1}^{r_{\mathbf{R}}}\frac{1}{\lambda_i}
{{e}^{-{\lvert\overline{h}_i\rvert^2}/{\lambda_i}}}.
\end{split}
\end{equation}
Thus, by performing the inverse Laplace transform of \eqref{Section_Performance_Analysis_OP_Asymptotic_Expression_Standard_Rician_Theorem_Proof_Trans5} and referring to \cite{oppenheim1997signals}, the PDF of $\tilde{a}$, i.e., $f_{\tilde{a}}(x)$ for $x\rightarrow0^+$ can be obtained as
\begin{equation}\label{Section_Performance_Analysis_OP_Asymptotic_Expression_Standard_Rician_Theorem_Proof_Trans5}
\begin{split}
\lim\nolimits_{x\rightarrow0^+}f_{\tilde{a}}(x)\simeq \frac{1}{\Gamma(r_{\mathbf{R}})}x^{r_{\mathbf{R}}-1}\prod_{i=1}^{r_{\mathbf{R}}}\frac{1}{\lambda_i}
{{e}^{-{\lvert\overline{h}_i\rvert^2}/{\lambda_i}}}.
\end{split}
\end{equation}
It follows that the CDF of $\tilde{a}$ for $x\rightarrow0^+$ satisfies
\begin{equation}\label{Section_Performance_Analysis_OP_Asymptotic_Expression_Standard_Rician_Theorem_Proof_Trans6}
\begin{split}
\lim\nolimits_{x\rightarrow0^+}F_{\tilde{a}}(x)\simeq \frac{1}{r_{\mathbf{R}}!}x^{r_{\mathbf{R}}}\prod_{i=1}^{r_{\mathbf{R}}}\frac{1}{\lambda_i}
{{e}^{-{\lvert\overline{h}_i\rvert^2}/{\lambda_i}}}.
\end{split}
\end{equation}
When $p\rightarrow{\frac{(2^{\mathcal{R}}-1)\sigma^2}{\tilde{a}_0}}$, we have $\frac{2^{\mathcal{R}}-1}{p/\sigma^2}-\tilde{a}_0\rightarrow0^+$ and obtain \eqref{Section_Performance_Analysis_OP_Asymptotic_Expression_Standard_Rician}. The proof is thus completed.
\section{Proof of Corollary \ref{Section_Performance_Analysis_EMI_Asymptotic_Expression_Standard_Rician_Theorem}}
\label{Section_Performance_Analysis_EMI_Asymptotic_Expression_Standard_Rician_Theorem_Proof}
Following the same approach as for obtaining \eqref{Section_Performance_Analysis_EMI_Trans_for_Asym}, we rewrite \eqref{Section_Performance_Analysis_EMI_Definition_MISO_Rician} as follows
\begin{equation}\label{Section_Performance_Analysis_EMI_Trans_for_Asym_Rician_Trans1}
\begin{split}
\bar{\mathcal{I}}_{\mathcal{X}}&=H_{{\mathbf{p}}_{\mathcal{X}}}-\int_{0}^{\infty}F_{\tilde{a}}\left(\frac{t}{p/\sigma^2 }\right)\frac{{\rm{MMSE}}_{\mathcal{X}}(t+p/\sigma^2\tilde{a}_0)}{\ln{2}}{\rm{d}}t.
\end{split}
\end{equation}
When $p\rightarrow\infty$, we have $\frac{t}{p/\sigma^2 }\rightarrow0$ and $t+p/\sigma^2\tilde{a}_0\rightarrow\infty$. Based on \eqref{Section_Performance_Analysis_OP_Asymptotic_Expression_Standard_Rician_Theorem_Proof_Trans6} and Lemma \ref{Lemma_EMI_Asy_Aux1}, we obtain
\begin{align}
\lim_{p\rightarrow\infty}F_{\tilde{a}}\left(\frac{\sigma^2t}{p}\right)\simeq\frac{1}{r_{\mathbf{R}}!}\left(\frac{\sigma^2t}{p}\right)^{r_{\mathbf{R}}}\prod_{i=1}^{r_{\mathbf{R}}}\frac{1}{\lambda_i}
{{e}^{-{\lvert\overline{h}_i\rvert^2}/{\lambda_i}}}
\label{Section_Performance_Analysis_EMI_Trans_for_Asym_Rician_Trans2},
\end{align}
\begin{align}
\lim_{p\rightarrow\infty}{\rm{MMSE}}_{\mathcal{X}}\left(t+\frac{p}{\sigma^2}\tilde{a}_0\right)\simeq
{\mathcal O}\left(
\frac{e^{\frac{t+\frac{p}{\sigma^2}\tilde{a}_0}{-8d_{{\mathcal X},{\min}}^{-2}}}}
{\sqrt{\frac{p}{\sigma^2}\tilde{a}_0}}
\right)
\label{Section_Performance_Analysis_EMI_Trans_for_Asym_Rician_Trans3}.
\end{align}
By substituting \eqref{Section_Performance_Analysis_EMI_Trans_for_Asym_Rician_Trans2} and \eqref{Section_Performance_Analysis_EMI_Trans_for_Asym_Rician_Trans3} into \eqref{Section_Performance_Analysis_EMI_Trans_for_Asym_Rician_Trans1} and solving the resulting integral with the help of \cite[Eq. (3.326.2)]{Ryzhik2007}, we obtain \eqref{Section_Performance_Analysis_EMI_Asymptotic_Expression_Standard_Rician}. For an equiprobable square $M$-QAM constellation, we have $H_{{\mathbf{p}}_{\mathcal{X}}}=\log_2{M}$. Besides, based on \cite{Alvarado2014}, we have
\begin{align}
\lim_{x\rightarrow\infty}{\rm{MMSE}}_{\mathcal{X}}(x)\simeq \frac{\sqrt{\pi}d_{\mathcal{X},\min}}{2\sqrt{2}}\frac{2\sqrt{M}-1}{\sqrt{M}}\frac{1}{\sqrt{x}}e^{-\frac{d_{\mathcal{X},\min}^2x}{8}},
\end{align}
which yields
\begin{equation}\label{Section_Performance_Analysis_EMI_Trans_for_Asym_Rician_Trans4}
\begin{split}
\lim_{p\rightarrow\infty}{\rm{MMSE}}_{\mathcal{X}}\left(t+\frac{p}{\sigma^2}\tilde{a}_0\right)&\simeq
\frac{e^{\frac{t+\frac{p}{\sigma^2}\tilde{a}_0}{-8d_{{\mathcal X},{\min}}^{-2}}}}
{\sqrt{\frac{p}{\sigma^2}\tilde{a}_0}}\times\\
&\frac{\sqrt{\pi}d_{\mathcal{X},\min}}{2\sqrt{2}}\frac{2\sqrt{M}-1}{\sqrt{M}}\frac{1}{\sqrt{t}}
\end{split}.
\end{equation}
By substituting \eqref{Section_Performance_Analysis_EMI_Trans_for_Asym_Rician_Trans2} and \eqref{Section_Performance_Analysis_EMI_Trans_for_Asym_Rician_Trans4} into \eqref{Section_Performance_Analysis_EMI_Trans_for_Asym_Rician_Trans1} and solving the resulting integral with the aid of \cite[Eq. (3.326.2)]{Ryzhik2007}, we obtain \eqref{Section_Performance_Analysis_EMI_Asymptotic_Expression_Standard_Rician_MQAM}. The proof is thus completed.

\end{appendices}

\balance
\bibliographystyle{IEEEtran}
\bibliography{mybib}

\end{document}